\documentclass[10pt,aps,pra,twocolumn,nofootinbib,superscriptaddress]{revtex4-2}

\usepackage{amsmath, amssymb}
\usepackage{xcolor}
\usepackage{titletoc}
\usepackage{url}
\usepackage{dirtytalk}
\usepackage{graphicx} 
\usepackage{subcaption}
\usepackage{wrapfig} 
\usepackage[T1]{fontenc}
\usepackage{tcolorbox}
\usepackage{mathtools}
\usepackage{svg}
\usepackage{bm}
\usepackage{multirow}
\usepackage{algorithm}
\usepackage{algorithmic}
\usepackage{xfrac}
\usepackage{qcircuit}
\usepackage{braket}
\usepackage{hyperref}
\hypersetup{breaklinks=true}
\usepackage{cleveref}
\allowdisplaybreaks

\newcommand{\subsubsubsection}[1]{\paragraph{#1}\mbox{}\\}
\hypersetup{
    colorlinks=true,
    linkcolor=blue,
    citecolor=red,
    urlcolor=magenta
}

\DeclareMathOperator{\Tr}{Tr}
\newtheorem{theorem}{Theorem}

\newtheorem{corollary}[theorem]{Corollary}

\newtheorem{definition}[theorem]{Definition}

\newtheorem{proposition}[theorem]{Proposition}
\newtheorem{remark}[theorem]{Remark}

\newenvironment{proof}[1][Proof]{\noindent\textbf{#1.} }{\ \rule{0.5em}{0.5em}}

\begin{document}

\title{Evolved Quantum Boltzmann Machines}

\author{Michele Minervini}
\affiliation{School of Electrical and Computer Engineering, Cornell University, Ithaca, New York 14850, USA}
\author{Dhrumil Patel}  
\affiliation{Department of Computer Science, Cornell University, Ithaca, New York 14850, USA}
\author{Mark M. Wilde}
\affiliation{School of Electrical and Computer Engineering, Cornell University, Ithaca, New York 14850, USA}

\date{\today}

\begin{abstract}
We introduce evolved quantum Boltzmann machines as a variational ansatz for quantum optimization and learning tasks. Given two parameterized Hamiltonians $G(\theta)$ and $H(\phi)$, an evolved quantum Boltzmann machine consists of preparing a thermal state of the first Hamiltonian $G(\theta)$ followed by unitary evolution according to the second Hamiltonian $H(\phi)$. Alternatively, one can think of it as first realizing imaginary time evolution according to $G(\theta)$ followed by real time evolution according to $H(\phi)$. After defining this ansatz, we provide analytical expressions for the gradient vector and illustrate their application in ground-state energy estimation and generative modeling, showing how the gradient for these tasks can be estimated by means of quantum algorithms that involve classical sampling, Hamiltonian simulation, and the Hadamard test. We also establish analytical expressions for the Fisher--Bures, Wigner--Yanase, and Kubo--Mori information matrix elements of evolved quantum Boltzmann machines, as well as quantum algorithms for estimating each of them, which leads to at least three different general natural gradient descent algorithms based on this ansatz, and we derive fundamental limitations on the estimation of time-evolved thermal states. Along the way, we establish a broad generalization of the main result of [Luo, \textit{Proc.~Am.~Math.~Soc.}~132, 885 (2004)], proving that the Fisher--Bures and Wigner--Yanase information matrices of general parameterized families of states differ by no more than a factor of two in the matrix (Loewner) order, making them essentially interchangeable for training when using natural gradient descent. 
\end{abstract}

\maketitle

\tableofcontents

\allowdisplaybreaks

\section{Introduction}

\subsection{Motivation}

Quantum computers offer a promising solution for various computational challenges~\cite{Montanaro2016overview_q_algo}. In particular, there is growing interest in how quantum computation can provide a speedup over classical algorithms in optimization and learning tasks or if it can be useful for purely quantum tasks in these areas~\cite{Biamonte2017qml}. The potential for quantum utility in these domains relies on efficiently representing quantum states, which requires the design of both expressive and tractable ansatzes. These representations not only determine the computational efficiency of an algorithm but also affect its ability to explore the solution space effectively.

In this context, several variational quantum ansatzes have been proposed~\cite{Peruzzo2014vqe,Amin2018qbm, Verdon2019qhbm,Ferguson2021mbqc}, with parameterized quantum circuits (PQCs) having emerged as a prominent heuristic~\cite{Peruzzo2014vqe,McClean2016VQA,Mitarai2018q_circuit_learning}. PQCs have been considered for solving practical problems such as ground-state energy estimation~\cite{Peruzzo2014vqe, Jones2019vqe,Cerezo2022vqe}, approximate combinatorial optimization~\cite{farhi2014qaoa,Wang2018qaoa,Hadfield2019qaoa}, and even machine-learning tasks like clustering~\cite{bermejo2024variational_clustering}, classification~\cite{Rebentrost2014qsvm,Havlicek2019supervised_learning,Schuld2020quantum_classifier}, and generative modeling~\cite{Benedetti2019pqc, Leadbeater2021gen_modell, Abbas2021qnn} (see \cite{Cerezo2021vqa} for a review). However, PQCs face significant hurdles, including the barren plateau problem~\cite{McClean2018barren_plateaus, Cerezo2021barren_plateaus, Holmes2022barren_plateaus, marrero2021barrenp_lateaus}, where gradient magnitudes decay exponentially with system size, making training infeasible for larger systems. These challenges underscore the importance of exploring alternative ansatzes that maintain expressivity while mitigating such optimization difficulties.

Quantum Boltzmann machines have emerged as an alternative ansatz, being an expressive and trainable model that incorporates ideas from both variational algorithms and statistical physics~\cite{Amin2018qbm,Benedetti2017qbm,kieferova2017qbm}. They generalize classical Boltzmann machines -- a classic machine learning technique that underpins many deep learning models~\cite{hinton1983optimal,hinton2006fast,salakhutdinov2009deep} -- by replacing the classical energy function with a parameterized quantum Hamiltonian. This quantum extension allows for the inclusion of non-commuting interaction terms, thereby enhancing the model’s expressivity, and the resulting quantum states take the form of thermal states. Interest in quantum Boltzmann machines has been further bolstered by recent breakthroughs in thermal-state preparation~\cite{chen2023q_Gibbs_sampl,chen2023thermalstatepreparation,rajakumar2024gibbssampling,bergamaschi2024gibbs_sampling,chen2024sim_Lindblad,rouze2024efficientthermalization,bakshi2024hightemperaturegibbsstates,ding2024preparationlowtemperaturegibbs}. Other recent works have shown how these models can be used in quantum machine learning tasks, such as generative modeling~\cite{Coopmans2024qbm_gen_learn,Tüysüz2024qbm_gen_model} and ground-state energy estimation~\cite{patel2024quantumboltzmannmachine}, as well as for solving semi-definite programs~\cite{liu2025quantumthermodynamicssemidefiniteoptimization,minervini2025constrainedfreeenergyminimization}. The geometry of parameterized thermal states has also been investigated, alongside quantum algorithms for estimating their information matrix elements~\cite{patel2024naturalgradientparameterestimation}, opening up applications like geometry-aware gradient descent algorithms.

\subsection{Main results}

In this paper, our first fundamental contribution is to establish a new ansatz for parameterizing quantum states, which generalizes quantum Boltzmann machines. For doing so, we consider general parameterized Hamiltonians of the form
\begin{align}
G(\theta) & \coloneqq\sum_{j=1}^J\theta_{j} G_{j},\label{eq:G}\\
H(\phi) & \coloneqq\sum_{k=1}^K \phi_{k}H_{k},\label{eq:H}
\end{align}
where $\theta_j \in \mathbb{R}$ for all $j \in \{1, \ldots, J\}$, $\phi_k \in \mathbb{R}$ for all $k \in \{1, \ldots, K\}$, and $G_j$ and $H_k$ are Hermitian operators. For Hamiltonians of physical interest, $G_j$ and $H_k $ act non-trivially on only a constant number of qubits.
We define the \textit{evolved quantum  Boltzmann machine} ansatz to be as follows:
\begin{align}
\omega(\theta,\phi) &  \coloneqq e^{-iH(\phi)}\rho(\theta)e^{iH(\phi)},\label{eq:ansatz}\\
\rho(\theta) &  \coloneqq\frac{e^{-G(\theta)}}{Z(\theta)},
\label{eq:param-thermal-state}
\end{align}
where $Z(\theta)\coloneqq\operatorname{Tr}[e^{-G(\theta)}]$ is the partition function.  Since the parameterized thermal state $\rho(\theta)$ defines a quantum Boltzmann machine~\cite{Amin2018qbm,Benedetti2017qbm,kieferova2017qbm} and $e^{-iH(\phi)}$ represents a unitary evolution, the state $\omega(\theta,\phi) $ is indeed an evolved quantum Boltzmann machine. In order for the Hamiltonian $H(\phi)$ to play a non-trivial role, it is necessary that $[H(\phi),G(\theta)] \neq 0$.

The main idea behind introducing the parameter $\phi$ via $H(\phi)$ is to enrich the representational capacity of quantum Boltzmann machines, enabling evolved quantum Boltzmann machines to explore a broader class of quantum states. Indeed, the unitary evolution allows for exploration of directions in the quantum state manifold that are inaccessible to thermal states of $G(\theta)$ alone, potentially enabling the proposed model to capture structures or correlations that would otherwise be difficult to represent. Moreover, the addition of real-time evolution introduces greater flexibility in the ansatz design: since preparing thermal states of arbitrary Hamiltonians can be computationally demanding \cite{devulapalli2026complexitythermalizationfinitequantum}, one can choose a $G(\theta)$ for which thermal state preparation is more tractable, and then enhance expressivity through a comparatively lightweight real-time evolution generated by $H(\phi)$. This trade-off allows for constructing expressive ansatz while mitigating the cost of thermal-state preparation.

Another contribution of our paper consists of analytic formulas for the gradient of~\eqref{eq:ansatz} with respect to the parameter vectors~$\theta$ and~$\phi$, as well as quantum algorithms for estimating the elements of the gradient vector. Similar to previously reported algorithms from~\cite{patel2024quantumboltzmannmachine,patel2024naturalgradientparameterestimation}, these algorithms involve a combination of classical random sampling, Hamiltonian simulation~\cite{lloyd1996universal,childs2018toward}, and the Hadamard test~\cite{Cleve1998}. These results support using evolved quantum Boltzmann machines for optimization and learning tasks, such as ground-state energy estimation, constrained Hamiltonian optimization~\cite{chen2023qslackslackvariableapproachvariational},  and generative modeling tasks. Table~\ref{table:grad-results} summarizes these findings.

We then explore various quantum generalizations of Fisher information (i.e., Fisher–Bures, Wigner–Yanase, and Kubo–Mori), establishing elegant expressions for the corresponding information matrix elements for evolved quantum Boltzmann machines. Of practical relevance, we prove that these matrix elements can be estimated using standard quantum subroutines -- specifically, a combination of the Hadamard test~\cite{Cleve1998}, classical random sampling, and Hamiltonian simulation~\cite{lloyd1996universal,childs2018toward}. Table~\ref{table:FB-WY-KM-results} summarizes our analytical findings for the Fisher--Bures, Wigner--Yanase, and Kubo--Mori information matrix elements of evolved quantum Boltzmann machines. 

These developments pave the way for further applications of evolved quantum Boltzmann machines. In this paper, we introduce two of them. Our first application is to develop a natural gradient descent algorithm applied specifically to evolved quantum Boltzmann machine learning. Notably, the gradient update step in our approach can be performed on a quantum computer using standard quantum subroutines. Our second application is to the problem of estimating a Hamiltonian when given access to time-evolved thermal-state samples of the form in~\eqref{eq:ansatz}.

Before moving on, let us note that fixing $\phi$ and varying $\theta$ leads to quantum Boltzmann machines as a special case, for which it is already known how to evaluate analytic gradients~\cite{Coopmans2024qbm_gen_learn,patel2024quantumboltzmannmachine} and information matrix elements~\cite{patel2024naturalgradientparameterestimation}. Alternatively, fixing the parameter vector $\theta$ and allowing the parameter vector $\phi$ to vary in the Hamiltonian evolution $e^{-iH(\phi)}$ leads to a special case of an evolved quantum Boltzmann machine, which we refer to as a \textit{quantum evolution machine}. It can be seen as an alternative parameterized ansatz itself, differing from the more common layered parameterized circuits. By Trotterization, a quantum evolution machine is related to the earlier proposed Hamiltonian variational ansatz~\cite{WH2015,Hadfield2019qaoa}, but there is a strong distinction in how one evaluates analytic gradients and information matrix elements for quantum evolution machines, as seen later on in~\eqref{eq:grad_wrt_phi} of Theorem~\ref{thm:gradient-eQBM} and Theorems~\ref{thm:FB-phi},~\ref{thm:WY-phi},~\ref{thm:KM-phi}, respectively, when compared to how it is done for the Hamiltonian variational ansatz. Quantum evolution machines are related to the general ansatz considered in~\cite[Section~2]{Banchi2021measuringanalytic} and  to \cite[Eq.~(2)]{Wiersema2024herecomessun}, but again there are distinctions in how we evaluate analytic gradients for them, and the initial state of a quantum evolution machine is a thermal state. Additionally, here we derive analytical expressions for quantum generalizations of Fisher information for this ansatz, and we establish quantum algorithms for estimating them. 

\setlength{\tabcolsep}{5pt} % Default value: 6pt
\renewcommand{\arraystretch}{1.8}
\begin{table*}
    \begin{tabular}
[c]{|l|l|c|c|}\hline\hline
Application & Gradient formula & Equation & Quantum Circuit
\\\hline\hline
\multirow{2}{*}{\shortstack{Ground-state energy\\estimation}} & $\begin{array}{l}\frac{\partial}{\partial\theta_{j}}\operatorname{Tr}[O\omega(\theta,\phi)]  =
         -\frac{1}{2}\left\langle\left\{e^{iH(\phi)} O e^{-iH(\phi)} , \Phi_{\theta}(G_{j})\right\} \right\rangle_{\rho(\theta)} \\ 
         \qquad \qquad \qquad \qquad \qquad + \left\langle O\right\rangle_{\omega(\theta,\phi)}\left\langle G_{j}\right\rangle_{\rho(\theta)}  \end{array}$
     & Equation~\eqref{eq:VQE-grad-theta} & Figure~\ref{fig:VQE-grad-theta} \\ \cline{2-4}
   & $ \frac{\partial}{\partial\phi_{k}}\operatorname{Tr}[O\omega(\theta,\phi)]  =
    i\left\langle\big[\Psi^{\dagger}_{\phi}(H_{k}),O\big] \right\rangle_{\omega(\theta,\phi)}$  & 
    Equation~\eqref{eq:VQE-grad-phi} & Figure~\ref{fig:VQE-grad-phi} \\ \hline
\multirow{2}{*}{Generative modeling} & $\frac{\partial}{\partial \theta_j} D( \eta \| \omega(\theta,\phi) ) = \left\langle G_j \right\rangle_{\eta(\phi)} - \left\langle G_j \right\rangle_{\rho(\theta)}$ & 
    Equation~\eqref{eq:gen_model_der_theta} & \\ \cline{2-4}
   & $\frac{\partial}{\partial \phi_k} D( \eta \| \omega(\theta,\phi) ) = i \left\langle  \left[  G(\theta)  ,\Psi_\phi (H_k) \right] \right\rangle_{\eta(\phi)} $ & Equation~\eqref{eq:gen_mod_grad_phi} & Figure~\ref{fig:gen-mod-grad-phi} \\\hline\hline
\end{tabular}
\caption{Summary of our analytical results for various gradients of evolved quantum Boltzmann machines, which includes  quantum Boltzmann machines (rows~1 and 3) and quantum evolution machines (rows~2 and~4) as special~cases.}
\label{table:grad-results}
\end{table*}

\setlength{\tabcolsep}{5pt} % Default value: 6pt
\renewcommand{\arraystretch}{1.8}
\begin{table*}
\centering
\begin{tabular}
[c]{|l|l|c|c|}\hline\hline
Quantity & Formula & Theorem & Quantum Circuit\\\hline\hline 
Fisher--Bures $\theta$ & $I_{ij}^{\operatorname{FB}}(\theta)=\frac{1}
{2}\left\langle \{\Phi_{\theta}(G_{i}),\Phi_{\theta}(G_{j})\}\right\rangle
_{\rho(\theta)}-\left\langle G_{i}\right\rangle _{\rho(\theta)}\left\langle
G_{j}\right\rangle _{\rho(\theta)}$ & Theorem~\ref{thm:FB-theta} & Figure~\ref{fig:FB-theta} \\\hline
Fisher--Bures $\phi$ & $I_{ij}^{\operatorname{FB}}(\phi)=\left\langle \left[  \Phi_{\theta}(\Psi_{\phi}(H_{i})), \left[  G(\theta) , \Psi_{\phi}
(H_{j})\right]  \right]
\right\rangle _{\rho(\theta)}$ & Theorem~\ref{thm:FB-phi} & Figure~\ref{fig:FB-phi} \\\hline
Fisher--Bures $\theta,\phi$ & $I_{ij}^{\operatorname{FB}}(\theta
,\phi)=i\left\langle \left[  \Phi_{\theta}(G_{i}),\Psi_{\phi}
(H_{j})\right]  \right\rangle _{\rho(\theta)}$ & Theorem~\ref{thm:FB-theta-phi} & Figure~\ref{fig:FB-theta-phi} \\\hline\hline
\multirow{2}{*}{Wigner--Yanase $\theta$} & $I_{ij}^{\operatorname{WY}}(\theta)=\frac{1}
{2}\operatorname{Tr}\!\left[  \Phi_{\frac{\theta}{2}}(G_{i})\sqrt{\rho
(\theta)}\Phi_{\frac{\theta}{2}}(G_{j})\sqrt{\rho(\theta)}\right]  $ & \multirow{2}{*}{Theorem~\ref{thm:WY-theta}} & Figure~\ref{fig:WY-theta-2} \\
& $\qquad\qquad+\frac{1}{4}\left\langle \left\{  \Phi_{\frac{\theta}{2}}
(G_{i}),\Phi_{\frac{\theta}{2}}(G_{j})\right\}  \right\rangle _{\rho(\theta
)}-\left\langle G_{i}\right\rangle _{\rho(\theta)}\left\langle G_{j}
\right\rangle _{\rho(\theta)}$ &  & Figure~\ref{fig:WY-theta}\\\hline
\multirow{2}{*}{Wigner--Yanase $\phi$} & $I_{ij}^{\operatorname{WY}}(\phi)= - 8 \Tr\!\left[ \Psi_{\phi}(H_{j}) \sqrt{\rho(\theta)} \Psi_{\phi}(H_{i}) \sqrt{\rho(\theta)} \right]$ & \multirow{2}{*}{Theorem~\ref{thm:WY-phi}} & Figure~\ref{fig:WY-phi-1}\\
& $\qquad\qquad\quad +\, 4\, \left\langle 
 \left\{ \Psi_{\phi}(H_{i}) , \Psi_{\phi}(H_{j}) \right\}   \right\rangle_{\rho(\theta)}$ &  & Figure~\ref{fig:WY-phi-2} \\\hline
Wigner--Yanase $\theta,\phi$ & $I_{ij}^{\operatorname{WY}}(\theta,\phi
)=i\left\langle \left[  \Phi_{\frac{\theta}{2}}(G_{j}),\Psi_{\phi
}(H_{i})\right]  \right\rangle _{\rho(\theta)}$ & Theorem~\ref{thm:WY-theta-phi} & Figure~\ref{fig:WY-theta-phi} \\\hline\hline
Kubo--Mori $\theta$ & $I_{ij}^{\operatorname{KM}}(\theta)=\frac{1}
{2}\left\langle \left\{  G_{i},\Phi_{\theta}(G_{j})\right\}  \right\rangle
_{\rho(\theta)}-\left\langle G_{i}\right\rangle _{\rho(\theta)}\left\langle
G_{j}\right\rangle _{\rho(\theta)}$ & Theorem~\ref{thm:KM-theta} & Figure~\ref{fig:KM-theta} \\\hline
Kubo--Mori $\phi$ & $I_{ij}^{\operatorname{KM}}(\phi)=\left\langle \left[  \Psi_{\phi}(H_{i}), \left[  G(\theta) , \Psi_{\phi}
(H_{j})\right]  \right]
\right\rangle _{\rho(\theta)}$ & Theorem~\ref{thm:KM-phi} & Figure~\ref{fig:KM-phi} \\\hline
Kubo--Mori $\theta,\phi$ & $I_{ij}^{\operatorname{KM}}(\theta,\phi)= \frac
{i}{2}\left\langle \left\{  \Phi_{\theta}(G_{i}),\left[ G(\theta), \Psi_{\phi}(H_{j})\right]  \right\}  \right\rangle _{\rho(\theta)}$ & Theorem~\ref{thm:KM-theta-phi} & Figure~\ref{fig:KM-theta-phi}
\\\hline\hline
\end{tabular}
\caption{Summary of our analytical results for the matrix elements of the Fisher--Bures, Wigner--Yanase, and Kubo--Mori information matrices for evolved quantum Boltzmann machines, which includes as special cases quantum Boltzmann machines (rows~1, 4, and 7) and quantum evolution machines (rows~2, 5, and 8).}
\label{table:FB-WY-KM-results}
\end{table*}

\subsection{Paper organization}

The rest of our paper is organized as follows. In Section~\ref{sec:gradient}, we present our analytical expressions for the gradient of evolved quantum Boltzmann machines, along with two applications for optimization and learning tasks. See Table~\ref{table:grad-results} for a summary of these findings and Figure~\ref{fig:VQE-circuits} for quantum circuits that estimate elements of the gradient for ground-state energy estimation and generative modeling. Section~\ref{sec:q_Fisher_info_theory} provides background on quantum generalizations of the Fisher information and their connections to smooth divergences. In particular, we introduce the Fisher--Bures, Wigner--Yanase, and Kubo--Mori information matrices, each of which is associated with Uhlmann fidelity~\cite{Uhl76}, Holevo fidelity~\cite{Kholevo1972}, and quantum relative entropy~\cite{Umegaki62}, respectively. Therein, we also  establish a broad generalization of \cite[Theorem~2]{Luo2004}, proving that the Fisher--Bures and Wigner--Yanase information matrices of general parameterized families of states differ by no more than a factor of two in the matrix (Loewner) order (Corollary~\ref{cor:WY-FB-ineqs}). In Section~\ref{sec:general-consids-info-mats}, we present general considerations regarding the geometry of evolved quantum Boltzmann machines. After that, we provide our formulas for the elements of the Fisher–Bures (Section~\ref{sec:FB-info}), Wigner-Yanase (Section~\ref{sec:WY-info}), and Kubo–Mori (Section~\ref{sec:KM-info}) information matrices for evolved quantum Boltzmann machines, only sketching the idea behind their proofs in the main text while including detailed proofs in the appendices. See Table~\ref{table:FB-WY-KM-results} for a summary of these findings. The quantum circuits involved in the quantum algorithms for estimating each matrix element are outlined in the respective sections corresponding to each information matrix (see Figures~\ref{fig:FB-circuits},~\ref{fig:WY-circuits}, and~\ref{fig:KM-circuits} also). In Section~\ref{sec:applications}, we finally delve into more detail about our two applications of the results related to the information matrices of evolved quantum Boltzmann machines: a natural gradient method for evolved quantum Boltzmann machine learning (Section~\ref{sec:app-nat-grad}) and estimating the parameters of a Hamiltonian from copies of a time-evolved thermal state (Section~\ref{sec:app-estimating}). Section~\ref{sec:special-cases} briefly discusses how quantum Boltzmann machines and quantum evolution machines are special cases of evolved quantum Boltzmann machines. 
We conclude in Section~\ref{sec:conclusion} with a summary and some directions for future work.

\section{Gradient of evolved quantum Boltzmann machines}

\label{sec:gradient}

For the purposes of optimization, we are interested in determining the elements of the gradient vector
$\nabla\!_{\theta,\phi}\omega(\theta,\phi)$. As we show in what follows, the following quantum channels appear in the expressions for the gradient:
\begin{align}
\Phi_{\theta}(X) &  \coloneqq\int_{\mathbb{R}}dt\ p(t)\ e^{-iG(\theta
)t}Xe^{iG(\theta)t},\label{eq:Phi}\\
\Psi_{\phi}(X) &  \coloneqq \int_{0}^{1}dt\ e^{iH(\phi)t}Xe^{-iH(\phi)t},\label{eq:Psi}
\end{align}
where
\begin{equation}
p(t)\coloneqq\frac{2}{\pi}\ln\left\vert \coth\!\left(  \frac{\pi t}{2}\right)
\right\vert 
\label{eq:high-peak-tent-density}
\end{equation}
is a probability density function known as the \textit{high-peak-tent} density~\cite{patel2024quantumboltzmannmachine}. We also make use of the Hilbert--Schmidt adjoint of~\eqref{eq:Psi}, which is given by
\begin{equation}
\Psi^\dagger_{\phi}(X)   = \int_{0}^{1}dt\ e^{-iH(\phi)t} X e^{iH(\phi)t}.
\label{eq:Psi-adjoint}
\end{equation}
These channels also appear later on in various expressions for quantum generalizations of Fisher information (see, e.g., Table~\ref{table:FB-WY-KM-results}).

The channel $\Phi_{\theta}$ was previously shown to be relevant for ground-state energy estimation~\cite{patel2024quantumboltzmannmachine} and for natural gradient descent~\cite{patel2024naturalgradientparameterestimation} using quantum Boltzmann machines, and related, it is relevant for the same tasks when using evolved quantum Boltzmann machines. It also appeared in prior work on generative modeling~\cite{Coopmans2024qbm_gen_learn}, Hamiltonian learning~\cite{Anshu2021ham_learning_qbm}, and quantum belief propagation~\cite{Hastings2007,Kim2012}.

As far as we are aware, the use of the channel $\Psi_{\phi}$ for general quantum machine learning tasks is original to the present paper. However, the technique behind the   derivations that lead to the channel $\Psi_{\phi}$ (i.e., Duhamel's formula) is the same technique used in the derivations of~\cite{Banchi2021measuringanalytic}. Furthermore, our quantum algorithms that make use of the channel $\Psi_{\phi}$ are, in a broad sense, similar to the algorithms proposed in~\cite{Banchi2021measuringanalytic}, in that they both make use of classical random sampling. The channel $\Psi_{\phi}$ has also appeared in the context of evaluating the Fisher--Bures information matrix of time-evolved pure states on quantum computers~\cite{Ho2023}, with applications to quantum metrology. 
In Section~\ref{sec:app-estimating}, we consider the same task but for general time-evolved states, representing a significant generalization of the task considered in~\cite{Ho2023}. This builds upon our expressions for the Fisher--Bures and Wigner--Yanase information matrices for general time-evolved states (Theorems~\ref{thm:FB-phi} and~\ref{thm:WY-phi}).

Observe that the channel $\Phi_{\theta}$ can be realized by picking $t$ at
random from $p(t)$ and then applying the Hamiltonian evolution $e^{-iG(\theta
)t}$. Similarly, $\Psi_{\phi}$ can be realized by picking $t$ uniformly at
random from $\left[  0,1\right]  $ and then applying the Hamiltonian evolution
$e^{iH(\phi)t}$. These observations play a role later on in our quantum algorithms for estimating the elements of the gradient and information matrices, as they require realizing them on a quantum computer.

Theorem~\ref{thm:gradient-eQBM} below presents our expressions for the gradient of evolved quantum Boltzmann machines:

\begin{theorem}
\label{thm:gradient-eQBM}
The partial derivatives for the parameterized family in~\eqref{eq:ansatz} are as follows:
\begin{align}
    \begin{split}
        \frac{\partial}{\partial\theta_{j}}\omega(\theta,\phi)  & = -\frac{1}{2}\left\{  e^{-iH(\phi)}\Phi_{\theta}(G_{j})e^{iH(\phi)},\omega(\theta,\phi)\right\} \\
        & \hspace{0.5cm} +\omega(\theta,\phi)\left\langle G_{j}\right\rangle_{\rho(\theta)},
    \end{split}\label{eq:grad_wrt_theta}\\
    \frac{\partial}{\partial\phi_{k}}\omega(\theta,\phi) & =i\left[\omega(\theta,\phi),\Psi^{\dagger}_{\phi}(H_{k})\right].\label{eq:grad_wrt_phi}
\end{align}

\end{theorem}

\begin{proof}
See Appendix~\ref{proof:gradient-eQBM}. For a proof of~\eqref{eq:grad_wrt_phi}, see also~\cite[Eq.~(5)]{Ho2023}.
\end{proof}

\medskip

Theorem~\ref{thm:gradient-eQBM} serves as a key result for implementing gradient-based optimization techniques when using evolved quantum Boltzmann machines. For a general parameterized family $(\sigma(\gamma))_{\gamma \in \mathbb{R}^L}$, let $\mathcal{L}(\gamma)$ denote a loss function, which is a function of the parameter vector $\gamma$. The goal of an optimization algorithm is to minimize $\mathcal{L}(\gamma)$. The standard gradient descent algorithm does so by means of the following update rule:
\begin{equation}
\label{eq:grad_descent}
    \gamma_{m+1}\coloneqq \gamma_m - \mu \nabla_\gamma \mathcal{L}(\gamma_m),
\end{equation}
where $\mu > 0$ is the learning rate or step size and $\nabla_\gamma \mathcal{L}(\gamma_m)$ is the gradient,  indicating the direction of the steepest descent.
Thus, access to the gradient is essential for optimization when using gradient descent.

The following subsections present two examples of problems for which Theorem~\ref{thm:gradient-eQBM} is useful.

\subsection{Ground-state energy estimation}

\label{sec:GSEE}

The goal of ground-state energy estimation is to minimize the following  objective
function:
\begin{equation}
\inf_{\theta,\phi}\operatorname{Tr}[O \omega(\theta,\phi)],
\label{eq:VQE-costfunction}
\end{equation}
where $O$ is an observable that is efficiently measurable. Theorem~\ref{thm:gradient-eQBM}  implies that the gradient of~\eqref{eq:VQE-costfunction} is given by
\begin{align}
    \begin{split}
        \frac{\partial}{\partial\theta_{j}}\operatorname{Tr}[O\omega(\theta,\phi)] & =
        \left\langle O\right\rangle_{\omega(\theta,\phi)}\left\langle G_{j}\right\rangle_{\rho(\theta)} \\
        & \hspace{-1cm} -\frac{1}{2}\left\langle\left\{e^{iH(\phi)} O e^{-iH(\phi)} , \Phi_{\theta}(G_{j})\right\} \right\rangle_{\rho(\theta)},
    \end{split}\label{eq:VQE-grad-theta} \\
    \frac{\partial}{\partial\phi_{k}}\operatorname{Tr}[O\omega(\theta,\phi)] & =
    i\left\langle\big[\Psi^{\dagger}_{\phi}(H_{k}),O\big] \right\rangle_{\omega(\theta,\phi)}.
    \label{eq:VQE-grad-phi}
\end{align}
See Appendix~\ref{app:VQE} for proofs of~\eqref{eq:VQE-grad-theta}--\eqref{eq:VQE-grad-phi}. 

Supposing that each $G_j$ in~\eqref{eq:G} and each $H_k$ in~\eqref{eq:H} are Hermitian, unitary, and efficiently realizable, both the partial derivatives can be estimated using a quantum computer using standard quantum subroutines. As in prior work \cite{patel2024quantumboltzmannmachine,patel2024naturalgradientparameterestimation}, we note that the circuit constructions presented throughout our paper can straightforwardly be generalized
beyond the case of each $G_{j}$ and $H_k$ being a Pauli string, if they instead are Hermitian operators 
block encoded into unitary circuits~\cite{Low2019hamiltonian,Gilyen2019}. Additionally, similar to prior work on quantum Boltzmann machines~\cite{Coopmans2024qbm_gen_learn,patel2024quantumboltzmannmachine,patel2024naturalgradientparameterestimation}, we assume here and throughout our paper that samples
of the thermal state~$\rho(\theta)$ in~\eqref{eq:param-thermal-state}  are available,
for every possible choice of $\theta$. Based on the assumption that both $H(\phi)$ and $G(\theta)$ are local Hamiltonians, one can efficiently implement the unitary evolutions $e^{-iH(\phi)t}$ and $e^{-iG(\theta)t}$, where $t \in \mathbb{R}$~\cite{childs2018toward}.

The first term of~\eqref{eq:VQE-grad-theta} can be straightforwardly estimated by a quantum algorithm described in detail in~\cite[Algorithm 2]{patel2024quantumboltzmannmachine}. The quantum circuit that plays a role in the procedure for estimating the second term in~\eqref{eq:VQE-grad-theta} is depicted in Figure~\ref{fig:VQE-grad-theta}. The procedure is described in detail in Appendix~\ref{app:VQE_grad_theta-est} as Algorithm~\ref{algo:VQE-grad-theta-2}. The quantum circuit used in the algorithm for estimating~\eqref{eq:VQE-grad-phi} is shown in Figure~\ref{fig:VQE-grad-phi}, with $S$ denoting the phase gate
\begin{equation}
S = 
\begingroup
\renewcommand{\arraystretch}{1.2} % Adjust row spacing
\setlength{\arraycolsep}{6pt}     % Adjust column spacing
\begin{pmatrix}
1 & 0 \\
0 & i
\end{pmatrix}.
\endgroup
\end{equation}
The corresponding algorithm is given in Appendix~\ref{app:VQE_grad_phi-est} as Algorithm~\ref{algo:VQE-grad-phi}. Note that one can alternatively use the quantum algorithm presented in~\cite{Banchi2021measuringanalytic} to evaluate the expression in~\eqref{eq:VQE-grad-phi}.

\begin{figure*}
    \centering
    \begin{subfigure}{\textwidth}
        \centering
        \scalebox{1.5}{
        \Qcircuit @C=1.em @R=1.2em {
        \lstick{\ket{1}} & \gate{\operatorname{Had}}  & \ctrl{1} &  \qw & \gate{\operatorname{Had}} & \meter & \rstick{\hspace{-1.2em}Z} \\
        \lstick{\rho(\theta)} & \qw  & \gate{G_{j}} & \gate{e^{-iG(\theta)t}} & \gate{e^{-iH(\phi)}}  & \meter & \rstick{\hspace{-1.2em}O}
    }
    }  
    \vspace{20pt}
    \caption{Quantum circuit that realizes an unbiased estimate of $-\frac{1}{2}\!\left\langle\left\{ e^{iH(\phi)} O e^{-iH(\phi)}, \Phi_{\theta}(G_{j})\right\} \right\rangle_{\rho(\theta)} $. For each run of the circuit, the time $t$ is sampled at random from the probability density $p(t)$ in~\eqref{eq:high-peak-tent-density}. For details of the algorithm, see Appendix~\ref{app:VQE_grad_theta-est}.}
    \label{fig:VQE-grad-theta}
    \end{subfigure}
    \vspace{10pt}
    
    \begin{subfigure}{\textwidth}
        \centering
        \scalebox{1.5}{
        \Qcircuit @C=1.em @R=1.2em {
        \lstick{\ket{1}} & \gate{\operatorname{Had}} & \gate{S} & \ctrl{1} & \gate{\operatorname{Had}} & \meter & \rstick{\hspace{-1.2em}Z} \\
        \lstick{\omega(\theta,\phi)} & \gate{e^{iH(\phi)t}} & \qw & \gate{H_{k}} & \gate{e^{-iH(\phi)t}} & \meter & \rstick{\hspace{-1.2em}O}
    }
    }
    \vspace{20pt}
    \caption{Quantum circuit that realizes an unbiased estimate of $\frac{i}{2}\left\langle\big[\Psi^{\dagger}_{\phi}(H_{k}),O\big] \right\rangle_{\omega(\theta,\phi)}$. For each run of the circuit, the time $t$ is sampled uniformly at random from $[0,1]$. For details of the algorithm, see Appendix~\ref{app:VQE_grad_phi-est}.}
    \label{fig:VQE-grad-phi}
    \end{subfigure}
    \vspace{10pt}
    
    \begin{subfigure}{\textwidth}
        \centering
        \scalebox{1.5}{
        \Qcircuit @C=1.em @R=1.2em {
        \lstick{\ket{0}} & \gate{\operatorname{Had}} & \gate{S} & \ctrl{1} & \gate{\operatorname{Had}} & \meter & \rstick{\hspace{-1.2em}Z} \\
        \lstick{\eta} & \gate{e^{iH(\phi)(1-t)}} & \qw & \gate{H_{k}} & \gate{e^{iH(\phi)t}} & \meter & \rstick{\hspace{-1.2em}G(\theta)}
    }
    }
    \vspace{20pt}
    \caption{Quantum circuit that realizes an unbiased estimate of $\frac{i}{2} \left\langle  \left[  G(\theta)  ,\Psi_\phi (H_k) \right] \right\rangle_{\eta(\phi)} $. For each run of the circuit, the time $t$ is sampled uniformly at random from $[0,1]$. For details of the algorithm, see Appendix~\ref{app:gen-mod_grad_phi-est}.}
    \label{fig:gen-mod-grad-phi}
    \end{subfigure}
    
    \caption{Quantum circuits involved in estimating  the gradient of the objective functions for the ground-state energy estimation problem and the generative modeling problem. (a) Quantum circuit involved in the estimation of $\partial_{\theta_j}\!\Tr\!\left[ O\omega(\theta,\phi) \right]$; (b) quantum circuit involved in the estimation of $\partial_{\phi_k}\!\Tr\!\left[ O\omega(\theta,\phi) \right]$; (c) quantum circuit involved in the estimation of $\partial_{\phi_k}\! D( \eta \| \omega(\theta,\phi) ) $.}
    \label{fig:VQE-circuits}
\end{figure*}
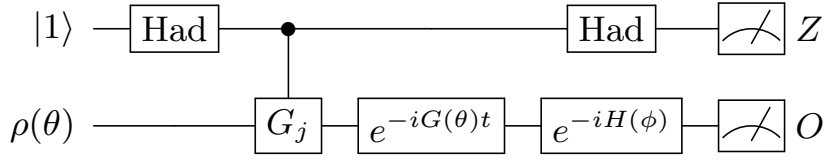
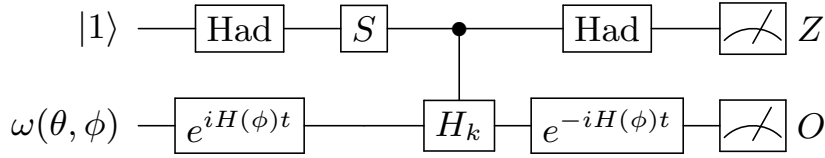
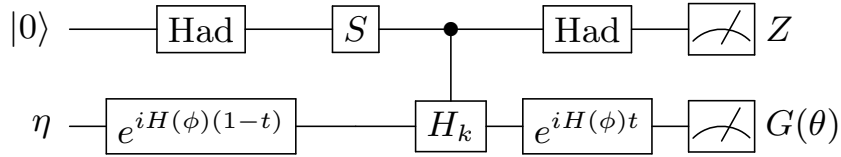

\subsection{Generative modeling}

\label{sec:gen-mod}

Generative modeling involves learning a target quantum state $\eta$ by approximating it with states from a parameterized family. Here, we use the parameterized family $(\omega(\theta,\phi))_{\theta, \phi}$ representing the evolved quantum Boltzmann machine states defined in~\eqref{eq:ansatz}. A natural measure of closeness between the target state~$\eta$ and the model state $\omega(\theta,\phi)$ is the quantum relative entropy, which is defined for general positive-definite states $\omega$ and $\tau$ as~\cite{Umegaki62}
\begin{equation}
\label{eq:quant_rel_entr}
    D( \omega \| \tau ) \coloneqq \Tr\!\left[ \omega (\ln \omega  - \ln \tau) \right].
\end{equation}
This is indeed a natural measure because it satisfies several properties discussed later on in~\eqref{eq:faithful-div}--\eqref{eq:div-non-neg}.
For the states $\eta$ and $\omega(\theta,\phi)$,  the quantum relative entropy can be written as follows:
\begin{equation}
\label{eq:rel_entr_alt}
    D( \eta \| \omega(\theta,\phi) ) \coloneqq \Tr\!\left[ \eta \ln \eta \right] + \Tr\!\left[ G(\theta) \eta(\phi)  \right] + \ln Z(\theta),
\end{equation}
where we have introduced the evolved target state 
\begin{equation}
\eta(\phi) \coloneqq e^{iH(\phi)} \eta e^{-iH(\phi)}.    
\end{equation}
See Appendix~\ref{app:gen_model} for a  derivation of~\eqref{eq:rel_entr_alt}. From~\eqref{eq:rel_entr_alt} or unitary invariance of the quantum relative entropy \cite[Exercise~11.8.6]{Wbook17}, it follows that 
\begin{equation}
     D( \eta \| \omega(\theta,\phi) ) = D( \eta(\phi) \| \rho(\theta) ).
     \label{eq:unitary-invariance-rel-ent}
\end{equation}
Thus, using evolved quantum Boltzmann machines for generative modeling can be seen as allowing the target state to evolve unitarily according to $e^{iH(\phi)}$, potentially improving the learning process when compared to standard quantum Boltzmann machines~\cite{Coopmans2024qbm_gen_learn}.

The quantum relative entropy in~\eqref{eq:quant_rel_entr} can be minimized using gradient descent, as defined by the update rule in~\eqref{eq:grad_descent}, which requires computing its derivatives with respect to the parameters of the model state. To this end, we now show the gradient expressions for the generative modeling problem when using evolved quantum Boltzmann machines.

\begin{theorem}
\label{thm:gen_model_der_theta}
The $j$th element of the gradient  of the quantum relative entropy in~\eqref{eq:quant_rel_entr} with respect to the $\theta$ parameter vector is as follows:
\begin{equation}
    \frac{\partial}{\partial \theta_j} D( \eta \| \omega(\theta,\phi) ) = \left\langle G_j \right\rangle_{\eta(\phi)} - \left\langle G_j \right\rangle_{\rho(\theta)}.
    \label{eq:gen_model_der_theta}
\end{equation}
\end{theorem}

\begin{proof}
    See Appendix~\ref{app:grad_gen_model}.
\end{proof}
\medskip 

For standard quantum Boltzmann machines, the gradient with respect to the parameter vector $\theta$ is equal to the difference between the target $\eta$ and model $\rho(\theta)$ expectation values of the operators in the Hamiltonian~\eqref{eq:G} of the ansatz \cite[Eq.~(4)]{Coopmans2024qbm_gen_learn}. However, for evolved quantum Boltzmann machines, it is  equal to the difference between the expectation values of the same operators but evaluated for the evolved target $\eta(\phi)$ and the thermal state~$\rho(\theta)$.

\begin{theorem}
\label{thm:gen_model_der_phi}
    The $k$th element of the gradient  of the quantum relative entropy in~\eqref{eq:quant_rel_entr} with respect to the $\phi$ parameter vector is as follows:
\begin{equation}
\label{eq:gen_mod_grad_phi}
   \frac{\partial}{\partial \phi_k} D( \eta \| \omega(\theta,\phi) ) = i \left\langle  \left[  G(\theta)  ,\Psi_\phi (H_k) \right] \right\rangle_{\eta(\phi)}. 
\end{equation}
\end{theorem}

\begin{proof}
    See Appendix~\ref{app:grad_gen_model}.
\end{proof}
\medskip 

The quantum circuit depicted in Figure~\ref{fig:gen-mod-grad-phi} can be used for estimating the quantity in~\eqref{eq:gen_mod_grad_phi}.
The algorithm for estimating~\eqref{eq:gen_mod_grad_phi} involves running this circuit multiple times, with the time $t$ sampled uniformly from the interval $[0,1]$ for each run. More details about the algorithm can be found in Appendix~\ref{app:gen-mod_grad_phi-est}. Note that one can alternatively use the quantum algorithm presented in~\cite{Banchi2021measuringanalytic} to evaluate the expression in~\eqref{eq:gen_mod_grad_phi}. 

An additional motivation behind choosing quantum relative entropy as a measure of closeness for generative modeling is that, for fixed $\phi$, it is strictly convex in the parameter vector $\theta$. The claim about strict convexity follows from~\cite[Lemma~6, Supp.~Inf.]{Coopmans2024qbm_gen_learn} and the unitary invariance of the objective function with respect to $\phi$ (see~\eqref{eq:unitary-invariance-rel-ent}). In practice, searching is often restricted to a subspace of the full Hilbert space that corresponds to some space of efficiently preparable thermal states. Now, if the global optimum lies within this subspace, then it can be reached using just the standard quantum Boltzmann machine as an ansatz. Otherwise, one can use evolved quantum Boltzmann machines as an ansatz because it provides access to a larger search space, which may include the global minimum.

\section{Definitions and properties of quantum generalizations of Fisher information}

\label{sec:q_Fisher_info_theory}

In this section, we provide an overview of various quantum generalizations of
Fisher information, which we employ in the remaining sections of our paper.
The approach that we adopt is information-theoretic in nature, viewing Fisher
information as quantifying the distinguishability between infinitesimally
close states chosen from a parameterized family. As such, our starting point
is the information-theoretic notion of a divergence, which quantifies the
distinguishability of two states, and then we develop definitions and
properties of quantum generalizations of Fisher information from this perspective. See~\cite{Bengtsson2006,Liu2019,Sidhu2020,Jarzyna2020,Meyer2021fisherinformationin,sbahi2022provablyefficientvariationalgenerative,scandi2024quantumfisherinformationdynamical} for various reviews, related notions, and background on quantum generalizations of Fisher information.

\subsection{Defining quantum generalizations of Fisher information from smooth
divergences}

To begin with, let us define a smooth divergence~$\boldsymbol{D}$ to be a
smooth function $\boldsymbol{D}\colon\mathcal{D}_{+}(\mathcal{H}
)\times\mathcal{D}_{+}(\mathcal{H})\mapsto\mathbb{R}$ (i.e., from two
positive-definite states $\omega$ and $\tau$ to the reals) such
that the following properties hold for all states $\omega$ and $\tau$ and
every quantum channel $\mathcal{N}$:
\begin{align}
\text{Faithfulness} &  \text{:}\qquad\boldsymbol{D}(\omega\Vert\tau
)=0\quad\Leftrightarrow\quad\omega=\tau,\label{eq:faithful-div}\\
\text{Data-processing} &  \text{:}\qquad\boldsymbol{D}(\omega\Vert\tau
)\geq\boldsymbol{D}(\mathcal{N}(\omega)\Vert\mathcal{N}(\tau
)).\label{eq:data-proc-div}
\end{align}
Here we focus nearly exclusively on
positive-definite states because the exposition is much simpler, and this
assumption is satisfied for evolved quantum Boltzmann machines.
Due to~\eqref{eq:faithful-div} and~\eqref{eq:data-proc-div}, the inequality
\begin{equation}
\boldsymbol{D}(\omega\Vert\tau)\geq0
\label{eq:div-non-neg}
\end{equation}
holds for every pair of states, by considering the trace and replace channel
$\mathcal{R}_{\xi}(X)\coloneqq\operatorname{Tr}[X]\xi$, where $\xi$ is a
state. Indeed, consider that
\begin{equation}
\boldsymbol{D}(\omega\Vert\tau)\geq\boldsymbol{D}(\mathcal{R}_{\xi}
(\omega)\Vert\mathcal{R}_{\xi}(\tau))=\boldsymbol{D}(\xi\Vert\xi)=0.
\end{equation}
Thus, it follows that the minimum value of the divergence~$\boldsymbol{D}$ is
equal to zero, a property used later on in arriving at~\eqref{eq:zeroth-first-order-zero-taylor}.  It is also a direct consequence of the data-processing inequality in~\eqref{eq:data-proc-div} that the smooth divergence $\boldsymbol{D}$ is invariant under the action of a unitary channel~$\mathcal{U}$:
\begin{equation}
    \boldsymbol{D}(\omega\Vert\tau) = \boldsymbol{D}(\mathcal{U}(\omega)\Vert\mathcal{U}(\tau)).
    \label{eq:unitary-inv-D}
\end{equation}

Particular
examples of smooth divergences include the quantum relative entropy~\cite{Umegaki62} (defined already in~\eqref{eq:quant_rel_entr}),
the Petz--R\'{e}nyi relative entropy~\cite{P85,P86}, and the sandwiched R\'{e}nyi
relative entropy~\cite{MDSFT13,WWY14}, which are respectively defined as follows:
\begin{align}
D_{\alpha}(\omega\Vert\tau) &  \coloneqq\frac{1}{\alpha-1}\ln\operatorname{Tr}
[\omega^{\alpha}\tau^{1-\alpha}],\\
\widetilde{D}_{\alpha}(\omega\Vert\tau) &  \coloneqq\frac{1}{\alpha-1}
\ln\operatorname{Tr}\!\left[  \left(  \tau^{\frac{1-\alpha}{2\alpha}}
\omega\tau^{\frac{1-\alpha}{2\alpha}}\right)  ^{\alpha}\right]  .
\end{align}
The quantum relative entropy indeed obeys the data-processing inequality~\cite{Lindblad1975}, 
the Petz--R\'{e}nyi relative entropy obeys it for
$\alpha\in(0,1)\cup(1,2]$~\cite{P85,P86}, and the sandwiched R\'{e}nyi relative
entropy obeys it for $\alpha\in\lbrack1/2,1)\cup(1,\infty)$~\cite{FL13} (see also~\cite{W18opt}). We
especially focus on the case $\alpha=1/2$ for both the Petz-- and sandwiched
R\'{e}nyi relative entropies, which are given by
\begin{align}
D_{1/2}(\omega\Vert\tau) &  \coloneqq-\ln F_{H}(\omega,\tau), \label{eq:Petz-Renyi-order-2}\\
\widetilde{D}_{1/2}(\omega\Vert\tau) &  \coloneqq-\ln F(\omega,\tau).\label{eq:sandwiched-Renyi-order-2}
\end{align}
We have written the quantities in~\eqref{eq:Petz-Renyi-order-2}--\eqref{eq:sandwiched-Renyi-order-2} in terms of the Holevo~\cite{Kholevo1972} and Uhlmann
\cite{Uhl76} fidelities, respectively, which are defined as
\begin{align}
F_{H}(\omega,\tau) &  \coloneqq\left(  \operatorname{Tr}\!\left[  \sqrt
{\omega}\sqrt{\tau}\right]  \right)  ^{2},\label{eq:holevo-fidelity-def}\\
F(\omega,\tau) &  \coloneqq\left\Vert \sqrt{\omega}
\sqrt{\tau}\right\Vert _{1}^{2}.\label{eq:uhlmann-fidelity-def}
\end{align}
Interestingly, the following inequalities relate the Holevo and Uhlmann
fidelities:
\begin{equation}
F_{H}(\omega,\tau)\leq F(\omega,\tau)\leq\sqrt{F_{H}(\omega,\tau)},
\end{equation}
with the first one following from the variational characterization of the trace norm~\cite[Property~9.1.6]{Wbook17} and the second one by adapting~\cite[Theorem~6]{audenaert2008asymptotic} with the choice $s=1/2$, $A=\rho$, and $B=\sigma$. 
They directly imply the following inequalities:
\begin{equation}
-2\ln F_{H}(\omega,\tau)\geq-2\ln F(\omega,\tau)\geq-\ln F_{H}(\omega
,\tau),\label{eq:ineqs-fid-hol-fid}
\end{equation}
which we use later on in Corollary~\ref{cor:WY-FB-ineqs} to relate the
Fisher--Bures and Wigner--Yanase information matrices (defined later in Definition~\ref{def:FB-WY-KM-2nd-derivs}).

Now suppose that $\left(  \sigma(\gamma)\right)  _{\gamma\in\mathbb{R}^{L}}$
is a parameterized, differentiable family of states such that $\gamma$ is an
$L$-dimensional parameter vector, where $L\in\mathbb{N}$. Letting
$\varepsilon\in\mathbb{R}^{L}$ be such that $\left\Vert \varepsilon\right\Vert $
is small (i.e., $\left\Vert \varepsilon\right\Vert \ll1$), the following
smooth divergence~$\boldsymbol{D}$ characterizes the distinguishability of the
nearby states $\sigma(\gamma)$ and $\sigma(\gamma+\varepsilon)$:
\begin{equation}
\boldsymbol{D}(\sigma(\gamma)\Vert\sigma(\gamma+\varepsilon)).
\end{equation}
Since the divergence $\boldsymbol{D}$ is assumed to be smooth, it has a Taylor
expansion of the following form:
\begin{align}
&  \boldsymbol{D}(\sigma(\gamma)\Vert\sigma(\gamma+\varepsilon))\nonumber\\
&  =\boldsymbol{D}(\sigma(\gamma)\Vert\sigma(\gamma))+\varepsilon^{T}
\nabla\boldsymbol{D}+\frac{1}{2}\varepsilon^{T}I^{\boldsymbol{D}}
(\gamma)\varepsilon+O(\left\Vert \varepsilon\right\Vert ^{3}) \notag \\
&  =\frac{1}{2}\varepsilon^{T}I^{\boldsymbol{D}}(\gamma)\varepsilon
+O(\left\Vert \varepsilon\right\Vert ^{3}
),\label{eq:zeroth-first-order-zero-taylor}
\end{align}
where the gradient $\nabla\boldsymbol{D}$ is defined as the $L\times 1$ vector with the
following components:
\begin{equation}
\left[  \nabla\boldsymbol{D}\right]  _{i}=\left.  \frac{\partial}
{\partial\varepsilon_{i}}\boldsymbol{D}(\sigma(\gamma)\Vert\sigma
(\gamma+\varepsilon))\right\vert _{\varepsilon=0}
\end{equation}
and $I^{\boldsymbol{D}}(\gamma)$ is an $L\times L$ Hessian matrix defined in terms of its
matrix elements as
\begin{equation}
\left[  I^{\boldsymbol{D}}(\gamma)\right]  _{ij}\coloneqq\left.
\frac{\partial^{2}}{\partial\varepsilon_{i}\partial\varepsilon_{j}
}\boldsymbol{D}(\sigma(\gamma)\Vert\sigma(\gamma+\varepsilon))\right\vert
_{\varepsilon=0}.\label{eq:D-based-Fisher-matrix}
\end{equation}
We call $I^{\boldsymbol{D}}(\gamma)$ the $\boldsymbol{D}$-based Fisher
information matrix. If we would like to denote the dependence of the matrix
$I^{\boldsymbol{D}}(\gamma)$ on the family $\left(  \sigma(\gamma)\right)
_{\gamma\in\mathbb{R}^{L}}$, then we employ the notation
\begin{equation}
I^{\boldsymbol{D}}\!\left(  \gamma;\left(  \sigma(\gamma)\right)  _{\gamma
\in\mathbb{R}^{L}}\right)  \coloneqq I^{\boldsymbol{D}}(\gamma).
\end{equation}
The equality in~\eqref{eq:zeroth-first-order-zero-taylor} is a direct
consequence of the faithfulness assumption in~\eqref{eq:faithful-div} and the
assumed smoothness of $\boldsymbol{D}$. Indeed, it is evident from
\eqref{eq:faithful-div} that the zeroth order term $\boldsymbol{D}
(\sigma(\gamma)\Vert\sigma(\gamma))$ vanishes, and the first order term $\varepsilon^{T}
\nabla\boldsymbol{D}$
vanishes because we have expanded the function $\varepsilon\mapsto
\boldsymbol{D}(\sigma(\gamma)\Vert\sigma(\gamma+\varepsilon))$ about the
critical point $\varepsilon=0$, i.e., where $\boldsymbol{D}$ takes its minimum
value. Eq.~\eqref{eq:zeroth-first-order-zero-taylor} thus establishes a
pivotal connection between a smooth divergence $\boldsymbol{D}$ and its
$\boldsymbol{D}$-based Fisher information matrix. The main divergences on
which we focus are the quantum relative entropy, twice the negative logarithm
of the Holevo fidelity, and twice the negative logarithm of the Uhlmann
fidelity, which respectively lead to the Kubo--Mori, Wigner--Yanase, and
Fisher--Bures information matrices.

The basic properties of a smooth divergence $\boldsymbol{D}$ imply related
basic properties of the $\boldsymbol{D}$-based Fisher information matrix
$I^{\boldsymbol{D}}(\gamma)$, which we state formally in
Proposition~\ref{prop:fisher-info-matrix-props} below.\ These properties
include non-negativity of $I^{\boldsymbol{D}}(\gamma)$, the data-processing
inequality, and an ordering property.

\begin{proposition}
\label{prop:fisher-info-matrix-props}
Let $\left(  \sigma(\gamma)\right)
_{\gamma\in\mathbb{R}^{L}}$ be a parameterized, differentiable family of
states, and let $I^{\boldsymbol{D}}(\gamma)$ be the $\boldsymbol{D}$-based
Fisher information matrix, as defined in~\eqref{eq:D-based-Fisher-matrix}.
Then the following matrix inequalities hold for all $\gamma\in\mathbb{R}^{L}$,
 corresponding to non-negativity and data-processing under a
quantum channel $\mathcal{N}$, respectively:
\begin{align}
I^{\boldsymbol{D}}(\gamma) &  \geq0,\label{eq:non-neg-D-based-Fisher}\\
I^{\boldsymbol{D}}\!\left(  \gamma;\left(  \sigma(\gamma)\right)  _{\gamma
\in\mathbb{R}^{L}}\right)   &  \geq I^{\boldsymbol{D}}\!\left(  \gamma;\left(
\mathcal{N}(\sigma(\gamma))\right)  _{\gamma\in\mathbb{R}^{L}}\right)
.\label{eq:data-proc-D-based-Fisher}
\end{align}
Furthermore, if $\boldsymbol{D}^{1}$ and $\boldsymbol{D}^{2}$ are smooth
divergences for which the following inequality holds for all states $\omega$
and~$\tau$:
\begin{equation}
\boldsymbol{D}^{1}(\omega\Vert\tau)\geq\boldsymbol{D}^{2}(\omega\Vert
\tau),\label{eq:assumption-ordering}
\end{equation}
then the following matrix inequality holds for all $\gamma\in\mathbb{R}^{L}$:
\begin{equation}
I^{\boldsymbol{D}^{1}}(\gamma)\geq I^{\boldsymbol{D}^{2}}(\gamma
).\label{eq:fisher-matrix-ineq}
\end{equation}

\end{proposition}

\begin{proof}
The proofs are short and based directly on
\eqref{eq:zeroth-first-order-zero-taylor}. Let $\delta\in\left(  0,1\right)  $,
and let $\hat{u}\in\mathbb{R}^{L}$ be a unit vector. Then, as mentioned previously, consider that
\eqref{eq:faithful-div},~\eqref{eq:div-non-neg}, and~\eqref{eq:zeroth-first-order-zero-taylor} imply that the following holds for
all such $\delta$, $\hat{u}$, and $\gamma\in\mathbb{R}^{L}$:
\begin{equation}
\frac{\delta^{2}}{2}\hat{u}^{T}I^{\boldsymbol{D}}(\gamma)\hat{u}+O(\delta
^{3})=\boldsymbol{D}(\sigma(\gamma)\Vert\sigma(\gamma+\delta\hat{u}))\geq0.
\end{equation}
Dividing by $\delta^{2}$ and taking the limit $\delta\rightarrow0$ then
implies that the following inequality holds for all $\hat{u}$ and $\gamma$:
\begin{equation}
\hat{u}^{T}I^{\boldsymbol{D}}(\gamma)\hat{u}\geq0,
\end{equation}
which is equivalent to~\eqref{eq:non-neg-D-based-Fisher}.

Similarly, Eq.~\eqref{eq:zeroth-first-order-zero-taylor} and the data-processing inequality in~\eqref{eq:data-proc-div} imply
that
\begin{align}
&  \frac{\delta^{2}}{2}\hat{u}^{T}I^{\boldsymbol{D}}\!\left(  \gamma;\left(
\sigma(\gamma)\right)  _{\gamma\in\mathbb{R}^{L}}\right)  \hat{u}+O(\delta
^{3})\nonumber\\
&  =\boldsymbol{D}(\sigma(\gamma)\Vert\sigma(\gamma+\delta\hat{u}))\\
&  \geq\boldsymbol{D}(\mathcal{N}(\sigma(\gamma))\Vert\mathcal{N}
(\sigma(\gamma+\delta\hat{u})))\\
&  =\frac{\delta^{2}}{2}\hat{u}^{T}I^{\boldsymbol{D}}\!\left(  \gamma;\left(
\mathcal{N}(\sigma(\gamma))\right)  _{\gamma\in\mathbb{R}^{L}}\right)  \hat
{u}+O(\delta^{3}).
\end{align}
Dividing by $\delta^{2}$ and taking the limit $\delta\rightarrow0$ then
implies that the following inequality holds for all $\hat{u}$ and $\gamma$:
\begin{equation}
\hat{u}^{T}I^{\boldsymbol{D}}\!\left(  \gamma;\left(  \sigma(\gamma)\right)
_{\gamma\in\mathbb{R}^{L}}\right)  \hat{u}\geq\hat{u}^{T}I^{\boldsymbol{D}
}\left(  \gamma;\left(  \mathcal{N}(\sigma(\gamma))\right)  _{\gamma
\in\mathbb{R}^{L}}\right)  \hat{u},
\end{equation}
which is equivalent to~\eqref{eq:data-proc-D-based-Fisher}.

Finally, the proof of~\eqref{eq:fisher-matrix-ineq} proceeds quite similarly
by making use of the assumption in~\eqref{eq:assumption-ordering}.
\end{proof}

\medskip

The following corollary is a direct consequence of the data-processing inequality in~\eqref{eq:data-proc-D-based-Fisher} and is also a consequence of the unitary invariance of the smooth divergence~$\boldsymbol{D}$:

\begin{corollary}
\label{cor:unitary-inv-Fisher-info}
    Let $\left(  \sigma(\gamma)\right)
_{\gamma\in\mathbb{R}^{L}}$ be a parameterized, differentiable family of
states, and let $I^{\boldsymbol{D}}(\gamma)$ be the $\boldsymbol{D}$-based
Fisher information matrix, as defined in~\eqref{eq:D-based-Fisher-matrix}.
Then the following matrix equality holds for all $\gamma\in\mathbb{R}^{L}$ and every unitary channel $\mathcal{U}$:
\begin{equation}
    I^{\boldsymbol{D}}\!\left(  \gamma;\left(  \sigma(\gamma)\right)  _{\gamma
\in\mathbb{R}^{L}}\right)     = I^{\boldsymbol{D}}\!\left(  \gamma;\left(
\mathcal{U}(\sigma(\gamma))\right)  _{\gamma\in\mathbb{R}^{L}}\right).
\label{eq:unitary-inv-Fisher-info}
\end{equation}
\end{corollary}

\subsection{Fisher--Bures, Wigner--Yanase, and Kubo--Mori information
matrices}

Throughout the rest of our paper, we consider three key quantum
generalizations of the Fisher information:\ the Fisher--Bures, Wigner--Yanase,
and Kubo--Mori information matrices. As mentioned previously, they are defined
in terms of the Uhlmann fidelity, Holevo fidelity, and quantum relative
entropy, respectively (see~\eqref{eq:uhlmann-fidelity-def},~\eqref{eq:holevo-fidelity-def}, and~\eqref{eq:quant_rel_entr}, respectively,
for definitions of the latter). These connections have been known for some time since~\cite{Hub92},~\cite{Gibilisco2003}, and~\cite{PetzToth1993}, respectively. Specifically, we define them as follows:

\begin{definition}
\label{def:FB-WY-KM-2nd-derivs}
For $\left(  \sigma(\gamma)\right)  _{\gamma\in\mathbb{R}^{L}}$ a
parameterized, differentiable family of states, the 
Fisher--Bures, Wigner--Yanase, and Kubo--Mori information matrix elements are defined
as follows, respectively:
\begin{align}
  I^{\operatorname{FB}}_{ij}(\gamma)  &  \coloneqq2\left.
\frac{\partial^{2}}{\partial\varepsilon_{i}\partial\varepsilon_{j}}
\widetilde{D}_{1/2}(\sigma(\gamma)\Vert\sigma(\gamma+\varepsilon))\right\vert
_{\varepsilon=0}\\
&  =2\left.  \frac{\partial^{2}}{\partial\varepsilon_{i}\partial
\varepsilon_{j}}\left[  -\ln F(\sigma(\gamma),\sigma(\gamma+\varepsilon
))\right]  \right\vert _{\varepsilon=0},\label{eq:FB-fid-2nd-deriv}\\
  I^{\operatorname{WY}}_{ij}(\gamma)  &  \coloneqq2\left.
\frac{\partial^{2}}{\partial\varepsilon_{i}\partial\varepsilon_{j}}
D_{1/2}(\sigma(\gamma)\Vert\sigma(\gamma+\varepsilon))\right\vert
_{\varepsilon=0}\\
&  =2\left.  \frac{\partial^{2}}{\partial\varepsilon_{i}\partial
\varepsilon_{j}}\left[  -\ln F_{H}(\sigma(\gamma),\sigma(\gamma+\varepsilon
))\right]  \right\vert _{\varepsilon=0},\label{eq:WY-Hol-fid-2nd-deriv}\\
  I^{\operatorname{KM}}_{ij}(\gamma)  &  \coloneqq\left.
\frac{\partial^{2}}{\partial\varepsilon_{i}\partial\varepsilon_{j}}
D(\sigma(\gamma)\Vert\sigma(\gamma+\varepsilon))\right\vert _{\varepsilon=0}.
\end{align}

\end{definition}

The particular prefactors used in the above definitions are related to those made when relating R\'enyi relative entropies to Fisher information (see~\cite[Eq.~(50)]{vEH2014}, \cite[Section~11]{Mat10fid}, and
\cite[Section~6.4]{Mat13,Matsumoto2018}).

These matrix elements have explicit expressions, which are given in Theorem~\ref{thm:QFI-formulas} below. Let us note that these expressions are well known (see, e.g.,~\cite{Sidhu2020,sbahi2022provablyefficientvariationalgenerative}), and for completeness of the exposition in this section, we provide explicit (and, in some cases, brief) proofs that connect the expressions in Definition~\ref{def:FB-WY-KM-2nd-derivs} to those in Theorem~\ref{thm:QFI-formulas} below.

\begin{theorem}
\label{thm:QFI-formulas}
For $\left(  \sigma(\gamma)\right)  _{\gamma\in\mathbb{R}^{L}}$ a
parameterized, differentiable family of states, the following equalities hold:
\begin{align}
  I^{\operatorname{FB}}_{ij}(\gamma) &  =\operatorname{Tr}
\!\left[  \mathcal{F}\!\left(  \partial_{i}\sigma(\gamma)\right)  
\partial_{j}\sigma(\gamma)  \right]  \label{eq:FB-formula-1}\\
&  =\sum_{k,\ell}\frac{2}{\lambda_{k}+\lambda_{\ell}}\langle k|\partial
_{i}\sigma(\gamma)|\ell\rangle\langle\ell|\partial_{j}\sigma(\gamma
)|k\rangle,\label{eq:FB-formula-2}\\
  I^{\operatorname{WY}}_{ij}(\gamma) &  =\operatorname{Tr}
\!\left[  \mathcal{W}\!\left(  \partial_{i}\sigma(\gamma)\right)  
\partial_{j}\sigma(\gamma)  \right] \label{eq:WY-formula-1} \\
&  =\sum_{k,\ell}\frac{4}{\left(  \sqrt{\lambda_{k}}+\sqrt{\lambda_{\ell}
}\right)  ^{2}}\langle k|\partial_{i}\sigma(\gamma)|\ell\rangle\langle
\ell|\partial_{j}\sigma(\gamma)|k\rangle,\label{eq:WY-formula-2}\\
  I^{\operatorname{KM}}_{ij}(\gamma) &  =\operatorname{Tr}
\!\left[  \mathcal{K}\!\left(  \partial_{i}\sigma(\gamma)\right)  
\partial_{j}\sigma(\gamma)  \right]  \label{eq:KM-formula-1}\\
&  =\sum_{k,\ell}\frac{\ln\lambda_{k}-\ln\lambda_{\ell}}{\lambda_{k}
-\lambda_{\ell}}\langle k|\partial_{i}\sigma(\gamma)|\ell\rangle\langle
\ell|\partial_{j}\sigma(\gamma)|k\rangle,
\label{eq:KM-formula-2}
\end{align}
where an eigendecomposition of $\sigma(\gamma)$ is given by $\sigma
(\gamma)=\sum_{k}\lambda_{k}|k\rangle\!\langle k|$ and we have employed the
shorthand $\partial_{i}\equiv\frac{\partial}{\partial\gamma_{i}}$. Also, the
superoperators $\mathcal{F}$, $\mathcal{W}$, and $\mathcal{K}$ are defined as
follows:
\begin{align}
 \mathcal{F}(X)&\coloneqq 2\int_{0}^{\infty}dt\ e^{-t\sigma(\gamma
)}Xe^{-t\sigma(\gamma)},\\
 \mathcal{W}(X)&\coloneqq4\int_{0}^{\infty}\int_{0}^{\infty}ds
\ dt\ e^{-\left(  s+t\right)  \sqrt{\sigma(\gamma)}}Xe^{-\left(
s+t\right)  \sqrt{\sigma(\gamma)}},\\
 \mathcal{K}(X)&\coloneqq\int_{0}^{\infty}dt\ \left(  \sigma(\gamma
)+tI\right)  ^{-1}X\left(  \sigma(\gamma)+tI\right)  ^{-1}.
\end{align}
\end{theorem}

\begin{proof}
See Appendix~\ref{app:QFI-formulas}.
\end{proof}

\medskip Based on~\eqref{eq:ineqs-fid-hol-fid},~\eqref{eq:fisher-matrix-ineq},
and the definitions in~\eqref{eq:FB-fid-2nd-deriv} and~\eqref{eq:WY-Hol-fid-2nd-deriv}, we arrive at the following corollary relating the Fisher--Bures and Wigner--Yanase information matrices,
representing a broad generalization of the inequalities reported in~\cite[Theorem~2]{Luo2004}, while noting that the normalization conventions in \cite{Luo2004} are non-standard:

\begin{corollary}
\label{cor:WY-FB-ineqs}For $\left(  \sigma(\gamma)\right)  _{\gamma
\in\mathbb{R}^{L}}$ a parameterized, differentiable family of states, the
following inequalities hold:
\begin{equation}
I^{\operatorname{WY}}(\gamma)\geq I^{\operatorname{FB}}(\gamma)\geq
\frac{1}{2} I^{\operatorname{WY}}(\gamma).
\end{equation}

\end{corollary}

As a consequence of Corollary~\ref{cor:WY-FB-ineqs}, we see that the
Fisher--Bures and Wigner--Yanase information matrices differ only by a
constant in the matrix order. As such, they are essentially exchangeable when
employing an optimization algorithm like natural gradient descent. This point
is discussed further in Section~\ref{sec:app-nat-grad}. Furthermore, if there are scenarios in
which the Wigner--Yanase information matrix is much simpler to compute than
the Fisher--Bures information matrix, then Corollary~\ref{cor:WY-FB-ineqs} is
useful.

Corollary~\ref{cor:WY-FB-ineqs} also has implications for estimation theory
and the multiparameter Cramer--Rao bound (see, e.g.,~\cite[Section~V]{Sidhu2020}), in particular implying that the
Fisher--Bures and Wigner--Yanase information matrices provide essentially the
same fundamental limitations on an estimation scheme, up to a factor of two. See Section~\ref{sec:app-estimating} for further discussions.

\subsection{Pure-state families, Wigner--Yanase information, and
canonical purifications}

Let us recall the known formula for a pure, parameterized, differentiable
family $\left(  \psi(\gamma)\right)  _{\gamma\in\mathbb{R}^{L}}$ of states and
connect to the Fisher--Bures and Wigner--Yanase information matrices above. In
particular, due to the fact that
\begin{equation}
F(\psi,\varphi)=\sqrt{F_{H}(\psi,\varphi)}=\left\vert \langle\psi
|\varphi\rangle\right\vert ^{2}
\end{equation}
for pure states $\psi$ and $\varphi$, the Fisher--Bures and Wigner--Yanase
information matrices, as defined in~\eqref{eq:FB-fid-2nd-deriv} and
\eqref{eq:WY-Hol-fid-2nd-deriv}, respectively, are proportional and expressed
in terms of the state vector $|\psi(\gamma)\rangle$ as follows:
\begin{align}
   I^{\operatorname{FB}}_{ij}(\gamma)  
&  =\frac{1}{2}
I^{\operatorname{WY}}_{ij}(\gamma)  \\
&  =2\left.  \frac{\partial^{2}}{\partial\varepsilon_{i}\partial
\varepsilon_{j}}\left[  -\ln F(\psi(\gamma),\psi(\gamma
+\varepsilon))\right]  \right\vert _{\varepsilon
=0}\\
&  =2\left.  \frac{\partial^{2}}{\partial\varepsilon_{i}\partial
\varepsilon_{j}}\left[  -\ln\left\vert \langle\psi(\gamma)|\psi(\gamma
+\varepsilon)\rangle\right\vert ^{2}\right]  \right\vert _{\varepsilon
=0}.\label{eq:pure-state-2nd-deriv}
\end{align}
Explicitly calculating~\eqref{eq:pure-state-2nd-deriv} for $\left(
\psi(\gamma)\right)  _{\gamma\in\mathbb{R}^{L}}$ then leads to the following
well known expression (see, e.g., \cite[Eq.~(129)]{Sidhu2020}):
\begin{align}
&   I^{\operatorname{FB}}_{ij}(\gamma)   \notag \\
&  =\frac{1}{2}
I^{\operatorname{WY}}_{ij}(\gamma)  \\
&  =4\operatorname{Re}\!\left[  \left\langle \partial_{i}\psi(\gamma
)|\partial_{j}\psi(\gamma)\right\rangle -\left\langle \partial_{i}\psi
(\gamma)|\psi(\gamma)\right\rangle \left\langle \psi(\gamma)|\partial_{j}
\psi(\gamma)\right\rangle \right]  .\label{eq:formula-FB-pure}
\end{align}
We detail this calculation in Appendix~\ref{app:proof-formula-FB-pure}.

We can furthermore establish an explicit connection between the Wigner--Yanase
information of a parameterized family $\left(  \sigma(\gamma)\right)
_{\gamma\in\mathbb{R}^{L}}$ and the Fisher--Bures information of the pure
parameterized family $\left(  \varphi^{\sigma}(\gamma)\right)  _{\gamma
\in\mathbb{R}^{L}}$, where $\varphi^{\sigma}(\gamma)$ is a canonical
purification of $\sigma(\gamma)$, defined as
\begin{equation}
\varphi^{\sigma}(\gamma)\coloneqq\left(  \sqrt{\sigma(\gamma)}\otimes
I\right)  \Gamma\left(  \sqrt{\sigma(\gamma)}\otimes I\right)
,\label{eq:canonical-purification-FB-WY}
\end{equation}
and $\Gamma$ is the standard maximally entangled operator:
\begin{equation}
\Gamma\coloneqq\sum_{k,\ell}|k\rangle\!\langle\ell|\otimes|k\rangle\!\langle\ell|.
\end{equation}
In particular, the following proposition holds:

\begin{proposition}
\label{prop:FB-WY-canonical-purifications}
For a parameterized, differentiable family $\left(  \sigma(\gamma)\right)
_{\gamma\in\mathbb{R}^{L}}$ of positive-definite states, the following
equality holds:
\begin{equation}
I^{\operatorname{WY}}\!\left(  \gamma;\left(  \sigma(\gamma)\right)
_{\gamma\in\mathbb{R}^{L}}\right)  =I^{\operatorname{FB}}\!\left(
\gamma;\left(  \varphi^{\sigma}(\gamma)\right)  _{\gamma\in\mathbb{R}^{L}
}\right)  ,
\end{equation}
where $\varphi^{\sigma}(\gamma)$ is defined in~\eqref{eq:canonical-purification-FB-WY}.
\end{proposition}

\begin{proof}
Defining $|\Gamma\rangle\coloneqq\sum_{k}|k\rangle\otimes|k\rangle$, the following equalities hold  for all states $\sigma_{1}$ and $\sigma_{2}$:
\begin{align}
F_{H}(\sigma_{1},\sigma_{2})  &  =\left(  \operatorname{Tr}[\sqrt{\sigma_{1}
}\sqrt{\sigma_{2}}]\right)  ^{2}\\
&  =\left(  \langle\Gamma|\left(  \sqrt{\sigma_{1}}\sqrt{\sigma_{2}}\otimes
I\right)  |\Gamma\rangle\right)  ^{2}\\
&  =\left(  \langle\varphi^{\sigma_{1}}|\varphi^{\sigma_{2}}\rangle\right)
^{2}\\
&  =\left\vert \langle\varphi^{\sigma_{1}}|\varphi^{\sigma_{2}}\rangle
\right\vert ^{2}\\
&  =F(\varphi^{\sigma_{1}},\varphi^{\sigma_{2}}).
\end{align}
So this implies that
\begin{equation}
-2\ln F_{H}(\sigma_{1},\sigma_{2})=-2\ln F(\varphi^{\sigma_{1}},\varphi
^{\sigma_{2}}).
\end{equation}
By the definitions in~\eqref{eq:FB-fid-2nd-deriv} and
\eqref{eq:WY-Hol-fid-2nd-deriv}, the claim follows. See Appendix~\ref{app:proof-FB-WY-canonical-purifications} for an
alternative proof.
\end{proof}

\section{Information matrices for evolved quantum Boltzmann machines}\label{sec:infor_matrix_eqbm}

In this section, we detail some further contributions of our paper, namely, analytical expressions for and quantum circuits and algorithms to estimate the matrix elements of the Fisher--Bures (Section~\ref{sec:FB-info}), Wigner--Yanase (Section~\ref{sec:WY-info}), and Kubo--Mori (Section~\ref{sec:KM-info}) information matrices for evolved quantum Boltzmann machines. As mentioned previously, these results are summarized in Table~\ref{table:FB-WY-KM-results}. The main applications of these findings are to natural gradient descent algorithms for quantum machine learning tasks (Section~\ref{sec:app-nat-grad}) and fundamental limitations on estimating time-evolved thermal states (Section~\ref{sec:app-estimating}).

Before delving into the results of Sections~\ref{sec:FB-info}--\ref{sec:KM-info}, we first present, in Section~\ref{sec:general-consids-info-mats}, some general considerations that apply to all of the information matrices for evolved quantum Boltzmann machines. 

\subsection{General considerations for the information matrices of evolved quantum Boltzmann machines}

\label{sec:general-consids-info-mats}

Let us first recall the parameterization in~\eqref{eq:ansatz}. The state $\omega(\theta, \phi)$ is parameterized by  $\theta \in \mathbb{R}^J$ and $\phi \in \mathbb{R}^K$. Thus, $\gamma = (\theta, \phi)$, $L = J+K$, and each information matrix is a block matrix of the following form
\begin{equation}
    I(\gamma) =
    \begin{pmatrix}
    \begin{bmatrix}
    \makebox[3em][c]{$I(\theta)$}
    \end{bmatrix} &
    \begin{bmatrix}
    \makebox[3em][c]{$I(\theta, \phi)$}
    \end{bmatrix} \\[1em]
    \begin{bmatrix}
    \makebox[3em][c]{$I(\phi, \theta)$}
    \end{bmatrix} &
    \begin{bmatrix}
    \makebox[3em][c]{$I(\phi)$}
    \end{bmatrix}
    \end{pmatrix},
    \label{eq:full-info-matrix-theta-phi}
\end{equation}
where the top left matrix is a $J\times J$ information matrix for the parameter vector $\theta$, the bottom right matrix is a $K\times K$  information matrix for the parameter vector~$\phi$, and the other $J\times K$ and $K\times J$ matrices capture cross terms between $\theta$ and $\phi$ (here note that $I(\phi, \theta) = I(\theta, \phi)^T$).

For the parameterized states $\omega(\theta,\phi)$, the  information matrices in Definition~\ref{def:FB-WY-KM-2nd-derivs} and Theorem~\ref{thm:QFI-formulas} have multiple components because the parameter vectors $\theta$ and $\phi$ are independent.
For the parametrized state $\omega(\theta,\phi)$, we denote the  various information matrices as follows:
\begin{equation}
    \begin{aligned}
I_{ij}(\theta) & \coloneqq \sum_{k,\ell}c(\lambda_{k},\lambda_{\ell})
\langle k|\partial_{\theta_i}\omega(\theta,\phi)|\ell\rangle\!\langle\ell|\partial_{\theta_j}\omega(\theta,\phi)|k\rangle,\\
I_{ij}(\phi) & \coloneqq \sum_{k,\ell}c(\lambda_{k},\lambda_{\ell})
\langle k|\partial_{\phi_i}\omega(\theta,\phi)|\ell\rangle\!\langle\ell|\partial_{\phi_j}\omega(\theta,\phi)|k\rangle,\\
I_{ij}(\theta,\phi) & \coloneqq \sum_{k,\ell}c(\lambda_{k},\lambda_{\ell})
\langle k|\partial_{\theta_i}\omega(\theta,\phi)|\ell\rangle\!\langle\ell|\partial_{\phi_j}\omega(\theta,\phi)|k\rangle,
\end{aligned}
\end{equation}
where $c(\lambda_{k},\lambda_{\ell})
$ is one of the functions in Theorem~\ref{thm:QFI-formulas},  $\partial_{\theta_i} \equiv \frac{\partial}{\partial_{\theta_i}}$, $\partial_{\phi_i}\equiv \frac{\partial}{\partial_{\phi_i}}$, and we used the spectral decomposition 
\begin{equation}
\label{eq:spectral_decomp}
    \omega(\theta,\phi) = \sum_k \lambda_k |k\rangle\! \langle k|.
\end{equation}
In the latter notation, we have suppressed the dependence of the eigenvalue $\lambda_k$ and the eigenvector $\ket{k}$ on the parameter vectors $\theta$ and $\phi$. In accordance with~\eqref{eq:ansatz} and~\eqref{eq:spectral_decomp}, the spectral decomposition for the parameterized thermal state $\rho(\theta)$ is as follows:
\begin{equation}
    \rho(\theta)=\sum_k \lambda_k |\tilde{k}\rangle\!\langle \tilde{k}|,
\end{equation}
where the eigenvalues $\lambda_k$ are the same as those of $\omega(\theta,\phi)$, while the eigenvectors $\ket{k}$ of $\omega(\theta,\phi)$ and $\ket{\tilde{k}}$ of $\rho(\theta)$ are related as follows: $\ket{k}=e^{-iH(\phi)}\ket{\tilde{k}}$.

By inspecting~\eqref{eq:FB-formula-2},~\eqref{eq:WY-formula-2}, and~\eqref{eq:KM-formula-2}, it is clear that the following calculations are helpful for evaluating all of the information matrix elements in Theorem~\ref{thm:QFI-formulas} for evolved quantum Boltzmann machines 
with respect to the $\theta$ parameters:
\begin{align}
\langle k|&\left[  \frac{\partial}{\partial\theta_{j}} \omega(\theta,\phi)\right]
|\ell\rangle   \nonumber\\
\begin{split}
& =\langle k| \Big( -\frac{1}
{2}\left\{  e^{-iH(\phi)}\Phi_{\theta}(G_{j})e^{iH(\phi)},\omega(\theta
,\phi)\right\} \\
& \hspace{0.4cm}+ \omega(\theta,\phi)\left\langle G_{j}\right\rangle
_{\rho(\theta)} \Big)   |\ell\rangle
\end{split}\\
\begin{split}
&  = -\frac{1}{2}\bra{k} e^{-iH(\phi)}\Phi_{\theta}(G_{j})e^{iH(\phi)}\omega(\theta,\phi) \ket{\ell} \\
&\hspace{0.4cm}-\frac{1}{2}\bra{k} \omega(\theta,\phi) e^{-iH(\phi)}\Phi_{\theta}(G_{j})e^{iH(\phi)} \ket{\ell} \\
&\hspace{0.4cm}+ \bra{k} \omega(\theta,\phi) \ket{\ell} \left\langle G_{j}\right\rangle
_{\rho(\theta)}
\end{split}\\
\begin{split}
& = -\frac{1}{2}\left( \lambda_\ell + \lambda_k\right) \bra{k} e^{-iH(\phi)}\Phi_{\theta}(G_{j})e^{iH(\phi)}\ket{\ell}\\
&\hspace{0.4cm}+ \delta_{k\ell} \lambda_\ell \left\langle G_{j}\right\rangle_{\rho(\theta)}\label{eq:Fisher-info-help-alt1}
\end{split}\\
& = -\frac{1}{2}\left( \lambda_\ell + \lambda_k\right) \bra{\tilde{k}} \Phi_{\theta}(G_{j})\ket{\tilde{\ell}}+ \delta_{k\ell} \lambda_\ell \left\langle G_{j}\right\rangle_{\rho(\theta)},\label{eq:Fisher-info-help-alt}
\end{align}
and with respect
to the $\phi$ parameters:
\begin{align}
\langle k|&\left[  \frac{\partial}{\partial\phi_{i}} \omega(\theta,\phi)\right]
|\ell\rangle   \nonumber\\
& =\langle k|  i\left[  \omega(\theta,\phi),\Psi_{\phi
}(H_{i})\right] |\ell\rangle\\
&  =i\langle k|\left(  \omega(\theta,\phi)\Psi^{\dagger}_{\phi}(H_{i})-\Psi_{\phi}
(H_{i})\omega(\theta,\phi)\right)  |\ell\rangle\\
&  =i\left(  \lambda_{k}-\lambda_{\ell}\right)  \langle k|\Psi_{\phi}
(H_{i})|\ell\rangle .\label{eq:Fisher-info-help-last}
\end{align}
These calculations follow by direct substitution using Theorem~\ref{thm:gradient-eQBM} and applying~\eqref{eq:spectral_decomp}.

Before proceeding to the forthcoming subsections, let us note that the expressions given for $I_{ij}^{\operatorname{FB}}(\theta)$, $I_{ij}^{\operatorname{WY}}(\theta)$, and $I_{ij}^{\operatorname{KM}}(\theta)$ in Theorems~\ref{thm:FB-theta}, \ref{thm:WY-theta}, and \ref{thm:KM-theta} below, respectively, generalize the classical formulas presented in \cite[Eq.~(6)]{Crooks2007}.

\subsection{Fisher–Bures information matrix of evolved quantum Boltzmann machines}

\label{sec:FB-info}

The main results of this section are Theorems~\ref{thm:FB-theta},~\ref{thm:FB-phi}, and~\ref{thm:FB-theta-phi}, which provide explicit analytical expressions for the Fisher--Bures information matrix elements of evolved quantum Boltzmann machines. Additionally, we give quantum circuits that play a key role in efficiently estimating the terms in these expressions (see Figure~\ref{fig:FB-circuits}).

\begin{theorem}
\label{thm:FB-theta}
For the parameterized family in~\eqref{eq:ansatz}, the Fisher--Bures information matrix elements, with respect to the $\theta$
parameters, are as follows:
\begin{equation}
I_{ij}^{\operatorname{FB}}(\theta) = \frac{1}{2} \left\langle \{\Phi_{\theta}(G_{i}),\Phi_{\theta}(G_{j})\}   \right\rangle_{\rho(\theta)}  - \left\langle G_{i} \right\rangle_{\rho(\theta)} \left\langle G_{j} \right\rangle_{\rho(\theta)}.
\label{eq:FB-theta}
\end{equation}
\end{theorem}

\begin{proof}
See Appendix~\ref{proof:FB-theta}.
\end{proof}

\medskip

This is the same result obtained in~\cite{patel2024naturalgradientparameterestimation} in the case of using the parameterized family of thermal states. $I_{ij}^{\operatorname{FB}}(\theta)$ can be efficiently estimated on a quantum computer (under the assumption that each $G_j$ is a local operator, acting on a constant number of qubits)~\cite{patel2024naturalgradientparameterestimation}. Supposing that each $G_j$ in~\eqref{eq:G} is not only Hermitian but also unitary, as it is for the common case in which each $G_j$ is a tensor product of Pauli operators, we can use the quantum circuit shown in Figure~\ref{fig:FB-theta} for estimating the first term of~\eqref{eq:FB-theta}. Also, the second term
$\left\langle G_{i}\right\rangle _{\rho(\theta)}\left\langle G_{j}
\right\rangle _{\rho(\theta)}$ in~\eqref{eq:FB-theta} can be
estimated by means of a quantum algorithm. Since it can be written as
\begin{equation}
\left\langle G_{i}\right\rangle _{\rho(\theta)}\left\langle G_{j}\right\rangle
_{\rho(\theta)}=\operatorname{Tr}[\left(  G_{i}\otimes G_{j}\right)  \left(
\rho(\theta)\otimes\rho(\theta)\right)  ], \label{eq:simple-term-to-estimate}
\end{equation}
a procedure for estimating it is to generate the state $\rho(\theta
)\otimes\rho(\theta)$ and then measure the observable $G_{i}\otimes G_{j}$ on
these two copies. Through repetition, the estimate of $\left\langle
G_{i}\right\rangle _{\rho(\theta)}\left\langle G_{j}\right\rangle
_{\rho(\theta)}$ can be made as precise as desired. This procedure is
described in detail as~\cite[Algorithm~2]{patel2024quantumboltzmannmachine}. 

The elements of the Fisher--Bures information matrix for the parameter vector $\phi$ are as follows:

\begin{theorem}
\label{thm:FB-phi}
For the parameterized family in~\eqref{eq:ansatz}, the Fisher--Bures information matrix elements, with respect to the $\phi$
parameters, are as follows:
\begin{equation}
I_{ij}^{\operatorname{FB}}(\phi) = \left\langle \left[  \Phi_{\theta}(\Psi_{\phi}(H_{i})), \left[  G(\theta) , \Psi_{\phi}
(H_{j})\right]  \right]
\right\rangle _{\rho(\theta)}.
\label{eq:FB-phi}
\end{equation}
\end{theorem}

\begin{proof}
See Appendix~\ref{proof:FB-phi}.
\end{proof}

\medskip

The quantity in~\eqref{eq:FB-phi} can be estimated on a quantum computer by a generalisation of the single-qubit Hadamard test circuit, which is useful for evaluating the expectation value of nested commutators~\cite[Algorithm~3]{liEfficientQuantumGradient2024}. The quantum circuit used in this case is depicted in Figure~\ref{fig:FB-phi}.

The elements of the Fisher--Bures information matrix for the cross terms involving $\theta$ and $\phi$ are as follows:

\begin{theorem}\label{thm:FB-theta-phi}
For the parameterized family in~\eqref{eq:ansatz}, the Fisher--Bures information matrix elements, with respect to the $\theta$ and $\phi$
parameters, are as follows:
\begin{equation}
    I_{ij}^{\operatorname{FB}}(\theta, \phi) = i \left\langle \left[ \Phi_{\theta}(G_{i}) , \Psi_{\phi}(H_{j}) \right] \right\rangle_{\rho(\theta)}.\label{eq:FB-theta-phi}
\end{equation}
\end{theorem}

\begin{proof}
See Appendix~\ref{proof:FB-theta-phi}.     
\end{proof}

\medskip 

The term in~\eqref{eq:FB-theta-phi} can be estimated by means of the quantum circuit depicted in Figure~\ref{fig:FB-theta-phi}.

\begin{figure*}
    \centering
    \begin{subfigure}{\textwidth}
        \centering
        \scalebox{1.5}{
        \Qcircuit @C=1.2em @R=1.5em {
        \lstick{\ket{0}} & \gate{\operatorname{Had}} & \ctrl{1} & \gate{\operatorname{Had}} & \meter & \rstick{\hspace{-1.2em}Z} \\
        \lstick{\rho(\theta)} & \qw & \gate{G_i} & \gate{e^{iG(\theta)(t_2-t_1)}} &  \meter & \rstick{\hspace{-1.2em}G_{j}}
        }
        }
        \vspace{20pt}
        \caption{Quantum circuit that realizes an unbiased estimate of $\frac{1}{2} \left\langle \{\Phi_{\theta}(G_{i}),\Phi_{\theta}(G_{j})\}   \right\rangle_{\rho(\theta)}$. For each run of the circuit, the times $t_1$ and $t_2$ are sampled independently at random from the probability density $p(t)$ in~\eqref{eq:high-peak-tent-density}. For details of the algorithm, see Appendix~\ref{app:FB-theta}.}
        \label{fig:FB-theta}
    \end{subfigure}
    \vspace{10pt}
    
    \begin{subfigure}{\textwidth}
        \centering
        \scalebox{1.29}{
        \Qcircuit @C=0.9em @R=1.0em {
            \lstick{\ket{1}} & \gate{\operatorname{Had}} & \gate{S} & \ctrl{2} & \qw & \qw & \qw & \qw & \gate{\operatorname{Had}} & \meter &\rstick{\hspace{-1.2em}Z} \\
            \lstick{\ket{1}} & \gate{\operatorname{Had}} & \gate{S} & \qw & \qw & \qw & \qw & \ctrl{1} & \gate{\operatorname{Had}} & \meter & \rstick{\hspace{-1.2em}Z}\\
            \lstick{\rho(\theta)} & \gate{e^{-iH(\phi)t_3}} & \qw & \gate{H_{i}} & \gate{e^{iH(\phi)t_3}} & \gate{e^{-iG(\theta)t_2}} & \gate{e^{-iH(\phi)t_1}} & \gate{H_{j}} & \gate{e^{iH(\phi)t_1}} & \meter & \rstick{\hspace{-1.2em}G(\theta)}
        }
        }
        \vspace{20pt}
         \caption{Quantum circuit that realizes an unbiased estimate of $\frac{1}{4} \left\langle \left[  \left[  \Psi_{\phi}(H_{j}),G(\theta)\right]  ,\Phi_{\theta}(\Psi_{\phi}(H_{i}))\right]\right\rangle _{\rho(\theta)}$. For each run of the circuit, the times $t_1$ and $t_3$ are sampled uniformly at random from $[0,1]$, and the time $t_2$ is sampled independently at random from the probability density $p(t)$ in~\eqref{eq:high-peak-tent-density}. For details of the algorithm, see Appendix~\ref{app:FB-phi}.}
        \label{fig:FB-phi}
    \end{subfigure}
    \vspace{10pt}
    
    \begin{subfigure}{\textwidth}
        \centering
        \scalebox{1.5}{
        \Qcircuit @C=1em @R=1.5em {
        \lstick{\ket{0}} & \gate{\operatorname{Had}} & \gate{S} & \ctrl{1} & \qw & \gate{\operatorname{Had}} & \meter & \rstick{\hspace{-1.2em}Z} \\
        \lstick{\rho(\theta)} & \gate{e^{-iH(\phi)t_2}} & \qw & \gate{H_{j}} & \gate{e^{iH(\phi)t_2}} & \gate{e^{iG(\theta)t_1}} & \meter & \rstick{\hspace{-1.2em}G_{i}}
        }
        }
        \vspace{20pt}
        \caption{Quantum circuit that realizes an unbiased estimate of $\frac{i}{2}\left\langle  \left[\Phi_{\theta}(G_{i}),\Psi_{\phi}(H_{j})\right]\right\rangle_{ \rho(\theta)}$. For each run of the circuit, the time $t_1$ is sampled independently at random from the probability density $p(t)$ in~\eqref{eq:high-peak-tent-density}, and $t_2$ is sampled uniformly at random from $[0,1]$. For details of the algorithm, see Appendix~\ref{app:FB-theta-phi}.}
        \label{fig:FB-theta-phi}
    \end{subfigure}
    
    \caption{Quantum circuits involved in the estimation of the Fisher--Bures information matrix elements of evolved quantum Boltzmann machines. (a) Quantum circuit involved in the estimation of $I_{ij}^{\operatorname{FB}}(\theta)$; (b) quantum circuit involved in the estimation of $I_{ij}^{\operatorname{FB}}(\phi)$; (c) quantum circuit involved in the estimation of $I_{ij}^{\operatorname{FB}}(\theta,\phi)$.}
    \label{fig:FB-circuits}
\end{figure*}
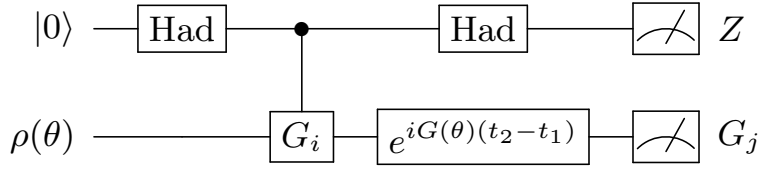
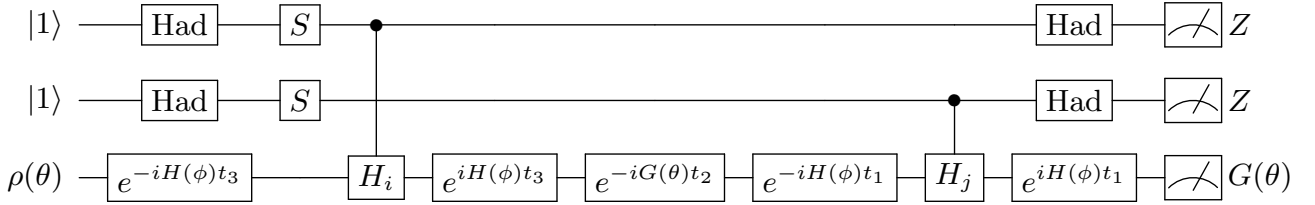
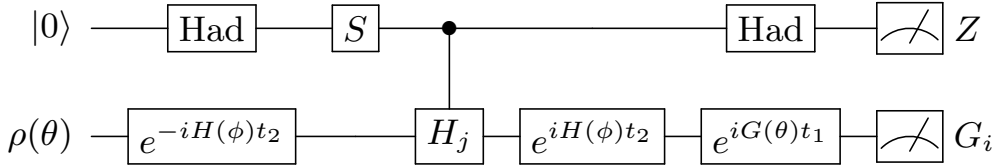

\subsection{Wigner--Yanase information matrix of evolved quantum Boltzmann machines}

\label{sec:WY-info}

The main results of this section are Theorems~\ref{thm:WY-theta},~\ref{thm:WY-phi}, and~\ref{thm:WY-theta-phi}, which provide explicit analytical expressions for the Wigner--Yanase information matrix elements of the parameterized family in~\eqref{eq:ansatz}. Additionally, we give quantum circuits that play a key role in efficiently estimating the terms in these expressions (see Figure~\ref{fig:WY-circuits}). We also provide a unique quantum algorithm for estimating the first term of~\eqref{eq:WY-theta}, which makes use of the canonical purification of a quantum Boltzmann machine. 

In order to obtain our findings here, we appeal to Proposition~\ref{prop:FB-WY-canonical-purifications} and, as such, we consider the canonically purified evolved quantum Boltzmann machine, defined as
\begin{align}
|\psi(\theta,\phi)\rangle & \coloneqq \left( \sqrt{\omega(\theta,\rho)}  \otimes
I\right)  |\Gamma\rangle\\
& = \left(  e^{-iH(\phi)} \sqrt{\rho(\theta)} \ e^{iH(\phi)} \otimes
I\right)  |\Gamma\rangle,\label{eq:purified-evolved-QBM}
\end{align}
where $\rho(\theta)$ is the parameterized thermal state defined in~\eqref{eq:param-thermal-state}, $H(\phi)$ is defined in~\eqref{eq:H}, and $\ket{\Gamma} \coloneqq \sum_k \ket{k}\ket{k}$ is the  maximally entangled vector. Observe that $|\psi(\theta,\phi)\rangle$ indeed purifies $\omega(\theta,\phi)$ because
\begin{equation}
    \omega(\theta,\phi) = \operatorname{Tr}_2[|\psi(\theta,\phi)\rangle\!\langle\psi(\theta,\phi)|].
\end{equation}

We have the following result for the partial derivatives of~\eqref{eq:purified-evolved-QBM}:
\begin{theorem}\label{thm:WY-partial-derivatives}
The partial derivatives for the parameterized family of canonically purified evolved quantum Boltzmann machines in~\eqref{eq:purified-evolved-QBM} are as follows:
\begin{align}
    \begin{split}
        |\partial_{\theta_j}\psi(\theta,\phi)\rangle 
        &= \frac{1}{2}\left\langle G_{j}\right\rangle _{\rho(\theta)}
        |\psi(\theta,\phi)\rangle \\
        &\hspace{-1.5cm} -\frac{1}{4}\left(  e^{-iH(\phi)}\left\{  \Phi_{\frac{\theta}{2}}
        (G_{j}),\sqrt{\rho(\theta)}\right\} e^{iH(\phi)} \otimes I\right)  |\Gamma\rangle,
    \end{split}\label{eq:part_der_theta_purified-evolved-QBM} \\
    |\partial_{\phi_i}\psi(\theta,\phi)\rangle 
    &= i\left(  \left[\sqrt{\omega(\theta,\phi)}, \Psi^{\dagger}_{\phi}(H_{i}) \right] \otimes I\right)  |\Gamma\rangle,\label{eq:part_der_phi_purified-evolved-QBM}
\end{align}
where we used the shorthands $\partial_{\theta_j} \equiv \frac{\partial}{\partial \theta_j}$ and $\partial_{\phi_i} \equiv \frac{\partial}{\partial \phi_i}$.
\end{theorem}

\begin{proof}
See Appendix~\ref{proof:WY-partial-derivatives}.
\end{proof}

\medskip 

The elements of the Wigner--Yanase information matrix for the parameter vector $\theta$ are as follows:

\begin{theorem}
\label{thm:WY-theta}
For the parameterized family in~\eqref{eq:ansatz}, the Wigner--Yanase information matrix elements, with respect to the $\theta$ parameter vector, are as follows:
\begin{align}
    \begin{split}
        I_{ij}^{\operatorname{WY}}(\theta) &= \frac{1}{2}\operatorname{Tr}\!\left[  \Phi_{\frac{\theta}{2}}(G_{i})
        \sqrt{\rho(\theta)}\Phi_{\frac{\theta}{2}}(G_{j})\sqrt{\rho(\theta)}\right] \\
        & \hspace{-1.37cm} +\frac{1}{4}\left\langle \left\{  \Phi_{\frac{\theta}{2}}(G_{i}),
        \Phi_{\frac{\theta}{2}}(G_{j})\right\}  \right\rangle _{\rho(\theta)} 
        - \left\langle G_{i}\right\rangle _{\rho(\theta)}\left\langle G_{j}\right\rangle _{\rho(\theta)}.
    \end{split}\label{eq:WY-theta}
\end{align}
\end{theorem}

\begin{proof}
See Appendix~\ref{proof:WY-theta}.
\end{proof}

\medskip

Note that we can rewrite the Wigner--Yanase
information matrix with respect to the parameter vector $\theta$ as
\begin{align}
\begin{split}
I_{ij}^{\operatorname{WY}}(\theta) & = - \left\langle G_{i}\right\rangle
_{\rho(\theta)}\left\langle G_{j}\right\rangle _{\rho(\theta)} \\
 & \hspace{-0.9cm} + \frac{1}{4}\operatorname{Tr}\!\left[  \left\{ \Phi_{\frac{\theta}{2}}
(G_{i}),\sqrt{\rho(\theta)}\right\}  \left\{  \Phi_{\frac{\theta}{2}}
(G_{j}),\sqrt{\rho(\theta)}\right\}  \right].
\end{split}
\end{align}

The last two terms of~\eqref{eq:WY-theta} can be easily estimated via a combination of classical sampling, Hamiltonian evolution, and the Hadamard test. The quantum circuit for estimating the second term of~\eqref{eq:WY-theta} is depicted in Figure~\ref{fig:WY-theta}, and the complete algorithm used to estimate this term is given in Appendix~\ref{app:WY-theta}. We defer the estimation of the first term of~\eqref{eq:WY-theta}  to Section~\ref{sec:1st-term-WY-theta} below.

The elements of the Wigner--Yanase information matrix for the parameter vector $\phi$ are as follows:

\begin{theorem}
\label{thm:WY-phi}
For the parameterized family in~\eqref{eq:ansatz}, the Wigner--Yanase information matrix elements, with respect to the $\phi$ parameter vector, are as follows:
\begin{align}\label{eq:WY-phi}
\begin{split}
I_{ij}^{\operatorname{WY}}(\phi)& = - 8 \Tr\!\left[ \Psi_{\phi}(H_{j}) \sqrt{\rho(\theta)} \Psi_{\phi}(H_{i}) \sqrt{\rho(\theta)} \right] \\
& \hspace{0.5cm} + 4 \left\langle 
 \left\{ \Psi_{\phi}(H_{i}) , \Psi_{\phi}(H_{j}) \right\}   \right\rangle_{\rho(\theta)}.
\end{split}
\end{align}
\end{theorem}

\begin{proof}
See Appendix~\ref{proof:WY-phi}.
\end{proof}

\medskip

The second term in~\eqref{eq:WY-phi} can be estimated on a quantum computer via a combination of classical
sampling, Hamiltonian evolution, and the Hadamard test. The quantum circuit that plays a role in its estimation is depicted in Figure~\ref{fig:WY-phi-2}. Regarding the first term, it resembles that of~\eqref{eq:WY-theta}, allowing the same procedure described in~\ref{sec:1st-term-WY-theta} to be applied. A detailed procedure for this case is provided in Appendix~\ref{app:1st-term-WY-phi}.

\begin{remark}
    The Wigner–Yanase information matrix e\-le\-men\-ts, given in~\eqref{eq:WY-phi}, can be equivalently expressed as
    \begin{align}
        I^{\operatorname{WY}}_{ij}(\phi) & = - \, 4 \Tr\!\bigg[  \left[ \Psi_{\phi}(H_i), \sqrt{\rho(\theta)}\right] \left[ \Psi_{\phi}\!(H_j), \sqrt{\rho(\theta)}\right] \bigg].
    \end{align}
    Notably, this alternative form is consistent with the common expression for the Wigner–Yanase metric found in the literature~\cite{Gibilisco2003,Hansen-Metric_adjusted_skew_information}.
\end{remark}

The elements of the Wigner--Yanase information matrix for the cross terms involving $\theta$ and $\phi$ are as follows:

\begin{theorem}
\label{thm:WY-theta-phi}
For the parameterized family in~\eqref{eq:ansatz}, the Wigner--Yanase information matrix elements, with respect to the $\theta$ and $\phi$ parameter vectors, are as follows:
\begin{equation}\label{eq:WY-theta-phi}
I_{ij}^{\operatorname{WY}}(\theta, \phi)=i\left\langle \left[   \Phi_{\frac{\theta}{2}}(G_{j}), \Psi_{\phi
}(H_{i})\right]  \right\rangle _{\rho(\theta)}.
\end{equation}
\end{theorem}

\begin{proof}
See Appendix~\ref{proof:WY-theta-phi}.
\end{proof}

\medskip 

The quantity in~\eqref{eq:WY-theta-phi} can be estimated by means of the quantum circuit depicted in Figure~\ref{fig:WY-theta-phi}.

\begin{figure*}
    \centering
        \begin{subfigure}{\textwidth}
        \centering
        \scalebox{1.5}{ % Adjust scale to fit the column width
        \Qcircuit @C=1.0em @R=1.0em {
            \lstick{} &  \gate{e^{iG\left(\frac{\theta}{2}\right)t_1}} & \meter & \rstick{\hspace{-1.2em}G_i} \\
            \lstick{} & \gate{e^{-iG^T\!\left(\frac{\theta}{2}\right)t_2}} &   \meter & \rstick{\hspace{-1.2em}G^T_{j}}
            \inputgroupv{1}{2}{0.9em}{1.3em}{\hspace{-1.2em}\ket{\psi(\theta)}}
        }
        }
        \vspace{15pt}
        \caption{Quantum circuit that realizes an unbiased estimate of $\operatorname{Tr}\!\left[  \Phi_{\frac{\theta}{2}}(G_{i})\sqrt{\rho(\theta)}\Phi_{\frac{\theta}{2}}(G_{j})\sqrt{\rho(\theta)}\right]$. For each run of the circuit, the times $t_1$ and $t_2$ are sampled independently at random from the probability density $p(t)$ in~\eqref{eq:high-peak-tent-density}.}
        \label{fig:WY-theta-2}
    \end{subfigure}
    \vspace{10pt}
    
    \begin{subfigure}{\textwidth}
        \centering
        \scalebox{1.5}{ % Adjust scale to fit the column width
        \Qcircuit @C=1.0em @R=1.0em {
            \lstick{\ket{0}} & \gate{\operatorname{Had}} & \ctrl{1} &  \gate{\operatorname{Had}} & \meter & \rstick{\hspace{-1.2em}Z} \\
            \lstick{\rho(\theta)} & \qw & \gate{G_j} & \gate{e^{iG(\theta/2)(t_2-t_1)}} &   \meter & \rstick{\hspace{-1.2em}G_{i}}
        }
        }
        \vspace{15pt}
        \caption{Quantum circuit that realizes an unbiased estimate of $\frac{1}{2}\left\langle \left\{  \Phi_{\frac{\theta}{2}}(G_{i}),\Phi_{\frac{\theta}{2}}(G_{j})\right\}  \right\rangle _{\rho(\theta)}$. For each run of the circuit, the times $t_1$ and $t_2$ are sampled independently at random from the probability density $p(t)$ in~\eqref{eq:high-peak-tent-density}. For details of the algorithm, see Appendix~\ref{app:WY-theta}.}
        \label{fig:WY-theta}
    \end{subfigure}
    \vspace{10pt}

    \begin{subfigure}{\textwidth}
        \centering
        \scalebox{1.5}{ % Adjust scale to fit the column width
        \Qcircuit @C=1.0em @R=1.0em {
            \lstick{} &  \gate{e^{-iH(\phi)t_1}} & \meter & \rstick{\hspace{-1.2em}H_j} \\
            \lstick{} & \gate{e^{iH^T\!(\phi)t_2}} &   \meter & \rstick{\hspace{-1.2em}H^T_{i}}
            \inputgroupv{1}{2}{0.9em}{1.3em}{\hspace{-1.2em}\ket{\psi(\theta)}}
        }
        }
        \vspace{15pt}
        \caption{Quantum circuit that realizes an unbiased estimate of $\Tr\!\left[ \Psi_{\phi}(H_{j}) \sqrt{\rho(\theta)} \Psi_{\phi}(H_{i}) \sqrt{\rho(\theta)} \right]$. For each run of the circuit, the times $t_1$ and $t_2$ are sampled uniformly at random from $[0,1]$.}
        \label{fig:WY-phi-1}
    \end{subfigure}
    \vspace{10pt}
    
    \begin{subfigure}{\textwidth}
        \centering
        \scalebox{1.5}{ % Adjust scale
        \Qcircuit @C=1.0em @R=1.0em {
            \lstick{\ket{0}} & \gate{\operatorname{Had}} & \ctrl{1} &  \gate{\operatorname{Had}} & \meter & \rstick{\hspace{-1.2em}Z} \\
            \lstick{\rho(\theta)} & \gate{e^{-iH(\phi)t_1}} & \gate{H_j} & \gate{e^{iH(\phi)(t_1-t_2)}} &   \meter & \rstick{\hspace{-1.2em}H_{i}}
        }
        }
        \vspace{15pt}
        \caption{Quantum circuit that realizes an unbiased estimate of $\frac{1}{2}\left\langle \left\{  \Psi_{\phi}(H_{i}),\Psi_{\phi}(H_{j}) \right\}  \right\rangle _{\rho(\theta)}$. For each run of the circuit, the times $t_1$ and $t_2$ are sampled uniformly at random from $[0,1]$. For details of the algorithm, see Appendix~\ref{app:WY-phi-2}.}
        \label{fig:WY-phi-2}
    \end{subfigure}
    \vspace{10pt}
    
    \begin{subfigure}{\textwidth}
        \centering
        \scalebox{1.5}{ % Adjust scale
        \Qcircuit @C=1.0em @R=1.0em {
            \lstick{\ket{0}} & \gate{\operatorname{Had}} & \gate{S} & \ctrl{1} &  \qw & \gate{\operatorname{Had}} & \meter & \rstick{\hspace{-1.2em}Z} \\
            \lstick{\rho(\theta)} & \gate{e^{-iH(\phi)t_2}} & \qw & \gate{H_i} & \gate{e^{iH(\phi)t_2}} & \gate{e^{iG(\theta/2)t_1}} &  \meter & \rstick{\hspace{-1.2em}G_{j}}
        }
        }
        \vspace{15pt}
        \caption{Quantum circuit that realizes an unbiased estimate of $\frac{i}{2}\left\langle \left[   \Phi_{\frac{\theta}{2}}(G_{j}), \Psi_{\phi}(H_{i})\right]  \right\rangle _{\rho(\theta)}$. For each run of the circuit, the time $t_1$ is sampled independently at random from the probability density $p(t)$ in~\eqref{eq:high-peak-tent-density}, and the time $t_2$ is sampled uniformly at random from $[0,1]$. For details of the algorithm, see Appendix~\ref{app:WY-theta-phi}.}
        \label{fig:WY-theta-phi}
    \end{subfigure}
    \caption{Quantum circuits involved in the estimation of the Wigner--Yanase information matrix elements. (a)-(b) Quantum circuits involved in the estimation of $I_{ij}^{\operatorname{WY}}(\theta)$; (c) quantum circuit involved in the estimation of $I_{ij}^{\operatorname{WY}}(\phi)$; (d) quantum circuit involved in the estimation of $I_{ij}^{\operatorname{WY}}(\theta,\phi)$.}
    \label{fig:WY-circuits}
\end{figure*}

\subsubsection{Evaluating the first term of the $\theta$-Wigner--Yanase information matrix}

\label{sec:1st-term-WY-theta}

In order to estimate the first term of~\eqref{eq:WY-theta}, we assume that one  has access to the canonical purification
$|\psi(\theta)\rangle$ of a quantum Boltzmann machine, defined as
\begin{equation}
    |\psi(\theta)\rangle \coloneqq \left(  \sqrt{\rho(\theta)}\otimes
I\right)  |\Gamma\rangle.
\label{eq:canon-pure-QBM}
\end{equation}
This is also known as a thermofield double state \cite[Eq.~(1.4)]{Cottrell2019}. 
Since many quantum algorithms for thermal-state preparation actually prepare this canonical purification~\cite{Holmes2022quantumalgorithms,chen2023q_Gibbs_sampl,rouze2024efficientthermalization} (see Appendix~\ref{app:canon-pure-thermofield-double} for further details of this point), this assumption is just as reasonable as our assumption of having sample access to the thermal state $\rho(\theta)$. 
Under this assumption, the following identity implies that one can estimate the first term of~\eqref{eq:WY-theta} efficiently:
\begin{align}
\begin{split}
&\operatorname{Tr}\!\left[  \Phi_{\frac{\theta}{2}}(G_{i}
)\sqrt{\rho(\theta)}\Phi_{\frac{\theta}{2}}(G_{j})\sqrt{\rho(\theta)}\right]  \\
& = \langle\psi(\theta)|\left(  \Phi_{\frac{\theta}{2}}(G_{i})\otimes\left[
\Phi_{\frac{\theta}{2}}(G_{j})\right]  ^{T}\right)  |\psi(\theta)\rangle.
\label{eq:identity-WY-theta-canon-pur}
\end{split}
\end{align}
The identity in~\eqref{eq:identity-WY-theta-canon-pur} follows because
\begin{align}
&  \langle\psi(\theta)|\left(  \Phi_{\frac{\theta}{2}}(G_{i})\otimes\left[
\Phi_{\frac{\theta}{2}}(G_{j})\right]  ^{T}\right)  |\psi(\theta)\rangle\nonumber\\
&  =\langle\Gamma|\left(  \sqrt{\rho(\theta)}\Phi_{\frac{\theta}{2}}
(G_{i})\sqrt{\rho(\theta)}\otimes\left[  \Phi_{\frac{\theta}{2}}
(G_{j})\right]  ^{T}\right)  |\Gamma\rangle\\
&  =\langle\Gamma|\left(  \sqrt{\rho(\theta)}\Phi_{\frac{\theta}{2}}
(G_{i})\sqrt{\rho(\theta)}\Phi_{\frac{\theta}{2}}(G_{j})\otimes I\right)
|\Gamma\rangle\\
&  =\operatorname{Tr}\!\left[  \sqrt{\rho(\theta)}\Phi_{\frac{\theta}{2}}
(G_{i})\sqrt{\rho(\theta)}\Phi_{\frac{\theta}{2}}(G_{j})\right]  .
\end{align}
The second equality follows from the transpose trick \cite[Exercise~3.7.12]{Wbook17}.
Thus, in order to estimate the right-hand side of~\eqref{eq:identity-WY-theta-canon-pur}, we need to be able to measure the expectation of the operator $\left[  \Phi_{\frac
{\theta}{2}}(G_{j})\right]  ^{T}$. Consider that
\begin{align}
& \left[  \Phi_{\frac{\theta}{2}}(G_{j})\right]  ^{T} \notag \\ & =\left[  \int
_{\mathbb{R}}dt\ p(t)\ e^{-iG(\theta/2)t}G_{j}e^{iG(\theta/2)t}\right]
^{T}\\
& =\int_{\mathbb{R}}dt\ p(t)\ \left[  e^{-iG(\theta/2)t}G_{j}e^{iG(\theta
/2)t}\right]  ^{T}\\
& =\int_{\mathbb{R}}dt\ p(t) \left(  e^{iG(\theta/2)t}\right)^{T}\!G_{j}
^{T}\!\left(  e^{-iG(\theta/2)t}\right)  ^{T}\\
& =\int_{\mathbb{R}}dt\ p(t)\ e^{iG^{T}(\theta/2)t}G_{j}^{T}e^{-iG^{T}
(\theta/2)t} \, .
\end{align}
If each $G_{j}$ is a Pauli string, this is easy to implement by noting that
$I^{T}=I$, $\sigma_{X}^{T}=\sigma_{X}$, $\sigma_{Y}^{T}=-\sigma_{Y}$, and
$\sigma_{Z}^{T}=\sigma_{Z}$.
Then, adopting the shorthand $\psi(\theta)\equiv|\psi(\theta)\rangle\!
\langle\psi(\theta)|$ and applying the definition of $\Phi_{\theta}$ in~\eqref{eq:Phi}
and cyclicity and linearity of trace, consider that
\begin{align}
& \langle\psi(\theta)|\Phi_{\frac{\theta}{2}}(G_{i})\otimes\left[  \Phi
_{\frac{\theta}{2}}(G_{i})\right]  ^{T}|\psi(\theta)\rangle \notag \\
& =\operatorname{Tr}\!\left[  \left(  \Phi_{\frac{\theta}{2}}(G_{i}
)\otimes\left[  \Phi_{\frac{\theta}{2}}(G_{i})\right]  ^{T}\right)
\psi(\theta)\right]  \\
& =\mathbb{E}_{\tau_{1},\tau_{2}}\!\left[  \operatorname{Tr}\!\left[  \left(
G_{i}\otimes G_{j}^{T}\right)  \mathcal{U}_{\tau_{1},\tau_{2}}(\psi
(\theta))\right]  \right]  ,
\end{align}
where $\tau_{1}$ and $\tau_{2}$ are independent random variables each chosen
according to the high-peak tent probability density $p(t)$ in~\eqref{eq:high-peak-tent-density} and
$\mathcal{U}_{\tau_{1},\tau_{2}}$ is the following unitary channel:
\begin{align}
\begin{split}
\mathcal{U}_{\tau_{1},\tau_{2}}(Y) & \coloneqq \left(  e^{iG\left(  \frac{\theta}
{2}\right)  \tau_{1}}\otimes e^{-iG^{T}\left(  \frac{\theta}{2}\right)
\tau_{2}}\right)  Y \\
& \hspace{1.5cm} \times\left(  e^{-iG\left(  \frac{\theta}{2}\right)  \tau_{1}
}\otimes e^{iG^{T}\left(  \frac{\theta}{2}\right)  \tau_{2}}\right)  .
\end{split}
\end{align}

Thus, a quantum algorithm for estimating the first term of~\eqref{eq:WY-theta} consists of
repeating the following steps and averaging: prepare the canonical
purification $\psi(\theta)$ in~\eqref{eq:canon-pure-QBM}, pick $\tau_{1}$ and $\tau_{2}$ independently at
random according to~\eqref{eq:high-peak-tent-density}, apply the Hamiltonian evolution $\mathcal{U}
_{\tau_{1},\tau_{2}}$ to $\psi(\theta)$, and measure the observable
$G_{i}\otimes G_{j}^{T}$. The respective quantum circuit is shown in Figure~\ref{fig:WY-theta-2}.

\subsection{Kubo–Mori information matrix of evolved quantum Boltzmann machines}\label{sec:KM-info}

The main results of this section are Theorems~\ref{thm:KM-theta},~\ref{thm:KM-phi}, and~\ref{thm:KM-theta-phi}, which provide explicit analytical expressions for the Kubo--Mori information matrix elements of the parameterized family in~\eqref{eq:ansatz}. Additionally, we give quantum circuits that play a key role in efficiently estimating the terms in these expressions (see Figure~\ref{fig:KM-circuits}).

\begin{theorem}
\label{thm:KM-theta}
For the parameterized family in~\eqref{eq:ansatz}, the Kubo--Mori information matrix elements, with respect to the $\theta$
parameters, are as follows:
\begin{equation}\label{eq:KM-theta}
I_{ij}^{\operatorname{KM}}(\theta) = \frac{1}{2} \left\langle \left\{ G_i, \Phi_{\theta}(G_{j}) \right\} \right\rangle\!_{\rho(\theta)} - \left\langle G_{i}\right\rangle \!_{\rho(\theta)} \left\langle G_{j}\right\rangle \!_{\rho(\theta)}.
\end{equation}
\end{theorem}

\begin{proof}
    See Appendix~\ref{proof:KM-theta}.
\end{proof}

\medskip 

This is the same result obtained in~\cite{patel2024naturalgradientparameterestimation} in the case of using parameterized thermal states. $I_{ij}^{\operatorname{KM}}(\theta)$ can be efficiently estimated on a quantum computer (under the assumption that each $G_i$ is a local operator, acting on a constant number of qubits)~\cite{patel2024naturalgradientparameterestimation}. We can use the quantum circuit shown in Figure~\ref{fig:KM-theta} for estimating the first term in~\eqref{eq:KM-theta}. The second term in~\eqref{eq:KM-theta} is the same as the second term in~\eqref{eq:FB-theta}. As such, we can use the same procedure delineated in the paragraph surrounding~\eqref{eq:simple-term-to-estimate} in order to estimate it.

The elements of the Kubo--Mori information matrix for the parameter vector $\phi$ are as follows:

\begin{theorem}
\label{thm:KM-phi}
For the parameterized family in~\eqref{eq:ansatz}, the Kubo--Mori information matrix elements, with respect to the $\phi$
parameter vector, are as follows:
\begin{align}\label{eq:KM-phi}
I_{ij}^{\operatorname{KM}}(\phi)= \left\langle \left[  \Psi_{\phi}(H_{i}), \left[  G(\theta) , \Psi_{\phi}
(H_{j})\right]  \right]
\right\rangle _{\rho(\theta)}.
\end{align}
\end{theorem}

\begin{proof}
    See Appendix~\ref{proof:KM-phi}.
\end{proof}

\medskip

The quantity in~\eqref{eq:KM-phi} can be estimated on a quantum computer by a generalization~\cite[Algorithm~3]{liEfficientQuantumGradient2024} of the standard Hadamard test, which  evaluates the expectation value of nested commutators. The quantum circuit used in this case is depicted in Figure~\ref{fig:KM-phi}.

The elements of the Kubo--Mori information matrix for the cross terms involving $\theta$ and $\phi$ are as follows:
\begin{theorem}
\label{thm:KM-theta-phi}
For the parameterized family in~\eqref{eq:ansatz}, the Kubo--Mori information matrix elements, with respect to $\theta$ and $\phi$ parameter vectors, are as follows:
\begin{equation}
\label{eq:KM-theta-phi}
I_{ij}^{\operatorname{KM}}(\theta, \phi) =  \frac{i}{2} \left\langle  \left\{  \Phi_{\theta}(G_{i}) , \left[G(\theta)  , \Psi_{\phi}(H_j)\right] \right\} \right\rangle_{\rho(\theta)} .
\end{equation}
\end{theorem}

\begin{proof}   
See Appendix~\ref{proof:KM-theta-phi}.
\end{proof}

\medskip 

The quantity in~\eqref{eq:KM-theta-phi} can also be estimated on a quantum computer by a generalization~\cite[Algorithm~3]{liEfficientQuantumGradient2024} of the standard Hadamard test, which  evaluates the expectation value of nested anticommutators and commutators. The quantum circuit used in this case is depicted in Figure~\ref{fig:KM-theta-phi}.

\begin{figure*}
    \centering
    \begin{subfigure}{\textwidth}
        \centering
        \scalebox{1.5}{ % Adjust scale to fit the column width
        \Qcircuit @C=1.0em @R=1.0em {
            \lstick{|0\rangle} & \gate{\operatorname{Had}} & \ctrl{1} &  \gate{\operatorname{Had}} & \meter & \rstick{\hspace{-1.2em}Z} \\
            \lstick{\rho(\theta)} & \qw & \gate{G_j} & \gate{e^{-iG(\theta)t}} &   \meter & \rstick{\hspace{-1.2em}G_{i}}
        }
        }
        \vspace{20pt}
        \caption{Quantum circuit that realizes an unbiased estimate of $\frac{1}{2} \left\langle \left\{ G_i, \Phi_{\theta}(G_{j}) \right\}\right\rangle_{ \rho(\theta)} $. For each run of the circuit, $t$ is sampled independently at random from the probability density $p(t)$ in~\eqref{eq:high-peak-tent-density}. For details of the algorithm, see Appendix~\ref{app:KM-theta}.}
        \label{fig:KM-theta}
    \end{subfigure}
    \vspace{10pt}
    
    \begin{subfigure}{\textwidth}
        \centering
        \scalebox{1.5}{ % Adjust scale
        \Qcircuit @C=0.9em @R=1.0em {
            \lstick{\ket{1}} & \gate{\operatorname{Had}} & \gate{S} & \ctrl{2} & \qw & \qw &  \gate{\operatorname{Had}} & \meter &\rstick{\hspace{-1.2em}Z} \\
            \lstick{\ket{1}} & \gate{\operatorname{Had}} & \gate{S} & \qw & \qw & \ctrl{1} & \gate{\operatorname{Had}} & \meter & \rstick{\hspace{-1.2em}Z}\\
            \lstick{\rho(\theta)} &  \gate{e^{-iH(\phi)t_2}} & \qw & \gate{H_{i}} & \gate{e^{iH(\phi)(t_2-t_1)}}  & \gate{H_{j}} & \gate{e^{iH(\phi)t_1}} &  \meter & \rstick{\hspace{-1.2em}G(\theta)}
        }
        }
        \vspace{20pt}
        \caption{Quantum circuit that realizes an unbiased estimate of $ \frac{1}{4} \left\langle\left[ \left[\Psi_{\phi}(H_{j}) ,G(\theta) \right],\Psi_{\phi}(H_{i})   \right]\right\rangle_{ \rho(\theta)} $. For each run of the circuit, $t_1$ and $t_2$ are sampled uniformly at random from $[0,1]$. For details of the algorithm, see Appendix~\ref{app:KM-phi}.}
        \label{fig:KM-phi}
    \end{subfigure}
    \vspace{10pt}
    
    \begin{subfigure}{\textwidth}
        \centering
        \scalebox{1.5}{ % Adjust scale
        \Qcircuit @C=0.9em @R=1.0em {
            \lstick{\ket{1}} & \gate{\operatorname{Had}} & \qw & \ctrl{2} & \qw & \qw & \qw & \gate{\operatorname{Had}} & \meter &\rstick{\hspace{-1.2em}Z} \\
            \lstick{\ket{1}} & \gate{\operatorname{Had}} & \gate{S} & \qw & \qw & \qw & \ctrl{1} & \gate{\operatorname{Had}} & \meter & \rstick{\hspace{-1.2em}Z}\\
            \lstick{\rho(\theta)} &  \qw & \qw & \gate{G_{i}} & \gate{e^{-iG(\theta)t_1}} & \gate{e^{-iH(\phi)t_2}} & \gate{H_{j}} & \gate{e^{iH(\phi)t_2}} & \meter & \rstick{\hspace{-1.2em}G(\theta)}
        }
        }
        \vspace{20pt}
        \caption{Quantum circuit that realizes an unbiased estimate of $\frac{i}{4} \left\langle \left\{  \Phi_{\theta}(G_{i}) , \left[G(\theta) , \Psi_{\phi}(H_j)  \right] \right\}\right\rangle_{ \rho(\theta)} $. For each run of the circuit, $t_1$ is sampled independently at random from the probability density $p(t)$ in~\eqref{eq:high-peak-tent-density}, and $t_2$ is sampled uniformly at random from $[0,1]$. For details of the algorithm, see Appendix~\ref{app:KM-theta-phi}.}
        \label{fig:KM-theta-phi}
    \end{subfigure}
    \caption{Quantum circuits involved in the estimation of the Kubo--Mori information matrix elements. (a) Quantum circuit involved in the estimation of $I_{ij}^{\operatorname{KM}}(\theta)$; (b) quantum circuit involved in the estimation of $I_{ij}^{\operatorname{KM}}(\phi)$; (c) quantum circuit involved in the estimation of $I_{ij}^{\operatorname{KM}}(\theta,\phi)$.}
    \label{fig:KM-circuits}
\end{figure*}
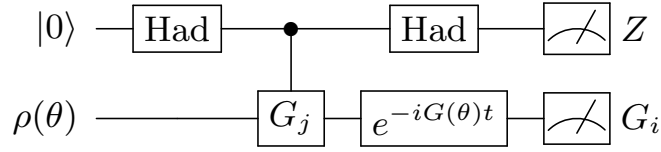
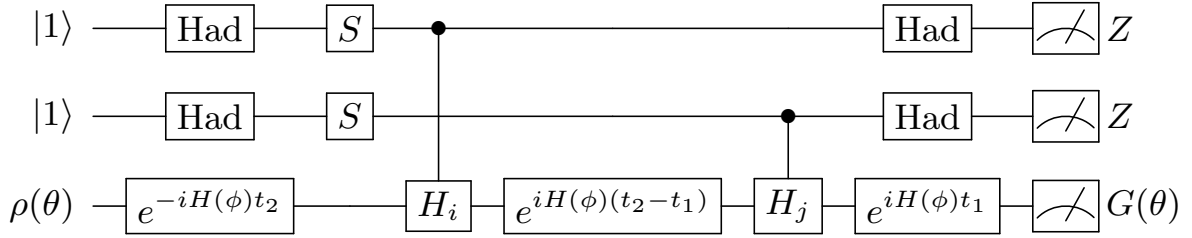
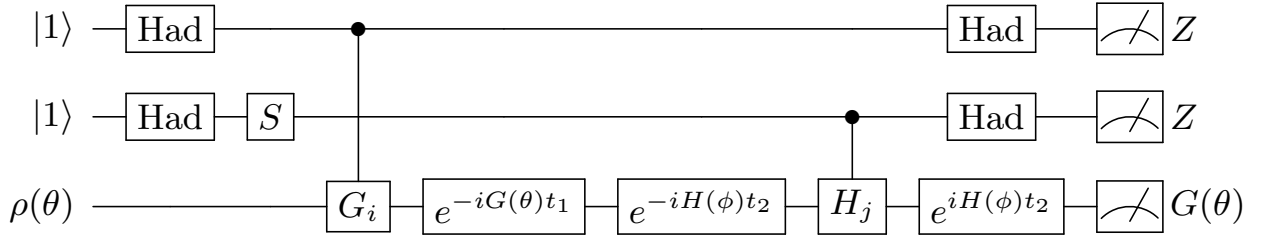

\section{Applications of information matrices of evolved quantum Boltzmann machines}

\label{sec:applications}

\subsection{Natural gradient for evolved quantum Boltzmann machine learning}

\label{sec:app-nat-grad}

The quantum generalizations of the Fisher information matrix considered in Section~\ref{sec:infor_matrix_eqbm} provide the foundation for a metric-aware optimization algorithm, known as  natural gradient descent, which can be effectively applied to evolved quantum Boltzmann machine learning. This type of optimization algorithm was introduced in the classical setting~\cite{amari1998nat_grad} and has been generalized to the quantum setting to account for various quantum generalizations of Fisher information~\cite{Stokes2020quantumnatural,sbahi2022provablyefficientvariationalgenerative,KS2022,sohail2024quantumnaturalstochasticpairwise,patel2024naturalgradientparameterestimation}. A convergence analysis of quantum natural gradient under certain assumptions has been presented in \cite{sohail2024quantumnaturalstochasticpairwise}.

Standard gradient descent, whose update rule is shown in~\eqref{eq:grad_descent}, relies on the Euclidean geometry of the parameter space. However, this geometry is typically not suited for the space of quantum states (see \cite[Section II.A]{patel2024naturalgradientparameterestimation}), resulting in slow convergence and difficulty escaping saddle points in the optimization landscape. To address these limitations, the natural gradient descent algorithm incorporates the geometry induced by the parameterization of quantum states~\cite{amari1998nat_grad}. For a parameterized family $(\sigma(\gamma))_{\gamma \in \mathbb{R}^L}$ and a loss function $\mathcal{L}(\gamma)$, the update rule of a quantum natural descent algorithm is given as follows:
\begin{equation}\label{eq:qngd}
    \gamma_{m+1}\coloneqq \gamma_m - \mu \left[ I^{\boldsymbol{D}}(\gamma_m) \right]^{-1} \nabla_\gamma \mathcal{L}(\gamma_m),
\end{equation}
where $I^{\boldsymbol{D}}(\gamma_m)$ is the $\boldsymbol{D}$-based Fisher information matrix defined in~\eqref{eq:D-based-Fisher-matrix}, which encodes the curvature of the parameter space, and $\mu$ is the learning rate. By incorporating the inverse of $I^{\boldsymbol{D}}(\gamma_m)$, the gradient is rescaled to account for the underlying geometry, enabling the optimization steps to align with the steepest descent direction with respect to the geometry defined by the smooth divergence $\boldsymbol{D}$, rather than the Euclidean geometry. This helps to navigate the optimization landscape more effectively, potentially avoiding getting trapped in local minima. 

Applying natural gradient to evolved quantum Boltzmann machines  involves estimating information matrices and computing their inverses. While this additional computation introduces some overhead, it is often offset by the improved convergence and ability to escape saddle points, potentially leading to fewer iterations and a faster overall optimization process~\cite{vSK2021}.
We have shown in Section~\ref{sec:infor_matrix_eqbm} how to evaluate the Fisher--Bures $I^{FB}(\gamma)$ (Section~\ref{sec:FB-info}), Wigner--Yanase $I^{WY}(\gamma)$ (Section~\ref{sec:WY-info}), and Kubo--Mori $I^{KM}(\gamma)$ (Section~\ref{sec:KM-info}) information matrix elements for evolved quantum Boltzmann machines. Thus, each of these quantum generalizations of the Fisher information matrix can be directly applied in the quantum natural gradient descent algorithm for evolved quantum Boltzmann machine learning by incorporating the chosen matrix into the update rule in~\eqref{eq:qngd}.

Let us also note that, due to Corollary~\ref{cor:WY-FB-ineqs}, the Fisher--Bures and Wigner--Yanase information matrices yield optimization steps that differ only by a constant factor. This equivalence makes these two information matrices essentially interchangeable for practical implementations of quantum natural gradient, offering flexibility in choosing the metric without affecting the overall optimization trajectory. Such flexibility is particularly valuable when computational constraints favor one metric over the other.

\subsection{Fundamental limitations on estimating time-evolved thermal states}

\label{sec:app-estimating}

In this section, we discuss how our findings are relevant for multiparameter estimation, in particular, to estimating the parameters of a time-evolved thermal state of the form in \eqref{eq:ansatz}. Previous studies have focused on parameter estimation of time-evolved states \cite{BD2016,Ho2023} or thermal states \cite{GarciaPintos2024ham_learning_qbm,abiuso2024limitsmetrologythermal}, but to the best of our knowledge, the problem of estimating the parameters of time-evolved thermal states has not been studied previously.

To briefly review the problem and similar to the review provided in \cite[Section~I-C-2]{patel2024naturalgradientparameterestimation}, consider that the following
multiparameter Cramer--Rao bound holds for an arbitrary unbiased estimator and
for a general parameterized family $(\sigma(\gamma))_{\gamma\in\mathbb{R}^{L}
}$ of states:
\begin{equation}
\operatorname{Cov}^{(n)}(\hat{\gamma},\gamma)\geq\frac{1}{n}\left[  I^{\operatorname{FB}
}(\gamma)\right]  ^{-1},
\label{eq:Cramer--Rao-multiple}
\end{equation}
where $n\in\mathbb{N}$ is the number of copies of the state $\sigma(\gamma)$
available, $\hat{\gamma}$ is an estimate of the parameter vector~$\gamma$, the matrix $I^{\operatorname{FB}}(\gamma)$ denotes the Fisher--Bures
information matrix, and the covariance matrix
$\operatorname{Cov}^{(n)}(\hat{\gamma},\gamma)$ measures errors in estimation and is
defined in terms of its matrix elements as
\begin{multline}
\lbrack\operatorname{Cov}^{(n)}(\hat{\gamma},\gamma)]_{k,\ell}\coloneqq\\
\sum_{m}\operatorname{Tr}[M_{m}^{(n)}\sigma(\gamma)^{\otimes n}](\hat{\gamma
}_{k}(m)-\gamma_{k})(\hat{\gamma}_{\ell}(m)-\gamma_{\ell}).
\label{eq:cov-mat-def}
\end{multline}
In \eqref{eq:cov-mat-def}, $(M_{m}^{(n)})_{m}$ is an arbitrary positive operator-value measure used for estimation, i.e.,
satisfying $M_{m}^{(n)}\geq0$ for all $m$ and $\sum_{m}M_{m}^{(n)}=I^{\otimes n}$. This measurement acts, in general, collectively on all $n$ copies of the
state $\sigma(\gamma)^{\otimes n}$. Additionally,
\begin{equation}
\hat{\gamma}(m)\coloneqq(\hat{\gamma}_{1}(m),\hat{\gamma}_{2}(m),\ldots
,\hat{\gamma}_{J}(m))
\label{eq:parameter-estimate-function}
\end{equation}
is a function that maps the measurement outcome $m$ to an estimate
$\hat{\gamma}(m)$ of the parameter vector $\gamma$. The inequality in \eqref{eq:Cramer--Rao-multiple} exploits the additivity of the Fisher--Bures information matrix, as reviewed in \cite[Appendix~A]{patel2024naturalgradientparameterestimation}. As noted in \cite[Eq.~(C11)]{sbahi2022provablyefficientvariationalgenerative}, the multiparameter Cramer--Rao bound in \eqref{eq:Cramer--Rao-multiple} can be written as follows:
\begin{equation}
\begin{pmatrix}
    \operatorname{Cov}^{(n)}(\hat{\gamma},\gamma) & I \\
    I & n  I^{\operatorname{FB}}(\gamma)
\end{pmatrix} \geq 0,
\end{equation}
which is a direct consequence of the Schur complement lemma.

There are several implications of our findings for parameter estimation of time-evolved thermal states:
\begin{enumerate}

    \item Our analytical expressions for the Fisher--Bures information matrix from Theorems~\ref{thm:FB-theta}, \ref{thm:FB-phi}, and \ref{thm:FB-theta-phi} can be plugged directly into \eqref{eq:Cramer--Rao-multiple} in order to obtain fundamental limits on the performance of an arbitrary  scheme for estimating time-evolved thermal states. Future work could conduct numerical studies of various schemes and examples of time-evolved thermal states in order to determine how close such schemes come to the fundamental limits established here. Moreover, one could explore the estimation of time-evolved bosonic Gaussian thermal states, as a generalization of the setting recently considered in~\cite{huang2024informationgeometrybosonicgaussian}.
    
    \item Similar in spirit to the main application of \cite{Ho2023}, it could be the case that the analytical expressions might be difficult to evaluate or even difficult computationally for a classical algorithm to approximate. In this case, our quantum algorithms could be helpful: one could perform them on time-evolved thermal states in order to estimate their Fisher--Bures information matrix elements, in order to have an understanding of the fundamental limits. Here our results also imply a broad generalization of the main findings and application of \cite{Ho2023}, given that our expressions and algorithms apply to  time-evolved mixed states (i.e., quantum evolution machines), whereas the results of \cite{Ho2023} apply exclusively to time-evolved pure states.
    
    \item Corollary~\ref{cor:WY-FB-ineqs} is useful for parameter estimation, because it indicates that the Fisher--Bures and Wigner--Yanase information matrices differ only by a factor of two in the matrix (Loewner) order, and thus one can be used as a substitute for the other while giving similar bounds in the low-error regime. This is also advantageous in the case that the Fisher--Bures information matrix elements are difficult to evaluate but the Wigner--Yanase information matrix elements are not.

    \item As discussed in \cite[Appendix~D]{sbahi2022provablyefficientvariationalgenerative} (specifically, Theorem~D.4 therein) by building on findings from the classical case \cite[Theorem~2]{amari1998nat_grad}, the natural gradient descent algorithm, under certain conditions, attains the Cramer--Rao bound asymptotically and thus is an optimal algorithm for asymptotic parameter estimation. Here we also note the similarity of natural gradient descent and the iterative scoring algorithm mentioned in~\cite[Eq.~(121)]{Sidhu2020}. As such, one could investigate certain classes of time-evolved thermal states to determine whether the Cramer--Rao bound could be attained asymptotically for them, by means of natural gradient descent.  However, based on \cite[Theorem~3]{patel2024naturalgradientparameterestimation}, we do not expect the aforementioned conditions for asymptotic optimality to be satisfied generally. Regardless, one could still possibly make effective use of our analytical expressions and quantum algorithms in order to devise an estimation strategy for time-evolved thermal states.
\end{enumerate}

\section{Quantum Boltzmann machines and quantum evolution machines as special cases}

\label{sec:special-cases}

In this section, we remark briefly on how quantum Boltzmann machines and quantum evolution machines are special cases of evolved quantum Boltzmann machines.

By fixing the parameter vector $\phi$ and allowing $\theta$ to vary, we obtain the following parameterized family of states:
\begin{equation}
    (\omega(\theta, \phi))_{\theta \in \mathbb{R}^J},
\end{equation}
where the state $\omega(\theta, \phi)$ is defined in~\eqref{eq:ansatz}. When $\phi = 0$, the resulting state is a parameterized thermal state or, equivalently, a quantum Boltzmann machine. As such, all of our findings in this paper apply to quantum Boltzmann machines. Even if $\phi \neq 0$, the resulting state is essentially a quantum Boltzmann machine, up to a fixed unitary evolution $e^{-iH(\phi)}$. All of our previous developments apply to this special case and recover previously reported results from~\cite{patel2024quantumboltzmannmachine,patel2024naturalgradientparameterestimation}. As such, one can perform gradient descent using the gradients reported in~\eqref{eq:grad_wrt_theta},~\eqref{eq:VQE-grad-theta}, and~\eqref{eq:gen_model_der_theta}. Also, one can perform natural gradient descent by using the information matrices in~\eqref{eq:FB-theta},~\eqref{eq:WY-theta}, or~\eqref{eq:KM-theta} (see also the caption of Table~\ref{table:FB-WY-KM-results}). Furthermore, we notice that all of the information matrices in these last referenced equations have no dependence on the parameter vector~$\phi$, consistent with the fact that they are unitarily invariant under the action of a unitary channel (Corollary~\ref{cor:unitary-inv-Fisher-info}), which in this case is $(\cdot) \to e^{-iH(\phi)} (\cdot)e^{iH(\phi)}$. 

By fixing the parameter vector $\theta$ and allowing $\phi$ to vary, we obtain the following parameterized family of states:
\begin{equation}
    (\omega(\theta, \phi))_{\phi \in \mathbb{R}^K},
\end{equation}
where the state $\omega(\theta, \phi)$ is defined in~\eqref{eq:ansatz}. We refer to this family as quantum evolution machines because they arise from the action of the parameterized unitary evolution $e^{-iH(\phi)}$ acting on the fixed state $\rho(\theta)$. All of our previous developments apply to this special case.  As such, one can perform gradient descent with quantum evolution machines by using the gradients reported in~\eqref{eq:grad_wrt_phi},~\eqref{eq:VQE-grad-phi}, and~\eqref{eq:gen_mod_grad_phi}. Also, one can perform natural gradient descent with quantum evolution machines by using the information matrices in~\eqref{eq:FB-phi},~\eqref{eq:WY-phi}, or~\eqref{eq:KM-phi} (see also the caption of Table~\ref{table:FB-WY-KM-results}).

\section{Conclusion}

\label{sec:conclusion}

\subsection{Summary and discussion}

In this paper, we propose evolved quantum Boltzmann machines as a variational ansatz for quantum optimization and learning. The main idea is captured by~\eqref{eq:ansatz}: beginning with two parameterized Hamiltonians $G(\theta)$ and $H(\phi)$, prepare a thermal state of $G(\theta)$ and then time-evolve it according to $H(\phi)$. One can alternatively think of it as imaginary time evolution according to $G(\theta)$, followed by real time evolution according to $H(\phi)$. In this way, the combination of imaginary-time and real-time dynamics gives rise to a broader class of expressible and trainable quantum states. Theorem~\ref{thm:gradient-eQBM} provides expressions for the gradient of the state $\omega(\theta, \phi)$ in~\eqref{eq:ansatz}. We subsequently applied them for the task of ground-state energy estimation in Section~\ref{sec:GSEE}, therein providing quantum algorithms for estimating the gradient (see Figure~\ref{fig:VQE-circuits}). We also considered the task of generative modeling (Section~\ref{sec:gen-mod}) and established expressions for the gradient for this task, as well as quantum algorithms for estimating it. As summarized in Table~\ref{table:FB-WY-KM-results}, we then established analytical expressions for the Fisher--Bures, Wigner--Yanase, and Kubo--Mori information matrices of evolved quantum Boltzmann machines. Figures~\ref{fig:FB-circuits},~\ref{fig:WY-circuits}, and~\ref{fig:KM-circuits} depict quantum algorithms that play an essential role in estimating their matrix elements, respectively. These results have applications in natural gradient descent algorithms when using evolved quantum Boltzmann machines---these algorithms are general purpose and can thus be employed for a broad variety of quantum optimization and learning tasks. Along the way, we also proved a broad generalization of the main finding of~\cite{Luo2004}, showing that the Fisher--Bures and Wigner--Yanase information matrices differ by no more than a factor of two in the matrix Loewner order, and are thus essentially interchangeable in natural gradient-descent algorithms.

Given that our paper develops full details of three different information matrices for evolved quantum Boltzmann machines, i.e., Fisher--Bures, Wigner--Yanase, and Kubo--Mori, a question arises as to which one is most suitable for a given application. While we have left this question somewhat open, at the least Corollary~\ref{cor:WY-FB-ineqs} indicates that the Fisher--Bures and Wigner--Yanase information matrices are essentially interchangeable for an algorithm like natural gradient descent. Given this, one would then employ whichever one of these is simpler to estimate. 
For the application of generative modeling, one might favor using the Kubo--Mori information matrix because the objective function is the quantum relative entropy, coinciding with the divergence that defines the Kubo--Mori information matrix (see Definition~\ref{def:FB-WY-KM-2nd-derivs}). Furthermore, the Kubo--Mori information matrix in this case is equal to the Hessian of the objective function in \eqref{eq:rel_entr_alt}, as observed from \cite[Theorem~2]{patel2024naturalgradientparameterestimation} and \cite[Eq.~(15), Supp.~Inf.]{Coopmans2024qbm_gen_learn}, indicating that natural gradient is equivalent to a second-order Newton search in this case. As such, the Kubo--Mori information matrix seems quite well aligned with the generative modeling problem, when compared to the other two information matrices.

\subsection{Future directions}

Going forward from here, there are several open directions and questions to address. First,
we have left it open to simulate the performance of evolved quantum Boltzmann machines for tasks of interest, such as ground-state energy estimation and generative modeling, with our main focus here being on developing the aforementioned fundamental theoretical findings. We plan to address this direction in future work. 
In parallel, recent work~\cite{wilde2025generativemodelingusingevolved} has delineated practical hybrid quantum--classical algorithms for training evolved quantum Boltzmann machines for generative modeling, by combining the evolved quantum Boltzmann gradient estimator with the Donsker--Varadhan variational representation of the relative entropy and related distinguishability measures. Building on this, it would be interesting to investigate how the information-geometric tools developed here can be incorporated into these training schemes, for example by replacing Euclidean gradient updates with natural-gradient steps or by extending the objective to broader families of distinguishability measures.
Beyond generative modeling, it is also open to employ evolved quantum Boltzmann machines as an ansatz for other optimization tasks, such as constrained Hamiltonian optimization or other semi-definite programming problems considered in~\cite{chen2023qslackslackvariableapproachvariational}.

An additional and important direction is to carry out a comparative study between evolved quantum Boltzmann machines and other existing variational ansatzes, in particular standard quantum Boltzmann machines. While evolved quantum Boltzmann machines clearly generalize both standard quantum Boltzmann machines and quantum evolution machines (as discussed in Section~\ref{sec:special-cases}), it remains to be determined whether this generalization translates into concrete performance advantages in practice. Future work will aim to benchmark these models under a fixed number of trainable parameters, analyzing expressivity versus trainability trade-offs, convergence speed, robustness to noise, and overall resource requirements. Such analysis will be crucial to positioning evolved quantum Boltzmann machines within the broader landscape of variational quantum algorithms for learning and optimization. 

In this context, the choice of the imaginary-time Hamiltonian $G(\theta)$ and the real-time evolution Hamiltonian 
$H(\phi)$ is expected to significantly influence the expressive power and optimization landscape of the evolved quantum Boltzmann machine. Richer or more structured choices of the Hamiltonians -- such as those incorporating multiple interaction terms or longer-range couplings -- can increase expressivity by enabling the ansatz to represent more complex correlations. However, greater expressivity often comes at the cost of introducing more trainable parameters, which can make optimization more challenging or less stable in practice. Conversely, simpler Hamiltonians may lead to easier training but at the expense of a more limited representational capacity. Exploring these trade-offs systematically, both analytically and numerically, remains an open and promising direction for future work.

Next, it is an open question to determine whether evolved quantum Boltzmann machines suffer from the barren plateau problem~\cite{McClean2018barren_plateaus, Cerezo2021barren_plateaus, Holmes2022barren_plateaus, marrero2021barrenp_lateaus}. The results of~\cite{Coopmans2024qbm_gen_learn} indicate that quantum Boltzmann machines are not subject to this problem for the generative modeling task, and so one should not encounter it when optimizing evolved quantum Boltzmann machines with respect to the $\theta$ parameter vector. It is open to determine whether the problem applies when optimizing with respect to the $\phi$ parameter vector. At the least, the standard unitary two-design argument~\cite{McClean2018barren_plateaus} for the onset of barren plateaus in parameterized quantum circuits does not seem to apply to evolved quantum Boltzmann machines, given that picking the elements of $\theta$ and $\phi$ randomly according to a multivariate Gaussian does not clearly lead to an averaged state that is maximally mixed. Further investigation is certainly required to determine whether this is the case, both for the generative modeling task, as well as for the ground-state energy estimation task.

Interestingly, our quantum algorithms for estimating the Fisher--Bures, Wigner--Yanase, and Kubo--Mori information matrix elements of evolved quantum Boltzmann machines rely on standard quantum subroutines -- specifically, the Hadamard test, classical random sampling, and Hamiltonian simulation -- and are efficient under the assumption that one can prepare samples of evolved quantum Boltzmann machines or, for the first term in~\eqref{eq:WY-theta}, their purifications. We remark that recent work~\cite{wilde2025quantumfisherinformationmatrices}, which develops broader quantum generalizations of the Fisher information matrix from Rényi relative entropies, has also provided a new expression for the Wigner--Yanase information matrix of time-evolved states that no longer requires access to purifications, thereby simplifying the estimation procedure presented here for quantum evolution machines.

As mentioned in~\eqref{eq:data-proc-D-based-Fisher}, the aforementioned information matrices obey the data-processing inequality for arbitrary quantum channels. These results stand in distinction to the complexity-theoretic barriers~\cite{watrous2002qszk,Wat06} in place for estimating other distinguishability measures like fidelity and trace distance, which obey the data-processing inequality but are not efficiently estimable in general. Thus, we have found distinguishability measures that both 1) obey the data-processing inequality and 2) are efficiently estimable on a quantum computer for a large class of states. As a future direction, we wonder whether there is an efficient algorithm for estimating distinguishability measures like quantum relative entropy, trace distance, and fidelity of evolved quantum Boltzmann machines. It seems like this might be difficult, given that the partition function $Z(\theta)$ appears in analytical expressions for each of these distinguishability measures. However, there do exist quantum algorithms for estimating various distinguishability measures that depend on the condition number of the underlying states~\cite{WGLZY22,WZ23,liu2024quantumalgorithmsmatrixgeometric}, and one could investigate whether these would be efficient for evolved quantum Boltzmann machines.

Given the ubiquitous role that time-evolved thermal states play in physics and the increasing relevance of quantum generalizations of Fisher information in various areas of physics like high energy and condensed matter, there is a distinct possibility that the findings of our paper could find applications well beyond those presented here. For example, in the AdS/CFT correspondence, the Kubo--Mori information of a perturbation in the conformal field theory is dual to the
bulk canonical energy of the  linearized gravitational perturbation~\cite{Lashkari2016} (see also~\cite{Kibe2022} for a review). Additionally, in condensed matter physics, quantum generalizations of Fisher information have played an essential role in understanding and detecting phase transitions~\cite{CVZ2007,ZGC2007,CAROLLO20201}. We suspect that our quantum algorithms for estimating information matrix elements should find use in both of these settings, and we leave such investigations to future work. 

\begin{acknowledgments}
    We thank Alvaro Alhambra, Pierre-Luc Dallaire-Demers, Zoe Holmes, James Watson, and Nicole Yunger Halpern   for helpful discussions. MMW is also grateful to Francesco Buscemi for bringing~\cite[Theorem~1.1]{DelMoral2018} to his attention.

MM and DP  acknowledge support from the Air Force Office of Scientific
Research under agreement no. FA2386-24-1-4069.
DP and MMW acknowledge support from
Air Force Research Laboratory under agreement no.~FA8750-23-2-0031.

The U.S.~Government is authorized to reproduce and
distribute reprints for Governmental purposes notwithstanding any copyright
notation thereon. The views and conclusions contained herein are those of the
authors and should not be interpreted as necessarily representing the official
policies or endorsements, either expressed or implied, of Air Force Research
Laboratory or the United States Air Force or the U.S.~Government.
\end{acknowledgments}

\section*{Author contributions}
\noindent
\textbf{Author Contributions}: The following describes the
different contributions of the authors of this work, using
roles defined by the CRediT (Contributor Roles Taxonomy) project~\cite{CRediT}:

\noindent\textbf{MM:} Formal analysis, Investigation, Methodology, Validation, Writing – original draft, Writing – review \& editing.

\noindent\textbf{DP:} Formal analysis, Investigation, Validation, Writing – review \& editing.

\noindent\textbf{MMW:} Conceptualization, Formal analysis, Funding acquisition, Investigation, Methodology, Supervision, Validation, Writing – original draft, Writing – review \& editing.

\bibliography{references}

%apsrev4-2.bst 2019-01-14 (MD) hand-edited version of apsrev4-1.bst
%Control: key (0)
%Control: author (8) initials jnrlst
%Control: editor formatted (1) identically to author
%Control: production of article title (0) allowed
%Control: page (0) single
%Control: year (1) truncated
%Control: production of eprint (0) enabled
\begin{thebibliography}{115}%
\makeatletter
\providecommand \@ifxundefined [1]{%
 \@ifx{#1\undefined}
}%
\providecommand \@ifnum [1]{%
 \ifnum #1\expandafter \@firstoftwo
 \else \expandafter \@secondoftwo
 \fi
}%
\providecommand \@ifx [1]{%
 \ifx #1\expandafter \@firstoftwo
 \else \expandafter \@secondoftwo
 \fi
}%
\providecommand \natexlab [1]{#1}%
\providecommand \enquote  [1]{``#1''}%
\providecommand \bibnamefont  [1]{#1}%
\providecommand \bibfnamefont [1]{#1}%
\providecommand \citenamefont [1]{#1}%
\providecommand \href@noop [0]{\@secondoftwo}%
\providecommand \href [0]{\begingroup \@sanitize@url \@href}%
\providecommand \@href[1]{\@@startlink{#1}\@@href}%
\providecommand \@@href[1]{\endgroup#1\@@endlink}%
\providecommand \@sanitize@url [0]{\catcode `\\12\catcode `\$12\catcode `\&12\catcode `\#12\catcode `\^12\catcode `\_12\catcode `\%12\relax}%
\providecommand \@@startlink[1]{}%
\providecommand \@@endlink[0]{}%
\providecommand \url  [0]{\begingroup\@sanitize@url \@url }%
\providecommand \@url [1]{\endgroup\@href {#1}{\urlprefix }}%
\providecommand \urlprefix  [0]{URL }%
\providecommand \Eprint [0]{\href }%
\providecommand \doibase [0]{https://doi.org/}%
\providecommand \selectlanguage [0]{\@gobble}%
\providecommand \bibinfo  [0]{\@secondoftwo}%
\providecommand \bibfield  [0]{\@secondoftwo}%
\providecommand \translation [1]{[#1]}%
\providecommand \BibitemOpen [0]{}%
\providecommand \bibitemStop [0]{}%
\providecommand \bibitemNoStop [0]{.\EOS\space}%
\providecommand \EOS [0]{\spacefactor3000\relax}%
\providecommand \BibitemShut  [1]{\csname bibitem#1\endcsname}%
\let\auto@bib@innerbib\@empty
%</preamble>
\bibitem [{\citenamefont {Montanaro}(2016)}]{Montanaro2016overview_q_algo}%
  \BibitemOpen
  \bibfield  {author} {\bibinfo {author} {\bibfnamefont {A.}~\bibnamefont {Montanaro}},\ }\bibfield  {title} {\bibinfo {title} {Quantum algorithms: An overview},\ }\href {https://doi.org/10.1038/npjqi.2015.23} {\bibfield  {journal} {\bibinfo  {journal} {npj Quantum Information}\ }\textbf {\bibinfo {volume} {2}},\ \bibinfo {pages} {15023} (\bibinfo {year} {2016})}\BibitemShut {NoStop}%
\bibitem [{\citenamefont {Biamonte}\ \emph {et~al.}(2017)\citenamefont {Biamonte}, \citenamefont {Wittek}, \citenamefont {Pancotti}, \citenamefont {Rebentrost}, \citenamefont {Wiebe},\ and\ \citenamefont {Lloyd}}]{Biamonte2017qml}%
  \BibitemOpen
  \bibfield  {author} {\bibinfo {author} {\bibfnamefont {J.}~\bibnamefont {Biamonte}}, \bibinfo {author} {\bibfnamefont {P.}~\bibnamefont {Wittek}}, \bibinfo {author} {\bibfnamefont {N.}~\bibnamefont {Pancotti}}, \bibinfo {author} {\bibfnamefont {P.}~\bibnamefont {Rebentrost}}, \bibinfo {author} {\bibfnamefont {N.}~\bibnamefont {Wiebe}},\ and\ \bibinfo {author} {\bibfnamefont {S.}~\bibnamefont {Lloyd}},\ }\bibfield  {title} {\bibinfo {title} {Quantum machine learning},\ }\href {https://doi.org/10.1038/nature23474} {\bibfield  {journal} {\bibinfo  {journal} {Nature}\ }\textbf {\bibinfo {volume} {549}},\ \bibinfo {pages} {195–202} (\bibinfo {year} {2017})}\BibitemShut {NoStop}%
\bibitem [{\citenamefont {Peruzzo}\ \emph {et~al.}(2014)\citenamefont {Peruzzo}, \citenamefont {McClean}, \citenamefont {Shadbolt}, \citenamefont {Yung}, \citenamefont {Zhou}, \citenamefont {Love}, \citenamefont {Aspuru-Guzik},\ and\ \citenamefont {O’Brien}}]{Peruzzo2014vqe}%
  \BibitemOpen
  \bibfield  {author} {\bibinfo {author} {\bibfnamefont {A.}~\bibnamefont {Peruzzo}}, \bibinfo {author} {\bibfnamefont {J.}~\bibnamefont {McClean}}, \bibinfo {author} {\bibfnamefont {P.}~\bibnamefont {Shadbolt}}, \bibinfo {author} {\bibfnamefont {M.-H.}\ \bibnamefont {Yung}}, \bibinfo {author} {\bibfnamefont {X.-Q.}\ \bibnamefont {Zhou}}, \bibinfo {author} {\bibfnamefont {P.~J.}\ \bibnamefont {Love}}, \bibinfo {author} {\bibfnamefont {A.}~\bibnamefont {Aspuru-Guzik}},\ and\ \bibinfo {author} {\bibfnamefont {J.~L.}\ \bibnamefont {O’Brien}},\ }\bibfield  {title} {\bibinfo {title} {A variational eigenvalue solver on a photonic quantum processor},\ }\href {https://doi.org/10.1038/ncomms5213} {\bibfield  {journal} {\bibinfo  {journal} {Nature Communications}\ }\textbf {\bibinfo {volume} {5}},\ \bibinfo {pages} {4213} (\bibinfo {year} {2014})}\BibitemShut {NoStop}%
\bibitem [{\citenamefont {Amin}\ \emph {et~al.}(2018)\citenamefont {Amin}, \citenamefont {Andriyash}, \citenamefont {Rolfe}, \citenamefont {Kulchytskyy},\ and\ \citenamefont {Melko}}]{Amin2018qbm}%
  \BibitemOpen
  \bibfield  {author} {\bibinfo {author} {\bibfnamefont {M.~H.}\ \bibnamefont {Amin}}, \bibinfo {author} {\bibfnamefont {E.}~\bibnamefont {Andriyash}}, \bibinfo {author} {\bibfnamefont {J.}~\bibnamefont {Rolfe}}, \bibinfo {author} {\bibfnamefont {B.}~\bibnamefont {Kulchytskyy}},\ and\ \bibinfo {author} {\bibfnamefont {R.}~\bibnamefont {Melko}},\ }\bibfield  {title} {\bibinfo {title} {Quantum {B}oltzmann machine},\ }\href {https://doi.org/10.1103/PhysRevX.8.021050} {\bibfield  {journal} {\bibinfo  {journal} {Physical Review X}\ }\textbf {\bibinfo {volume} {8}},\ \bibinfo {pages} {021050} (\bibinfo {year} {2018})}\BibitemShut {NoStop}%
\bibitem [{\citenamefont {Verdon}\ \emph {et~al.}(2019)\citenamefont {Verdon}, \citenamefont {Marks}, \citenamefont {Nanda}, \citenamefont {Leichenauer},\ and\ \citenamefont {Hidary}}]{Verdon2019qhbm}%
  \BibitemOpen
  \bibfield  {author} {\bibinfo {author} {\bibfnamefont {G.}~\bibnamefont {Verdon}}, \bibinfo {author} {\bibfnamefont {J.}~\bibnamefont {Marks}}, \bibinfo {author} {\bibfnamefont {S.}~\bibnamefont {Nanda}}, \bibinfo {author} {\bibfnamefont {S.}~\bibnamefont {Leichenauer}},\ and\ \bibinfo {author} {\bibfnamefont {J.}~\bibnamefont {Hidary}},\ }\href {https://arxiv.org/abs/1910.02071} {\bibinfo {title} {Quantum {H}amiltonian-based models and the variational quantum thermalizer algorithm}} (\bibinfo {year} {2019}),\ \Eprint {https://arxiv.org/abs/1910.02071} {arXiv:1910.02071 [quant-ph]} \BibitemShut {NoStop}%
\bibitem [{\citenamefont {Ferguson}\ \emph {et~al.}(2021)\citenamefont {Ferguson}, \citenamefont {Dellantonio}, \citenamefont {Balushi}, \citenamefont {Jansen}, \citenamefont {D\"ur},\ and\ \citenamefont {Muschik}}]{Ferguson2021mbqc}%
  \BibitemOpen
  \bibfield  {author} {\bibinfo {author} {\bibfnamefont {R.~R.}\ \bibnamefont {Ferguson}}, \bibinfo {author} {\bibfnamefont {L.}~\bibnamefont {Dellantonio}}, \bibinfo {author} {\bibfnamefont {A.~A.}\ \bibnamefont {Balushi}}, \bibinfo {author} {\bibfnamefont {K.}~\bibnamefont {Jansen}}, \bibinfo {author} {\bibfnamefont {W.}~\bibnamefont {D\"ur}},\ and\ \bibinfo {author} {\bibfnamefont {C.~A.}\ \bibnamefont {Muschik}},\ }\bibfield  {title} {\bibinfo {title} {Measurement-based variational quantum eigensolver},\ }\href {https://doi.org/10.1103/PhysRevLett.126.220501} {\bibfield  {journal} {\bibinfo  {journal} {Physical Review Letters}\ }\textbf {\bibinfo {volume} {126}},\ \bibinfo {pages} {220501} (\bibinfo {year} {2021})}\BibitemShut {NoStop}%
\bibitem [{\citenamefont {McClean}\ \emph {et~al.}(2016)\citenamefont {McClean}, \citenamefont {Romero}, \citenamefont {Babbush},\ and\ \citenamefont {Aspuru-Guzik}}]{McClean2016VQA}%
  \BibitemOpen
  \bibfield  {author} {\bibinfo {author} {\bibfnamefont {J.~R.}\ \bibnamefont {McClean}}, \bibinfo {author} {\bibfnamefont {J.}~\bibnamefont {Romero}}, \bibinfo {author} {\bibfnamefont {R.}~\bibnamefont {Babbush}},\ and\ \bibinfo {author} {\bibfnamefont {A.}~\bibnamefont {Aspuru-Guzik}},\ }\bibfield  {title} {\bibinfo {title} {The theory of variational hybrid quantum-classical algorithms},\ }\href {https://doi.org/10.1088/1367-2630/18/2/023023} {\bibfield  {journal} {\bibinfo  {journal} {New Journal of Physics}\ }\textbf {\bibinfo {volume} {18}},\ \bibinfo {pages} {023023} (\bibinfo {year} {2016})}\BibitemShut {NoStop}%
\bibitem [{\citenamefont {Mitarai}\ \emph {et~al.}(2018)\citenamefont {Mitarai}, \citenamefont {Negoro}, \citenamefont {Kitagawa},\ and\ \citenamefont {Fujii}}]{Mitarai2018q_circuit_learning}%
  \BibitemOpen
  \bibfield  {author} {\bibinfo {author} {\bibfnamefont {K.}~\bibnamefont {Mitarai}}, \bibinfo {author} {\bibfnamefont {M.}~\bibnamefont {Negoro}}, \bibinfo {author} {\bibfnamefont {M.}~\bibnamefont {Kitagawa}},\ and\ \bibinfo {author} {\bibfnamefont {K.}~\bibnamefont {Fujii}},\ }\bibfield  {title} {\bibinfo {title} {Quantum circuit learning},\ }\href {https://doi.org/10.1103/PhysRevA.98.032309} {\bibfield  {journal} {\bibinfo  {journal} {Phyical Review A}\ }\textbf {\bibinfo {volume} {98}},\ \bibinfo {pages} {032309} (\bibinfo {year} {2018})}\BibitemShut {NoStop}%
\bibitem [{\citenamefont {Jones}\ \emph {et~al.}(2019)\citenamefont {Jones}, \citenamefont {Endo}, \citenamefont {McArdle}, \citenamefont {Yuan},\ and\ \citenamefont {Benjamin}}]{Jones2019vqe}%
  \BibitemOpen
  \bibfield  {author} {\bibinfo {author} {\bibfnamefont {T.}~\bibnamefont {Jones}}, \bibinfo {author} {\bibfnamefont {S.}~\bibnamefont {Endo}}, \bibinfo {author} {\bibfnamefont {S.}~\bibnamefont {McArdle}}, \bibinfo {author} {\bibfnamefont {X.}~\bibnamefont {Yuan}},\ and\ \bibinfo {author} {\bibfnamefont {S.~C.}\ \bibnamefont {Benjamin}},\ }\bibfield  {title} {\bibinfo {title} {Variational quantum algorithms for discovering {H}amiltonian spectra},\ }\href {https://doi.org/10.1103/PhysRevA.99.062304} {\bibfield  {journal} {\bibinfo  {journal} {Physical Review A}\ }\textbf {\bibinfo {volume} {99}},\ \bibinfo {pages} {062304} (\bibinfo {year} {2019})}\BibitemShut {NoStop}%
\bibitem [{\citenamefont {Cerezo}\ \emph {et~al.}(2022)\citenamefont {Cerezo}, \citenamefont {Sharma}, \citenamefont {Arrasmith},\ and\ \citenamefont {Coles}}]{Cerezo2022vqe}%
  \BibitemOpen
  \bibfield  {author} {\bibinfo {author} {\bibfnamefont {M.}~\bibnamefont {Cerezo}}, \bibinfo {author} {\bibfnamefont {K.}~\bibnamefont {Sharma}}, \bibinfo {author} {\bibfnamefont {A.}~\bibnamefont {Arrasmith}},\ and\ \bibinfo {author} {\bibfnamefont {P.~J.}\ \bibnamefont {Coles}},\ }\bibfield  {title} {\bibinfo {title} {Variational quantum state eigensolver},\ }\href {https://doi.org/10.1038/s41534-022-00611-6} {\bibfield  {journal} {\bibinfo  {journal} {npj Quantum Information}\ }\textbf {\bibinfo {volume} {8}},\ \bibinfo {pages} {113} (\bibinfo {year} {2022})}\BibitemShut {NoStop}%
\bibitem [{\citenamefont {Farhi}\ \emph {et~al.}(2014)\citenamefont {Farhi}, \citenamefont {Goldstone},\ and\ \citenamefont {Gutmann}}]{farhi2014qaoa}%
  \BibitemOpen
  \bibfield  {author} {\bibinfo {author} {\bibfnamefont {E.}~\bibnamefont {Farhi}}, \bibinfo {author} {\bibfnamefont {J.}~\bibnamefont {Goldstone}},\ and\ \bibinfo {author} {\bibfnamefont {S.}~\bibnamefont {Gutmann}},\ }\href {https://arxiv.org/abs/1411.4028} {\bibinfo {title} {A quantum approximate optimization algorithm}} (\bibinfo {year} {2014}),\ \Eprint {https://arxiv.org/abs/1411.4028} {arXiv:1411.4028 [quant-ph]} \BibitemShut {NoStop}%
\bibitem [{\citenamefont {Wang}\ \emph {et~al.}(2018)\citenamefont {Wang}, \citenamefont {Hadfield}, \citenamefont {Jiang},\ and\ \citenamefont {Rieffel}}]{Wang2018qaoa}%
  \BibitemOpen
  \bibfield  {author} {\bibinfo {author} {\bibfnamefont {Z.}~\bibnamefont {Wang}}, \bibinfo {author} {\bibfnamefont {S.}~\bibnamefont {Hadfield}}, \bibinfo {author} {\bibfnamefont {Z.}~\bibnamefont {Jiang}},\ and\ \bibinfo {author} {\bibfnamefont {E.~G.}\ \bibnamefont {Rieffel}},\ }\bibfield  {title} {\bibinfo {title} {Quantum approximate optimization algorithm for {MaxCut}: A fermionic view},\ }\href {https://doi.org/10.1103/physreva.97.022304} {\bibfield  {journal} {\bibinfo  {journal} {Physical Review A}\ }\textbf {\bibinfo {volume} {97}},\ \bibinfo {pages} {022304} (\bibinfo {year} {2018})}\BibitemShut {NoStop}%
\bibitem [{\citenamefont {Hadfield}\ \emph {et~al.}(2019)\citenamefont {Hadfield}, \citenamefont {Wang}, \citenamefont {O’Gorman}, \citenamefont {Rieffel}, \citenamefont {Venturelli},\ and\ \citenamefont {Biswas}}]{Hadfield2019qaoa}%
  \BibitemOpen
  \bibfield  {author} {\bibinfo {author} {\bibfnamefont {S.}~\bibnamefont {Hadfield}}, \bibinfo {author} {\bibfnamefont {Z.}~\bibnamefont {Wang}}, \bibinfo {author} {\bibfnamefont {B.}~\bibnamefont {O’Gorman}}, \bibinfo {author} {\bibfnamefont {E.~G.}\ \bibnamefont {Rieffel}}, \bibinfo {author} {\bibfnamefont {D.}~\bibnamefont {Venturelli}},\ and\ \bibinfo {author} {\bibfnamefont {R.}~\bibnamefont {Biswas}},\ }\bibfield  {title} {\bibinfo {title} {From the quantum approximate optimization algorithm to a quantum alternating operator ansatz},\ }\href {https://doi.org/10.3390/a12020034} {\bibfield  {journal} {\bibinfo  {journal} {Algorithms}\ }\textbf {\bibinfo {volume} {12}},\ \bibinfo {pages} {34} (\bibinfo {year} {2019})}\BibitemShut {NoStop}%
\bibitem [{\citenamefont {Bermejo}\ and\ \citenamefont {Or{\'u}s}(2023)}]{bermejo2024variational_clustering}%
  \BibitemOpen
  \bibfield  {author} {\bibinfo {author} {\bibfnamefont {P.}~\bibnamefont {Bermejo}}\ and\ \bibinfo {author} {\bibfnamefont {R.}~\bibnamefont {Or{\'u}s}},\ }\bibfield  {title} {\bibinfo {title} {Variational quantum and quantum-inspired clustering},\ }\href {https://doi.org/https://doi.org/10.1038/s41598-023-39771-6} {\bibfield  {journal} {\bibinfo  {journal} {Scientific Reports}\ }\textbf {\bibinfo {volume} {13}},\ \bibinfo {pages} {13284} (\bibinfo {year} {2023})}\BibitemShut {NoStop}%
\bibitem [{\citenamefont {Rebentrost}\ \emph {et~al.}(2014)\citenamefont {Rebentrost}, \citenamefont {Mohseni},\ and\ \citenamefont {Lloyd}}]{Rebentrost2014qsvm}%
  \BibitemOpen
  \bibfield  {author} {\bibinfo {author} {\bibfnamefont {P.}~\bibnamefont {Rebentrost}}, \bibinfo {author} {\bibfnamefont {M.}~\bibnamefont {Mohseni}},\ and\ \bibinfo {author} {\bibfnamefont {S.}~\bibnamefont {Lloyd}},\ }\bibfield  {title} {\bibinfo {title} {Quantum support vector machine for big data classification},\ }\href {https://doi.org/10.1103/physrevlett.113.130503} {\bibfield  {journal} {\bibinfo  {journal} {Physical Review Letters}\ }\textbf {\bibinfo {volume} {113}},\ \bibinfo {pages} {130503} (\bibinfo {year} {2014})}\BibitemShut {NoStop}%
\bibitem [{\citenamefont {Havlíček}\ \emph {et~al.}(2019)\citenamefont {Havlíček}, \citenamefont {Córcoles}, \citenamefont {Temme}, \citenamefont {Harrow}, \citenamefont {Kandala}, \citenamefont {Chow},\ and\ \citenamefont {Gambetta}}]{Havlicek2019supervised_learning}%
  \BibitemOpen
  \bibfield  {author} {\bibinfo {author} {\bibfnamefont {V.}~\bibnamefont {Havlíček}}, \bibinfo {author} {\bibfnamefont {A.~D.}\ \bibnamefont {Córcoles}}, \bibinfo {author} {\bibfnamefont {K.}~\bibnamefont {Temme}}, \bibinfo {author} {\bibfnamefont {A.~W.}\ \bibnamefont {Harrow}}, \bibinfo {author} {\bibfnamefont {A.}~\bibnamefont {Kandala}}, \bibinfo {author} {\bibfnamefont {J.~M.}\ \bibnamefont {Chow}},\ and\ \bibinfo {author} {\bibfnamefont {J.~M.}\ \bibnamefont {Gambetta}},\ }\bibfield  {title} {\bibinfo {title} {Supervised learning with quantum-enhanced feature spaces},\ }\href {https://doi.org/10.1038/s41586-019-0980-2} {\bibfield  {journal} {\bibinfo  {journal} {Nature}\ }\textbf {\bibinfo {volume} {567}},\ \bibinfo {pages} {209–212} (\bibinfo {year} {2019})}\BibitemShut {NoStop}%
\bibitem [{\citenamefont {Schuld}\ \emph {et~al.}(2020)\citenamefont {Schuld}, \citenamefont {Bocharov}, \citenamefont {Svore},\ and\ \citenamefont {Wiebe}}]{Schuld2020quantum_classifier}%
  \BibitemOpen
  \bibfield  {author} {\bibinfo {author} {\bibfnamefont {M.}~\bibnamefont {Schuld}}, \bibinfo {author} {\bibfnamefont {A.}~\bibnamefont {Bocharov}}, \bibinfo {author} {\bibfnamefont {K.~M.}\ \bibnamefont {Svore}},\ and\ \bibinfo {author} {\bibfnamefont {N.}~\bibnamefont {Wiebe}},\ }\bibfield  {title} {\bibinfo {title} {Circuit-centric quantum classifiers},\ }\href {https://doi.org/10.1103/PhysRevA.101.032308} {\bibfield  {journal} {\bibinfo  {journal} {Physical Review A}\ }\textbf {\bibinfo {volume} {101}},\ \bibinfo {pages} {032308} (\bibinfo {year} {2020})}\BibitemShut {NoStop}%
\bibitem [{\citenamefont {Benedetti}\ \emph {et~al.}(2019)\citenamefont {Benedetti}, \citenamefont {Lloyd}, \citenamefont {Sack},\ and\ \citenamefont {Fiorentini}}]{Benedetti2019pqc}%
  \BibitemOpen
  \bibfield  {author} {\bibinfo {author} {\bibfnamefont {M.}~\bibnamefont {Benedetti}}, \bibinfo {author} {\bibfnamefont {E.}~\bibnamefont {Lloyd}}, \bibinfo {author} {\bibfnamefont {S.}~\bibnamefont {Sack}},\ and\ \bibinfo {author} {\bibfnamefont {M.}~\bibnamefont {Fiorentini}},\ }\bibfield  {title} {\bibinfo {title} {Parameterized quantum circuits as machine learning models},\ }\href {https://doi.org/10.1088/2058-9565/ab4eb5} {\bibfield  {journal} {\bibinfo  {journal} {Quantum Science and Technology}\ }\textbf {\bibinfo {volume} {4}},\ \bibinfo {pages} {043001} (\bibinfo {year} {2019})}\BibitemShut {NoStop}%
\bibitem [{\citenamefont {Leadbeater}\ \emph {et~al.}(2021)\citenamefont {Leadbeater}, \citenamefont {Sharrock}, \citenamefont {Coyle},\ and\ \citenamefont {Benedetti}}]{Leadbeater2021gen_modell}%
  \BibitemOpen
  \bibfield  {author} {\bibinfo {author} {\bibfnamefont {C.}~\bibnamefont {Leadbeater}}, \bibinfo {author} {\bibfnamefont {L.}~\bibnamefont {Sharrock}}, \bibinfo {author} {\bibfnamefont {B.}~\bibnamefont {Coyle}},\ and\ \bibinfo {author} {\bibfnamefont {M.}~\bibnamefont {Benedetti}},\ }\bibfield  {title} {\bibinfo {title} {$f$-divergences and cost function locality in generative modelling with quantum circuits},\ }\href {https://doi.org/10.3390/e23101281} {\bibfield  {journal} {\bibinfo  {journal} {Entropy}\ }\textbf {\bibinfo {volume} {23}},\ \bibinfo {pages} {1281} (\bibinfo {year} {2021})}\BibitemShut {NoStop}%
\bibitem [{\citenamefont {Abbas}\ \emph {et~al.}(2021)\citenamefont {Abbas}, \citenamefont {Sutter}, \citenamefont {Zoufal}, \citenamefont {Lucchi}, \citenamefont {Figalli},\ and\ \citenamefont {Woerner}}]{Abbas2021qnn}%
  \BibitemOpen
  \bibfield  {author} {\bibinfo {author} {\bibfnamefont {A.}~\bibnamefont {Abbas}}, \bibinfo {author} {\bibfnamefont {D.}~\bibnamefont {Sutter}}, \bibinfo {author} {\bibfnamefont {C.}~\bibnamefont {Zoufal}}, \bibinfo {author} {\bibfnamefont {A.}~\bibnamefont {Lucchi}}, \bibinfo {author} {\bibfnamefont {A.}~\bibnamefont {Figalli}},\ and\ \bibinfo {author} {\bibfnamefont {S.}~\bibnamefont {Woerner}},\ }\bibfield  {title} {\bibinfo {title} {The power of quantum neural networks},\ }\href {https://doi.org/10.1038/s43588-021-00084-1} {\bibfield  {journal} {\bibinfo  {journal} {Nature Computational Science}\ }\textbf {\bibinfo {volume} {1}},\ \bibinfo {pages} {403–409} (\bibinfo {year} {2021})}\BibitemShut {NoStop}%
\bibitem [{\citenamefont {Cerezo}\ \emph {et~al.}(2021{\natexlab{a}})\citenamefont {Cerezo}, \citenamefont {Arrasmith}, \citenamefont {Babbush}, \citenamefont {Benjamin}, \citenamefont {Endo}, \citenamefont {Fujii}, \citenamefont {McClean}, \citenamefont {Mitarai}, \citenamefont {Yuan}, \citenamefont {Cincio},\ and\ \citenamefont {Coles}}]{Cerezo2021vqa}%
  \BibitemOpen
  \bibfield  {author} {\bibinfo {author} {\bibfnamefont {M.}~\bibnamefont {Cerezo}}, \bibinfo {author} {\bibfnamefont {A.}~\bibnamefont {Arrasmith}}, \bibinfo {author} {\bibfnamefont {R.}~\bibnamefont {Babbush}}, \bibinfo {author} {\bibfnamefont {S.~C.}\ \bibnamefont {Benjamin}}, \bibinfo {author} {\bibfnamefont {S.}~\bibnamefont {Endo}}, \bibinfo {author} {\bibfnamefont {K.}~\bibnamefont {Fujii}}, \bibinfo {author} {\bibfnamefont {J.~R.}\ \bibnamefont {McClean}}, \bibinfo {author} {\bibfnamefont {K.}~\bibnamefont {Mitarai}}, \bibinfo {author} {\bibfnamefont {X.}~\bibnamefont {Yuan}}, \bibinfo {author} {\bibfnamefont {L.}~\bibnamefont {Cincio}},\ and\ \bibinfo {author} {\bibfnamefont {P.~J.}\ \bibnamefont {Coles}},\ }\bibfield  {title} {\bibinfo {title} {Variational quantum algorithms},\ }\href {https://doi.org/10.1038/s42254-021-00348-9} {\bibfield  {journal} {\bibinfo  {journal} {Nature Reviews Physics}\ }\textbf {\bibinfo {volume} {3}},\ \bibinfo {pages} {625–644} (\bibinfo {year}
  {2021}{\natexlab{a}})}\BibitemShut {NoStop}%
\bibitem [{\citenamefont {McClean}\ \emph {et~al.}(2018)\citenamefont {McClean}, \citenamefont {Boixo}, \citenamefont {Smelyanskiy}, \citenamefont {Babbush},\ and\ \citenamefont {Neven}}]{McClean2018barren_plateaus}%
  \BibitemOpen
  \bibfield  {author} {\bibinfo {author} {\bibfnamefont {J.~R.}\ \bibnamefont {McClean}}, \bibinfo {author} {\bibfnamefont {S.}~\bibnamefont {Boixo}}, \bibinfo {author} {\bibfnamefont {V.~N.}\ \bibnamefont {Smelyanskiy}}, \bibinfo {author} {\bibfnamefont {R.}~\bibnamefont {Babbush}},\ and\ \bibinfo {author} {\bibfnamefont {H.}~\bibnamefont {Neven}},\ }\bibfield  {title} {\bibinfo {title} {Barren plateaus in quantum neural network training landscapes},\ }\href {https://doi.org/10.1038/s41467-018-07090-4} {\bibfield  {journal} {\bibinfo  {journal} {Nature Communications}\ }\textbf {\bibinfo {volume} {9}},\ \bibinfo {pages} {4812} (\bibinfo {year} {2018})}\BibitemShut {NoStop}%
\bibitem [{\citenamefont {Cerezo}\ \emph {et~al.}(2021{\natexlab{b}})\citenamefont {Cerezo}, \citenamefont {Sone}, \citenamefont {Volkoff}, \citenamefont {Cincio},\ and\ \citenamefont {Coles}}]{Cerezo2021barren_plateaus}%
  \BibitemOpen
  \bibfield  {author} {\bibinfo {author} {\bibfnamefont {M.}~\bibnamefont {Cerezo}}, \bibinfo {author} {\bibfnamefont {A.}~\bibnamefont {Sone}}, \bibinfo {author} {\bibfnamefont {T.}~\bibnamefont {Volkoff}}, \bibinfo {author} {\bibfnamefont {L.}~\bibnamefont {Cincio}},\ and\ \bibinfo {author} {\bibfnamefont {P.~J.}\ \bibnamefont {Coles}},\ }\bibfield  {title} {\bibinfo {title} {Cost function dependent barren plateaus in shallow parametrized quantum circuits},\ }\href {https://doi.org/10.1038/s41467-021-21728-w} {\bibfield  {journal} {\bibinfo  {journal} {Nature Communications}\ }\textbf {\bibinfo {volume} {12}},\ \bibinfo {pages} {1791} (\bibinfo {year} {2021}{\natexlab{b}})}\BibitemShut {NoStop}%
\bibitem [{\citenamefont {Holmes}\ \emph {et~al.}(2022{\natexlab{a}})\citenamefont {Holmes}, \citenamefont {Sharma}, \citenamefont {Cerezo},\ and\ \citenamefont {Coles}}]{Holmes2022barren_plateaus}%
  \BibitemOpen
  \bibfield  {author} {\bibinfo {author} {\bibfnamefont {Z.}~\bibnamefont {Holmes}}, \bibinfo {author} {\bibfnamefont {K.}~\bibnamefont {Sharma}}, \bibinfo {author} {\bibfnamefont {M.}~\bibnamefont {Cerezo}},\ and\ \bibinfo {author} {\bibfnamefont {P.~J.}\ \bibnamefont {Coles}},\ }\bibfield  {title} {\bibinfo {title} {Connecting ansatz expressibility to gradient magnitudes and barren plateaus},\ }\href {https://doi.org/10.1103/prxquantum.3.010313} {\bibfield  {journal} {\bibinfo  {journal} {PRX Quantum}\ }\textbf {\bibinfo {volume} {3}},\ \bibinfo {pages} {010313} (\bibinfo {year} {2022}{\natexlab{a}})}\BibitemShut {NoStop}%
\bibitem [{\citenamefont {Marrero}\ \emph {et~al.}(2021)\citenamefont {Marrero}, \citenamefont {Kieferov{\'a}},\ and\ \citenamefont {Wiebe}}]{marrero2021barrenp_lateaus}%
  \BibitemOpen
  \bibfield  {author} {\bibinfo {author} {\bibfnamefont {C.~O.}\ \bibnamefont {Marrero}}, \bibinfo {author} {\bibfnamefont {M.}~\bibnamefont {Kieferov{\'a}}},\ and\ \bibinfo {author} {\bibfnamefont {N.}~\bibnamefont {Wiebe}},\ }\bibfield  {title} {\bibinfo {title} {Entanglement-induced barren plateaus},\ }\href {https://doi.org/10.1103/PRXQuantum.2.040316} {\bibfield  {journal} {\bibinfo  {journal} {PRX Quantum}\ }\textbf {\bibinfo {volume} {2}},\ \bibinfo {pages} {040316} (\bibinfo {year} {2021})}\BibitemShut {NoStop}%
\bibitem [{\citenamefont {Benedetti}\ \emph {et~al.}(2017)\citenamefont {Benedetti}, \citenamefont {Realpe-Gómez}, \citenamefont {Biswas},\ and\ \citenamefont {Perdomo-Ortiz}}]{Benedetti2017qbm}%
  \BibitemOpen
  \bibfield  {author} {\bibinfo {author} {\bibfnamefont {M.}~\bibnamefont {Benedetti}}, \bibinfo {author} {\bibfnamefont {J.}~\bibnamefont {Realpe-Gómez}}, \bibinfo {author} {\bibfnamefont {R.}~\bibnamefont {Biswas}},\ and\ \bibinfo {author} {\bibfnamefont {A.}~\bibnamefont {Perdomo-Ortiz}},\ }\bibfield  {title} {\bibinfo {title} {Quantum-assisted learning of hardware-embedded probabilistic graphical models},\ }\href {https://doi.org/10.1103/physrevx.7.041052} {\bibfield  {journal} {\bibinfo  {journal} {Physical Review X}\ }\textbf {\bibinfo {volume} {7}},\ \bibinfo {pages} {041052} (\bibinfo {year} {2017})}\BibitemShut {NoStop}%
\bibitem [{\citenamefont {Kieferova}\ and\ \citenamefont {Wiebe}(2017)}]{kieferova2017qbm}%
  \BibitemOpen
  \bibfield  {author} {\bibinfo {author} {\bibfnamefont {M.}~\bibnamefont {Kieferova}}\ and\ \bibinfo {author} {\bibfnamefont {N.}~\bibnamefont {Wiebe}},\ }\bibfield  {title} {\bibinfo {title} {Tomography and generative data modeling via quantum {B}oltzmann training},\ }\href {https://doi.org/10.1103/PhysRevA.96.062327} {\bibfield  {journal} {\bibinfo  {journal} {Physical Review A}\ }\textbf {\bibinfo {volume} {96}},\ \bibinfo {pages} {062327} (\bibinfo {year} {2017})}\BibitemShut {NoStop}%
\bibitem [{\citenamefont {Hinton}\ and\ \citenamefont {Sejnowski}(1983)}]{hinton1983optimal}%
  \BibitemOpen
  \bibfield  {author} {\bibinfo {author} {\bibfnamefont {G.~E.}\ \bibnamefont {Hinton}}\ and\ \bibinfo {author} {\bibfnamefont {T.~J.}\ \bibnamefont {Sejnowski}},\ }\bibfield  {title} {\bibinfo {title} {Optimal perceptual inference},\ }in\ \href {https://api.semanticscholar.org/CorpusID:10379672} {\emph {\bibinfo {booktitle} {Proceedings of the IEEE conference on Computer Vision and Pattern Recognition}}},\ Vol.\ \bibinfo {volume} {448}\ (\bibinfo {organization} {Washington},\ \bibinfo {year} {1983})\ pp.\ \bibinfo {pages} {448--453}\BibitemShut {NoStop}%
\bibitem [{\citenamefont {Hinton}\ \emph {et~al.}(2006)\citenamefont {Hinton}, \citenamefont {Osindero},\ and\ \citenamefont {Teh}}]{hinton2006fast}%
  \BibitemOpen
  \bibfield  {author} {\bibinfo {author} {\bibfnamefont {G.~E.}\ \bibnamefont {Hinton}}, \bibinfo {author} {\bibfnamefont {S.}~\bibnamefont {Osindero}},\ and\ \bibinfo {author} {\bibfnamefont {Y.-W.}\ \bibnamefont {Teh}},\ }\bibfield  {title} {\bibinfo {title} {A fast learning algorithm for deep belief nets},\ }\href {https://doi.org/10.1162/neco.2006.18.7.1527} {\bibfield  {journal} {\bibinfo  {journal} {Neural computation}\ }\textbf {\bibinfo {volume} {18}},\ \bibinfo {pages} {1527} (\bibinfo {year} {2006})}\BibitemShut {NoStop}%
\bibitem [{\citenamefont {Salakhutdinov}\ and\ \citenamefont {Hinton}(2009)}]{salakhutdinov2009deep}%
  \BibitemOpen
  \bibfield  {author} {\bibinfo {author} {\bibfnamefont {R.}~\bibnamefont {Salakhutdinov}}\ and\ \bibinfo {author} {\bibfnamefont {G.}~\bibnamefont {Hinton}},\ }\bibfield  {title} {\bibinfo {title} {Deep {B}oltzmann machines},\ }in\ \href {https://proceedings.mlr.press/v5/salakhutdinov09a.html} {\emph {\bibinfo {booktitle} {Proceedings of the Twelfth International Conference on Artificial Intelligence and Statistics}}},\ \bibinfo {series} {Proceedings of Machine Learning Research}, Vol.~\bibinfo {volume} {5}\ (\bibinfo  {publisher} {PMLR},\ \bibinfo {year} {2009})\ pp.\ \bibinfo {pages} {448--455}\BibitemShut {NoStop}%
\bibitem [{\citenamefont {Chen}\ \emph {et~al.}(2023{\natexlab{a}})\citenamefont {Chen}, \citenamefont {Kastoryano},\ and\ \citenamefont {Gilyén}}]{chen2023q_Gibbs_sampl}%
  \BibitemOpen
  \bibfield  {author} {\bibinfo {author} {\bibfnamefont {C.-F.}\ \bibnamefont {Chen}}, \bibinfo {author} {\bibfnamefont {M.~J.}\ \bibnamefont {Kastoryano}},\ and\ \bibinfo {author} {\bibfnamefont {A.}~\bibnamefont {Gilyén}},\ }\href {https://arxiv.org/abs/2311.09207} {\bibinfo {title} {An efficient and exact noncommutative quantum {G}ibbs sampler}} (\bibinfo {year} {2023}{\natexlab{a}}),\ \Eprint {https://arxiv.org/abs/2311.09207} {arXiv:2311.09207 [quant-ph]} \BibitemShut {NoStop}%
\bibitem [{\citenamefont {Chen}\ \emph {et~al.}(2023{\natexlab{b}})\citenamefont {Chen}, \citenamefont {Kastoryano}, \citenamefont {Brandão},\ and\ \citenamefont {Gilyén}}]{chen2023thermalstatepreparation}%
  \BibitemOpen
  \bibfield  {author} {\bibinfo {author} {\bibfnamefont {C.-F.}\ \bibnamefont {Chen}}, \bibinfo {author} {\bibfnamefont {M.~J.}\ \bibnamefont {Kastoryano}}, \bibinfo {author} {\bibfnamefont {F.~G. S.~L.}\ \bibnamefont {Brandão}},\ and\ \bibinfo {author} {\bibfnamefont {A.}~\bibnamefont {Gilyén}},\ }\href {https://arxiv.org/abs/2303.18224} {\bibinfo {title} {Quantum thermal state preparation}} (\bibinfo {year} {2023}{\natexlab{b}}),\ \Eprint {https://arxiv.org/abs/2303.18224} {arXiv:2303.18224 [quant-ph]} \BibitemShut {NoStop}%
\bibitem [{\citenamefont {Rajakumar}\ and\ \citenamefont {Watson}(2026)}]{rajakumar2024gibbssampling}%
  \BibitemOpen
  \bibfield  {author} {\bibinfo {author} {\bibfnamefont {J.}~\bibnamefont {Rajakumar}}\ and\ \bibinfo {author} {\bibfnamefont {J.~D.}\ \bibnamefont {Watson}},\ }\bibfield  {title} {\bibinfo {title} {Gibbs sampling gives quantum advantage at constant temperatures with {O}(1)-local {H}amiltonians},\ }\href {http://dx.doi.org/10.22331/q-2026-01-22-1981} {\bibfield  {journal} {\bibinfo  {journal} {Quantum}\ }\textbf {\bibinfo {volume} {10}},\ \bibinfo {pages} {1981} (\bibinfo {year} {2026})}\BibitemShut {NoStop}%
\bibitem [{\citenamefont {Bergamaschi}\ \emph {et~al.}(2024)\citenamefont {Bergamaschi}, \citenamefont {Chen},\ and\ \citenamefont {Liu}}]{bergamaschi2024gibbs_sampling}%
  \BibitemOpen
  \bibfield  {author} {\bibinfo {author} {\bibfnamefont {T.}~\bibnamefont {Bergamaschi}}, \bibinfo {author} {\bibfnamefont {C.-F.}\ \bibnamefont {Chen}},\ and\ \bibinfo {author} {\bibfnamefont {Y.}~\bibnamefont {Liu}},\ }\bibfield  {title} {\bibinfo {title} {Quantum computational advantage with constant-temperature {G}ibbs sampling},\ }in\ \href {http://dx.doi.org/10.1109/FOCS61266.2024.00071} {\emph {\bibinfo {booktitle} {2024 IEEE 65th Annual Symposium on Foundations of Computer Science (FOCS)}}}\ (\bibinfo  {publisher} {IEEE},\ \bibinfo {year} {2024})\ p.\ \bibinfo {pages} {1063–1085}\BibitemShut {NoStop}%
\bibitem [{\citenamefont {Chen}\ \emph {et~al.}(2024)\citenamefont {Chen}, \citenamefont {Li}, \citenamefont {Lu},\ and\ \citenamefont {Ying}}]{chen2024sim_Lindblad}%
  \BibitemOpen
  \bibfield  {author} {\bibinfo {author} {\bibfnamefont {H.}~\bibnamefont {Chen}}, \bibinfo {author} {\bibfnamefont {B.}~\bibnamefont {Li}}, \bibinfo {author} {\bibfnamefont {J.}~\bibnamefont {Lu}},\ and\ \bibinfo {author} {\bibfnamefont {L.}~\bibnamefont {Ying}},\ }\href {https://doi.org/10.48550/arXiv.2407.06594} {\bibinfo {title} {A randomized method for simulating {L}indblad equations and thermal state preparation}} (\bibinfo {year} {2024})\BibitemShut {NoStop}%
\bibitem [{\citenamefont {Rouz\'{e}}\ \emph {et~al.}(2025)\citenamefont {Rouz\'{e}}, \citenamefont {Fran\c{c}a},\ and\ \citenamefont {Alhambra}}]{rouze2024efficientthermalization}%
  \BibitemOpen
  \bibfield  {author} {\bibinfo {author} {\bibfnamefont {C.}~\bibnamefont {Rouz\'{e}}}, \bibinfo {author} {\bibfnamefont {D.~S.}\ \bibnamefont {Fran\c{c}a}},\ and\ \bibinfo {author} {\bibfnamefont {A.~M.}\ \bibnamefont {Alhambra}},\ }\bibfield  {title} {\bibinfo {title} {Efficient thermalization and universal quantum computing with quantum {G}ibbs samplers},\ }in\ \href {https://doi.org/10.1145/3717823.3718268} {\emph {\bibinfo {booktitle} {Proceedings of the 57th Annual ACM Symposium on Theory of Computing}}},\ \bibinfo {series and number} {STOC '25}\ (\bibinfo  {publisher} {Association for Computing Machinery},\ \bibinfo {address} {New York, NY, USA},\ \bibinfo {year} {2025})\ p.\ \bibinfo {pages} {1488–1495}\BibitemShut {NoStop}%
\bibitem [{\citenamefont {Bakshi}\ \emph {et~al.}(2024)\citenamefont {Bakshi}, \citenamefont {Liu}, \citenamefont {Moitra},\ and\ \citenamefont {Tang}}]{bakshi2024hightemperaturegibbsstates}%
  \BibitemOpen
  \bibfield  {author} {\bibinfo {author} {\bibfnamefont {A.}~\bibnamefont {Bakshi}}, \bibinfo {author} {\bibfnamefont {A.}~\bibnamefont {Liu}}, \bibinfo {author} {\bibfnamefont {A.}~\bibnamefont {Moitra}},\ and\ \bibinfo {author} {\bibfnamefont {E.}~\bibnamefont {Tang}},\ }\href {https://arxiv.org/abs/2403.16850} {\bibinfo {title} {High-temperature {G}ibbs states are unentangled and efficiently preparable}} (\bibinfo {year} {2024}),\ \Eprint {https://arxiv.org/abs/2403.16850} {arXiv:2403.16850 [quant-ph]} \BibitemShut {NoStop}%
\bibitem [{\citenamefont {Ding}\ \emph {et~al.}(2024)\citenamefont {Ding}, \citenamefont {Li}, \citenamefont {Lin},\ and\ \citenamefont {Zhang}}]{ding2024preparationlowtemperaturegibbs}%
  \BibitemOpen
  \bibfield  {author} {\bibinfo {author} {\bibfnamefont {Z.}~\bibnamefont {Ding}}, \bibinfo {author} {\bibfnamefont {B.}~\bibnamefont {Li}}, \bibinfo {author} {\bibfnamefont {L.}~\bibnamefont {Lin}},\ and\ \bibinfo {author} {\bibfnamefont {R.}~\bibnamefont {Zhang}},\ }\href {https://arxiv.org/abs/2410.01206} {\bibinfo {title} {Polynomial-time preparation of low-temperature {G}ibbs states for {2D} toric code}} (\bibinfo {year} {2024}),\ \Eprint {https://arxiv.org/abs/2410.01206} {arXiv:2410.01206 [quant-ph]} \BibitemShut {NoStop}%
\bibitem [{\citenamefont {Coopmans}\ and\ \citenamefont {Benedetti}(2024)}]{Coopmans2024qbm_gen_learn}%
  \BibitemOpen
  \bibfield  {author} {\bibinfo {author} {\bibfnamefont {L.}~\bibnamefont {Coopmans}}\ and\ \bibinfo {author} {\bibfnamefont {M.}~\bibnamefont {Benedetti}},\ }\bibfield  {title} {\bibinfo {title} {On the sample complexity of quantum {B}oltzmann machine learning},\ }\href {https://doi.org/10.1038/s42005-024-01763-x} {\bibfield  {journal} {\bibinfo  {journal} {Communications Physics}\ }\textbf {\bibinfo {volume} {7}},\ \bibinfo {pages} {274} (\bibinfo {year} {2024})}\BibitemShut {NoStop}%
\bibitem [{\citenamefont {Tüysüz}\ \emph {et~al.}(2024)\citenamefont {Tüysüz}, \citenamefont {Demidik}, \citenamefont {Coopmans}, \citenamefont {Rinaldi}, \citenamefont {Croft}, \citenamefont {Haddad}, \citenamefont {Rosenkranz},\ and\ \citenamefont {Jansen}}]{Tüysüz2024qbm_gen_model}%
  \BibitemOpen
  \bibfield  {author} {\bibinfo {author} {\bibfnamefont {C.}~\bibnamefont {Tüysüz}}, \bibinfo {author} {\bibfnamefont {M.}~\bibnamefont {Demidik}}, \bibinfo {author} {\bibfnamefont {L.}~\bibnamefont {Coopmans}}, \bibinfo {author} {\bibfnamefont {E.}~\bibnamefont {Rinaldi}}, \bibinfo {author} {\bibfnamefont {V.}~\bibnamefont {Croft}}, \bibinfo {author} {\bibfnamefont {Y.}~\bibnamefont {Haddad}}, \bibinfo {author} {\bibfnamefont {M.}~\bibnamefont {Rosenkranz}},\ and\ \bibinfo {author} {\bibfnamefont {K.}~\bibnamefont {Jansen}},\ }\href {https://doi.org/10.48550/arXiv.2410.16363} {\bibinfo {title} {Learning to generate high-dimensional distributions with low-dimensional quantum {B}oltzmann machines}} (\bibinfo {year} {2024})\BibitemShut {NoStop}%
\bibitem [{\citenamefont {Patel}\ \emph {et~al.}(2024)\citenamefont {Patel}, \citenamefont {Koch}, \citenamefont {Patel},\ and\ \citenamefont {Wilde}}]{patel2024quantumboltzmannmachine}%
  \BibitemOpen
  \bibfield  {author} {\bibinfo {author} {\bibfnamefont {D.}~\bibnamefont {Patel}}, \bibinfo {author} {\bibfnamefont {D.}~\bibnamefont {Koch}}, \bibinfo {author} {\bibfnamefont {S.}~\bibnamefont {Patel}},\ and\ \bibinfo {author} {\bibfnamefont {M.~M.}\ \bibnamefont {Wilde}},\ }\href {https://arxiv.org/abs/2410.12935} {\bibinfo {title} {Quantum {B}oltzmann machine learning of ground-state energies}} (\bibinfo {year} {2024}),\ \Eprint {https://arxiv.org/abs/2410.12935} {arXiv:2410.12935 [quant-ph]} \BibitemShut {NoStop}%
\bibitem [{\citenamefont {Liu}\ \emph {et~al.}(2025{\natexlab{a}})\citenamefont {Liu}, \citenamefont {Minervini}, \citenamefont {Patel},\ and\ \citenamefont {Wilde}}]{liu2025quantumthermodynamicssemidefiniteoptimization}%
  \BibitemOpen
  \bibfield  {author} {\bibinfo {author} {\bibfnamefont {N.}~\bibnamefont {Liu}}, \bibinfo {author} {\bibfnamefont {M.}~\bibnamefont {Minervini}}, \bibinfo {author} {\bibfnamefont {D.}~\bibnamefont {Patel}},\ and\ \bibinfo {author} {\bibfnamefont {M.~M.}\ \bibnamefont {Wilde}},\ }\href {https://arxiv.org/abs/2505.04514} {\bibinfo {title} {Quantum thermodynamics and semi-definite optimization}} (\bibinfo {year} {2025}{\natexlab{a}}),\ \Eprint {https://arxiv.org/abs/2505.04514} {arXiv:2505.04514 [quant-ph]} \BibitemShut {NoStop}%
\bibitem [{\citenamefont {Minervini}\ \emph {et~al.}(2025)\citenamefont {Minervini}, \citenamefont {Chin}, \citenamefont {Kupperman}, \citenamefont {Liu}, \citenamefont {Luo}, \citenamefont {Ly}, \citenamefont {Rethinasamy}, \citenamefont {Wang},\ and\ \citenamefont {Wilde}}]{minervini2025constrainedfreeenergyminimization}%
  \BibitemOpen
  \bibfield  {author} {\bibinfo {author} {\bibfnamefont {M.}~\bibnamefont {Minervini}}, \bibinfo {author} {\bibfnamefont {M.}~\bibnamefont {Chin}}, \bibinfo {author} {\bibfnamefont {J.}~\bibnamefont {Kupperman}}, \bibinfo {author} {\bibfnamefont {N.}~\bibnamefont {Liu}}, \bibinfo {author} {\bibfnamefont {I.}~\bibnamefont {Luo}}, \bibinfo {author} {\bibfnamefont {M.}~\bibnamefont {Ly}}, \bibinfo {author} {\bibfnamefont {S.}~\bibnamefont {Rethinasamy}}, \bibinfo {author} {\bibfnamefont {K.}~\bibnamefont {Wang}},\ and\ \bibinfo {author} {\bibfnamefont {M.~M.}\ \bibnamefont {Wilde}},\ }\href {https://arxiv.org/abs/2508.09103} {\bibinfo {title} {Constrained free energy minimization for the design of thermal states and stabilizer thermodynamic systems}} (\bibinfo {year} {2025}),\ \Eprint {https://arxiv.org/abs/2508.09103} {arXiv:2508.09103 [quant-ph]} \BibitemShut {NoStop}%
\bibitem [{\citenamefont {Patel}\ and\ \citenamefont {Wilde}(2025)}]{patel2024naturalgradientparameterestimation}%
  \BibitemOpen
  \bibfield  {author} {\bibinfo {author} {\bibfnamefont {D.}~\bibnamefont {Patel}}\ and\ \bibinfo {author} {\bibfnamefont {M.~M.}\ \bibnamefont {Wilde}},\ }\bibfield  {title} {\bibinfo {title} {Natural gradient and parameter estimation for quantum {B}oltzmann machines},\ }\href {http://dx.doi.org/10.1103/j8nb-by4l} {\bibfield  {journal} {\bibinfo  {journal} {Physical Review A}\ }\textbf {\bibinfo {volume} {112}},\ \bibinfo {pages} {052421} (\bibinfo {year} {2025})}\BibitemShut {NoStop}%
\bibitem [{\citenamefont {Devulapalli}\ \emph {et~al.}(2026)\citenamefont {Devulapalli}, \citenamefont {Mooney},\ and\ \citenamefont {Watson}}]{devulapalli2026complexitythermalizationfinitequantum}%
  \BibitemOpen
  \bibfield  {author} {\bibinfo {author} {\bibfnamefont {D.}~\bibnamefont {Devulapalli}}, \bibinfo {author} {\bibfnamefont {T.~C.}\ \bibnamefont {Mooney}},\ and\ \bibinfo {author} {\bibfnamefont {J.~D.}\ \bibnamefont {Watson}},\ }\href {https://arxiv.org/abs/2507.00405} {\bibinfo {title} {The complexity of thermalization in finite quantum systems}} (\bibinfo {year} {2026}),\ \Eprint {https://arxiv.org/abs/2507.00405} {arXiv:2507.00405 [quant-ph]} \BibitemShut {NoStop}%
\bibitem [{\citenamefont {Lloyd}(1996)}]{lloyd1996universal}%
  \BibitemOpen
  \bibfield  {author} {\bibinfo {author} {\bibfnamefont {S.}~\bibnamefont {Lloyd}},\ }\bibfield  {title} {\bibinfo {title} {Universal quantum simulators},\ }\href {https://doi.org/10.1126/science.273.5278.1073} {\bibfield  {journal} {\bibinfo  {journal} {Science}\ }\textbf {\bibinfo {volume} {273}},\ \bibinfo {pages} {1073} (\bibinfo {year} {1996})}\BibitemShut {NoStop}%
\bibitem [{\citenamefont {Childs}\ \emph {et~al.}(2018)\citenamefont {Childs}, \citenamefont {Maslov}, \citenamefont {Nam}, \citenamefont {Ross},\ and\ \citenamefont {Su}}]{childs2018toward}%
  \BibitemOpen
  \bibfield  {author} {\bibinfo {author} {\bibfnamefont {A.~M.}\ \bibnamefont {Childs}}, \bibinfo {author} {\bibfnamefont {D.}~\bibnamefont {Maslov}}, \bibinfo {author} {\bibfnamefont {Y.}~\bibnamefont {Nam}}, \bibinfo {author} {\bibfnamefont {N.~J.}\ \bibnamefont {Ross}},\ and\ \bibinfo {author} {\bibfnamefont {Y.}~\bibnamefont {Su}},\ }\bibfield  {title} {\bibinfo {title} {Toward the first quantum simulation with quantum speedup},\ }\href {https://doi.org/10.1073/pnas.1801723115} {\bibfield  {journal} {\bibinfo  {journal} {Proceedings of the National Academy of Sciences}\ }\textbf {\bibinfo {volume} {115}},\ \bibinfo {pages} {9456} (\bibinfo {year} {2018})}\BibitemShut {NoStop}%
\bibitem [{\citenamefont {Cleve}\ \emph {et~al.}(1998)\citenamefont {Cleve}, \citenamefont {Ekert}, \citenamefont {Macchiavello},\ and\ \citenamefont {Mosca}}]{Cleve1998}%
  \BibitemOpen
  \bibfield  {author} {\bibinfo {author} {\bibfnamefont {R.}~\bibnamefont {Cleve}}, \bibinfo {author} {\bibfnamefont {A.}~\bibnamefont {Ekert}}, \bibinfo {author} {\bibfnamefont {C.}~\bibnamefont {Macchiavello}},\ and\ \bibinfo {author} {\bibfnamefont {M.}~\bibnamefont {Mosca}},\ }\bibfield  {title} {\bibinfo {title} {Quantum algorithms revisited},\ }\href {https://doi.org/10.1098/rspa.1998.0164} {\bibfield  {journal} {\bibinfo  {journal} {Proceedings of the Royal Society A}\ }\textbf {\bibinfo {volume} {454}},\ \bibinfo {pages} {339} (\bibinfo {year} {1998})}\BibitemShut {NoStop}%
\bibitem [{\citenamefont {Chen}\ \emph {et~al.}(2025)\citenamefont {Chen}, \citenamefont {Westerheim}, \citenamefont {Holmes}, \citenamefont {Luo}, \citenamefont {Nuradha}, \citenamefont {Patel}, \citenamefont {Rethinasamy}, \citenamefont {Wang},\ and\ \citenamefont {Wilde}}]{chen2023qslackslackvariableapproachvariational}%
  \BibitemOpen
  \bibfield  {author} {\bibinfo {author} {\bibfnamefont {J.}~\bibnamefont {Chen}}, \bibinfo {author} {\bibfnamefont {H.}~\bibnamefont {Westerheim}}, \bibinfo {author} {\bibfnamefont {Z.}~\bibnamefont {Holmes}}, \bibinfo {author} {\bibfnamefont {I.}~\bibnamefont {Luo}}, \bibinfo {author} {\bibfnamefont {T.}~\bibnamefont {Nuradha}}, \bibinfo {author} {\bibfnamefont {D.}~\bibnamefont {Patel}}, \bibinfo {author} {\bibfnamefont {S.}~\bibnamefont {Rethinasamy}}, \bibinfo {author} {\bibfnamefont {K.}~\bibnamefont {Wang}},\ and\ \bibinfo {author} {\bibfnamefont {M.~M.}\ \bibnamefont {Wilde}},\ }\bibfield  {title} {\bibinfo {title} {Qslack: A slack-variable approach for variational quantum semi-definite programming},\ }\href {http://dx.doi.org/10.1103/lwxq-4myj} {\bibfield  {journal} {\bibinfo  {journal} {Physical Review A}\ }\textbf {\bibinfo {volume} {112}},\ \bibinfo {pages} {022607} (\bibinfo {year} {2025})}\BibitemShut {NoStop}%
\bibitem [{\citenamefont {Wecker}\ \emph {et~al.}(2015)\citenamefont {Wecker}, \citenamefont {Hastings},\ and\ \citenamefont {Troyer}}]{WH2015}%
  \BibitemOpen
  \bibfield  {author} {\bibinfo {author} {\bibfnamefont {D.}~\bibnamefont {Wecker}}, \bibinfo {author} {\bibfnamefont {M.~B.}\ \bibnamefont {Hastings}},\ and\ \bibinfo {author} {\bibfnamefont {M.}~\bibnamefont {Troyer}},\ }\bibfield  {title} {\bibinfo {title} {Progress towards practical quantum variational algorithms},\ }\href {https://doi.org/10.1103/PhysRevA.92.042303} {\bibfield  {journal} {\bibinfo  {journal} {Physical Review A}\ }\textbf {\bibinfo {volume} {92}},\ \bibinfo {pages} {042303} (\bibinfo {year} {2015})}\BibitemShut {NoStop}%
\bibitem [{\citenamefont {Banchi}\ and\ \citenamefont {Crooks}(2021)}]{Banchi2021measuringanalytic}%
  \BibitemOpen
  \bibfield  {author} {\bibinfo {author} {\bibfnamefont {L.}~\bibnamefont {Banchi}}\ and\ \bibinfo {author} {\bibfnamefont {G.~E.}\ \bibnamefont {Crooks}},\ }\bibfield  {title} {\bibinfo {title} {Measuring analytic gradients of general quantum evolution with the stochastic parameter shift rule},\ }\href {https://doi.org/10.22331/q-2021-01-25-386} {\bibfield  {journal} {\bibinfo  {journal} {{Quantum}}\ }\textbf {\bibinfo {volume} {5}},\ \bibinfo {pages} {386} (\bibinfo {year} {2021})}\BibitemShut {NoStop}%
\bibitem [{\citenamefont {Wiersema}\ \emph {et~al.}(2024)\citenamefont {Wiersema}, \citenamefont {Lewis}, \citenamefont {Wierichs}, \citenamefont {Carrasquilla},\ and\ \citenamefont {Killoran}}]{Wiersema2024herecomessun}%
  \BibitemOpen
  \bibfield  {author} {\bibinfo {author} {\bibfnamefont {R.}~\bibnamefont {Wiersema}}, \bibinfo {author} {\bibfnamefont {D.}~\bibnamefont {Lewis}}, \bibinfo {author} {\bibfnamefont {D.}~\bibnamefont {Wierichs}}, \bibinfo {author} {\bibfnamefont {J.}~\bibnamefont {Carrasquilla}},\ and\ \bibinfo {author} {\bibfnamefont {N.}~\bibnamefont {Killoran}},\ }\bibfield  {title} {\bibinfo {title} {Here comes the {SU}({N}): multivariate quantum gates and gradients},\ }\href {https://doi.org/10.22331/q-2024-03-07-1275} {\bibfield  {journal} {\bibinfo  {journal} {{Quantum}}\ }\textbf {\bibinfo {volume} {8}},\ \bibinfo {pages} {1275} (\bibinfo {year} {2024})}\BibitemShut {NoStop}%
\bibitem [{\citenamefont {Uhlmann}(1976)}]{Uhl76}%
  \BibitemOpen
  \bibfield  {author} {\bibinfo {author} {\bibfnamefont {A.}~\bibnamefont {Uhlmann}},\ }\bibfield  {title} {\bibinfo {title} {The `transition probability' in the state space of a *-algebra},\ }\href {https://www.sciencedirect.com/science/article/pii/0034487776900604} {\bibfield  {journal} {\bibinfo  {journal} {Reports on Mathematical Physics}\ }\textbf {\bibinfo {volume} {9}},\ \bibinfo {pages} {273} (\bibinfo {year} {1976})}\BibitemShut {NoStop}%
\bibitem [{\citenamefont {Holevo}(1972)}]{Kholevo1972}%
  \BibitemOpen
  \bibfield  {author} {\bibinfo {author} {\bibfnamefont {A.~S.}\ \bibnamefont {Holevo}},\ }\bibfield  {title} {\bibinfo {title} {On quasiequivalence of locally normal states},\ }\href {https://doi.org/10.1007/BF01035528} {\bibfield  {journal} {\bibinfo  {journal} {Theoretical and Mathematical Physics}\ }\textbf {\bibinfo {volume} {13}},\ \bibinfo {pages} {1071} (\bibinfo {year} {1972})}\BibitemShut {NoStop}%
\bibitem [{\citenamefont {Umegaki}(1962)}]{Umegaki62}%
  \BibitemOpen
  \bibfield  {author} {\bibinfo {author} {\bibfnamefont {H.}~\bibnamefont {Umegaki}},\ }\bibfield  {title} {\bibinfo {title} {Conditional expectations in an operator algebra {IV} (entropy and information)},\ }\href {https://doi.org/10.2996/kmj/1138844604} {\bibfield  {journal} {\bibinfo  {journal} {Kodai Mathematical Seminar Reports}\ }\textbf {\bibinfo {volume} {14}},\ \bibinfo {pages} {59} (\bibinfo {year} {1962})}\BibitemShut {NoStop}%
\bibitem [{\citenamefont {Luo}(2004)}]{Luo2004}%
  \BibitemOpen
  \bibfield  {author} {\bibinfo {author} {\bibfnamefont {S.}~\bibnamefont {Luo}},\ }\bibfield  {title} {\bibinfo {title} {{W}igner-{Y}anase skew information vs. quantum {F}isher information},\ }\href {http://www.jstor.org/stable/1194711} {\bibfield  {journal} {\bibinfo  {journal} {Proceedings of the American Mathematical Society}\ }\textbf {\bibinfo {volume} {132}},\ \bibinfo {pages} {885} (\bibinfo {year} {2004})}\BibitemShut {NoStop}%
\bibitem [{\citenamefont {Anshu}\ \emph {et~al.}(2021)\citenamefont {Anshu}, \citenamefont {Arunachalam}, \citenamefont {Kuwahara},\ and\ \citenamefont {Soleimanifar}}]{Anshu2021ham_learning_qbm}%
  \BibitemOpen
  \bibfield  {author} {\bibinfo {author} {\bibfnamefont {A.}~\bibnamefont {Anshu}}, \bibinfo {author} {\bibfnamefont {S.}~\bibnamefont {Arunachalam}}, \bibinfo {author} {\bibfnamefont {T.}~\bibnamefont {Kuwahara}},\ and\ \bibinfo {author} {\bibfnamefont {M.}~\bibnamefont {Soleimanifar}},\ }\bibfield  {title} {\bibinfo {title} {Sample-efficient learning of interacting quantum systems},\ }\href {https://doi.org/10.1038/s41567-021-01232-0} {\bibfield  {journal} {\bibinfo  {journal} {Nature Physics}\ }\textbf {\bibinfo {volume} {17}},\ \bibinfo {pages} {931–935} (\bibinfo {year} {2021})}\BibitemShut {NoStop}%
\bibitem [{\citenamefont {Hastings}(2007)}]{Hastings2007}%
  \BibitemOpen
  \bibfield  {author} {\bibinfo {author} {\bibfnamefont {M.~B.}\ \bibnamefont {Hastings}},\ }\bibfield  {title} {\bibinfo {title} {Quantum belief propagation: An algorithm for thermal quantum systems},\ }\href {https://doi.org/10.1103/PhysRevB.76.201102} {\bibfield  {journal} {\bibinfo  {journal} {Physical Review B}\ }\textbf {\bibinfo {volume} {76}},\ \bibinfo {pages} {201102} (\bibinfo {year} {2007})}\BibitemShut {NoStop}%
\bibitem [{\citenamefont {Kim}(2012)}]{Kim2012}%
  \BibitemOpen
  \bibfield  {author} {\bibinfo {author} {\bibfnamefont {I.~H.}\ \bibnamefont {Kim}},\ }\bibfield  {title} {\bibinfo {title} {Perturbative analysis of topological entanglement entropy from conditional independence},\ }\href {https://doi.org/10.1103/PhysRevB.86.245116} {\bibfield  {journal} {\bibinfo  {journal} {Physical Review B}\ }\textbf {\bibinfo {volume} {86}},\ \bibinfo {pages} {245116} (\bibinfo {year} {2012})}\BibitemShut {NoStop}%
\bibitem [{\citenamefont {Ho}(2023)}]{Ho2023}%
  \BibitemOpen
  \bibfield  {author} {\bibinfo {author} {\bibfnamefont {L.~B.}\ \bibnamefont {Ho}},\ }\bibfield  {title} {\bibinfo {title} {A stochastic evaluation of quantum {F}isher information matrix with generic {H}amiltonians},\ }\href {https://doi.org/10.1140/epjqt/s40507-023-00195-w} {\bibfield  {journal} {\bibinfo  {journal} {EPJ Quantum Technology}\ }\textbf {\bibinfo {volume} {10}},\ \bibinfo {pages} {37} (\bibinfo {year} {2023})}\BibitemShut {NoStop}%
\bibitem [{\citenamefont {Low}\ and\ \citenamefont {Chuang}(2019)}]{Low2019hamiltonian}%
  \BibitemOpen
  \bibfield  {author} {\bibinfo {author} {\bibfnamefont {G.~H.}\ \bibnamefont {Low}}\ and\ \bibinfo {author} {\bibfnamefont {I.~L.}\ \bibnamefont {Chuang}},\ }\bibfield  {title} {\bibinfo {title} {Hamiltonian simulation by qubitization},\ }\href {https://doi.org/10.22331/q-2019-07-12-163} {\bibfield  {journal} {\bibinfo  {journal} {{Quantum}}\ }\textbf {\bibinfo {volume} {3}},\ \bibinfo {pages} {163} (\bibinfo {year} {2019})}\BibitemShut {NoStop}%
\bibitem [{\citenamefont {Gily\'{e}n}\ \emph {et~al.}(2019)\citenamefont {Gily\'{e}n}, \citenamefont {Su}, \citenamefont {Low},\ and\ \citenamefont {Wiebe}}]{Gilyen2019}%
  \BibitemOpen
  \bibfield  {author} {\bibinfo {author} {\bibfnamefont {A.}~\bibnamefont {Gily\'{e}n}}, \bibinfo {author} {\bibfnamefont {Y.}~\bibnamefont {Su}}, \bibinfo {author} {\bibfnamefont {G.~H.}\ \bibnamefont {Low}},\ and\ \bibinfo {author} {\bibfnamefont {N.}~\bibnamefont {Wiebe}},\ }\bibfield  {title} {\bibinfo {title} {Quantum singular value transformation and beyond: exponential improvements for quantum matrix arithmetics},\ }in\ \href {https://doi.org/10.1145/3313276.3316366} {\emph {\bibinfo {booktitle} {Proceedings of the 51st Annual ACM SIGACT Symposium on Theory of Computing}}},\ \bibinfo {series and number} {STOC 2019}\ (\bibinfo  {publisher} {Association for Computing Machinery},\ \bibinfo {address} {New York, NY, USA},\ \bibinfo {year} {2019})\ pp.\ \bibinfo {pages} {193--204}\BibitemShut {NoStop}%
\bibitem [{\citenamefont {Wilde}(2017)}]{Wbook17}%
  \BibitemOpen
  \bibfield  {author} {\bibinfo {author} {\bibfnamefont {M.~M.}\ \bibnamefont {Wilde}},\ }\href {https://doi.org/10.1017/CBO9781139525343} {\emph {\bibinfo {title} {Quantum Information Theory}}},\ \bibinfo {edition} {2nd}\ ed.\ (\bibinfo  {publisher} {Cambridge University Press},\ \bibinfo {year} {2017})\ \Eprint {https://arxiv.org/abs/1106.1445} {arXiv:1106.1445} \BibitemShut {NoStop}%
\bibitem [{\citenamefont {Bengtsson}\ and\ \citenamefont {Zyczkowski}(2006)}]{Bengtsson2006}%
  \BibitemOpen
  \bibfield  {author} {\bibinfo {author} {\bibfnamefont {I.}~\bibnamefont {Bengtsson}}\ and\ \bibinfo {author} {\bibfnamefont {K.}~\bibnamefont {Zyczkowski}},\ }\href {https://doi.org/10.1017/CBO9780511535048} {\emph {\bibinfo {title} {Geometry of Quantum States: An Introduction to Quantum Entanglement}}}\ (\bibinfo  {publisher} {Cambridge University Press},\ \bibinfo {year} {2006})\BibitemShut {NoStop}%
\bibitem [{\citenamefont {Liu}\ \emph {et~al.}(2019)\citenamefont {Liu}, \citenamefont {Yuan}, \citenamefont {Lu},\ and\ \citenamefont {Wang}}]{Liu2019}%
  \BibitemOpen
  \bibfield  {author} {\bibinfo {author} {\bibfnamefont {J.}~\bibnamefont {Liu}}, \bibinfo {author} {\bibfnamefont {H.}~\bibnamefont {Yuan}}, \bibinfo {author} {\bibfnamefont {X.-M.}\ \bibnamefont {Lu}},\ and\ \bibinfo {author} {\bibfnamefont {X.}~\bibnamefont {Wang}},\ }\bibfield  {title} {\bibinfo {title} {Quantum {F}isher information matrix and multiparameter estimation},\ }\href {https://doi.org/10.1088/1751-8121/ab5d4d} {\bibfield  {journal} {\bibinfo  {journal} {Journal of Physics A: Mathematical and Theoretical}\ }\textbf {\bibinfo {volume} {53}},\ \bibinfo {pages} {023001} (\bibinfo {year} {2019})}\BibitemShut {NoStop}%
\bibitem [{\citenamefont {Sidhu}\ and\ \citenamefont {Kok}(2020)}]{Sidhu2020}%
  \BibitemOpen
  \bibfield  {author} {\bibinfo {author} {\bibfnamefont {J.~S.}\ \bibnamefont {Sidhu}}\ and\ \bibinfo {author} {\bibfnamefont {P.}~\bibnamefont {Kok}},\ }\bibfield  {title} {\bibinfo {title} {Geometric perspective on quantum parameter estimation},\ }\href {https://doi.org/10.1116/1.5119961} {\bibfield  {journal} {\bibinfo  {journal} {AVS Quantum Science}\ }\textbf {\bibinfo {volume} {2}},\ \bibinfo {pages} {014701} (\bibinfo {year} {2020})}\BibitemShut {NoStop}%
\bibitem [{\citenamefont {Jarzyna}\ and\ \citenamefont {Kolodynski}(2020)}]{Jarzyna2020}%
  \BibitemOpen
  \bibfield  {author} {\bibinfo {author} {\bibfnamefont {M.}~\bibnamefont {Jarzyna}}\ and\ \bibinfo {author} {\bibfnamefont {J.}~\bibnamefont {Kolodynski}},\ }\bibfield  {title} {\bibinfo {title} {Geometric approach to quantum statistical inference},\ }\href {https://doi.org/10.1109/JSAIT.2020.3017469} {\bibfield  {journal} {\bibinfo  {journal} {IEEE Journal on Selected Areas in Information Theory}\ }\textbf {\bibinfo {volume} {1}},\ \bibinfo {pages} {367} (\bibinfo {year} {2020})}\BibitemShut {NoStop}%
\bibitem [{\citenamefont {Meyer}(2021)}]{Meyer2021fisherinformationin}%
  \BibitemOpen
  \bibfield  {author} {\bibinfo {author} {\bibfnamefont {J.~J.}\ \bibnamefont {Meyer}},\ }\bibfield  {title} {\bibinfo {title} {Fisher information in noisy intermediate-scale quantum applications},\ }\href {https://doi.org/10.22331/q-2021-09-09-539} {\bibfield  {journal} {\bibinfo  {journal} {{Quantum}}\ }\textbf {\bibinfo {volume} {5}},\ \bibinfo {pages} {539} (\bibinfo {year} {2021})}\BibitemShut {NoStop}%
\bibitem [{\citenamefont {Sbahi}\ \emph {et~al.}(2022)\citenamefont {Sbahi}, \citenamefont {Martinez}, \citenamefont {Patel}, \citenamefont {Saberi}, \citenamefont {Yoo}, \citenamefont {Roeder},\ and\ \citenamefont {Verdon}}]{sbahi2022provablyefficientvariationalgenerative}%
  \BibitemOpen
  \bibfield  {author} {\bibinfo {author} {\bibfnamefont {F.~M.}\ \bibnamefont {Sbahi}}, \bibinfo {author} {\bibfnamefont {A.~J.}\ \bibnamefont {Martinez}}, \bibinfo {author} {\bibfnamefont {S.}~\bibnamefont {Patel}}, \bibinfo {author} {\bibfnamefont {D.}~\bibnamefont {Saberi}}, \bibinfo {author} {\bibfnamefont {J.~H.}\ \bibnamefont {Yoo}}, \bibinfo {author} {\bibfnamefont {G.}~\bibnamefont {Roeder}},\ and\ \bibinfo {author} {\bibfnamefont {G.}~\bibnamefont {Verdon}},\ }\href {https://arxiv.org/abs/2206.04663v1} {\bibinfo {title} {Provably efficient variational generative modeling of quantum many-body systems via quantum-probabilistic information geometry}} (\bibinfo {year} {2022}),\ \Eprint {https://arxiv.org/abs/2206.04663v1} {arXiv:2206.04663v1 [quant-ph]} \BibitemShut {NoStop}%
\bibitem [{\citenamefont {Scandi}\ \emph {et~al.}(2025)\citenamefont {Scandi}, \citenamefont {Abiuso}, \citenamefont {Surace},\ and\ \citenamefont {De~Santis}}]{scandi2024quantumfisherinformationdynamical}%
  \BibitemOpen
  \bibfield  {author} {\bibinfo {author} {\bibfnamefont {M.}~\bibnamefont {Scandi}}, \bibinfo {author} {\bibfnamefont {P.}~\bibnamefont {Abiuso}}, \bibinfo {author} {\bibfnamefont {J.}~\bibnamefont {Surace}},\ and\ \bibinfo {author} {\bibfnamefont {D.}~\bibnamefont {De~Santis}},\ }\bibfield  {title} {\bibinfo {title} {Quantum {F}isher information and its dynamical nature},\ }\href {https://doi.org/10.1088/1361-6633/ade453} {\bibfield  {journal} {\bibinfo  {journal} {Reports on Progress in Physics}\ }\textbf {\bibinfo {volume} {88}},\ \bibinfo {pages} {076001} (\bibinfo {year} {2025})}\BibitemShut {NoStop}%
\bibitem [{\citenamefont {Petz}(1985)}]{P85}%
  \BibitemOpen
  \bibfield  {author} {\bibinfo {author} {\bibfnamefont {D.}~\bibnamefont {Petz}},\ }\bibfield  {title} {\bibinfo {title} {Quasi-entropies for states of a von {N}eumann algebra},\ }\href {https://doi.org/10.2977/prims/1195178929} {\bibfield  {journal} {\bibinfo  {journal} {Publications of the Research Institute for Mathematical Sciences}\ }\textbf {\bibinfo {volume} {21}},\ \bibinfo {pages} {787} (\bibinfo {year} {1985})}\BibitemShut {NoStop}%
\bibitem [{\citenamefont {Petz}(1986)}]{P86}%
  \BibitemOpen
  \bibfield  {author} {\bibinfo {author} {\bibfnamefont {D.}~\bibnamefont {Petz}},\ }\bibfield  {title} {\bibinfo {title} {Quasi-entropies for finite quantum systems},\ }\href {https://doi.org/10.1016/0034-4877(86)90067-4} {\bibfield  {journal} {\bibinfo  {journal} {Reports in Mathematical Physics}\ }\textbf {\bibinfo {volume} {23}},\ \bibinfo {pages} {57} (\bibinfo {year} {1986})}\BibitemShut {NoStop}%
\bibitem [{\citenamefont {{M\"uller}-Lennert}\ \emph {et~al.}(2013)\citenamefont {{M\"uller}-Lennert}, \citenamefont {Dupuis}, \citenamefont {Szehr}, \citenamefont {Fehr},\ and\ \citenamefont {Tomamichel}}]{MDSFT13}%
  \BibitemOpen
  \bibfield  {author} {\bibinfo {author} {\bibfnamefont {M.}~\bibnamefont {{M\"uller}-Lennert}}, \bibinfo {author} {\bibfnamefont {F.}~\bibnamefont {Dupuis}}, \bibinfo {author} {\bibfnamefont {O.}~\bibnamefont {Szehr}}, \bibinfo {author} {\bibfnamefont {S.}~\bibnamefont {Fehr}},\ and\ \bibinfo {author} {\bibfnamefont {M.}~\bibnamefont {Tomamichel}},\ }\bibfield  {title} {\bibinfo {title} {On quantum {R\'enyi} entropies: a new generalization and some properties},\ }\href {https://doi.org/10.1063/1.4838856} {\bibfield  {journal} {\bibinfo  {journal} {Journal of Mathematical Physics}\ }\textbf {\bibinfo {volume} {54}},\ \bibinfo {pages} {122203} (\bibinfo {year} {2013})},\ \Eprint {https://arxiv.org/abs/1306.3142} {arXiv:1306.3142} \BibitemShut {NoStop}%
\bibitem [{\citenamefont {Wilde}\ \emph {et~al.}(2014)\citenamefont {Wilde}, \citenamefont {Winter},\ and\ \citenamefont {Yang}}]{WWY14}%
  \BibitemOpen
  \bibfield  {author} {\bibinfo {author} {\bibfnamefont {M.~M.}\ \bibnamefont {Wilde}}, \bibinfo {author} {\bibfnamefont {A.}~\bibnamefont {Winter}},\ and\ \bibinfo {author} {\bibfnamefont {D.}~\bibnamefont {Yang}},\ }\bibfield  {title} {\bibinfo {title} {Strong converse for the classical capacity of entanglement-breaking and {Hadamard} channels via a sandwiched {R\'enyi} relative entropy},\ }\href {https://doi.org/10.1007/s00220-014-2122-x} {\bibfield  {journal} {\bibinfo  {journal} {Communications in Mathematical Physics}\ }\textbf {\bibinfo {volume} {331}},\ \bibinfo {pages} {593} (\bibinfo {year} {2014})},\ \Eprint {https://arxiv.org/abs/1306.1586} {arXiv:1306.1586} \BibitemShut {NoStop}%
\bibitem [{\citenamefont {Lindblad}(1975)}]{Lindblad1975}%
  \BibitemOpen
  \bibfield  {author} {\bibinfo {author} {\bibfnamefont {G.}~\bibnamefont {Lindblad}},\ }\bibfield  {title} {\bibinfo {title} {Completely positive maps and entropy inequalities},\ }\href {https://doi.org/10.1007/BF01609396} {\bibfield  {journal} {\bibinfo  {journal} {Communications in Mathematical Physics}\ }\textbf {\bibinfo {volume} {40}},\ \bibinfo {pages} {147} (\bibinfo {year} {1975})}\BibitemShut {NoStop}%
\bibitem [{\citenamefont {Frank}\ and\ \citenamefont {Lieb}(2013)}]{FL13}%
  \BibitemOpen
  \bibfield  {author} {\bibinfo {author} {\bibfnamefont {R.~L.}\ \bibnamefont {Frank}}\ and\ \bibinfo {author} {\bibfnamefont {E.~H.}\ \bibnamefont {Lieb}},\ }\bibfield  {title} {\bibinfo {title} {Monotonicity of a relative {R\'enyi} entropy},\ }\href {https://doi.org/10.1063/1.4838835} {\bibfield  {journal} {\bibinfo  {journal} {Journal of Mathematical Physics}\ }\textbf {\bibinfo {volume} {54}},\ \bibinfo {pages} {122201} (\bibinfo {year} {2013})},\ \Eprint {https://arxiv.org/abs/1306.5358} {arXiv:1306.5358} \BibitemShut {NoStop}%
\bibitem [{\citenamefont {Wilde}(2018)}]{W18opt}%
  \BibitemOpen
  \bibfield  {author} {\bibinfo {author} {\bibfnamefont {M.~M.}\ \bibnamefont {Wilde}},\ }\bibfield  {title} {\bibinfo {title} {Optimized quantum $f$-divergences and data processing},\ }\href {https://doi.org/10.1088/1751-8121/aad5a1} {\bibfield  {journal} {\bibinfo  {journal} {Journal of Physics A}\ }\textbf {\bibinfo {volume} {51}},\ \bibinfo {pages} {374002} (\bibinfo {year} {2018})},\ \Eprint {https://arxiv.org/abs/1710.10252} {arXiv:1710.10252} \BibitemShut {NoStop}%
\bibitem [{\citenamefont {Audenaert}\ \emph {et~al.}(2008)\citenamefont {Audenaert}, \citenamefont {Nussbaum}, \citenamefont {Szko{\l}a},\ and\ \citenamefont {Verstraete}}]{audenaert2008asymptotic}%
  \BibitemOpen
  \bibfield  {author} {\bibinfo {author} {\bibfnamefont {K.~M.~R.}\ \bibnamefont {Audenaert}}, \bibinfo {author} {\bibfnamefont {M.}~\bibnamefont {Nussbaum}}, \bibinfo {author} {\bibfnamefont {A.}~\bibnamefont {Szko{\l}a}},\ and\ \bibinfo {author} {\bibfnamefont {F.}~\bibnamefont {Verstraete}},\ }\bibfield  {title} {\bibinfo {title} {Asymptotic error rates in quantum hypothesis testing},\ }\href {https://doi.org/10.1007/s00220-008-0417-5} {\bibfield  {journal} {\bibinfo  {journal} {Communications in Mathematical Physics}\ }\textbf {\bibinfo {volume} {279}},\ \bibinfo {pages} {251} (\bibinfo {year} {2008})},\ \Eprint {https://arxiv.org/abs/0708.4282} {arXiv:0708.4282} \BibitemShut {NoStop}%
\bibitem [{\citenamefont {{{Hubner}}}(1992)}]{Hub92}%
  \BibitemOpen
  \bibfield  {author} {\bibinfo {author} {\bibfnamefont {M.}~\bibnamefont {{{Hubner}}}},\ }\bibfield  {title} {\bibinfo {title} {Explicit computation of the {B}ures distance for density matrices},\ }\href {https://doi.org/https://doi.org/10.1016/0375-9601(92)91004-B} {\bibfield  {journal} {\bibinfo  {journal} {Physics Letters A}\ }\textbf {\bibinfo {volume} {163}},\ \bibinfo {pages} {239} (\bibinfo {year} {1992})}\BibitemShut {NoStop}%
\bibitem [{\citenamefont {Gibilisco}\ and\ \citenamefont {Isola}(2003)}]{Gibilisco2003}%
  \BibitemOpen
  \bibfield  {author} {\bibinfo {author} {\bibfnamefont {P.}~\bibnamefont {Gibilisco}}\ and\ \bibinfo {author} {\bibfnamefont {T.}~\bibnamefont {Isola}},\ }\bibfield  {title} {\bibinfo {title} {{Wigner--Yanase} information on quantum state space: The geometric approach},\ }\href {https://doi.org/10.1063/1.1598279} {\bibfield  {journal} {\bibinfo  {journal} {Journal of Mathematical Physics}\ }\textbf {\bibinfo {volume} {44}},\ \bibinfo {pages} {3752} (\bibinfo {year} {2003})}\BibitemShut {NoStop}%
\bibitem [{\citenamefont {Petz}\ and\ \citenamefont {Toth}(1993)}]{PetzToth1993}%
  \BibitemOpen
  \bibfield  {author} {\bibinfo {author} {\bibfnamefont {D.}~\bibnamefont {Petz}}\ and\ \bibinfo {author} {\bibfnamefont {G.}~\bibnamefont {Toth}},\ }\bibfield  {title} {\bibinfo {title} {The {B}ogoliubov inner product in quantum statistics},\ }\href {https://doi.org/10.1007/BF00739578} {\bibfield  {journal} {\bibinfo  {journal} {Letters in Mathematical Physics}\ }\textbf {\bibinfo {volume} {27}},\ \bibinfo {pages} {205} (\bibinfo {year} {1993})}\BibitemShut {NoStop}%
\bibitem [{\citenamefont {van Erven}\ and\ \citenamefont {Harremos}(2014)}]{vEH2014}%
  \BibitemOpen
  \bibfield  {author} {\bibinfo {author} {\bibfnamefont {T.}~\bibnamefont {van Erven}}\ and\ \bibinfo {author} {\bibfnamefont {P.}~\bibnamefont {Harremos}},\ }\bibfield  {title} {\bibinfo {title} {Rényi divergence and {K}ullback-{L}eibler divergence},\ }\href {https://doi.org/10.1109/TIT.2014.2320500} {\bibfield  {journal} {\bibinfo  {journal} {IEEE Transactions on Information Theory}\ }\textbf {\bibinfo {volume} {60}},\ \bibinfo {pages} {3797} (\bibinfo {year} {2014})}\BibitemShut {NoStop}%
\bibitem [{\citenamefont {Matsumoto}(2010)}]{Mat10fid}%
  \BibitemOpen
  \bibfield  {author} {\bibinfo {author} {\bibfnamefont {K.}~\bibnamefont {Matsumoto}},\ }\href@noop {} {\bibinfo {title} {Reverse test and quantum analogue of classical fidelity and generalized fidelity}} (\bibinfo {year} {2010}),\ \Eprint {https://arxiv.org/abs/1006.0302} {arXiv:1006.0302} \BibitemShut {NoStop}%
\bibitem [{\citenamefont {Matsumoto}(2013)}]{Mat13}%
  \BibitemOpen
  \bibfield  {author} {\bibinfo {author} {\bibfnamefont {K.}~\bibnamefont {Matsumoto}},\ }\href@noop {} {\bibinfo {title} {A new quantum version of $f$-divergence}} (\bibinfo {year} {2013}),\ \Eprint {https://arxiv.org/abs/1311.4722} {arXiv:1311.4722} \BibitemShut {NoStop}%
\bibitem [{\citenamefont {Matsumoto}(2018)}]{Matsumoto2018}%
  \BibitemOpen
  \bibfield  {author} {\bibinfo {author} {\bibfnamefont {K.}~\bibnamefont {Matsumoto}},\ }\bibfield  {title} {\bibinfo {title} {A new quantum version of $f$-divergence},\ }in\ \href {https://doi.org/10.1007/978-981-13-2487-1_10} {\emph {\bibinfo {booktitle} {Reality and Measurement in Algebraic Quantum Theory}}},\ Vol.\ \bibinfo {volume} {261},\ \bibinfo {editor} {edited by\ \bibinfo {editor} {\bibfnamefont {M.}~\bibnamefont {Ozawa}}, \bibinfo {editor} {\bibfnamefont {J.}~\bibnamefont {Butterfield}}, \bibinfo {editor} {\bibfnamefont {H.}~\bibnamefont {Halvorson}}, \bibinfo {editor} {\bibfnamefont {M.}~\bibnamefont {R{\'e}dei}}, \bibinfo {editor} {\bibfnamefont {Y.}~\bibnamefont {Kitajima}},\ and\ \bibinfo {editor} {\bibfnamefont {F.}~\bibnamefont {Buscemi}}}\ (\bibinfo  {publisher} {Springer Singapore},\ \bibinfo {address} {Singapore},\ \bibinfo {year} {2018})\ pp.\ \bibinfo {pages} {229--273},\ \bibinfo {note} {series Title: Springer Proceedings in Mathematics \& Statistics}\BibitemShut {NoStop}%
\bibitem [{\citenamefont {Crooks}(2007)}]{Crooks2007}%
  \BibitemOpen
  \bibfield  {author} {\bibinfo {author} {\bibfnamefont {G.~E.}\ \bibnamefont {Crooks}},\ }\bibfield  {title} {\bibinfo {title} {Measuring thermodynamic length},\ }\href {https://doi.org/10.1103/PhysRevLett.99.100602} {\bibfield  {journal} {\bibinfo  {journal} {Physical Review Letters}\ }\textbf {\bibinfo {volume} {99}},\ \bibinfo {pages} {100602} (\bibinfo {year} {2007})}\BibitemShut {NoStop}%
\bibitem [{\citenamefont {Li}\ \emph {et~al.}(2025)\citenamefont {Li}, \citenamefont {Dulal}, \citenamefont {Ohorodnikov}, \citenamefont {Wang},\ and\ \citenamefont {Ding}}]{liEfficientQuantumGradient2024}%
  \BibitemOpen
  \bibfield  {author} {\bibinfo {author} {\bibfnamefont {D.}~\bibnamefont {Li}}, \bibinfo {author} {\bibfnamefont {D.}~\bibnamefont {Dulal}}, \bibinfo {author} {\bibfnamefont {M.}~\bibnamefont {Ohorodnikov}}, \bibinfo {author} {\bibfnamefont {H.}~\bibnamefont {Wang}},\ and\ \bibinfo {author} {\bibfnamefont {Y.}~\bibnamefont {Ding}},\ }\bibfield  {title} {\bibinfo {title} {Efficient quantum gradient and higher-order derivative estimation via generalized {H}adamard test},\ }in\ \href {https://doi.ieeecomputersociety.org/10.1109/QCE65121.2025.00024} {\emph {\bibinfo {booktitle} {2025 IEEE International Conference on Quantum Computing and Engineering (QCE)}}}\ (\bibinfo  {publisher} {IEEE Computer Society},\ \bibinfo {address} {Los Alamitos, CA, USA},\ \bibinfo {year} {2025})\ pp.\ \bibinfo {pages} {130--141}\BibitemShut {NoStop}%
\bibitem [{\citenamefont {Hansen}(2008)}]{Hansen-Metric_adjusted_skew_information}%
  \BibitemOpen
  \bibfield  {author} {\bibinfo {author} {\bibfnamefont {F.}~\bibnamefont {Hansen}},\ }\bibfield  {title} {\bibinfo {title} {Metric adjusted skew information},\ }\href {https://doi.org/10.1073/pnas.0803323105} {\bibfield  {journal} {\bibinfo  {journal} {Proceedings of the National Academy of Sciences}\ }\textbf {\bibinfo {volume} {105}},\ \bibinfo {pages} {9909} (\bibinfo {year} {2008})}\BibitemShut {NoStop}%
\bibitem [{\citenamefont {Cottrell}\ \emph {et~al.}(2019)\citenamefont {Cottrell}, \citenamefont {Freivogel}, \citenamefont {Hofman},\ and\ \citenamefont {Lokhande}}]{Cottrell2019}%
  \BibitemOpen
  \bibfield  {author} {\bibinfo {author} {\bibfnamefont {W.}~\bibnamefont {Cottrell}}, \bibinfo {author} {\bibfnamefont {B.}~\bibnamefont {Freivogel}}, \bibinfo {author} {\bibfnamefont {D.~M.}\ \bibnamefont {Hofman}},\ and\ \bibinfo {author} {\bibfnamefont {S.~F.}\ \bibnamefont {Lokhande}},\ }\bibfield  {title} {\bibinfo {title} {How to build the thermofield double state},\ }\href {https://doi.org/10.1007/JHEP02(2019)058} {\bibfield  {journal} {\bibinfo  {journal} {Journal of High Energy Physics}\ }\textbf {\bibinfo {volume} {2019}},\ \bibinfo {pages} {58} (\bibinfo {year} {2019})}\BibitemShut {NoStop}%
\bibitem [{\citenamefont {Holmes}\ \emph {et~al.}(2022{\natexlab{b}})\citenamefont {Holmes}, \citenamefont {Muraleedharan}, \citenamefont {Somma}, \citenamefont {Subasi},\ and\ \citenamefont {{\c{S}}ahino{\u{g}}lu}}]{Holmes2022quantumalgorithms}%
  \BibitemOpen
  \bibfield  {author} {\bibinfo {author} {\bibfnamefont {Z.}~\bibnamefont {Holmes}}, \bibinfo {author} {\bibfnamefont {G.}~\bibnamefont {Muraleedharan}}, \bibinfo {author} {\bibfnamefont {R.~D.}\ \bibnamefont {Somma}}, \bibinfo {author} {\bibfnamefont {Y.}~\bibnamefont {Subasi}},\ and\ \bibinfo {author} {\bibfnamefont {B.}~\bibnamefont {{\c{S}}ahino{\u{g}}lu}},\ }\bibfield  {title} {\bibinfo {title} {Quantum algorithms from fluctuation theorems: {T}hermal-state preparation},\ }\href {https://doi.org/10.22331/q-2022-10-06-825} {\bibfield  {journal} {\bibinfo  {journal} {{Quantum}}\ }\textbf {\bibinfo {volume} {6}},\ \bibinfo {pages} {825} (\bibinfo {year} {2022}{\natexlab{b}})}\BibitemShut {NoStop}%
\bibitem [{\citenamefont {Amari}(1998)}]{amari1998nat_grad}%
  \BibitemOpen
  \bibfield  {author} {\bibinfo {author} {\bibfnamefont {S.-I.}\ \bibnamefont {Amari}},\ }\bibfield  {title} {\bibinfo {title} {Natural gradient works efficiently in learning},\ }\href {https://doi.org/10.1162/089976698300017746} {\bibfield  {journal} {\bibinfo  {journal} {Neural Computation}\ }\textbf {\bibinfo {volume} {10}},\ \bibinfo {pages} {251} (\bibinfo {year} {1998})}\BibitemShut {NoStop}%
\bibitem [{\citenamefont {Stokes}\ \emph {et~al.}(2020)\citenamefont {Stokes}, \citenamefont {Izaac}, \citenamefont {Killoran},\ and\ \citenamefont {Carleo}}]{Stokes2020quantumnatural}%
  \BibitemOpen
  \bibfield  {author} {\bibinfo {author} {\bibfnamefont {J.}~\bibnamefont {Stokes}}, \bibinfo {author} {\bibfnamefont {J.}~\bibnamefont {Izaac}}, \bibinfo {author} {\bibfnamefont {N.}~\bibnamefont {Killoran}},\ and\ \bibinfo {author} {\bibfnamefont {G.}~\bibnamefont {Carleo}},\ }\bibfield  {title} {\bibinfo {title} {Quantum natural gradient},\ }\href {https://doi.org/10.22331/q-2020-05-25-269} {\bibfield  {journal} {\bibinfo  {journal} {{Quantum}}\ }\textbf {\bibinfo {volume} {4}},\ \bibinfo {pages} {269} (\bibinfo {year} {2020})}\BibitemShut {NoStop}%
\bibitem [{\citenamefont {Koczor}\ and\ \citenamefont {Benjamin}(2022)}]{KS2022}%
  \BibitemOpen
  \bibfield  {author} {\bibinfo {author} {\bibfnamefont {B.}~\bibnamefont {Koczor}}\ and\ \bibinfo {author} {\bibfnamefont {S.~C.}\ \bibnamefont {Benjamin}},\ }\bibfield  {title} {\bibinfo {title} {Quantum natural gradient generalized to noisy and nonunitary circuits},\ }\href {https://doi.org/10.1103/PhysRevA.106.062416} {\bibfield  {journal} {\bibinfo  {journal} {Physical Review A}\ }\textbf {\bibinfo {volume} {106}},\ \bibinfo {pages} {062416} (\bibinfo {year} {2022})}\BibitemShut {NoStop}%
\bibitem [{\citenamefont {Sohail}\ \emph {et~al.}(2025)\citenamefont {Sohail}, \citenamefont {Heidari},\ and\ \citenamefont {Pradhan}}]{sohail2024quantumnaturalstochasticpairwise}%
  \BibitemOpen
  \bibfield  {author} {\bibinfo {author} {\bibfnamefont {M.~A.}\ \bibnamefont {Sohail}}, \bibinfo {author} {\bibfnamefont {M.}~\bibnamefont {Heidari}},\ and\ \bibinfo {author} {\bibfnamefont {S.~S.}\ \bibnamefont {Pradhan}},\ }\bibfield  {title} {\bibinfo {title} {Quantum natural stochastic pairwise coordinate descent},\ }\href {https://doi.org/10.1038/s41534-025-01047-4} {\bibfield  {journal} {\bibinfo  {journal} {npj Quantum Information}\ }\textbf {\bibinfo {volume} {11}},\ \bibinfo {pages} {109} (\bibinfo {year} {2025})}\BibitemShut {NoStop}%
\bibitem [{\citenamefont {van Straaten}\ and\ \citenamefont {Koczor}(2021)}]{vSK2021}%
  \BibitemOpen
  \bibfield  {author} {\bibinfo {author} {\bibfnamefont {B.}~\bibnamefont {van Straaten}}\ and\ \bibinfo {author} {\bibfnamefont {B.}~\bibnamefont {Koczor}},\ }\bibfield  {title} {\bibinfo {title} {Measurement cost of metric-aware variational quantum algorithms},\ }\href {https://doi.org/10.1103/PRXQuantum.2.030324} {\bibfield  {journal} {\bibinfo  {journal} {PRX Quantum}\ }\textbf {\bibinfo {volume} {2}},\ \bibinfo {pages} {030324} (\bibinfo {year} {2021})}\BibitemShut {NoStop}%
\bibitem [{\citenamefont {Baumgratz}\ and\ \citenamefont {Datta}(2016)}]{BD2016}%
  \BibitemOpen
  \bibfield  {author} {\bibinfo {author} {\bibfnamefont {T.}~\bibnamefont {Baumgratz}}\ and\ \bibinfo {author} {\bibfnamefont {A.}~\bibnamefont {Datta}},\ }\bibfield  {title} {\bibinfo {title} {Quantum enhanced estimation of a multidimensional field},\ }\href {https://doi.org/10.1103/PhysRevLett.116.030801} {\bibfield  {journal} {\bibinfo  {journal} {Physical Review Letters}\ }\textbf {\bibinfo {volume} {116}},\ \bibinfo {pages} {030801} (\bibinfo {year} {2016})}\BibitemShut {NoStop}%
\bibitem [{\citenamefont {Garc\'{\i}a-Pintos}\ \emph {et~al.}(2024)\citenamefont {Garc\'{\i}a-Pintos}, \citenamefont {Bharti}, \citenamefont {Bringewatt}, \citenamefont {Dehghani}, \citenamefont {Ehrenberg}, \citenamefont {Yunger~Halpern},\ and\ \citenamefont {Gorshkov}}]{GarciaPintos2024ham_learning_qbm}%
  \BibitemOpen
  \bibfield  {author} {\bibinfo {author} {\bibfnamefont {L.~P.}\ \bibnamefont {Garc\'{\i}a-Pintos}}, \bibinfo {author} {\bibfnamefont {K.}~\bibnamefont {Bharti}}, \bibinfo {author} {\bibfnamefont {J.}~\bibnamefont {Bringewatt}}, \bibinfo {author} {\bibfnamefont {H.}~\bibnamefont {Dehghani}}, \bibinfo {author} {\bibfnamefont {A.}~\bibnamefont {Ehrenberg}}, \bibinfo {author} {\bibfnamefont {N.}~\bibnamefont {Yunger~Halpern}},\ and\ \bibinfo {author} {\bibfnamefont {A.~V.}\ \bibnamefont {Gorshkov}},\ }\bibfield  {title} {\bibinfo {title} {Estimation of {H}amiltonian parameters from thermal states},\ }\href {https://doi.org/10.1103/PhysRevLett.133.040802} {\bibfield  {journal} {\bibinfo  {journal} {Physical Review Letters}\ }\textbf {\bibinfo {volume} {133}},\ \bibinfo {pages} {040802} (\bibinfo {year} {2024})}\BibitemShut {NoStop}%
\bibitem [{\citenamefont {Abiuso}\ \emph {et~al.}(2025)\citenamefont {Abiuso}, \citenamefont {Sekatski}, \citenamefont {Calsamiglia},\ and\ \citenamefont {Perarnau-Llobet}}]{abiuso2024limitsmetrologythermal}%
  \BibitemOpen
  \bibfield  {author} {\bibinfo {author} {\bibfnamefont {P.}~\bibnamefont {Abiuso}}, \bibinfo {author} {\bibfnamefont {P.}~\bibnamefont {Sekatski}}, \bibinfo {author} {\bibfnamefont {J.}~\bibnamefont {Calsamiglia}},\ and\ \bibinfo {author} {\bibfnamefont {M.}~\bibnamefont {Perarnau-Llobet}},\ }\bibfield  {title} {\bibinfo {title} {Fundamental limits of metrology at thermal equilibrium},\ }\href {http://dx.doi.org/10.1103/PhysRevLett.134.010801} {\bibfield  {journal} {\bibinfo  {journal} {Physical Review Letters}\ }\textbf {\bibinfo {volume} {134}},\ \bibinfo {pages} {010801} (\bibinfo {year} {2025})}\BibitemShut {NoStop}%
\bibitem [{\citenamefont {Huang}\ and\ \citenamefont {Wilde}(2024)}]{huang2024informationgeometrybosonicgaussian}%
  \BibitemOpen
  \bibfield  {author} {\bibinfo {author} {\bibfnamefont {Z.}~\bibnamefont {Huang}}\ and\ \bibinfo {author} {\bibfnamefont {M.~M.}\ \bibnamefont {Wilde}},\ }\href {https://arxiv.org/abs/2411.18268} {\bibinfo {title} {Information geometry of bosonic {G}aussian thermal states}} (\bibinfo {year} {2024}),\ \Eprint {https://arxiv.org/abs/2411.18268} {arXiv:2411.18268 [quant-ph]} \BibitemShut {NoStop}%
\bibitem [{\citenamefont {Wilde}(2025{\natexlab{a}})}]{wilde2025generativemodelingusingevolved}%
  \BibitemOpen
  \bibfield  {author} {\bibinfo {author} {\bibfnamefont {M.~M.}\ \bibnamefont {Wilde}},\ }\href {https://arxiv.org/abs/2512.02721} {\bibinfo {title} {Generative modeling using evolved quantum {B}oltzmann machines}} (\bibinfo {year} {2025}{\natexlab{a}}),\ \Eprint {https://arxiv.org/abs/2512.02721} {arXiv:2512.02721 [quant-ph]} \BibitemShut {NoStop}%
\bibitem [{\citenamefont {Wilde}(2025{\natexlab{b}})}]{wilde2025quantumfisherinformationmatrices}%
  \BibitemOpen
  \bibfield  {author} {\bibinfo {author} {\bibfnamefont {M.~M.}\ \bibnamefont {Wilde}},\ }\href {https://arxiv.org/abs/2510.02218} {\bibinfo {title} {Quantum {F}isher information matrices from {R}\'enyi relative entropies}} (\bibinfo {year} {2025}{\natexlab{b}}),\ \Eprint {https://arxiv.org/abs/2510.02218} {arXiv:2510.02218 [quant-ph]} \BibitemShut {NoStop}%
\bibitem [{\citenamefont {Watrous}(2002)}]{watrous2002qszk}%
  \BibitemOpen
  \bibfield  {author} {\bibinfo {author} {\bibfnamefont {J.}~\bibnamefont {Watrous}},\ }\bibfield  {title} {\bibinfo {title} {Limits on the power of quantum statistical zero-knowledge},\ }in\ \href {https://doi.org/10.1109/SFCS.2002.1181970} {\emph {\bibinfo {booktitle} {Proceedings of the 43rd Annual IEEE Symposium on Foundations of Computer Science}}}\ (\bibinfo {year} {2002})\ pp.\ \bibinfo {pages} {459--468},\ \Eprint {https://arxiv.org/abs/quant-ph/0202111} {arXiv:quant-ph/0202111} \BibitemShut {NoStop}%
\bibitem [{\citenamefont {Watrous}(2006)}]{Wat06}%
  \BibitemOpen
  \bibfield  {author} {\bibinfo {author} {\bibfnamefont {J.}~\bibnamefont {Watrous}},\ }\bibfield  {title} {\bibinfo {title} {Zero-knowledge against quantum attacks},\ }in\ \href {https://doi.org/10.1145/1132516.1132560} {\emph {\bibinfo {booktitle} {Proceedings of the Thirty-Eighth Annual ACM Symposium on Theory of Computing}}},\ \bibinfo {series and number} {STOC '06}\ (\bibinfo  {publisher} {Association for Computing Machinery},\ \bibinfo {address} {New York, NY, USA},\ \bibinfo {year} {2006})\ p.\ \bibinfo {pages} {296–305}\BibitemShut {NoStop}%
\bibitem [{\citenamefont {Wang}\ \emph {et~al.}(2024)\citenamefont {Wang}, \citenamefont {Guan}, \citenamefont {Liu}, \citenamefont {Zhang},\ and\ \citenamefont {Ying}}]{WGLZY22}%
  \BibitemOpen
  \bibfield  {author} {\bibinfo {author} {\bibfnamefont {Q.}~\bibnamefont {Wang}}, \bibinfo {author} {\bibfnamefont {J.}~\bibnamefont {Guan}}, \bibinfo {author} {\bibfnamefont {J.}~\bibnamefont {Liu}}, \bibinfo {author} {\bibfnamefont {Z.}~\bibnamefont {Zhang}},\ and\ \bibinfo {author} {\bibfnamefont {M.}~\bibnamefont {Ying}},\ }\bibfield  {title} {\bibinfo {title} {New quantum algorithms for computing quantum entropies and distances},\ }\href {https://doi.org/10.1109/TIT.2024.3399014} {\bibfield  {journal} {\bibinfo  {journal} {IEEE Transactions on Information Theory}\ }\textbf {\bibinfo {volume} {70}},\ \bibinfo {pages} {5653} (\bibinfo {year} {2024})}\BibitemShut {NoStop}%
\bibitem [{\citenamefont {Wang}\ and\ \citenamefont {Zhang}(2024)}]{WZ23}%
  \BibitemOpen
  \bibfield  {author} {\bibinfo {author} {\bibfnamefont {Q.}~\bibnamefont {Wang}}\ and\ \bibinfo {author} {\bibfnamefont {Z.}~\bibnamefont {Zhang}},\ }\bibfield  {title} {\bibinfo {title} {Fast quantum algorithms for trace distance estimation},\ }\href {https://doi.org/10.1109/TIT.2023.3321121} {\bibfield  {journal} {\bibinfo  {journal} {IEEE Transactions on Information Theory}\ }\textbf {\bibinfo {volume} {70}},\ \bibinfo {pages} {2720} (\bibinfo {year} {2024})}\BibitemShut {NoStop}%
\bibitem [{\citenamefont {Liu}\ \emph {et~al.}(2025{\natexlab{b}})\citenamefont {Liu}, \citenamefont {Wang}, \citenamefont {Wilde},\ and\ \citenamefont {Zhang}}]{liu2024quantumalgorithmsmatrixgeometric}%
  \BibitemOpen
  \bibfield  {author} {\bibinfo {author} {\bibfnamefont {N.}~\bibnamefont {Liu}}, \bibinfo {author} {\bibfnamefont {Q.}~\bibnamefont {Wang}}, \bibinfo {author} {\bibfnamefont {M.~M.}\ \bibnamefont {Wilde}},\ and\ \bibinfo {author} {\bibfnamefont {Z.}~\bibnamefont {Zhang}},\ }\bibfield  {title} {\bibinfo {title} {Quantum algorithms for matrix geometric means},\ }\href {https://doi.org/10.1038/s41534-025-00973-7} {\bibfield  {journal} {\bibinfo  {journal} {npj Quantum Information}\ }\textbf {\bibinfo {volume} {11}},\ \bibinfo {pages} {101} (\bibinfo {year} {2025}{\natexlab{b}})},\ \Eprint {https://arxiv.org/abs/2405.00673} {arXiv:2405.00673 [quant-ph]} \BibitemShut {NoStop}%
\bibitem [{\citenamefont {Lashkari}\ and\ \citenamefont {Van~Raamsdonk}(2016)}]{Lashkari2016}%
  \BibitemOpen
  \bibfield  {author} {\bibinfo {author} {\bibfnamefont {N.}~\bibnamefont {Lashkari}}\ and\ \bibinfo {author} {\bibfnamefont {M.}~\bibnamefont {Van~Raamsdonk}},\ }\bibfield  {title} {\bibinfo {title} {Canonical energy is quantum {F}isher information},\ }\href {https://doi.org/10.1007/JHEP04(2016)153} {\bibfield  {journal} {\bibinfo  {journal} {Journal of High Energy Physics}\ }\textbf {\bibinfo {volume} {2016}},\ \bibinfo {pages} {153} (\bibinfo {year} {2016})}\BibitemShut {NoStop}%
\bibitem [{\citenamefont {Kibe}\ \emph {et~al.}(2022)\citenamefont {Kibe}, \citenamefont {Mandayam},\ and\ \citenamefont {Mukhopadhyay}}]{Kibe2022}%
  \BibitemOpen
  \bibfield  {author} {\bibinfo {author} {\bibfnamefont {T.}~\bibnamefont {Kibe}}, \bibinfo {author} {\bibfnamefont {P.}~\bibnamefont {Mandayam}},\ and\ \bibinfo {author} {\bibfnamefont {A.}~\bibnamefont {Mukhopadhyay}},\ }\bibfield  {title} {\bibinfo {title} {Holographic spacetime, black holes and quantum error correcting codes: a review},\ }\href {https://doi.org/10.1140/epjc/s10052-022-10382-1} {\bibfield  {journal} {\bibinfo  {journal} {The European Physical Journal C}\ }\textbf {\bibinfo {volume} {82}},\ \bibinfo {pages} {463} (\bibinfo {year} {2022})}\BibitemShut {NoStop}%
\bibitem [{\citenamefont {Campos~Venuti}\ and\ \citenamefont {Zanardi}(2007)}]{CVZ2007}%
  \BibitemOpen
  \bibfield  {author} {\bibinfo {author} {\bibfnamefont {L.}~\bibnamefont {Campos~Venuti}}\ and\ \bibinfo {author} {\bibfnamefont {P.}~\bibnamefont {Zanardi}},\ }\bibfield  {title} {\bibinfo {title} {Quantum critical scaling of the geometric tensors},\ }\href {https://doi.org/10.1103/PhysRevLett.99.095701} {\bibfield  {journal} {\bibinfo  {journal} {Physical Review Letters}\ }\textbf {\bibinfo {volume} {99}},\ \bibinfo {pages} {095701} (\bibinfo {year} {2007})}\BibitemShut {NoStop}%
\bibitem [{\citenamefont {Zanardi}\ \emph {et~al.}(2007)\citenamefont {Zanardi}, \citenamefont {Giorda},\ and\ \citenamefont {Cozzini}}]{ZGC2007}%
  \BibitemOpen
  \bibfield  {author} {\bibinfo {author} {\bibfnamefont {P.}~\bibnamefont {Zanardi}}, \bibinfo {author} {\bibfnamefont {P.}~\bibnamefont {Giorda}},\ and\ \bibinfo {author} {\bibfnamefont {M.}~\bibnamefont {Cozzini}},\ }\bibfield  {title} {\bibinfo {title} {Information-theoretic differential geometry of quantum phase transitions},\ }\href {https://doi.org/10.1103/PhysRevLett.99.100603} {\bibfield  {journal} {\bibinfo  {journal} {Physical Review Letters}\ }\textbf {\bibinfo {volume} {99}},\ \bibinfo {pages} {100603} (\bibinfo {year} {2007})}\BibitemShut {NoStop}%
\bibitem [{\citenamefont {Carollo}\ \emph {et~al.}(2020)\citenamefont {Carollo}, \citenamefont {Valenti},\ and\ \citenamefont {Spagnolo}}]{CAROLLO20201}%
  \BibitemOpen
  \bibfield  {author} {\bibinfo {author} {\bibfnamefont {A.}~\bibnamefont {Carollo}}, \bibinfo {author} {\bibfnamefont {D.}~\bibnamefont {Valenti}},\ and\ \bibinfo {author} {\bibfnamefont {B.}~\bibnamefont {Spagnolo}},\ }\bibfield  {title} {\bibinfo {title} {Geometry of quantum phase transitions},\ }\href {https://doi.org/https://doi.org/10.1016/j.physrep.2019.11.002} {\bibfield  {journal} {\bibinfo  {journal} {Physics Reports}\ }\textbf {\bibinfo {volume} {838}},\ \bibinfo {pages} {1} (\bibinfo {year} {2020})},\ \bibinfo {note} {geometry of quantum phase transitions}\BibitemShut {NoStop}%
\bibitem [{\citenamefont {{Del Moral}}\ and\ \citenamefont {Niclas}(2018)}]{DelMoral2018}%
  \BibitemOpen
  \bibfield  {author} {\bibinfo {author} {\bibfnamefont {P.}~\bibnamefont {{Del Moral}}}\ and\ \bibinfo {author} {\bibfnamefont {A.}~\bibnamefont {Niclas}},\ }\bibfield  {title} {\bibinfo {title} {A {T}aylor expansion of the square root matrix function},\ }\href {https://doi.org/https://doi.org/10.1016/j.jmaa.2018.05.005} {\bibfield  {journal} {\bibinfo  {journal} {Journal of Mathematical Analysis and Applications}\ }\textbf {\bibinfo {volume} {465}},\ \bibinfo {pages} {259} (\bibinfo {year} {2018})}\BibitemShut {NoStop}%
\bibitem [{\citenamefont {{NISO}}()}]{CRediT}%
  \BibitemOpen
  \bibfield  {author} {\bibinfo {author} {\bibnamefont {{NISO}}},\ }\href@noop {} {\bibinfo {title} {Credit -- contributor roles taxonomy}},\ \bibinfo {howpublished} {\url{https://credit.niso.org/}},\ \bibinfo {note} {accessed: 2024-10-28}\BibitemShut {NoStop}%
\bibitem [{\citenamefont {Anshu}\ \emph {et~al.}(2020)\citenamefont {Anshu}, \citenamefont {Arunachalam}, \citenamefont {Kuwahara},\ and\ \citenamefont {Soleimanifar}}]{Anshu2020arXiv}%
  \BibitemOpen
  \bibfield  {author} {\bibinfo {author} {\bibfnamefont {A.}~\bibnamefont {Anshu}}, \bibinfo {author} {\bibfnamefont {S.}~\bibnamefont {Arunachalam}}, \bibinfo {author} {\bibfnamefont {T.}~\bibnamefont {Kuwahara}},\ and\ \bibinfo {author} {\bibfnamefont {M.}~\bibnamefont {Soleimanifar}},\ }\href {https://arxiv.org/abs/2004.07266v1} {\bibinfo {title} {Sample-efficient learning of quantum many-body systems}} (\bibinfo {year} {2020}),\ \Eprint {https://arxiv.org/abs/2004.07266v1} {arXiv:2004.07266v1 [quant-ph]} \BibitemShut {NoStop}%
\bibitem [{\citenamefont {Feldman}(2007)}]{Feldman2007}%
  \BibitemOpen
  \bibfield  {author} {\bibinfo {author} {\bibfnamefont {J.}~\bibnamefont {Feldman}},\ }\href@noop {} {\bibinfo {title} {Duhamel's formula}},\ \bibinfo {howpublished} {Lecture notes} (\bibinfo {year} {2007}),\ \bibinfo {note} {\url{https://personal.math.ubc.ca/~feldman/m428/duhamel.pdf}}\BibitemShut {NoStop}%
\end{thebibliography}%

\appendix

\onecolumngrid 

\section{Proof of Theorem~\ref{thm:gradient-eQBM}}

\label{proof:gradient-eQBM}

From previous work~\cite[Eq.~(10)]{patel2024quantumboltzmannmachine}, it is known that
\begin{align}
\frac{\partial}{\partial\theta_{j}}\omega(\theta,\phi) &  =e^{-iH(\phi
)}\left(  \frac{\partial}{\partial\theta_{j}}\rho(\theta)\right)  e^{iH(\phi
)}\\
&  =e^{-iH(\phi)}\left(  -\frac{1}{2}\left\{  \Phi_{\theta}(G_{j}),\rho
(\theta)\right\}  +\rho(\theta)\left\langle G_{j}\right\rangle _{\rho(\theta
)}\right)  e^{iH(\phi)}\\
&  =-\frac{1}{2}\left\{  e^{-iH(\phi)}\Phi_{\theta}(G_{j})e^{iH(\phi
)},e^{-iH(\phi)}\rho(\theta)e^{iH(\phi)}\right\}  +e^{-iH(\phi)}\rho(\theta)e^{iH(\phi)}\left\langle G_{j}\right\rangle
_{\rho(\theta)}\\
&  =-\frac{1}{2}\left\{  e^{-iH(\phi)}\Phi_{\theta}(G_{j})e^{iH(\phi)}
,\omega(\theta,\phi)\right\}  +\omega(\theta,\phi)\left\langle G_{j}
\right\rangle _{\rho(\theta)}.
\end{align}
See also~\cite[Eq.~(9)]{Hastings2007},~\cite[Proposition~20]{Anshu2020arXiv}, and~\cite[Lemma~5]{Coopmans2024qbm_gen_learn}.
Now we consider the derivative with respect the variable $\phi_k$:
\begin{align}
\frac{\partial}{\partial\phi_{k}}\omega(\theta,\phi) &  =\frac{\partial
}{\partial\phi_{k}}\left(  e^{-iH(\phi)}\rho(\theta)e^{iH(\phi)}\right)  \\
&  =\left(  \frac{\partial}{\partial\phi_{k}}e^{-iH(\phi)}\right)  \rho
(\theta)e^{iH(\phi)}+e^{-iH(\phi)}\rho(\theta)\left(  \frac{\partial}
{\partial\phi_{k}}e^{iH(\phi)}\right)  .
\end{align}
According to Duhamel's formula~\cite{Feldman2007}, the partial derivative of a matrix exponential
$e^{A(x)}$ with respect to some parameter $x$ is given as follows:
\begin{align}
\frac{\partial}{\partial x} e^{A(x)} &  =\int_{0}^{1}e^{(1-t)A(x)}\left(  \frac{\partial}{\partial x} A(x)\right)  e^{tA(x)}\ dt\label{eq:duh_for}\\
&  =\int_{0}^{1}e^{tA(x)}\left(  \frac{\partial}{\partial x} A(x)\right)  e^{\left(
1-t\right)  A(x)}\ dt.
\end{align}
From Duhamel's formula, consider that
\begin{align}
\frac{\partial}{\partial\phi_{k}}e^{-iH(\phi)} &  =\int_{0}^{1}dt\ e^{t\left(
-iH(\phi)\right)  }\left(  \frac{\partial}{\partial\phi_{k}}\left(
-iH(\phi)\right)  \right)  e^{\left(  1-t\right)  \left(  -iH(\phi)\right)
} \label{eq:deriv-phi-proof-1}\\
&  =\left[  \int_{0}^{1}dt\ e^{-itH(\phi)}\left(  -iH_{k}\right)
e^{itH(\phi)}\right]  e^{-iH(\phi)}\\
&  =-i\left[  \int_{0}^{1}dt\ e^{-itH(\phi)}H_{k}e^{itH(\phi)}\right]
e^{-iH(\phi)}\\
&  =-i\Psi^{\dagger}_{\phi}(H_{k})e^{-iH(\phi)}.
\label{eq:deriv-phi-proof-1-2}
\end{align}
Also, we have that
\begin{align}
\frac{\partial}{\partial\phi_{k}}e^{iH(\phi)} &  =\int_{0}^{1}dt\ e^{\left(
1-t\right)  iH(\phi)}\left(  \frac{\partial}{\partial\phi_{k}}iH(\phi)\right)
e^{itH(\phi)}\\
&  =ie^{iH(\phi)}\left[  \int_{0}^{1}dt\ e^{-itH(\phi)}H_{k}e^{itH(\phi
)}\right]  \\
&  =ie^{iH(\phi)}\Psi^{\dagger}_{\phi}(H_{k}).
\label{eq:deriv-phi-proof-last}
\end{align}
Then we find that
\begin{align}
  \frac{\partial}{\partial\phi_{k}}\omega(\theta,\phi)
&  =\left(  \frac{\partial}{\partial\phi_{k}}e^{-iH(\phi)}\right)  \rho
(\theta)e^{iH(\phi)}+e^{-iH(\phi)}\rho(\theta)\left(  \frac{\partial}
{\partial\phi_{k}}e^{iH(\phi)}\right)  \\
&  =-i\Psi^{\dagger}_{\phi}(H_{k})e^{-iH(\phi)}\rho(\theta)e^{iH(\phi)}+ie^{-iH(\phi
)}\rho(\theta)e^{iH(\phi)}\Psi^{\dagger}_{\phi}(H_{k})\\
&  =-i\Psi^{\dagger}_{\phi}(H_{k})\omega(\theta,\phi)+i\omega(\theta,\phi)\Psi_{\phi
}(H_{k})\\
&  =i\left[  \omega(\theta,\phi),\Psi^{\dagger}_{\phi}(H_{k})\right]  .
\end{align}

\section{Hadamard tests for expected values of commutators,  anticommutators, and nestings}
\label{app:Hadamard_test}

In this appendix, we present a generalized Hadamard test for estimating the expectation values of commutators and anticommutators of two operators. We then show how to extend this method to nested commutators and anticommutators involving multiple operators. These circuits form the foundation of the algorithms developed for estimating the gradient and the information matrix elements of  evolved quantum Boltzmann machines.  
\\Let us start with the case of two operators $U$ and $H$. We assume that $U$ is both unitary and Hermitian and that $H$ is Hermitian. We now present a quantum circuit used to estimate the following quantity:
\begin{equation}
    -\frac{1}{2}\big\langle\big\{U ,H\big\}\big\rangle_\rho\label{eq:anticomm-U0_U1},
\end{equation}
where $\rho$ is a generic quantum state, illustrated in Figure~\ref{fig:qc-primitive-anticomm}. 
The circuit consists of two quantum registers:
\begin{itemize}
    \item a control register, initialized in the state $|1\rangle \!\langle 1|$,
    \item a system register, initialized  in the state $\rho$.
\end{itemize}
To demonstrate that the output matches the desired quantity, let us track the state of the circuit in Figure~\ref{fig:qc-primitive-anticomm} as it progresses through
the various steps:
\begin{align}
|1\rangle\!\langle1|\otimes\rho & \rightarrow\frac{1}{2}\sum_{j,k\in\left\{
0,1\right\}  } (-1)^{j+k}|j\rangle\!\langle k|\otimes\rho\\
& \rightarrow\frac{1}{2}\sum_{j,k\in\left\{  0,1\right\}  } (-1)^{j+k} |j\rangle\!\langle
k|\otimes U^{j}\rho U^{k}.
\end{align}
Given that $X=\left(  \text{Had}\right)  Z\left(  \text{Had}\right)  $, where
$X$ and $Z$ are Pauli matrices, the final step is equivalent to determining
the expectation of the observable $X\otimes H$. This expectation is as
follows:
\begin{align}
& \operatorname{Tr}\!\left[  \left(  X\otimes H\right)  \left(  \frac{1}{2}
\sum_{j,k\in\left\{  0,1\right\}  }(-1)^{j+k}|j\rangle\!\langle k|\otimes U^{j}\rho
U^{k}\right)  \right]  \nonumber\\
& =\frac{1}{2}\sum_{j,k\in\left\{  0,1\right\}  }\operatorname{Tr}\!\left[
\left(  X\otimes H\right)  \left(  (-1)^{j+k}|j\rangle\!\langle k|\otimes U^{j}\rho
U^{k}\right)  \right]  \\
& =\frac{1}{2}\sum_{j,k\in\left\{  0,1\right\}  }(-1)^{j+k}\langle k|X|j\rangle
\operatorname{Tr}\!\left[  HU^{j}\rho U^{k}\right]  \\
& =\frac{1}{2}\sum_{j,k\in\left\{  0,1\right\}  :j\neq k}(-1)^{j+k}\langle
k|X|j\rangle\operatorname{Tr}\!\left[  HU^{j}\rho U^{k}\right]  \\
& =-\frac{1}{2}\left(  \operatorname{Tr}\!\left[  HU^{0}\rho U^{1}\right]
+\operatorname{Tr}\!\left[  HU^{1}\rho U^{0}\right]  \right)  \\
& =-\frac{1}{2}\left(  \operatorname{Tr}\!\left[  H\rho U\right]
+\operatorname{Tr}\!\left[  HU\rho\right]  \right)  \\
& =-\frac{1}{2}\left(  \operatorname{Tr}\!\left[  UH\rho\right]
+\operatorname{Tr}\!\left[  HU\rho\right]  \right)  \\
& =-\frac{1}{2}\left\langle \left\{  U,H\right\}  \right\rangle _{\rho}.
\end{align}
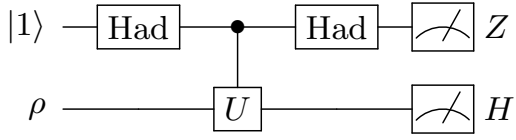
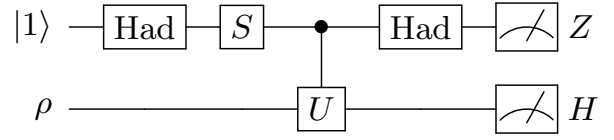
\begin{figure*}
    \centering
    % First subfigure
    \begin{subfigure}[t]{0.48\textwidth}
        \centering
        \scalebox{1.4}{ 
        \Qcircuit @C=1em @R=1.em {
            \lstick{\ket{1}} & \gate{\operatorname{Had}} & \ctrl{1} & \gate{\operatorname{Had}} & \meter & \rstick{\hspace{-1.2em}Z} \\
            \lstick{\rho} & \qw & \gate{U} & \qw & \meter & \rstick{\hspace{-1.2em}H}
        }
        }
        \vspace{15pt}
        \caption{Quantum circuit for estimating $-\tfrac{1}{2}\langle\{U ,H\}\rangle_\rho$ when $U$ is unitary and Hermitian and $H$ is Hermitian.}
        \label{fig:qc-primitive-anticomm}
    \end{subfigure}
    \hfill
    % Second subfigure
    \begin{subfigure}[t]{0.48\textwidth}
        \centering
        \scalebox{1.4}{ 
        \Qcircuit @C=1em @R=1em {
            \lstick{\ket{1}} & \gate{\operatorname{Had}} & \gate{S} & \ctrl{1} & \gate{\operatorname{Had}} & \meter & \rstick{\hspace{-1.2em}Z} \\
            \lstick{\rho} & \qw & \qw & \gate{U} & \qw & \meter  & \rstick{\hspace{-1.2em}H} 
        }
        }
        \vspace{15pt}
        \caption{Quantum circuit for estimating $\frac{i}{2}\langle[U,H]\rangle_\rho$ when $U$ is unitary and Hermitian and $H$ is Hermitian. }
        \label{fig:qc-primitive-comm}
    \end{subfigure}
    \caption{Quantum circuits for estimating the expected value of anticommutator and commutator of operators.}
    \label{fig:side-by-side-circuits-comm-anticomm}
\end{figure*}

Analogously, we can construct a related quantum circuit to estimate the expected value of the commutator of $U$ and $H$. If $U$ is unitary and Hermitian and $H$ is Hermitian, then the following quantity is of interest:
\begin{equation}
    \frac{i}{2}\big\langle\big[U ,H\big]\big\rangle_\rho \label{eq:U0_U1_comm}.
\end{equation}
The quantum circuit for estimating this quantity is shown in Figure~\ref{fig:qc-primitive-comm}. The setup is essentially the same as in the anticommutator case: we have a control qubit initially in the state $\ket{1}$ and a system register initialized in the state $\rho$. The only difference compared to the circuit in Figure~\ref{fig:qc-primitive-anticomm} is the addition of an $S$ gate applied to the control qubit immediately after the first Hadamard gate. To demonstrate that the output matches the desired quantity, let us track the state of the circuit in Figure~\ref{fig:qc-primitive-comm} as it
progresses through the various steps:
\begin{align}
|1\rangle\!\langle1|\otimes\rho & \rightarrow\frac{1}{2}\sum_{j,k\in\left\{
0,1\right\}  } (-1)^{j+k}|j\rangle\!\langle k|\otimes\rho\\
& \rightarrow\frac{1}{2}\sum_{j,k\in\left\{  0,1\right\}  } (-1)^{j+k}i^{j}\left(
-i\right)  ^{k}|j\rangle\!\langle k|\otimes\rho\\
& \rightarrow\frac{1}{2}\sum_{j,k\in\left\{  0,1\right\}  } (-1)^{j+k}i^{j}\left(
-i\right)  ^{k}|j\rangle\!\langle k|\otimes U^{j}\rho U^{k}.
\end{align}
As before, the final step is equivalent to determining the expectation of the
observable $X\otimes H$. This expectation is as follows:
\begin{align}
& \operatorname{Tr}\!\left[  \left(  X\otimes H\right)  \left(  \frac{1}{2}
\sum_{j,k\in\left\{  0,1\right\}  }(-1)^{j+k}i^{j}\left(  -i\right)  ^{k}|j\rangle
\langle k|\otimes U^{j}\rho U^{k}\right)  \right]  \nonumber\\
& =\frac{1}{2}\sum_{j,k\in\left\{  0,1\right\}  }(-1)^{j+k}i^{j}\left(  -i\right)
^{k}\operatorname{Tr}\!\left[  \left(  X\otimes H\right)  \left(  |j\rangle
\langle k|\otimes U^{j}\rho U^{k}\right)  \right]  \\
& =\frac{1}{2}\sum_{j,k\in\left\{  0,1\right\}  }(-1)^{j+k}i^{j}\left(  -i\right)
^{k}\langle k|X|j\rangle\operatorname{Tr}\!\left[  HU^{j}\rho U^{k}\right]  \\
& =\frac{1}{2}\sum_{j,k\in\left\{  0,1\right\}  :j\neq k}(-1)^{j+k}i^{j}\left(
-i\right)  ^{k}\langle k|X|j\rangle\operatorname{Tr}\!\left[  HU^{j}\rho
U^{k}\right]  \\
& =\frac{1}{2}\left(  i\operatorname{Tr}\!\left[  HU^{0}\rho U^{1}\right]
-i\operatorname{Tr}\!\left[  HU^{1}\rho U^{0}\right]  \right)  \\
& =\frac{i}{2}\left(  \operatorname{Tr}\!\left[  H\rho U\right]
-\operatorname{Tr}\!\left[  HU\rho\right]  \right)  \\
& =\frac{i}{2}\left(  \operatorname{Tr}\!\left[  UH\rho\right]
-\operatorname{Tr}\!\left[  HU\rho\right]  \right)  \\
& =\frac{i}{2}\left\langle \left[  U,H \right]  \right\rangle _{\rho}.
\end{align}

Thus, by repeatedly running the quantum circuits shown in Figures~\ref{fig:qc-primitive-anticomm} and~\ref{fig:qc-primitive-comm} with independent copies of $\rho$, one can obtain unbiased estimators for~\eqref{eq:anticomm-U0_U1} and~\eqref{eq:U0_U1_comm}.

A generalisation of this single-control-qubit Hadamard test circuit using multiple control qubits can be used to estimate the expectation of nested commutators and anticommutators~\cite[Algorithm 3]{liEfficientQuantumGradient2024}. We show here the quantum circuits that can be used to estimate the expected values of two nested commutators and of nested commutator and anticommutator, of interest for their application in this paper. The quantum circuit shown in Figure~\ref{fig:prim-nest-comm} is used for estimating the following quantity
\begin{align}
    \frac{1}{4} \left\langle \big[ \left[U_1 , H \right], U_0 \big]\right\rangle_\rho,
    \label{eq:nest-comm-app}
\end{align}
where $U_0$ and $U_1$ are Hermitian unitaries, $H$ is Hermitian and $\rho$ is a generic quantum state. The quantum circuit shown in Figure~\ref{fig:prim-nest-anticomm-comm} is used for estimating the following quantity
\begin{align}
    \frac{i}{4} \left\langle \big\{ U_0, \left[H, U_1 \right]\big\}\right\rangle_\rho,
    \label{eq:nest-comm-anticomm-app}
\end{align}
with the same requirements for $U_0$, $U_1$, and $H$ as in the previous circuit. 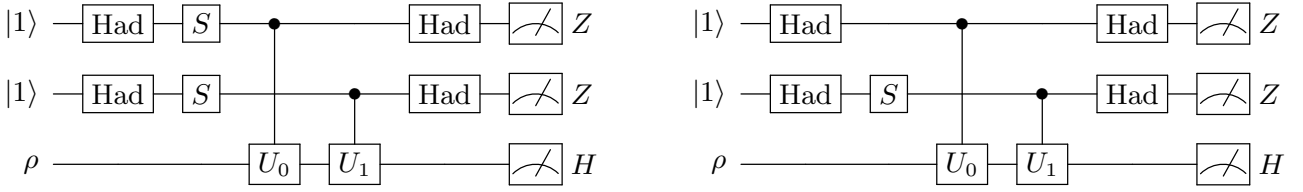
\begin{figure*}
    \centering
    % First subfigure
    \begin{subfigure}[t]{0.49\textwidth}
        \centering
        \scalebox{1.2}{ 
        \Qcircuit @C=1.em @R=1.em {
            \lstick{\ket{1}} & \gate{\operatorname{Had}} & \gate{S} & \ctrl{2} & \qw & \gate{\operatorname{Had}} & \meter &\rstick{\hspace{-1.2em}Z} \\
            \lstick{\ket{1}} & \gate{\operatorname{Had}} & \gate{S} & \qw & \ctrl{1} & \gate{\operatorname{Had}} & \meter & \rstick{\hspace{-1.2em}Z}\\
            \lstick{\rho} & \qw & \qw  & \gate{U_0} & \gate{U_1} & \qw  & \meter & \rstick{\hspace{-1.2em}H}
        }
        }
        \vspace{15pt}
        \caption{Quantum circuit for estimating $\frac{1}{4} \left\langle \big[ \left[U_1 ,H  \right], U_0 \big]\right\rangle_\rho$ when $U_0$ and $U_1$ are unitary and Hermitian and $H$ is Hermitian.}
        \label{fig:prim-nest-comm}
    \end{subfigure}
    \hfill
    % Second subfigure
    \begin{subfigure}[t]{0.49\textwidth}
        \centering
        \scalebox{1.2}{ 
        \Qcircuit @C=1.em @R=1.em {
            \lstick{\ket{1}} & \gate{\operatorname{Had}} & \qw & \ctrl{2} & \qw & \gate{\operatorname{Had}} & \meter &\rstick{\hspace{-1.2em}Z} \\
            \lstick{\ket{1}} & \gate{\operatorname{Had}} & \gate{S} & \qw & \ctrl{1} & \gate{\operatorname{Had}} & \meter & \rstick{\hspace{-1.2em}Z}\\
            \lstick{\rho} & \qw & \qw  & \gate{U_0} & \gate{U_1} & \qw  & \meter & \rstick{\hspace{-1.2em}H}
        }
        }
        \vspace{15pt}
        \caption{Quantum circuit for estimating $\frac{i}{4} \left\langle \big\{ U_0, \left[H, U_1 \right]\big\}\right\rangle_\rho$ when $U_0$ and $U_1$ are unitary and Hermitian and $H$ is Hermitian.}
        \label{fig:prim-nest-anticomm-comm}
    \end{subfigure}
    \caption{Quantum circuits for estimating expected values of nested commutators and anticommutator-commutators.}
    \label{fig:side-by-side-circuits}
\end{figure*}The approach to obtaining unbiased estimators of~\eqref{eq:nest-comm-app} and~\eqref{eq:nest-comm-anticomm-app} follows a similar procedure described previously for estimating~\eqref{eq:anticomm-U0_U1}. For completeness, below we detail justifications of the claims in~\eqref{eq:nest-comm-app} and~\eqref{eq:nest-comm-anticomm-app}. Let us track the state of the circuit in Figure~\ref{fig:prim-nest-comm}, as it progresses through the
various steps:
\begin{align}
 |1\rangle\!\langle1|\otimes|1\rangle\!\langle 1|\otimes\rho 
& \rightarrow\frac{1}{4}\sum_{j,k,\ell,m\in\left\{  0,1\right\}  }\left(
-1\right)  ^{j+k+\ell+m}|j\rangle\!\langle  k|\otimes|\ell\rangle\!\langle 
m|\otimes\rho\\
& \rightarrow\frac{1}{4}\sum_{j,k,\ell,m\in\left\{  0,1\right\}  }\left(
-1\right)  ^{j+k+\ell+m}i^{j+\ell}\left(  -i\right)  ^{k+m}|j\rangle\!\langle 
k|\otimes|\ell\rangle\!\langle  m|\otimes\rho\\
& \rightarrow\frac{1}{4}\sum_{j,k,\ell,m\in\left\{  0,1\right\}  }\left(
-1\right)  ^{j+k+\ell+m}i^{j+\ell}\left(  -i\right)  ^{k+m}|j\rangle\!\langle 
k|\otimes|\ell\rangle\!\langle  m|\otimes U_{1}^{\ell}U_{0}^{j}\rho U_{0}
^{k}U_{1}^{m}.
\end{align}
Given that $X=\ \left(  \text{Had}\right)  Z\left(  \text{Had}\right)  $, where $X$ and $Z$ are Pauli matrices, the
final step is equivalent to determining the expectation of the observable
$X\otimes X\otimes H$. This expectation is as follows:
\begin{align}
& \frac{1}{4}\operatorname{Tr}\!\left[  \left(  X\otimes X\otimes H\right)
\left(  \sum_{j,k,\ell,m\in\left\{  0,1\right\}  }\left(  -1\right)
^{j+k+\ell+m}i^{j+\ell}\left(  -i\right)  ^{k+m}|j\rangle\!\langle 
k|\otimes|\ell\rangle\!\langle  m|\otimes U_{1}^{\ell}U_{0}^{j}\rho U_{0}
^{k}U_{1}^{m}\right)  \right]  \nonumber\\
& =\frac{1}{4}\sum_{j,k,\ell,m\in\left\{  0,1\right\}  }\left(  -1\right)
^{j+k+\ell+m}i^{j+\ell}\left(  -i\right)  ^{k+m}\operatorname{Tr}\!\left[
\left(  X\otimes X\otimes H\right)  \left(  |j\rangle\!\langle  k|\otimes
|\ell\rangle\!\langle  m|\otimes U_{1}^{\ell}U_{0}^{j}\rho U_{0}^{k}U_{1}
^{m}\right)  \right]  \\
& =\frac{1}{4}\sum_{j,k,\ell,m\in\left\{  0,1\right\}  }\left(  -1\right)
^{j+k+\ell+m}i^{j+\ell}\left(  -i\right)  ^{k+m}\langle k|X|j\rangle\ \langle
m|X|\ell\rangle\operatorname{Tr}\!\left[  HU_{1}^{\ell}U_{0}^{j}\rho U_{0}
^{k}U_{1}^{m}\right]  \\
& =\frac{1}{4}\sum_{\substack{j,k,\ell,m\in\left\{  0,1\right\}  \ :\\j\neq
k,\ell\neq m}}\left(  -1\right)  ^{j+k+\ell+m}i^{j+\ell}\left(  -i\right)
^{k+m}\operatorname{Tr}\!\left[  HU_{1}^{\ell}U_{0}^{j}\rho U_{0}^{k}U_{1}
^{m}\right]  \\
& =\frac{1}{4}\left[
\begin{array}
[c]{c}
\left(  -1\right)  ^{0+1+0+1}i^{0+0}\left(  -i\right)  ^{1+1}\operatorname{Tr}
\left[  HU_{1}^{0}U_{0}^{0}\rho U_{0}^{1}U_{1}^{1}\right]  +\left(  -1\right)
^{0+1+1+0}i^{0+1}\left(  -i\right)  ^{1+0}\operatorname{Tr}\!\left[  HU_{1}
^{1}U_{0}^{0}\rho U_{0}^{1}U_{1}^{0}\right]  \\
+\left(  -1\right)  ^{1+0+0+1}i^{1+0}\left(  -i\right)  ^{0+1}
\operatorname{Tr}\!\left[  HU_{1}^{0}U_{0}^{1}\rho U_{0}^{0}U_{1}^{1}\right]
+\left(  -1\right)  ^{1+0+1+0}i^{1+1}\left(  -i\right)  ^{0+0}
\operatorname{Tr}\!\left[  HU_{1}^{1}U_{0}^{1}\rho U_{0}^{0}U_{1}^{0}\right]
\end{array}
\right]  \\
& =\frac{1}{4}\left[  -\operatorname{Tr}\!\left[  H\rho U_{0}U_{1}\right]
+\operatorname{Tr}\!\left[  HU_{1}\rho U_{0}\right]  +\operatorname{Tr}\!\left[
HU_{0}\rho U_{1}\right]  -\operatorname{Tr}\!\left[  HU_{1}U_{0}\rho\right]
\right]  \\
& =\frac{1}{4}\operatorname{Tr}\!\left[  \left[  \left[  U_{1},H\right], U_{0}
\right]  \rho\right]  \\
& = \frac{1}{4}\left\langle \left[  \left[  U_{1},H\right], U_{0}  \right]
\right\rangle _{\rho}.
\end{align}

We do the same for the circuit depicted in Figure~\ref{fig:prim-nest-anticomm-comm}. Again, tracking the state
as it progress through the circuit, consider that
\begin{align}
 |1\rangle\!\langle 1|\otimes|1\rangle\!\langle 1|\otimes\rho
& \rightarrow\frac{1}{4}\sum_{j,k,\ell,m\in\left\{  0,1\right\}  }\left(
-1\right)  ^{j+k+\ell+m}|j\rangle\!\langle  k|\otimes|\ell\rangle\!\langle 
m|\otimes\rho\\
& \rightarrow\frac{1}{4}\sum_{j,k,\ell,m\in\left\{  0,1\right\}  }\left(
-1\right)  ^{j+k+\ell+m}i^{\ell}\left(  -i\right)  ^{m}|j\rangle\!\langle 
k|\otimes|\ell\rangle\!\langle  m|\otimes\rho\\
& \rightarrow\frac{1}{4}\sum_{j,k,\ell,m\in\left\{  0,1\right\}  }\left(
-1\right)  ^{j+k+\ell+m}i^{\ell}\left(  -i\right)  ^{m}|j\rangle\!\langle 
k|\otimes|\ell\rangle\!\langle  m|\otimes U_{1}^{\ell}U_{0}^{j}\rho U_{0}
^{k}U_{1}^{m}.
\end{align}
The expectation of the observable $X\otimes X\otimes H$ is then as follows:
\begin{align}
& \operatorname{Tr}\!\left[  \left(  X\otimes X\otimes H\right)  \left(
\frac{1}{4}\sum_{j,k,\ell,m\in\left\{  0,1\right\}  }\left(  -1\right)
^{j+k+\ell+m}i^{\ell}\left(  -i\right)  ^{m}|j\rangle\!\langle  k|\otimes|\ell
\rangle\!\langle  m|\otimes U_{1}^{\ell}U_{0}^{j}\rho U_{0}^{k}U_{1}^{m}\right)
\right]  \nonumber\\
& =\frac{1}{4}\sum_{j,k,\ell,m\in\left\{  0,1\right\}  }\left(  -1\right)
^{j+k+\ell+m}i^{\ell}\left(  -i\right)  ^{m}\operatorname{Tr}\!\left[  \left(
X\otimes X\otimes H\right)  \left(  |j\rangle\!\langle  k|\otimes|\ell
\rangle\!\langle  m|\otimes U_{1}^{\ell}U_{0}^{j}\rho U_{0}^{k}U_{1}^{m}\right)
\right]  \\
& =\frac{1}{4}\sum_{j,k,\ell,m\in\left\{  0,1\right\}  }\left(  -1\right)
^{j+k+\ell+m}i^{\ell}\left(  -i\right)  ^{m}\langle k|X|j\rangle\ \langle
m|X|\ell\rangle\operatorname{Tr}\!\left[  HU_{1}^{\ell}U_{0}^{j}\rho U_{0}
^{k}U_{1}^{m}\right]  \\
& =\frac{1}{4}\sum_{\substack{j,k,\ell,m\in\left\{  0,1\right\}  :\\j\neq
k,\ell\neq m}}\left(  -1\right)  ^{j+k+\ell+m}i^{\ell}\left(  -i\right)
^{m}\operatorname{Tr}\!\left[
HU_{1}^{\ell}U_{0}^{j}\rho U_{0}^{k}U_{1}^{m}\right]  \\
& =\frac{1}{4}\left[
\begin{array}
[c]{c}
\left(  -1\right)  ^{0+1+0+1}i^{0}\left(  -i\right)  ^{1}\operatorname{Tr}\!\left[
HU_{1}^{0}U_{0}^{0}\rho U_{0}^{1}U_{1}^{1}\right]  +\left(  -1\right)
^{0+1+1+0}i^{1}\left(  -i\right)  ^{0}\operatorname{Tr}\!\left[  HU_{1}^{1}U_{0}
^{0}\rho U_{0}^{1}U_{1}^{0}\right]  \\
+\left(  -1\right)  ^{1+0+0+1}i^{0}\left(  -i\right)  ^{1}\operatorname{Tr}\!\left[
HU_{1}^{0}U_{0}^{1}\rho U_{0}^{0}U_{1}^{1}\right]  +\left(  -1\right)
^{1+0+1+0}i^{1}\left(  -i\right)  ^{0}\operatorname{Tr}\!\left[  HU_{1}^{1}U_{0}
^{1}\rho U_{0}^{0}U_{1}^{0}\right]
\end{array}
\right]  \\
& =\frac{1}{4}\left[  -i\operatorname{Tr}\!\left[  H\rho U_{0}U_{1}\right]
+i\operatorname{Tr}\!\left[  HU_{1}\rho U_{0}\right]  -i\operatorname{Tr}\!\left[
HU_{0}\rho U_{1}\right]  +i\operatorname{Tr}\!\left[  HU_{1}U_{0}\rho\right]
\right]  \\
& =\frac{i}{4}\left[  - \operatorname{Tr}\!\left[  U_{0}U_{1}H\rho \right]
+ \operatorname{Tr}\!\left[U_{0}  HU_{1}\rho \right]  -\operatorname{Tr}\!\left[U_{1}
HU_{0}\rho \right]  +\operatorname{Tr}\!\left[  HU_{1}U_{0}\rho\right]
\right]  \\
& =\frac{i}{4}\operatorname{Tr}\!\left[  \left\{  U_{0},\left[ H, U_{1}\right]
\right\}  \rho\right]  \\
& =\frac{i}{4}\left\langle \left\{  U_{0},\left[ H, U_{1}\right]  \right\}
\right\rangle _{\rho}.
\end{align}

\section{Ground-state energy estimation using evolved quantum Boltzmann machines}

\label{app:VQE}

\subsection{Gradient of the expected energy of an evolved quantum Boltzmann machine}

Let us start by giving detailed calculations on how to obtain~\eqref{eq:VQE-grad-theta}, that is, the partial derivative of the objective function in~\eqref{eq:VQE-costfunction} with respect to the parameter vector $\theta$. Using~\eqref{eq:grad_wrt_theta}, we find that
\begin{align}
\frac{\partial}{\partial\theta_{j}}\operatorname{Tr}[O\omega(\theta
,\phi)] & = \operatorname{Tr}\!\left[  O\left(  -\frac{1}{2}\left\{  e^{-iH(\phi)} \Phi_{\theta}(G_{j})e^{iH(\phi)},\omega(\theta,\phi)\right\} +\omega (\theta,\phi)\left\langle G_{j}\right\rangle _{\rho(\theta)}\right)  \right]
\\
&  =-\frac{1}{2}\operatorname{Tr}\!\left[  O\left\{  e^{-iH(\phi)}\Phi_{\theta
}(G_{j})e^{iH(\phi)},\omega(\theta,\phi)\right\}  \right] + \operatorname{Tr}\!
\left[  O\omega(\theta,\phi)\left\langle G_{j}\right\rangle _{\rho(\theta
)}\right]  \\
&  =-\frac{1}{2}\operatorname{Tr}\!\left[  e^{iH(\phi)} O e^{-iH(\phi)}  \left\{  \Phi_{\theta
}(G_{j}),\rho(\theta)\right\}  \right] + \operatorname{Tr}\!
\left[  O\omega(\theta,\phi)\left\langle G_{j}\right\rangle _{\rho(\theta
)}\right]  \\
&  =-\frac{1}{2}\operatorname{Tr}\!\left[  \left\{ e^{iH(\phi)} O e^{-iH(\phi)} , \Phi_{\theta
}(G_{j}) \right\} \rho(\theta)  \right] +\operatorname{Tr}\!
\left[  O\omega(\theta,\phi)\right]  \left\langle G_{j}\right\rangle
_{\rho(\theta)}\\
& = -\frac{1}{2}\left\langle \left\{ e^{iH(\phi)} O e^{-iH(\phi)} , \Phi_{\theta
}(G_{j})\right\} \right\rangle_{\rho(\theta)} + \left\langle O\right\rangle
_{\omega(\theta,\phi)}\left\langle G_{j}\right\rangle
_{\rho(\theta)}.
\end{align}

Using~\eqref{eq:grad_wrt_phi}, it is straightforward to prove the partial derivative of the cost function with respect to the parameter vector $\phi$ shown in~\eqref{eq:VQE-grad-phi}:
\begin{align}
\frac{\partial}{\partial\phi_{k}}\operatorname{Tr}[O\omega(\theta
,\phi)]& =i\operatorname{Tr}[O\left[  \omega(\theta,\phi),\Psi_{\phi}
(H_{k})\right]  ]\\
& = i\operatorname{Tr}\!\left[\left[\Psi_{\phi}
(H_{k}),O\right]  \omega(\theta,\phi)  \right]\\
& = i\left\langle\big[\Psi^{\dagger}_{\phi}(H_{k}),O\big] \right\rangle_{\omega(\theta,\phi)} .
\end{align}

\subsubsection{Quantum algorithm for estimating the partial derivative with respect to $\theta$}\label{app:VQE_grad_theta-est}

Here we show how to estimate the second term appearing in~\eqref{eq:VQE-grad-theta}. Consider that
\begin{align}
    & -\frac{1}{2} \left\langle\big\{ e^{iH(\phi)} O e^{-iH(\phi)}, \Phi_{\theta}(G_{j})\big\} \right\rangle_{\rho(\theta)} \nonumber \\
    & = -\frac{1}{2} \Tr\! \left[ \left\{ e^{iH(\phi)} O e^{-iH(\phi)}, \Phi_{\theta}(G_{j}) \right\}\rho(\theta)\right] \\
    & = \int_\mathbb{R} dt\, p(t) \left( -\frac{1}{2} \Tr\! \left[ \left\{ e^{iH(\phi)} O e^{-iH(\phi)} , e^{-iG(\theta)t} G_{j} e^{iG(\theta)t} \right\} \rho(\theta)\right] \right).\label{eq:VQE-grad-theta-eqiv}
\end{align}
We are now in a position to present an algorithm (Algorithm~\ref{algo:VQE-grad-theta-2}) to estimate the second term of~\eqref{eq:VQE-grad-theta} using its equivalent form shown in~\eqref{eq:VQE-grad-theta-eqiv}. At the core of our algorithm lies the quantum circuit that estimates the expected value of the anticommutator of two operators, $-\frac{1}{2} \left\langle \left\{ H, U\right\} \right\rangle_{\rho}$, where $H$ is Hermitian and $U$ is Hermitian and unitary (refer to Appendix~\ref{app:Hadamard_test} and Figure~\ref{fig:qc-primitive-anticomm}).  In this case, we choose $\rho=\rho(\theta)$, $U=e^{-iG(\theta)t} G_{j} e^{iG(\theta)t}$, $H=e^{iH(\phi)} O e^{-iH(\phi)}$. We then make some further simplifications that follow because $\rho(\theta)$ commutes with $e^{-iG(\theta)t}$. Accordingly, the quantum circuit that estimates the integrand of~\eqref{eq:VQE-grad-theta-eqiv} is depicted in Figure~\ref{fig:VQE-grad-theta}.

\begin{algorithm}[H]
\caption{\texorpdfstring{$\mathtt{gradient\_\theta\_ground\_state\_energy}(j, \theta, \{G_\ell\}_{\ell=1}^{J}, \phi, \{H_m\}_{m=1}^{K}, p(\cdot), \varepsilon, \delta)$}{estimate first term}}
\label{algo:VQE-grad-theta-2}
\begin{algorithmic}[1]
\STATE \textbf{Input:} Index $j \in [J]$, parameter vectors $\theta = \left( \theta_{1}, \ldots,  \theta_{J}\right)^{\mathsf{T}} \in \mathbb{R}^{J}$ and $\phi = \left( \phi_{1}, \ldots,  \phi_{K}\right)^{\mathsf{T}} \in \mathbb{R}^{K}$, Gibbs local Hamiltonians $\{G_\ell\}_{\ell=1}^{J}$ and $\{H_m\}_{m=1}^{K}$, probability distribution $p(t)$ over $\mathbb{R}$, precision $\varepsilon > 0$, error probability $\delta \in (0,1)$
\STATE $N \leftarrow \lceil\sfrac{2 \ln(\sfrac{2}{\delta})}{\varepsilon^2}\rceil$
\FOR{$n = 0$ to $N-1$}
\STATE Initialize the control register to $|1\rangle\!\langle 1 |$
\STATE Prepare the system register in the state $\rho(\theta)$
\STATE Sample $t$ at random with probability $p(t)$ (defined in~\eqref{eq:high-peak-tent-density})
\STATE Apply the Hadamard gate to the control register
\STATE Apply the following unitaries to the control and system registers:
\STATE \hspace{0.6cm} \textbullet~Controlled-$G_j$: $G_j$ is a local unitary acting on the system register, controlled by the control register
\STATE \hspace{0.6cm} \textbullet~$e^{-iG(\theta)t}$: Hamiltonian simulation for time $t$ on the system register
\STATE \hspace{0.6cm} \textbullet~$e^{-iH(\phi)}$: Hamiltonian simulation on the system register
\STATE Apply the Hadamard gate to the control register
\STATE Measure the control register in the computational basis and store the measurement outcome~$b_n$
\STATE Measure the system register in the eigenbasis of $O$ and store the measurement outcome $\lambda_n$
\STATE $Y_{n}^{(\theta)} \leftarrow (-1)^{b_n}\lambda_n$
\ENDFOR

\STATE \textbf{return} $\overline{Y}^{(\theta)} \leftarrow \frac{1}{N}\sum_{n=0}^{N-1}Y_{n}^{(\operatorname{\theta})}$
\end{algorithmic}
\end{algorithm}

\begin{remark}\label{remark_O}
If it is not straightforward to measure in the eigenbasis of $O$, but instead $O$ is  a linear combination of simpler observables that are each easy to measure, then one can adopt a sampling approach along the lines of \cite[Algorithm~1]{patel2024quantumboltzmannmachine}. See Remark~\ref{rem:measuring-G-theta} for further discussions of this point.
\end{remark}

\subsubsection{Quantum algorithm for estimating the partial derivative with respect to $\phi$}\label{app:VQE_grad_phi-est}

Here we show how to estimate the quantity in~\eqref{eq:VQE-grad-phi}. Consider that
\begin{align}
     i\left\langle\big[\Psi^{\dagger}_{\phi}(H_{k}),O\big] \right\rangle_{\omega(\theta,\phi)}
    & =i\operatorname{Tr}\!\left[\left[\Psi^{\dagger}_{\phi}(H_{k}), O\right]  \omega(\theta,\phi) \right]  \\
    & = \int_0^1 dt \left(i\operatorname{Tr}\!\left[\left[ e^{-iH(\phi)t}H_{k} e^{iH(\phi)t},  O \right] \omega(\theta,\phi)  \right] \right).\label{eq:VQE-grad-phi-eqiv}
\end{align}

We are now in a position to present an algorithm (Algorithm~\ref{algo:VQE-grad-phi}) to estimate~\eqref{eq:VQE-grad-phi} using its equivalent form shown in~\eqref{eq:VQE-grad-phi-eqiv}. At the core of our algorithm lies the quantum circuit that estimates the expected value of the commutator of two operators, $\frac{i}{2} \left\langle \left[ U, H\right] \right\rangle_{\rho}$, where $H$ is Hermitian and $U$ is Hermitian and unitary (refer to Appendix~\ref{app:Hadamard_test} and Figure~\ref{fig:qc-primitive-comm}). In this case, we choose $\rho=\omega(\theta,\phi)$, $U=e^{-iH(\phi)t}H_{k} e^{iH(\phi)t}$, $H=O$. Accordingly, the quantum circuit that estimates the integrand of~\eqref{eq:VQE-grad-phi-eqiv} is depicted in Figure~\ref{fig:VQE-grad-phi}.

\begin{algorithm}[H]
\caption{\texorpdfstring{$\mathtt{gradient\_\phi\_ground\_state\_energy}(j, \theta, \{G_\ell\}_{\ell=1}^{J}, \phi, \{H_m\}_{m=1}^{K}, \varepsilon, \delta)$}{estimate first term}}
\label{algo:VQE-grad-phi}
\begin{algorithmic}[1]
\STATE \textbf{Input:} Index $k \in [K]$, parameter vectors $\theta = \left( \theta_{1}, \ldots,  \theta_{J}\right)^{\mathsf{T}} \in \mathbb{R}^{J}$ and $\phi = \left( \phi_{1}, \ldots,  \phi_{K}\right)^{\mathsf{T}} \in \mathbb{R}^{K}$, Gibbs local Hamiltonians $\{G_\ell\}_{\ell=1}^{J}$ and $\{H_m\}_{m=1}^{K}$, precision $\varepsilon > 0$, error probability $\delta \in (0,1)$
\STATE $N \leftarrow \lceil\sfrac{2 \ln(\sfrac{2}{\delta})}{\varepsilon^2}\rceil$
\FOR{$n = 0$ to $N-1$}
\STATE Initialize the control register to $|1\rangle\!\langle 1 |$
\STATE Prepare the system register in the state $\omega(\theta,\phi)$
\STATE Sample $t$ uniformly at random from the interval $[0,1]$
\STATE Apply the Hadamard gate and the phase gate $S$ to the control register
\STATE Apply the following unitaries to the control and system registers:
\STATE \hspace{0.6cm} \textbullet~$e^{iH(\phi)t}$: Hamiltonian simulation for time $t$ on the system register
\STATE \hspace{0.6cm} \textbullet~Controlled-$H_k$: $H_k$ is a local unitary acting on the system register, controlled by the control register
\STATE \hspace{0.6cm} \textbullet~$e^{-iH(\phi)t}$: Hamiltonian simulation for time $t$ on the system register
\STATE Apply the Hadamard gate to the control register
\STATE Measure the control register in the computational basis and store the measurement outcome~$b_n$
\STATE Measure the system register in the eigenbasis of $O$ and store the measurement outcome $\lambda_n$ (see Remark~\ref{remark_O})
\STATE $Y_{n}^{(\phi)} \leftarrow (-1)^{b_n}\lambda_n$
\ENDFOR

\STATE \textbf{return} $\overline{Y}^{(\phi)} \leftarrow 2\times\frac{1}{N}\sum_{n=0}^{N-1}Y_{n}^{(\operatorname{\phi})}$
\end{algorithmic}
\end{algorithm}

\section{Evolved quantum Boltzmann machines for generative modeling}

\label{app:gen_model}

Here we prove the alternative formulation of the quantum relative entropy in~\eqref{eq:rel_entr_alt}. Consider that
\begin{align}
    D( \eta \| \omega(\theta,\phi) ) & = \Tr\!\left[ \eta \ln \eta \right] - \Tr\!\left[ \eta \ln \omega(\theta,\phi) \right]\\
    & =  \Tr\!\left[ \eta \ln \eta \right] - \Tr\!\left[ \eta \ln \!\left( e^{-iH(\phi)}\rho(\theta) e^{iH(\phi)} \right)\right]\\
    & = \Tr\!\left[ \eta \ln \eta \right] - \Tr\!\left[ \eta  e^{-iH(\phi)} \ln \rho(\theta) e^{iH(\phi)} \right]\label{pass:log_trio}\\
    & = \Tr\!\left[ \eta \ln \eta \right] - \Tr\!\left[e^{iH(\phi)}  \eta  e^{-iH(\phi)} (-G(\theta) )\right] - \Tr\!\left[e^{iH(\phi)}  \eta  e^{-iH(\phi)} (-\ln Z(\theta)) \right]\\
    & = \Tr\!\left[ \eta \ln \eta \right] + \Tr\!\left[ G(\theta)e^{iH(\phi)} \eta e^{-iH(\phi)}  \right] + \ln Z(\theta)\\
    & = \Tr\!\left[ \eta \ln \eta \right] + \Tr\!\left[ G(\theta) \eta(\phi)  \right] + \ln Z(\theta),
\end{align}
where in~\eqref{pass:log_trio} we used the fact that $\ln\!\left( U\!AU^\dag\right)=U \ (\ln A) \ U^\dag$ when $A$ is a positive semidefinite matrix and $U$ is a unitary matrix.

\subsection{Gradient of the quantum relative entropy}

\label{app:grad_gen_model}

We first prove Theorem~\ref{thm:gen_model_der_theta}, that is, how to obtain the $j$th element of the gradient of the quantum relative entropy with respect to the $\theta$ parameter. Using~\eqref{eq:rel_entr_alt}, we find that
    \begin{align}
    \partial_{\theta_j} D( \eta \| \omega(\theta,\phi) ) & =  \partial_{\theta_j}\! \Tr\!\left[    G(\theta) \eta(\phi) \right] + \partial_{\theta_j}\! \ln Z(\theta)\\
    & = \Tr\!\left[    G_j \eta(\phi) \right] - \Tr\!\left[ G_j \rho(\theta)\right]\\
    & = \left\langle G_j \right\rangle_{\eta(\phi)} - \left\langle G_j \right\rangle_{\rho(\theta)},
\end{align}
concluding the proof of Theorem~\ref{thm:gen_model_der_theta}.

Now, we prove Theorem~\ref{thm:gen_model_der_phi}, that is, the analytical expression of  the $k$th element of the gradient of the quantum relative entropy with respect to  $\phi$. Using~\eqref{eq:rel_entr_alt}, we find that
\begin{align}
    \partial_{\phi_k} D( \eta \| \omega(\theta,\phi) ) & = \partial_{\phi_k}\! \Tr\!\left[ e^{iH(\phi)} \eta e^{-iH(\phi)} G(\theta) \right]\\
    & = \Tr\!\left[ \left(\partial_{\phi_k}\! e^{iH(\phi)} \right) \eta e^{-iH(\phi)} G(\theta) + e^{iH(\phi)} \eta \left(\partial_{\phi_k}\! e^{-iH(\phi)}\right) G(\theta) \right]\\
    & = \Tr\!\left[ ie^{iH(\phi)} \Psi^\dag_\phi (H_k) \eta e^{-iH(\phi)} G(\theta) - i e^{iH(\phi)} \eta \Psi^\dag_\phi (H_k)e^{-iH(\phi)} G(\theta) \right]\label{pass:partials_phi}\\
    & = i \Tr\!\left[ e^{iH(\phi)} \Psi^\dag_\phi (H_k) \eta e^{-iH(\phi)} G(\theta) \right] - i \Tr\!\left[ e^{iH(\phi)} \eta \Psi^\dag_\phi (H_k)e^{-iH(\phi)} G(\theta)  \right]\\
    & = i \Tr\!\left[ e^{iH(\phi)} \Psi^\dag_\phi (H_k) e^{-iH(\phi)} e^{iH(\phi)} \eta e^{-iH(\phi)} G(\theta) \right] \nonumber\\
    & \hspace{1cm} - i \Tr\!\left[ e^{iH(\phi)} \eta e^{-iH(\phi)} e^{iH(\phi)} \Psi^\dag_\phi (H_k)e^{-iH(\phi)} G(\theta)  \right]\\
    & = i \Tr\!\left[ \Psi_\phi (H_k) \eta(\phi) G(\theta) \right] - i \Tr\!\left[ \eta(\phi) \Psi_\phi (H_k) G(\theta)  \right]\\
    & = i \Tr\!\left[ G(\theta) \Psi_\phi (H_k) \eta(\phi)  \right] - i \Tr\!\left[ \Psi_\phi (H_k) G(\theta)  \eta(\phi) \right]\\
    & = i \Tr\!\left[ \left[  G(\theta) ,\Psi_\phi (H_k) \right] \eta(\phi)  \right]\\
    & = i \left\langle  \left[ G(\theta) ,\Psi_\phi (H_k) \right] \right\rangle_{\eta(\phi)},
\end{align}
where, in~\eqref{pass:partials_phi}, we have used the facts that $\partial_{\phi_{k}} e^{iH(\phi)}
=ie^{iH(\phi)}\Psi^{\dagger}_{\phi}(H_{k})$ and $\partial_{\phi_{k}} e^{-iH(\phi)} = 
-i\Psi^{\dagger}_{\phi}(H_{k})e^{-iH(\phi)}$, as derived in~\eqref{eq:deriv-phi-proof-1}--\eqref{eq:deriv-phi-proof-last}.

\subsubsection{Quantum algorithm for estimating the partial derivative with respect to $\phi$}\label{app:gen-mod_grad_phi-est}

Here we show how to estimate the quantity in~\eqref{eq:gen_mod_grad_phi}. Consider that
\begin{align}
     i \left\langle  \left[  G(\theta)  ,\Psi_\phi (H_k) \right] \right\rangle_{\eta(\phi)}
    & =i\operatorname{Tr}\!\left[\left[G(\theta),\Psi_{\phi}(H_{k}) \right]  \eta(\phi) \right] \\
    & = \int_0^1 dt \left(i\operatorname{Tr}\!\left[\left[ G(\theta), e^{iH(\phi)t}H_{k} e^{-iH(\phi)t} \right] \eta(\phi)  \right] \right).\label{eq:gen_mod_grad_phi-eqiv}
\end{align}
We are now in a position to present an algorithm to estimate~\eqref{eq:gen_mod_grad_phi} using its equivalent form shown in~\eqref{eq:gen_mod_grad_phi-eqiv}. The algorithm is similar to Algorithm~\ref{algo:VQE-grad-phi}, so here we provide a high-level description of how it works.  At its core, the algorithm relies on a quantum circuit that estimates the expected value of the commutator of two operators (see Appendix~\ref{app:Hadamard_test}). Specifically, if the control register in Figure~\ref{fig:qc-primitive-comm} is initialized in the state $\ket{0}$ instead of $\ket{1}$, the output of the circuit is $\frac{i}{2} \left\langle \left[ H, U\right] \right\rangle_{\rho}$, where $H$ is Hermitian and $U$ is Hermitian and unitary.  In this case, we choose $U=e^{iH(\phi)t} H_{k} e^{-iH(\phi)t}$, $H=G(\theta)$, and $\rho=\eta(\phi)$, where $\eta(\phi)$ is obtained by applying $e^{iH(\phi)}$ to $\eta$. Accordingly, the quantum circuit that plays a role in estimating the integrand of~\eqref{eq:gen_mod_grad_phi-eqiv} is depicted in Figure~\ref{fig:gen-mod-grad-phi}. The algorithm involves running this circuit $N$ times, where $N$ is determined by the desired precision and error probability. During each run, the time $t$ for the Hamiltonian evolution is sampled uniformly at random from the interval $[0,1]$. The final estimation of \eqref{eq:gen_mod_grad_phi} is obtained by averaging the outputs of the $N$ runs and multiplying the result by $2\left \| \theta \right\|_1$. For measuring $G(\theta) $, we adopt a sampling approach (see Remark~\ref{rem:measuring-G-theta}).

\section{Proof of Theorem~\ref{thm:QFI-formulas}}

\label{app:QFI-formulas}

\subsection{Proof of Equations~\eqref{eq:FB-formula-1} and~\eqref{eq:FB-formula-2}}

Consider that
\begin{align}
 \frac{\partial^{2}}{\partial\varepsilon_{i}\partial\varepsilon_{j}}\left[
-2\ln F(\sigma(\gamma),\sigma(\gamma+\varepsilon))\right]  
&  =-2\frac{\partial^{2}}{\partial\varepsilon_{i}\partial\varepsilon_{j}}
\ln\left(  \operatorname{Tr}\!\left[  \sqrt{\sqrt{\sigma(\gamma)}\sigma
(\gamma+\varepsilon)\sqrt{\sigma(\gamma)}}\right]  \right)  ^{2}\\
&  =-4\frac{\partial^{2}}{\partial\varepsilon_{i}\partial\varepsilon_{j}}
\ln\operatorname{Tr}\!\left[  \sqrt{\sqrt{\sigma(\gamma)}\sigma(\gamma
+\varepsilon)\sqrt{\sigma(\gamma)}}\right]  \\
&  =-4\frac{\partial}{\partial\varepsilon_{i}}\left(  \frac{\partial}
{\partial\varepsilon_{j}}\ln\operatorname{Tr}\!\left[  \sqrt{\sqrt
{\sigma(\gamma)}\sigma(\gamma+\varepsilon)\sqrt{\sigma(\gamma)}}\right]
\right)  \\
&  =-4\frac{\partial}{\partial\varepsilon_{i}}\left(  \frac{\frac{\partial
}{\partial\varepsilon_{j}}\operatorname{Tr}\!\left[  \sqrt{\sqrt{\sigma
(\gamma)}\sigma(\gamma+\varepsilon)\sqrt{\sigma(\gamma)}}\right]
}{\operatorname{Tr}\!\left[  \sqrt{\sqrt{\sigma(\gamma)}\sigma(\gamma
+\varepsilon)\sqrt{\sigma(\gamma)}}\right]  }\right)
.\label{eq:FB-deriv-full-2}
\end{align}
Recalling from \cite[Theorem~1.1]{DelMoral2018}  that
\begin{equation}
\frac{\partial}{\partial\gamma_{j}}\sqrt{\sigma(\gamma)}=\int_{0}^{\infty
}dt\ e^{-t\sqrt{\sigma(\gamma)}}\left(  \frac{\partial}{\partial\gamma_{j}
}\sigma(\gamma)\right)  e^{-t\sqrt{\sigma(\gamma)}},
\end{equation}
now consider that
\begin{align}
&  \frac{\partial}{\partial\varepsilon_{j}}\operatorname{Tr}\!\left[
\sqrt{\sqrt{\sigma(\gamma)}\sigma(\gamma+\varepsilon)\sqrt{\sigma(\gamma)}
}\right]  \nonumber\\
&  =\operatorname{Tr}\!\left[  \int_{0}^{\infty}dt\ e^{-t\sqrt{\sqrt
{\sigma(\gamma)}\sigma(\gamma+\varepsilon)\sqrt{\sigma(\gamma)}}}\left(
\frac{\partial}{\partial\varepsilon_{j}}\sqrt{\sigma(\gamma)}\sigma
(\gamma+\varepsilon)\sqrt{\sigma(\gamma)}\right)  e^{-t\sqrt{\sqrt
{\sigma(\gamma)}\sigma(\gamma+\varepsilon)\sqrt{\sigma(\gamma)}}}\right]  \\
&  =\int_{0}^{\infty}dt\ \operatorname{Tr}\!\left[  e^{-2t\sqrt{\sqrt
{\sigma(\gamma)}\sigma(\gamma+\varepsilon)\sqrt{\sigma(\gamma)}}}\sqrt
{\sigma(\gamma)}\left(  \frac{\partial}{\partial\varepsilon_{j}}\sigma
(\gamma+\varepsilon)\right)  \sqrt{\sigma(\gamma)}\right]  \\
&  =\frac{1}{2}\operatorname{Tr}\!\left[  \left(  \sqrt{\sigma(\gamma)}
\sigma(\gamma+\varepsilon)\sqrt{\sigma(\gamma)}\right)  ^{-\frac{1}{2}}
\sqrt{\sigma(\gamma)}\left(  \frac{\partial}{\partial\varepsilon_{j}}
\sigma(\gamma+\varepsilon)\right)  \sqrt{\sigma(\gamma)}\right]
.\label{eq:FB-deriv-full-1}
\end{align}
Substituting~\eqref{eq:FB-deriv-full-1} into the numerator of
\eqref{eq:FB-deriv-full-2}, we find that
\begin{align}
&  \frac{\partial^{2}}{\partial\varepsilon_{i}\partial\varepsilon_{j}}\left[
-2\ln F(\sigma(\gamma),\sigma(\gamma+\varepsilon))\right]  \nonumber\\
&  =-4\frac{\partial}{\partial\varepsilon_{i}}\left(  \frac{\frac{\partial
}{\partial\varepsilon_{j}}\operatorname{Tr}\!\left[  \sqrt{\sqrt{\sigma
(\gamma)}\sigma(\gamma+\varepsilon)\sqrt{\sigma(\gamma)}}\right]
}{\operatorname{Tr}\!\left[  \sqrt{\sqrt{\sigma(\gamma)}\sigma(\gamma
+\varepsilon)\sqrt{\sigma(\gamma)}}\right]  }\right)  \\
&  =-2\frac{\partial}{\partial\varepsilon_{i}}\left(  \frac{\operatorname{Tr}\!
\left[  \left(  \sqrt{\sigma(\gamma)}\sigma(\gamma+\varepsilon)\sqrt
{\sigma(\gamma)}\right)  ^{-\frac{1}{2}}\sqrt{\sigma(\gamma)}\left(
\frac{\partial}{\partial\varepsilon_{j}}\sigma(\gamma+\varepsilon)\right)
\sqrt{\sigma(\gamma)}\right]  }{\operatorname{Tr}\!\left[  \sqrt{\sqrt
{\sigma(\gamma)}\sigma(\gamma+\varepsilon)\sqrt{\sigma(\gamma)}}\right]
}\right)  \\
&  =2\frac{\frac{\partial}{\partial\varepsilon_{i}}\operatorname{Tr}\!\left[
\sqrt{\sqrt{\sigma(\gamma)}\sigma(\gamma+\varepsilon)\sqrt{\sigma(\gamma)}
}\right]  \operatorname{Tr}\!\left[  \left(  \sqrt{\sigma(\gamma)}
\sigma(\gamma+\varepsilon)\sqrt{\sigma(\gamma)}\right)  ^{-\frac{1}{2}}
\sqrt{\sigma(\gamma)}\left(  \frac{\partial}{\partial\varepsilon_{j}}
\sigma(\gamma+\varepsilon)\right)  \sqrt{\sigma(\gamma)}\right]
}{\operatorname{Tr}\!\left[  \sqrt{\sqrt{\sigma(\gamma)}\sigma(\gamma
+\varepsilon)\sqrt{\sigma(\gamma)}}\right]  ^{2}}\nonumber\\
&  \qquad-2\frac{\frac{\partial}{\partial\varepsilon_{i}}\operatorname{Tr}\!
\left[  \left(  \sqrt{\sigma(\gamma)}\sigma(\gamma+\varepsilon)\sqrt
{\sigma(\gamma)}\right)  ^{-\frac{1}{2}}\sqrt{\sigma(\gamma)}\left(
\frac{\partial}{\partial\varepsilon_{j}}\sigma(\gamma+\varepsilon)\right)
\sqrt{\sigma(\gamma)}\right]  }{\operatorname{Tr}\!\left[  \sqrt{\sqrt
{\sigma(\gamma)}\sigma(\gamma+\varepsilon)\sqrt{\sigma(\gamma)}}\right]  }\\
&  =\frac{\left(
\begin{array}
[c]{c}
\operatorname{Tr}\!\left[  \left(  \sqrt{\sigma(\gamma)}\sigma(\gamma
+\varepsilon)\sqrt{\sigma(\gamma)}\right)  ^{-\frac{1}{2}}\sqrt{\sigma
(\gamma)}\left(  \frac{\partial}{\partial\varepsilon_{i}}\sigma(\gamma
+\varepsilon)\right)  \sqrt{\sigma(\gamma)}\right]  \times\\
\operatorname{Tr}\!\left[  \left(  \sqrt{\sigma(\gamma)}\sigma(\gamma
+\varepsilon)\sqrt{\sigma(\gamma)}\right)  ^{-\frac{1}{2}}\sqrt{\sigma
(\gamma)}\left(  \frac{\partial}{\partial\varepsilon_{j}}\sigma(\gamma
+\varepsilon)\right)  \sqrt{\sigma(\gamma)}\right]
\end{array}
\right)  }{\operatorname{Tr}\!\left[  \sqrt{\sqrt{\sigma(\gamma)}\sigma
(\gamma+\varepsilon)\sqrt{\sigma(\gamma)}}\right]  ^{2}}\nonumber\\
&  \qquad-2\frac{\operatorname{Tr}\!\left[  \left(  \frac{\partial}
{\partial\varepsilon_{i}}\left(  \sqrt{\sigma(\gamma)}\sigma(\gamma
+\varepsilon)\sqrt{\sigma(\gamma)}\right)  ^{-\frac{1}{2}}\right)
\sqrt{\sigma(\gamma)}\left(  \frac{\partial}{\partial\varepsilon_{j}}
\sigma(\gamma+\varepsilon)\right)  \sqrt{\sigma(\gamma)}\right]
}{\operatorname{Tr}\!\left[  \sqrt{\sqrt{\sigma(\gamma)}\sigma(\gamma
+\varepsilon)\sqrt{\sigma(\gamma)}}\right]  }\nonumber\\
&  \qquad-2\frac{\operatorname{Tr}\!\left[  \left(  \sqrt{\sigma(\gamma
)}\sigma(\gamma+\varepsilon)\sqrt{\sigma(\gamma)}\right)  ^{-\frac{1}{2}}
\sqrt{\sigma(\gamma)}\left(  \frac{\partial^{2}}{\partial\varepsilon
_{i}\partial\varepsilon_{j}}\sigma(\gamma+\varepsilon)\right)  \sqrt
{\sigma(\gamma)}\right]  }{\operatorname{Tr}\!\left[  \sqrt{\sqrt
{\sigma(\gamma)}\sigma(\gamma+\varepsilon)\sqrt{\sigma(\gamma)}}\right]  }.
\end{align}
Then it follows that
\begin{align}
&  \left.  \frac{\partial^{2}}{\partial\varepsilon_{i}\partial\varepsilon_{j}
}\left[  -2\ln F(\sigma(\gamma),\sigma(\gamma+\varepsilon))\right]
\right\vert _{\varepsilon=0}\nonumber\\
&  =\frac{\left(
\begin{array}
[c]{c}
\operatorname{Tr}\!\left[  \left(  \sqrt{\sigma(\gamma)}\sigma(\gamma
)\sqrt{\sigma(\gamma)}\right)  ^{-\frac{1}{2}}\sqrt{\sigma(\gamma)}\left(
\left.  \frac{\partial}{\partial\varepsilon_{i}}\sigma(\gamma+\varepsilon
)\right\vert _{\varepsilon=0}\right)  \sqrt{\sigma(\gamma)}\right]  \times\\
\operatorname{Tr}\!\left[  \left(  \sqrt{\sigma(\gamma)}\sigma(\gamma
)\sqrt{\sigma(\gamma)}\right)  ^{-\frac{1}{2}}\sqrt{\sigma(\gamma)}\left(
\left.  \frac{\partial}{\partial\varepsilon_{j}}\sigma(\gamma+\varepsilon
)\right\vert _{\varepsilon=0}\right)  \sqrt{\sigma(\gamma)}\right]
\end{array}
\right)  }{\operatorname{Tr}\!\left[  \sqrt{\sqrt{\sigma(\gamma)}\sigma
(\gamma)\sqrt{\sigma(\gamma)}}\right]  ^{2}}\nonumber\\
&  \qquad-2\frac{\operatorname{Tr}\!\left[  \left(  \left.  \frac{\partial
}{\partial\varepsilon_{i}}\left(  \sqrt{\sigma(\gamma)}\sigma(\gamma
+\varepsilon)\sqrt{\sigma(\gamma)}\right)  ^{-\frac{1}{2}}\right\vert
_{\varepsilon=0}\right)  \sqrt{\sigma(\gamma)}\left(  \left.  \frac{\partial
}{\partial\varepsilon_{j}}\sigma(\gamma+\varepsilon)\right\vert _{\varepsilon
=0}\right)  \sqrt{\sigma(\gamma)}\right]  }{\operatorname{Tr}\!\left[
\sqrt{\sqrt{\sigma(\gamma)}\sigma(\gamma)\sqrt{\sigma(\gamma)}}\right]
}\nonumber\\
&  \qquad-2\frac{\operatorname{Tr}\!\left[  \left(  \sqrt{\sigma(\gamma
)}\sigma(\gamma)\sqrt{\sigma(\gamma)}\right)  ^{-\frac{1}{2}}\sqrt
{\sigma(\gamma)}\left(  \left.  \frac{\partial^{2}}{\partial\varepsilon
_{i}\partial\varepsilon_{j}}\sigma(\gamma+\varepsilon)\right\vert
_{\varepsilon=0}\right)  \sqrt{\sigma(\gamma)}\right]  }{\operatorname{Tr}
\!\left[  \sqrt{\sqrt{\sigma(\gamma)}\sigma(\gamma)\sqrt{\sigma(\gamma)}
}\right]  }\\
&  =\operatorname{Tr}\!\left[  \left(  \left.  \frac{\partial}{\partial
\varepsilon_{i}}\sigma(\gamma+\varepsilon)\right\vert _{\varepsilon=0}\right)
\right]  \operatorname{Tr}\!\left[  \left(  \left.  \frac{\partial}
{\partial\varepsilon_{j}}\sigma(\gamma+\varepsilon)\right\vert _{\varepsilon
=0}\right)  \right]  \nonumber\\
&  \qquad-2\operatorname{Tr}\!\left[  \left(  \left.  \frac{\partial}
{\partial\varepsilon_{i}}\left(  \sqrt{\sigma(\gamma)}\sigma(\gamma
+\varepsilon)\sqrt{\sigma(\gamma)}\right)  ^{-\frac{1}{2}}\right\vert
_{\varepsilon=0}\right)  \sqrt{\sigma(\gamma)}\left(  \left.  \frac{\partial
}{\partial\varepsilon_{j}}\sigma(\gamma+\varepsilon)\right\vert _{\varepsilon
=0}\right)  \sqrt{\sigma(\gamma)}\right]  \nonumber\\
&  \qquad-2\operatorname{Tr}\!\left[  \left(  \left.  \frac{\partial^{2}
}{\partial\varepsilon_{i}\partial\varepsilon_{j}}\sigma(\gamma+\varepsilon
)\right\vert _{\varepsilon=0}\right)  \right]  \\
&  =-2\operatorname{Tr}\!\left[  \left(  \left.  \frac{\partial}
{\partial\varepsilon_{i}}\left(  \sqrt{\sigma(\gamma)}\sigma(\gamma
+\varepsilon)\sqrt{\sigma(\gamma)}\right)  ^{-\frac{1}{2}}\right\vert
_{\varepsilon=0}\right)  \sqrt{\sigma(\gamma)}\left(  \left.  \frac{\partial
}{\partial\varepsilon_{j}}\sigma(\gamma+\varepsilon)\right\vert _{\varepsilon
=0}\right)  \sqrt{\sigma(\gamma)}\right]  .
\end{align}
In the transition to the last line above, we observed that
\begin{align}
\operatorname{Tr}\!\left[  \left(  \left.  \frac{\partial}{\partial
\varepsilon_{i}}\sigma(\gamma+\varepsilon)\right\vert _{\varepsilon=0}\right)
\right]   &  =\frac{\partial}{\partial\varepsilon_{i}}\left.
\operatorname{Tr}\!\left[  \sigma(\gamma+\varepsilon)\right]  \right\vert
_{\varepsilon=0}  =0,\\
\operatorname{Tr}\!\left[  \left(  \left.  \frac{\partial^{2}}{\partial
\varepsilon_{i}\partial\varepsilon_{j}}\sigma(\gamma+\varepsilon)\right\vert
_{\varepsilon=0}\right)  \right]   &  =\frac{\partial^{2}}{\partial
\varepsilon_{i}\partial\varepsilon_{j}}\left.  \operatorname{Tr}\!\left[
\sigma(\gamma+\varepsilon)\right]  \right\vert _{\varepsilon=0}  =0.
\end{align}
So then
\begin{multline}
\left.  \frac{\partial^{2}}{\partial\varepsilon_{i}\partial\varepsilon_{j}
}\left[  -2\ln F(\sigma(\gamma),\sigma(\gamma+\varepsilon))\right]
\right\vert _{\varepsilon=0}\label{eq:FB-deriv-full-4}\\
=-2\operatorname{Tr}\!\left[  \left(  \left.  \frac{\partial}{\partial
\varepsilon_{i}}\left(  \sqrt{\sigma(\gamma)}\sigma(\gamma+\varepsilon
)\sqrt{\sigma(\gamma)}\right)  ^{-\frac{1}{2}}\right\vert _{\varepsilon
=0}\right)  \sqrt{\sigma(\gamma)}\left(  \left.  \frac{\partial}
{\partial\varepsilon_{j}}\sigma(\gamma+\varepsilon)\right\vert _{\varepsilon
=0}\right)  \sqrt{\sigma(\gamma)}\right]  .
\end{multline}
Now, recalling that
\begin{equation}
\frac{\partial}{\partial\gamma_{j}}\left(  \sigma(\gamma)^{-1}\right)
=-\sigma(\gamma)^{-1}\left(  \frac{\partial}{\partial\gamma_{j}}\sigma
(\gamma)\right)  \sigma(\gamma)^{-1},
\end{equation}
which follows from applying $\frac{\partial}{\partial\gamma_{j}}$ to the equation $I = \sigma(\gamma)\sigma(\gamma)^{-1}$ and solving for $\frac{\partial}{\partial\gamma_{j}}\left(  \sigma(\gamma)^{-1}\right)$,
consider that
\begin{align}
&  \frac{\partial}{\partial\varepsilon_{i}}\left(  \sqrt{\sigma(\gamma)}
\sigma(\gamma+\varepsilon)\sqrt{\sigma(\gamma)}\right)  ^{-\frac{1}{2}
}\nonumber\\
&  =\frac{\partial}{\partial\varepsilon_{i}}\left[  \sqrt{\sqrt{\sigma
(\gamma)}\sigma(\gamma+\varepsilon)\sqrt{\sigma(\gamma)}}\right]  ^{-1}\\
&  =-\sqrt{\sqrt{\sigma(\gamma)}\sigma(\gamma+\varepsilon)\sqrt{\sigma
(\gamma)}}^{-1}\left(  \frac{\partial}{\partial\varepsilon_{i}}\sqrt
{\sqrt{\sigma(\gamma)}\sigma(\gamma+\varepsilon)\sqrt{\sigma(\gamma)}}\right)
\sqrt{\sqrt{\sigma(\gamma)}\sigma(\gamma+\varepsilon)\sqrt{\sigma(\gamma)}
}^{-1}\\
&  =-\sqrt{\sqrt{\sigma(\gamma)}\sigma(\gamma+\varepsilon)\sqrt{\sigma
(\gamma)}}^{-1}\int_{0}^{\infty}dt\ e^{-t\sqrt{\sqrt{\sigma(\gamma)}
\sigma(\gamma+\varepsilon)\sqrt{\sigma(\gamma)}}}\left(  \frac{\partial
}{\partial\varepsilon_{i}}\sqrt{\sigma(\gamma)}\sigma(\gamma+\varepsilon
)\sqrt{\sigma(\gamma)}\right)  \times\nonumber\\
&  \qquad e^{-t\sqrt{\sqrt{\sigma(\gamma)}\sigma(\gamma+\varepsilon
)\sqrt{\sigma(\gamma)}}}\sqrt{\sqrt{\sigma(\gamma)}\sigma(\gamma
+\varepsilon)\sqrt{\sigma(\gamma)}}^{-1}\\
&  =-\sqrt{\sqrt{\sigma(\gamma)}\sigma(\gamma+\varepsilon)\sqrt{\sigma
(\gamma)}}^{-1}\int_{0}^{\infty}dt\ e^{-t\sqrt{\sqrt{\sigma(\gamma)}
\sigma(\gamma+\varepsilon)\sqrt{\sigma(\gamma)}}}\sqrt{\sigma(\gamma)}\left(
\frac{\partial}{\partial\varepsilon_{i}}\sigma(\gamma+\varepsilon)\right)
\sqrt{\sigma(\gamma)}\times\nonumber\\
&  \qquad e^{-t\sqrt{\sqrt{\sigma(\gamma)}\sigma(\gamma+\varepsilon
)\sqrt{\sigma(\gamma)}}}\sqrt{\sqrt{\sigma(\gamma)}\sigma(\gamma
+\varepsilon)\sqrt{\sigma(\gamma)}}^{-1},
\end{align}
which implies that
\begin{align}
&  \left.  \frac{\partial}{\partial\varepsilon_{i}}\left(  \sqrt{\sigma
(\gamma)}\sigma(\gamma+\varepsilon)\sqrt{\sigma(\gamma)}\right)  ^{-\frac
{1}{2}}\right\vert _{\varepsilon=0}\nonumber\\
&  =-\sqrt{\sqrt{\sigma(\gamma)}\sigma(\gamma)\sqrt{\sigma(\gamma)}}^{-1}
\int_{0}^{\infty}dt\ e^{-t\sqrt{\sqrt{\sigma(\gamma)}\sigma(\gamma
)\sqrt{\sigma(\gamma)}}}\sqrt{\sigma(\gamma)}\left(  \left.  \frac{\partial
}{\partial\varepsilon_{i}}\sigma(\gamma+\varepsilon)\right\vert _{\varepsilon
=0}\right)  \times\nonumber\\
&  \qquad\sqrt{\sigma(\gamma)}e^{-t\sqrt{\sqrt{\sigma(\gamma)}\sigma
(\gamma)\sqrt{\sigma(\gamma)}}}\sqrt{\sqrt{\sigma(\gamma)}\sigma(\gamma
)\sqrt{\sigma(\gamma)}}^{-1}\\
&  =-\sqrt{\sigma(\gamma)}^{-1}\int_{0}^{\infty}dt\ e^{-t\sigma(\gamma
)}\left(  \left.  \frac{\partial}{\partial\varepsilon_{i}}\sigma
(\gamma+\varepsilon)\right\vert _{\varepsilon=0}\right)  e^{-t\sigma(\gamma
)}\sqrt{\sigma(\gamma)}^{-1}\\
&  =-\sqrt{\sigma(\gamma)}^{-1}\int_{0}^{\infty}dt\ e^{-t\sigma(\gamma
)}\left(  \frac{\partial}{\partial\gamma_{i}}\sigma(\gamma)\right)
e^{-t\sigma(\gamma)}\sqrt{\sigma(\gamma)}^{-1}.\label{eq:FB-deriv-full-3}
\end{align}
Substituting~\eqref{eq:FB-deriv-full-3}\ into~\eqref{eq:FB-deriv-full-4}, we
find that
\begin{align}
&  \left.  \frac{\partial^{2}}{\partial\varepsilon_{i}\partial\varepsilon_{j}
}\left[  -2\ln F(\sigma(\gamma),\sigma(\gamma+\varepsilon))\right]
\right\vert _{\varepsilon=0}\nonumber\\
&  =-2\operatorname{Tr}\!\left[  \left(  -\sqrt{\sigma(\gamma)}^{-1}\int
_{0}^{\infty}dt\ e^{-t\sigma(\gamma)}\left(  \frac{\partial}{\partial
\gamma_{i}}\sigma(\gamma)\right)  e^{-t\sigma(\gamma)}\sqrt{\sigma(\gamma
)}^{-1}\right)  \sqrt{\sigma(\gamma)}\left(  \left.  \frac{\partial}
{\partial\varepsilon_{j}}\sigma(\gamma+\varepsilon)\right\vert _{\varepsilon
=0}\right)  \sqrt{\sigma(\gamma)}\right]  \\
&  =2\int_{0}^{\infty}dt\ \operatorname{Tr}\!\left[  e^{-t\sigma(\gamma
)}\left(  \frac{\partial}{\partial\gamma_{i}}\sigma(\gamma)\right)
e^{-t\sigma(\gamma)}\left(  \frac{\partial}{\partial\gamma_{j}}\sigma
(\gamma)\right)  \right]  ,
\end{align}
thus establishing~\eqref{eq:FB-formula-1}.

Now let us consider substituting in the eigenbasis of $\sigma(\gamma)$, taken
as $\sum_{k}\lambda_{k}|k\rangle\!\langle k|$. Then we find that 
\begin{align}
&  2\int_{0}^{\infty}dt\ \operatorname{Tr}\!\left[  e^{-t\sigma(\gamma
)}\left(  \frac{\partial}{\partial\gamma_{i}}\sigma(\gamma)\right)
e^{-t\sigma(\gamma)}\left(  \frac{\partial}{\partial\gamma_{j}}\sigma
(\gamma)\right)  \right] \nonumber\\
&  =2\int_{0}^{\infty}dt\ \operatorname{Tr}\!\left[  \left(  \sum
_{k}e^{-t\lambda_{k}}|k\rangle\langle k|\right)  \left(  \frac{\partial
}{\partial\gamma_{i}}\sigma(\gamma)\right)  \left(  \sum_{\ell}e^{-t\lambda
_{\ell}}|\ell\rangle\langle\ell|\right)  \left(  \frac{\partial}
{\partial\gamma_{j}}\sigma(\gamma)\right)  \right] \\
&  =2\sum_{k,\ell}\left(  \int_{0}^{\infty}dt\ e^{-t\left(  \lambda
_{k}+\lambda_{\ell}\right)  }\right)  \operatorname{Tr}\!\left[
|k\rangle\langle k|\left(  \frac{\partial}{\partial\gamma_{i}}\sigma
(\gamma)\right)  |\ell\rangle\langle\ell|\left(  \frac{\partial}
{\partial\gamma_{j}}\sigma(\gamma)\right)  \right] \\
&  =\sum_{k,\ell}\frac{2}{\lambda_{k}+\lambda_{\ell}}\langle k|\left(
\partial_{i}\sigma(\gamma)\right)  |\ell\rangle\langle\ell|\left(
\partial_{j}\sigma(\gamma)\right)  |k\rangle,
\end{align}
thus establishing~\eqref{eq:FB-formula-2}.

\subsection{Proof of Equations~\eqref{eq:WY-formula-1} and~\eqref{eq:WY-formula-2}}

Consider that
\begin{align}
  \frac{\partial^{2}}{\partial\varepsilon_{i}\partial\varepsilon_{j}}\left[
-2\ln F_{H}(\sigma(\gamma),\sigma(\gamma+\varepsilon))\right]  
&  =-2\frac{\partial^{2}}{\partial\varepsilon_{i}\partial\varepsilon_{j}}
\ln\operatorname{Tr}\!\left[  \sqrt{\sigma(\gamma)}\sqrt{\sigma(\gamma
+\varepsilon)}\right]  ^{2}\\
&  =-4\frac{\partial^{2}}{\partial\varepsilon_{i}\partial\varepsilon_{j}}
\ln\operatorname{Tr}\!\left[  \sqrt{\sigma(\gamma)}\sqrt{\sigma(\gamma
+\varepsilon)}\right]  \\
&  =-4\frac{\partial}{\partial\varepsilon_{i}}\left(  \frac{\partial}
{\partial\varepsilon_{j}}\ln\operatorname{Tr}\!\left[  \sqrt{\sigma(\gamma
)}\sqrt{\sigma(\gamma+\varepsilon)}\right]  \right)  \\
&  =-4\frac{\partial}{\partial\varepsilon_{i}}\left(  \frac{\frac{\partial
}{\partial\varepsilon_{j}}\operatorname{Tr}\!\left[  \sqrt{\sigma(\gamma)}
\sqrt{\sigma(\gamma+\varepsilon)}\right]  }{\operatorname{Tr}\!\left[
\sqrt{\sigma(\gamma)}\sqrt{\sigma(\gamma+\varepsilon)}\right]  }\right)  \\
&  =-4\frac{\partial}{\partial\varepsilon_{i}}\left(  \frac{\operatorname{Tr}\!
\left[  \sqrt{\sigma(\gamma)}\frac{\partial}{\partial\varepsilon_{j}}\sqrt
{\sigma(\gamma+\varepsilon)}\right]  }{\operatorname{Tr}\!\left[  \sqrt
{\sigma(\gamma)}\sqrt{\sigma(\gamma+\varepsilon)}\right]  }\right)  \\
&  =\frac{4\operatorname{Tr}\!\left[  \sqrt{\sigma(\gamma)}\frac{\partial
}{\partial\varepsilon_{i}}\sqrt{\sigma(\gamma+\varepsilon)}\right]
\operatorname{Tr}\!\left[  \sqrt{\sigma(\gamma)}\frac{\partial}{\partial
\varepsilon_{j}}\sqrt{\sigma(\gamma+\varepsilon)}\right]  }{\left(
\operatorname{Tr}\!\left[  \sqrt{\sigma(\gamma)}\sqrt{\sigma(\gamma+\varepsilon
)}\right]  \right)  ^{2}}\nonumber\\
&  \qquad-\frac{4\operatorname{Tr}\!\left[  \sqrt{\sigma(\gamma)}\frac
{\partial^{2}}{\partial\varepsilon_{i}\partial\varepsilon_{j}}\sqrt
{\sigma(\gamma+\varepsilon)}\right]  }{\operatorname{Tr}\!\left[  \sqrt
{\sigma(\gamma)}\sqrt{\sigma(\gamma+\varepsilon)}\right]  }.
\end{align}
Then we find that
\begin{align}
  \left.  \frac{\partial^{2}}{\partial\varepsilon_{i}\partial\varepsilon_{j}
}\left(  -2\ln F_{H}(\sigma(\gamma),\sigma(\gamma+\varepsilon))\right)
\right\vert _{\varepsilon=0}
&  =\frac{4\operatorname{Tr}\!\left[  \sqrt{\sigma(\gamma)}\left.
\frac{\partial}{\partial\varepsilon_{i}}\sqrt{\sigma(\gamma+\varepsilon
)}\right\vert _{\varepsilon=0}\right]  \operatorname{Tr}\!\left[  \sqrt
{\sigma(\gamma)}\left.  \frac{\partial}{\partial\varepsilon_{j}}\sqrt
{\sigma(\gamma+\varepsilon)}\right\vert _{\varepsilon=0}\right]  }{\left(
\operatorname{Tr}\!\left[  \sqrt{\sigma(\gamma)}\sqrt{\sigma(\gamma)}\right]
\right)  ^{2}}\nonumber\\
&  \qquad-\frac{4\operatorname{Tr}\!\left[  \sqrt{\sigma(\gamma)}\left.
\frac{\partial^{2}}{\partial\varepsilon_{i}\partial\varepsilon_{j}}\sqrt
{\sigma(\gamma+\varepsilon)}\right\vert _{\varepsilon=0}\right]  }
{\operatorname{Tr}\!\left[  \sqrt{\sigma(\gamma)}\sqrt{\sigma(\gamma)}\right]  }\\
&  =4\operatorname{Tr}\!\left[  \sqrt{\sigma(\gamma)}\left.  \frac{\partial
}{\partial\varepsilon_{i}}\sqrt{\sigma(\gamma+\varepsilon)}\right\vert
_{\varepsilon=0}\right]  \operatorname{Tr}\!\left[  \sqrt{\sigma(\gamma)}\left.
\frac{\partial}{\partial\varepsilon_{j}}\sqrt{\sigma(\gamma+\varepsilon
)}\right\vert _{\varepsilon=0}\right]  \nonumber\\
&  \qquad-4\operatorname{Tr}\!\left[  \sqrt{\sigma(\gamma)}\left.
\frac{\partial^{2}}{\partial\varepsilon_{i}\partial\varepsilon_{j}}\sqrt
{\sigma(\gamma+\varepsilon)}\right\vert _{\varepsilon=0}\right]  .
\end{align}
Recalling from \cite[Theorem~1.1]{DelMoral2018} that
\begin{equation}
\frac{\partial}{\partial\gamma_{j}}\sqrt{\sigma(\gamma)}=\int_{0}^{\infty
}dt\ e^{-t\sqrt{\sigma(\gamma)}}\left(  \frac{\partial}{\partial\gamma_{j}}
\sigma(\gamma)\right)  e^{-t\sqrt{\sigma(\gamma)}},
\end{equation}
now consider that
\begin{align}
  \operatorname{Tr}\!\left[  \sqrt{\sigma(\gamma)}\left.  \frac{\partial
}{\partial\varepsilon_{i}}\sqrt{\sigma(\gamma+\varepsilon)}\right\vert
_{\varepsilon=0}\right]  
&  =\operatorname{Tr}\!\left[  \sqrt{\sigma(\gamma)}\left.  \int_{0}^{\infty
}dt\ e^{-t\sqrt{\sigma(\gamma+\varepsilon)}}\left(  \frac{\partial}
{\partial\varepsilon_{i}}\sigma(\gamma+\varepsilon)\right)  e^{-t\sqrt
{\sigma(\gamma+\varepsilon)}}\right\vert _{\varepsilon=0}\right]  \\
&  =\int_{0}^{\infty}dt\ \operatorname{Tr}\!\left[  \sqrt{\sigma(\gamma
)}e^{-t\sqrt{\sigma(\gamma)}}\left(  \left.  \frac{\partial}{\partial
\varepsilon_{i}}\sigma(\gamma+\varepsilon)\right\vert _{\varepsilon=0}\right)
e^{-t\sqrt{\sigma(\gamma)}}\right]  \\
&  =\operatorname{Tr}\!\left[  \sqrt{\sigma(\gamma)}\int_{0}^{\infty
}dt\ e^{-2t\sqrt{\sigma(\gamma)}}\left(  \left.  \frac{\partial}{\partial
\varepsilon_{i}}\sigma(\gamma+\varepsilon)\right\vert _{\varepsilon=0}\right)
\right]  \\
&  =\operatorname{Tr}\!\left[  \sqrt{\sigma(\gamma)}\sqrt{\sigma(\gamma)}
^{-1}\left(  \left.  \frac{\partial}{\partial\varepsilon_{i}}\sigma
(\gamma+\varepsilon)\right\vert _{\varepsilon=0}\right)  \right]  \\
&  =\operatorname{Tr}\!\left[  \left.  \frac{\partial}{\partial\varepsilon
_{i}}\sigma(\gamma+\varepsilon)\right\vert _{\varepsilon=0}\right]  \\
&  =\left.  \operatorname{Tr}\!\left[  \frac{\partial}{\partial\varepsilon
_{i}}\sigma(\gamma+\varepsilon)\right]  \right\vert _{\varepsilon=0}\\
&  =0.
\end{align}
Thus, we conclude that
\begin{align}
\left.  \frac{\partial^{2}}{\partial\varepsilon_{i}\partial\varepsilon_{j}
}\left(  -2\ln F_{H}(\sigma(\gamma),\sigma(\gamma+\varepsilon))\right)
\right\vert _{\varepsilon=0}  & =-4\operatorname{Tr}\!\left[  \sqrt
{\sigma(\gamma)}\left.  \frac{\partial^{2}}{\partial\varepsilon_{i}
\partial\varepsilon_{j}}\sqrt{\sigma(\gamma+\varepsilon)}\right\vert
_{\varepsilon=0}\right]  \\
& =-4\operatorname{Tr}\!\left[  \sqrt{\sigma(\gamma)}\frac{\partial^{2}
}{\partial\gamma_{i}\partial\gamma_{j}}\sqrt{\sigma(\gamma)}\right]  ,
\end{align}
where we used that
\begin{equation}
\left.  \frac{\partial^{2}}{\partial\varepsilon_{i}\partial\varepsilon_{j}
}\sqrt{\sigma(\gamma+\varepsilon)}\right\vert _{\varepsilon=0}=\frac
{\partial^{2}}{\partial\gamma_{i}\partial\gamma_{j}}\sqrt{\sigma(\gamma)}.
\end{equation}
Now observe that
\begin{align}
0  & =\frac{\partial^{2}}{\partial\gamma_{i}\partial\gamma_{j}}
\operatorname{Tr}[\sigma(\gamma)]\\
& =\frac{\partial^{2}}{\partial\gamma_{i}\partial\gamma_{j}}\operatorname{Tr}\!
\left[  \sqrt{\sigma(\gamma)}\sqrt{\sigma(\gamma)}\right]  \\
& =\frac{\partial}{\partial\gamma_{i}}\operatorname{Tr}\!\left[  \frac{\partial
}{\partial\gamma_{j}}\left(  \sqrt{\sigma(\gamma)}\sqrt{\sigma(\gamma)}\right)
\right]  \\
& =\frac{\partial}{\partial\gamma_{i}}\operatorname{Tr}\!\left[  \left(
\frac{\partial}{\partial\gamma_{j}}\sqrt{\sigma(\gamma)}\right)  \sqrt
{\sigma(\gamma)}+\sqrt{\sigma(\gamma)}\left(  \frac{\partial}{\partial\gamma_{j}
}\sqrt{\sigma(\gamma)}\right)  \right]  \\
& =2\frac{\partial}{\partial\gamma_{i}}\operatorname{Tr}\!\left[  \left(
\frac{\partial}{\partial\gamma_{j}}\sqrt{\sigma(\gamma)}\right)  \sqrt
{\sigma(\gamma)}\right]  \\
& =2\operatorname{Tr}\!\left[  \left(  \frac{\partial^{2}}{\partial\gamma
_{i}\partial\gamma_{j}}\sqrt{\sigma(\gamma)}\right)  \sqrt{\sigma(\gamma)}+\left(
\frac{\partial}{\partial\gamma_{j}}\sqrt{\sigma(\gamma)}\right)  \left(
\frac{\partial}{\partial\gamma_{i}}\sqrt{\sigma(\gamma)}\right)  \right]  \\
& =2\operatorname{Tr}\!\left[  \left(  \frac{\partial^{2}}{\partial\gamma
_{i}\partial\gamma_{j}}\sqrt{\sigma(\gamma)}\right)  \sqrt{\sigma(\gamma)}\right]
+2\operatorname{Tr}\!\left[  \left(  \frac{\partial}{\partial\gamma_{i}}
\sqrt{\sigma(\gamma)}\right)  \left(  \frac{\partial}{\partial\gamma_{j}}
\sqrt{\sigma(\gamma)}\right)  \right]  .
\end{align}
So we conclude that
\begin{equation}
-\operatorname{Tr}\!\left[  \left(  \frac{\partial^{2}}{\partial\gamma
_{i}\partial\gamma_{j}}\sqrt{\sigma(\gamma)}\right)  \sqrt{\sigma(\gamma)}\right]
=\operatorname{Tr}\!\left[  \left(  \frac{\partial}{\partial\gamma_{i}}
\sqrt{\sigma(\gamma)}\right)  \left(  \frac{\partial}{\partial\gamma_{j}}
\sqrt{\sigma(\gamma)}\right)  \right]  ,
\end{equation}
which in turn implies that
\begin{align}
& \left.  \frac{\partial^{2}}{\partial\varepsilon_{i}\partial\varepsilon_{j}
}\left(  -2\ln F_{H}(\sigma(\gamma),\sigma(\gamma+\varepsilon))\right)
\right\vert _{\varepsilon=0}\nonumber\\
& =4\operatorname{Tr}\!\left[  \left(  \frac{\partial}{\partial\gamma_{i}}
\sqrt{\sigma(\gamma)}\right)  \left(  \frac{\partial}{\partial\gamma_{j}}
\sqrt{\sigma(\gamma)}\right)  \right]  \\
& =4\operatorname{Tr}\!\left[  \left(  \int_{0}^{\infty}dt_{1}\ e^{-t_{1}
\sqrt{\sigma(\gamma)}}\left(  \frac{\partial}{\partial\gamma_{i}}\sigma
(\gamma)\right)  e^{-t_{1}\sqrt{\sigma(\gamma)}}\right)  \left(  \int
_{0}^{\infty}dt_{2}\ e^{-t_{2}\sqrt{\sigma(\gamma)}}\left(  \frac{\partial
}{\partial\gamma_{j}}\sigma(\gamma)\right)  e^{-t_{2}\sqrt{\sigma(\gamma)}
}\right)  \right]  \\
& =4\int_{0}^{\infty}dt_{1}\ \int_{0}^{\infty}dt_{2}\ \operatorname{Tr}\!\left[
e^{-t_{1}\sqrt{\sigma(\gamma)}}\left(  \frac{\partial}{\partial\gamma_{i}}
\sigma(\gamma)\right)  e^{-t_{1}\sqrt{\sigma(\gamma)}}e^{-t_{2}\sqrt{\sigma(\gamma
)}}\left(  \frac{\partial}{\partial\gamma_{j}}\sigma(\gamma)\right)
e^{-t_{2}\sqrt{\sigma(\gamma)}}\right]  \\
& =4\int_{0}^{\infty}dt_{1}\ \int_{0}^{\infty}dt_{2}\ \operatorname{Tr}\!\left[
e^{-\left(  t_{1}+t_{2}\right)  \sqrt{\sigma(\gamma)}}\left(  \frac{\partial
}{\partial\gamma_{i}}\sigma(\gamma)\right)  e^{-\left(  t_{1}+t_{2}\right)
\sqrt{\sigma(\gamma)}}\left(  \frac{\partial}{\partial\gamma_{j}}\sigma
(\gamma)\right)  \right]  ,
\end{align}
thus establishing~\eqref{eq:WY-formula-1}.

We can derive~\eqref{eq:WY-formula-2} by considering a spectral
decomposition for $\sigma(\gamma)$ as
\begin{equation}
\sigma(\gamma)=\sum_{k}\lambda_{k}|k\rangle\!\langle k|,
\end{equation}
where we have suppressed the dependence of the eigenvalues and eigenvectors on
the parameter vector $\gamma$. Then we find that
\begin{align}
&  4\int_{0}^{\infty}\int_{0}^{\infty}dt_{1}\ dt_{2}\ \operatorname{Tr}
\!\left[  e^{-\left(  t_{1}+t_{2}\right)  \sqrt{\sigma(\gamma)}}\left(
\partial_{i}\sigma(\gamma)\right)  e^{-\left(  t_{1}+t_{2}\right)  \sqrt
{\sigma(\gamma)}}\left(  \partial_{j}\sigma(\gamma)\right)  \right]  \notag \\
&  =4\int_{0}^{\infty}\int_{0}^{\infty}dt_{1}\ dt_{2}\ \operatorname{Tr}
\!\left[  \sum_{k}e^{-\left(  t_{1}+t_{2}\right)  \sqrt{\lambda_{k}}}
|k\rangle\!\langle k|\left(  \partial_{i}\sigma(\gamma)\right)  \sum_{\ell
}e^{-\left(  t_{1}+t_{2}\right)  \sqrt{\lambda_{\ell}}}|\ell\rangle
\!\langle\ell|\left(  \partial_{j}\sigma(\gamma)\right)  \right]  \\
&  =4\sum_{k,\ell}\int_{0}^{\infty}\int_{0}^{\infty}dt_{1}\ dt_{2}\ e^{-\left(
t_{1}+t_{2}\right)  \left(  \sqrt{\lambda_{k}}+\sqrt{\lambda_{\ell}}\right)
}\langle k|\left(  \partial_{i}\sigma(\gamma)\right)  |\ell\rangle\!\langle
\ell|\left(  \partial_{j}\sigma(\gamma)\right)  |k\rangle\\
&  =\sum_{k,\ell}\frac{4}{\left(  \sqrt{\lambda_{k}}+\sqrt{\lambda_{\ell}
}\right)  ^{2}}\langle k|\left(  \partial_{i}\sigma(\gamma)\right)  |\ell
\rangle\!\langle\ell|\left(  \partial_{j}\sigma(\gamma)\right)  |k\rangle,
\end{align}
thus establishing~\eqref{eq:WY-formula-2}.
In the last line, we made use of the integral
\begin{equation}
\int_{0}^{\infty}\int_{0}^{\infty}dt_{1}\ dt_{2}\ e^{-\left(  t_{1}
+t_{2}\right)  x}=\frac{1}{x^{2}},
\end{equation}
holding for $x>0$.

\subsection{Proof of Equations~\eqref{eq:KM-formula-1} and~\eqref{eq:KM-formula-2}}

Consider that
\begin{align}
\frac{\partial^{2}}{\partial\varepsilon_{i}\partial\varepsilon_{j}}\left[
D(\sigma(\gamma)\Vert\sigma(\gamma+\varepsilon))\right]   &  =\frac
{\partial^{2}}{\partial\varepsilon_{i}\partial\varepsilon_{j}}\left(
\operatorname{Tr}\!\left[  \sigma(\gamma)\left(  \ln\sigma(\gamma)-\ln
\sigma(\gamma+\varepsilon)\right)  \right]  \right)  \\
&  =-\frac{\partial^{2}}{\partial\varepsilon_{i}\partial\varepsilon_{j}
}\operatorname{Tr}\!\left[  \sigma(\gamma)\ln\sigma(\gamma+\varepsilon
)\right]  \\
&  =-\operatorname{Tr}\!\left[  \sigma(\gamma)\left(  \frac{\partial^{2}
}{\partial\varepsilon_{i}\partial\varepsilon_{j}}\ln\sigma(\gamma
+\varepsilon)\right)  \right]  ,
\end{align}
which implies that
\begin{align}
\left.  \frac{\partial^{2}}{\partial\varepsilon_{i}\partial\varepsilon_{j}
}\left[  D(\sigma(\gamma)\Vert\sigma(\gamma+\varepsilon))\right]  \right\vert
_{\varepsilon=0}  & =-\operatorname{Tr}\!\left[  \sigma(\gamma)\left(  \left.
\frac{\partial^{2}}{\partial\varepsilon_{i}\partial\varepsilon_{j}}\ln
\sigma(\gamma+\varepsilon)\right\vert _{\varepsilon=0}\right)  \right]  \\
& =-\operatorname{Tr}\!\left[  \sigma(\gamma)\left(  \frac{\partial^{2}
}{\partial\gamma_{i}\partial\gamma_{j}}\ln\sigma(\gamma)\right)  \right]
.\label{eq:KM-proof-deriv-full-2}
\end{align}
Observe that the logarithm has the following integral representation for
$x>0$:
\begin{equation}
\ln x=\int_{0}^{\infty}dt\ \left(  1+t\right)  ^{-1}-\left(  x+t\right)
^{-1},
\end{equation}
which implies the following integral representation for the matrix logarithm
of a positive definite operator $\sigma$:
\begin{equation}
\ln\sigma=\int_{0}^{\infty}dt\ \left(  1+t\right)  ^{-1}I-\left(
\sigma+tI\right)  ^{-1}.\label{eq:KM-proof-deriv-full-5}
\end{equation}
It then follows from~\eqref{eq:KM-proof-deriv-full-5}\ and the following
\begin{equation}
\frac{\partial}{\partial\gamma_{j}}\left(  \sigma(\gamma)\right)
^{-1}=-\sigma(\gamma)^{-1}\left(  \frac{\partial}{\partial\gamma_{j}}
\sigma(\gamma)\right)  \sigma(\gamma)^{-1}
\end{equation}
that the derivative of the matrix logarithm is as follows:
\begin{equation}
\frac{\partial}{\partial\gamma_{j}}\ln\sigma(\gamma)=\int_{0}^{\infty
}dt\ \left(  \sigma(\gamma)+tI\right)  ^{-1}\left(  \frac{\partial}
{\partial\gamma_{j}}\sigma(\gamma)\right)  \left(  \sigma(\gamma)+tI\right)
^{-1}.\label{eq:KM-deriv-full-3}
\end{equation}
Then, by making use of the following integral for $x>0$,
\begin{equation}
\int_{0}^{\infty}dt\ \left(  x+t\right)  ^{-2}=\frac{1}{x},
\end{equation}
consider that
\begin{align}
  \operatorname{Tr}\!\left[  \sigma(\gamma)\left(  \frac{\partial}
{\partial\gamma_{j}}\ln\sigma(\gamma)\right)  \right]  
&  =\operatorname{Tr}\!\left[  \sigma(\gamma)\left(  \int_{0}^{\infty
}dt\ \left(  \sigma(\gamma)+tI\right)  ^{-1}\left(  \frac{\partial}
{\partial\gamma_{j}}\sigma(\gamma)\right)  \left(  \sigma(\gamma)+tI\right)
^{-1}\right)  \right]  \\
&  =\int_{0}^{\infty}dt\ \operatorname{Tr}\!\left[  \sigma(\gamma)\left(
\sigma(\gamma)+tI\right)  ^{-2}\left(  \frac{\partial}{\partial\gamma_{j}
}\sigma(\gamma)\right)  \right]  \\
&  =\operatorname{Tr}\!\left[  \sigma(\gamma)\left(  \int_{0}^{\infty
}dt\ \left(  \sigma(\gamma)+tI\right)  ^{-2}\right)  \left(  \frac{\partial
}{\partial\gamma_{j}}\sigma(\gamma)\right)  \right]  \\
&  =\operatorname{Tr}\!\left[  \sigma(\gamma)\sigma(\gamma)^{-1}\left(
\frac{\partial}{\partial\gamma_{j}}\sigma(\gamma)\right)  \right]  \\
&  =\operatorname{Tr}\!\left[  \frac{\partial}{\partial\gamma_{j}}
\sigma(\gamma)\right]  \\
&  =\frac{\partial}{\partial\gamma_{j}}\operatorname{Tr}\!\left[
\sigma(\gamma)\right]  \\
&  =0.
\end{align}
It then follows that
\begin{align}
0  & =\frac{\partial}{\partial\gamma_{i}}\operatorname{Tr}\!\left[
\sigma(\gamma)\left(  \frac{\partial}{\partial\gamma_{j}}\ln\sigma
(\gamma)\right)  \right]  \\
& =\operatorname{Tr}\!\left[  \left(  \frac{\partial}{\partial\gamma_{i}
}\sigma(\gamma)\right)  \left(  \frac{\partial}{\partial\gamma_{j}}\ln
\sigma(\gamma)\right)  \right]  +\operatorname{Tr}\!\left[  \sigma
(\gamma)\left(  \frac{\partial^{2}}{\partial\gamma_{i}\partial\gamma_{j}}
\ln\sigma(\gamma)\right)  \right]  ,
\end{align}
which is equivalent to
\begin{equation}
-\operatorname{Tr}\!\left[  \sigma(\gamma)\left(  \frac{\partial^{2}}
{\partial\gamma_{i}\partial\gamma_{j}}\ln\sigma(\gamma)\right)  \right]
=\operatorname{Tr}\!\left[  \left(  \frac{\partial}{\partial\gamma_{i}}
\sigma(\gamma)\right)  \left(  \frac{\partial}{\partial\gamma_{j}}\ln
\sigma(\gamma)\right)  \right]  .\label{eq:KM-proof-deriv-full-1}
\end{equation}
Substituting~\eqref{eq:KM-proof-deriv-full-1} into
\eqref{eq:KM-proof-deriv-full-2} and again making use of
\eqref{eq:KM-deriv-full-3}, we find that
\begin{align}
  \left.  \frac{\partial^{2}}{\partial\varepsilon_{i}\partial\varepsilon_{j}
}\left[  D(\sigma(\gamma)\Vert\sigma(\gamma+\varepsilon))\right]  \right\vert
_{\varepsilon=0}
&  =\operatorname{Tr}\!\left[  \left(  \frac{\partial}{\partial\gamma_{i}
}\sigma(\gamma)\right)  \left(  \frac{\partial}{\partial\gamma_{j}}\ln
\sigma(\gamma)\right)  \right]  \\
&  =\operatorname{Tr}\!\left[  \left(  \frac{\partial}{\partial\gamma_{i}
}\sigma(\gamma)\right)  \left(  \int_{0}^{\infty}dt\ \left(  \sigma
(\gamma)+tI\right)  ^{-1}\left(  \frac{\partial}{\partial\gamma_{j}}
\sigma(\gamma)\right)  \left(  \sigma(\gamma)+tI\right)  ^{-1}\right)
\right]  \\
&  =\int_{0}^{\infty}dt\ \operatorname{Tr}\!\left[  \frac{\partial}
{\partial\gamma_{i}}\sigma(\gamma)\left(  \sigma(\gamma)+tI\right)
^{-1}\left(  \frac{\partial}{\partial\gamma_{j}}\sigma(\gamma)\right)  \left(
\sigma(\gamma)+tI\right)  ^{-1}\right]  ,
\end{align}
which establishes~\eqref{eq:KM-formula-1}.

Finally, by making use of the following integral for $x,y>0$:
\begin{equation}
\int_{0}^{\infty}dt\ \frac{1}{\left(  x+t\right)  \left(  y+t\right)  }
=\frac{\ln x-\ln y}{x-y},
\end{equation}
consider that
\begin{align}
&  \int_{0}^{\infty}dt\ \operatorname{Tr}\!\left[  \left(  \sigma
(\gamma)+tI\right)  ^{-1}\left(  \partial_{i}\sigma(\gamma)\right)  \left(
\sigma(\gamma)+tI\right)  ^{-1}\left(  \partial_{j}\sigma(\gamma)\right)
\right]  \nonumber\\
&  =\int_{0}^{\infty}dt\ \operatorname{Tr}\!\left[  \left(  \sum_{k}\frac
{1}{\lambda_{k}+t}|k\rangle\langle k|\right)  \left(  \frac{\partial}
{\partial\gamma_{i}}\sigma(\gamma)\right)  \left(  \sum_{\ell}\frac{1}
{\lambda_{\ell}+t}|\ell\rangle\langle\ell|\right)  \left(  \frac{\partial
}{\partial\gamma_{j}}\sigma(\gamma)\right)  \right]  \\
&  =\sum_{k,\ell}\left(  \int_{0}^{\infty}dt\ \frac{1}{\left(  \lambda
_{k}+t\right)  \left(  \lambda_{\ell}+t\right)  }\right)  \operatorname{Tr}
\!\left[  |k\rangle\langle k|\left(  \frac{\partial}{\partial\gamma_{i}}
\sigma(\gamma)\right)  |\ell\rangle\langle\ell|\left(  \frac{\partial
}{\partial\gamma_{j}}\sigma(\gamma)\right)  \right]  \\
&  =\sum_{k,\ell}\left[  \frac{\ln\lambda_{k}-\ln\lambda_{\ell}}{\lambda
_{k}-\lambda_{\ell}}\right]  \langle k|\left(  \frac{\partial}{\partial
\gamma_{i}}\sigma(\gamma)\right)  |\ell\rangle\langle\ell|\left(
\frac{\partial}{\partial\gamma_{j}}\sigma(\gamma)\right)  |k\rangle,
\end{align}
thus establishing~\eqref{eq:KM-formula-2}.

\section{Proof of Equation~\eqref{eq:formula-FB-pure}}

\label{app:proof-formula-FB-pure}

To see~\eqref{eq:formula-FB-pure}, consider that
\begin{align}
&  2\frac{\partial^{2}}{\partial\varepsilon_{i}\partial\varepsilon_{j}}\left[
-\ln\left\vert \langle\psi(\gamma)|\psi(\gamma+\varepsilon)\rangle\right\vert
^{2}\right] \nonumber\\
&  =-2\frac{\partial}{\partial\varepsilon_{i}}\left[  \frac{\partial}
{\partial\varepsilon_{j}}\ln\langle\psi(\gamma)|\psi(\gamma+\varepsilon
)\rangle\langle\psi(\gamma+\varepsilon)|\psi(\gamma)\rangle\right] \\
&  =-2\frac{\partial}{\partial\varepsilon_{i}}\left[  \frac{\langle\psi
(\gamma)|\frac{\partial}{\partial\varepsilon_{j}}\psi(\gamma+\varepsilon
)\rangle\langle\psi(\gamma+\varepsilon)|\psi(\gamma)\rangle+\langle\psi
(\gamma)|\psi(\gamma+\varepsilon)\rangle\langle\frac{\partial}{\partial
\varepsilon_{j}}\psi(\gamma+\varepsilon)|\psi(\gamma)\rangle}{\langle
\psi(\gamma)|\psi(\gamma+\varepsilon)\rangle\langle\psi(\gamma+\varepsilon
)|\psi(\gamma)\rangle}\right] \\
&  =\frac{2\left[  \frac{\partial}{\partial\varepsilon_{i}}\left(  \langle
\psi(\gamma)|\psi(\gamma+\varepsilon)\rangle\langle\psi(\gamma+\varepsilon
)|\psi(\gamma)\rangle\right)  \right]  \left[
\begin{array}
[c]{c}
\langle\psi(\gamma)|\frac{\partial}{\partial\varepsilon_{j}}\psi
(\gamma+\varepsilon)\rangle\langle\psi(\gamma+\varepsilon)|\psi(\gamma
)\rangle\\
+\langle\psi(\gamma)|\psi(\gamma+\varepsilon)\rangle\langle\frac{\partial
}{\partial\varepsilon_{j}}\psi(\gamma+\varepsilon)|\psi(\gamma)\rangle
\end{array}
\right]  }{\langle\psi(\gamma)|\psi(\gamma+\varepsilon)\rangle\langle
\psi(\gamma+\varepsilon)|\psi(\gamma)\rangle^{2}} \notag \\
&  \qquad-2\frac{\frac{\partial}{\partial\varepsilon_{i}}\left[  \langle
\psi(\gamma)|\frac{\partial}{\partial\varepsilon_{j}}\psi(\gamma
+\varepsilon)\rangle\langle\psi(\gamma+\varepsilon)|\psi(\gamma)\rangle
+\langle\psi(\gamma)|\psi(\gamma+\varepsilon)\rangle\langle\frac{\partial
}{\partial\varepsilon_{j}}\psi(\gamma+\varepsilon)|\psi(\gamma)\rangle\right]
}{\langle\psi(\gamma)|\psi(\gamma+\varepsilon)\rangle\langle\psi
(\gamma+\varepsilon)|\psi(\gamma)\rangle}.
\end{align}
It then follows that
\begin{align}
&  \left.  2\frac{\partial^{2}}{\partial\varepsilon_{i}\partial\varepsilon
_{j}}\left[  -\ln\left\vert \langle\psi(\gamma)|\psi(\gamma+\varepsilon
)\rangle\right\vert ^{2}\right]  \right\vert _{\varepsilon=0}\nonumber\\
&  =\frac{2\left[  \left.  \frac{\partial}{\partial\varepsilon_{i}}\left(
\langle\psi(\gamma)|\psi(\gamma+\varepsilon)\rangle\langle\psi(\gamma
+\varepsilon)|\psi(\gamma)\rangle\right)  \right\vert _{\varepsilon=0}\right]
\left[
\begin{array}
[c]{c}
\langle\psi(\gamma)|\left.  \frac{\partial}{\partial\varepsilon_{j}}
\psi(\gamma+\varepsilon)\right\vert _{\varepsilon=0}\rangle\langle\psi
(\gamma)|\psi(\gamma)\rangle\\
+\langle\psi(\gamma)|\psi(\gamma)\rangle\langle\left.  \frac{\partial
}{\partial\varepsilon_{j}}\psi(\gamma+\varepsilon)\right\vert _{\varepsilon
=0}|\psi(\gamma)\rangle
\end{array}
\right]  }{\langle\psi(\gamma)|\psi(\gamma)\rangle\langle\psi(\gamma
)|\psi(\gamma)\rangle^{2}}\nonumber\\
&  \qquad-2\frac{\left.  \frac{\partial}{\partial\varepsilon_{i}}\left[
\langle\psi(\gamma)|\frac{\partial}{\partial\varepsilon_{j}}\psi
(\gamma+\varepsilon)\rangle\langle\psi(\gamma+\varepsilon)|\psi(\gamma
)\rangle+\langle\psi(\gamma)|\psi(\gamma+\varepsilon)\rangle\langle
\frac{\partial}{\partial\varepsilon_{j}}\psi(\gamma+\varepsilon)|\psi
(\gamma)\rangle\right]  \right\vert _{\varepsilon=0}}{\langle\psi(\gamma
)|\psi(\gamma)\rangle\langle\psi(\gamma)|\psi(\gamma)\rangle}\\
&  =2\left[  \left.  \langle\psi(\gamma)|\frac{\partial}{\partial
\varepsilon_{i}}\psi(\gamma+\varepsilon)\rangle\langle\psi(\gamma
+\varepsilon)|\psi(\gamma)\rangle+\langle\psi(\gamma)|\psi(\gamma
+\varepsilon)\rangle\langle\frac{\partial}{\partial\varepsilon_{i}}\psi
(\gamma+\varepsilon)|\psi(\gamma)\rangle\right\vert _{\varepsilon=0}\right]
\times\nonumber\\
&  \qquad\left[  \langle\psi(\gamma)|\frac{\partial}{\partial\gamma_{j}}
\psi(\gamma)\rangle+\langle\frac{\partial}{\partial\gamma_{j}}\psi
(\gamma)|\psi(\gamma)\rangle\right] \nonumber\\
&  \qquad-2\left.  \frac{\partial}{\partial\varepsilon_{i}}\left[  \langle
\psi(\gamma)|\frac{\partial}{\partial\varepsilon_{j}}\psi(\gamma
+\varepsilon)\rangle\langle\psi(\gamma+\varepsilon)|\psi(\gamma)\rangle
+\langle\psi(\gamma)|\psi(\gamma+\varepsilon)\rangle\langle\frac{\partial
}{\partial\varepsilon_{j}}\psi(\gamma+\varepsilon)|\psi(\gamma)\rangle\right]
\right\vert _{\varepsilon=0}\\
&  =2\left[  \langle\psi(\gamma)|\frac{\partial}{\partial\gamma_{i}}
\psi(\gamma)\rangle+\langle\frac{\partial}{\partial\gamma_{i}}\psi
(\gamma)|\psi(\gamma)\rangle\right]  \left[  \langle\psi(\gamma)|\frac
{\partial}{\partial\gamma_{j}}\psi(\gamma)\rangle+\langle\frac{\partial
}{\partial\gamma_{j}}\psi(\gamma)|\psi(\gamma)\rangle\right] \nonumber\\
&  \qquad-2\left.  \frac{\partial}{\partial\varepsilon_{i}}\left[  \langle
\psi(\gamma)|\frac{\partial}{\partial\varepsilon_{j}}\psi(\gamma
+\varepsilon)\rangle\langle\psi(\gamma+\varepsilon)|\psi(\gamma)\rangle
+\langle\psi(\gamma)|\psi(\gamma+\varepsilon)\rangle\langle\frac{\partial
}{\partial\varepsilon_{j}}\psi(\gamma+\varepsilon)|\psi(\gamma)\rangle\right]
\right\vert
\end{align}
Then consider that
\begin{align}
0  &  =\frac{\partial}{\partial\gamma_{i}}\left[  \langle\psi(\gamma
)|\psi(\gamma)\rangle\right] \\
&  =\langle\psi(\gamma)|\frac{\partial}{\partial\gamma_{i}}\psi(\gamma
)\rangle+\langle\frac{\partial}{\partial\gamma_{i}}\psi(\gamma)|\psi
(\gamma)\rangle.
\end{align}
So this implies that
\begin{multline}
\left.  2\frac{\partial^{2}}{\partial\varepsilon_{i}\partial\varepsilon_{j}
}\left[  -\ln\left\vert \langle\psi(\gamma)|\psi(\gamma+\varepsilon
)\rangle\right\vert ^{2}\right]  \right\vert _{\varepsilon=0}=\\
-2\left.  \frac{\partial}{\partial\varepsilon_{i}}\left[  \langle\psi
(\gamma)|\frac{\partial}{\partial\varepsilon_{j}}\psi(\gamma+\varepsilon
)\rangle\langle\psi(\gamma+\varepsilon)|\psi(\gamma)\rangle+\langle\psi
(\gamma)|\psi(\gamma+\varepsilon)\rangle\langle\frac{\partial}{\partial
\varepsilon_{j}}\psi(\gamma+\varepsilon)|\psi(\gamma)\rangle\right]
\right\vert _{\varepsilon=0}.
\end{multline}
So consider that
\begin{multline}
\frac{\partial}{\partial\varepsilon_{i}}\left[  \langle\psi(\gamma
)|\frac{\partial}{\partial\varepsilon_{j}}\psi(\gamma+\varepsilon
)\rangle\langle\psi(\gamma+\varepsilon)|\psi(\gamma)\rangle+\langle\psi
(\gamma)|\psi(\gamma+\varepsilon)\rangle\langle\frac{\partial}{\partial
\varepsilon_{j}}\psi(\gamma+\varepsilon)|\psi(\gamma)\rangle\right] \\
=\langle\psi(\gamma)|\frac{\partial^{2}}{\partial\varepsilon_{i}
\partial\varepsilon_{j}}\psi(\gamma+\varepsilon)\rangle\langle\psi
(\gamma+\varepsilon)|\psi(\gamma)\rangle+\langle\psi(\gamma)|\frac{\partial
}{\partial\varepsilon_{j}}\psi(\gamma+\varepsilon)\rangle\langle\frac
{\partial}{\partial\varepsilon_{i}}\psi(\gamma+\varepsilon)|\psi
(\gamma)\rangle\\
+\langle\psi(\gamma)|\frac{\partial}{\partial\varepsilon_{i}}\psi
(\gamma+\varepsilon)\rangle\langle\frac{\partial}{\partial\varepsilon_{j}}
\psi(\gamma+\varepsilon)|\psi(\gamma)\rangle+\langle\psi(\gamma)|\psi
(\gamma+\varepsilon)\rangle\langle\frac{\partial^{2}}{\partial\varepsilon
_{i}\partial\varepsilon_{j}}\psi(\gamma+\varepsilon)|\psi(\gamma)\rangle.
\end{multline}
Then
\begin{align}
&  -2\left.  \frac{\partial}{\partial\varepsilon_{i}}\left[  \langle
\psi(\gamma)|\frac{\partial}{\partial\varepsilon_{j}}\psi(\gamma
+\varepsilon)\rangle\langle\psi(\gamma+\varepsilon)|\psi(\gamma)\rangle
+\langle\psi(\gamma)|\psi(\gamma+\varepsilon)\rangle\langle\frac{\partial
}{\partial\varepsilon_{j}}\psi(\gamma+\varepsilon)|\psi(\gamma)\rangle\right]
\right\vert _{\varepsilon=0}\nonumber\\
&  =-2\left[
\begin{array}
[c]{c}
\langle\psi(\gamma)|\frac{\partial^{2}}{\partial\gamma_{i}\partial\gamma_{j}
}\psi(\gamma)\rangle\langle\psi(\gamma)|\psi(\gamma)\rangle+\langle\psi
(\gamma)|\frac{\partial}{\partial\gamma_{j}}\psi(\gamma)\rangle\langle
\frac{\partial}{\partial\gamma_{i}}\psi(\gamma)|\psi(\gamma)\rangle\\
+\langle\psi(\gamma)|\frac{\partial}{\partial\gamma_{i}}\psi(\gamma
)\rangle\langle\frac{\partial}{\partial\gamma_{j}}\psi(\gamma)|\psi
(\gamma)\rangle+\langle\psi(\gamma)|\psi(\gamma)\rangle\langle\frac
{\partial^{2}}{\partial\gamma_{i}\partial\gamma_{j}}\psi(\gamma)|\psi
(\gamma)\rangle
\end{array}
\right] \\
&  =-2\left[
\begin{array}
[c]{c}
\langle\psi(\gamma)|\frac{\partial^{2}}{\partial\gamma_{i}\partial\gamma_{j}
}\psi(\gamma)\rangle+\langle\psi(\gamma)|\frac{\partial}{\partial\gamma_{j}
}\psi(\gamma)\rangle\langle\frac{\partial}{\partial\gamma_{i}}\psi
(\gamma)|\psi(\gamma)\rangle\\
+\langle\psi(\gamma)|\frac{\partial}{\partial\gamma_{i}}\psi(\gamma
)\rangle\langle\frac{\partial}{\partial\gamma_{j}}\psi(\gamma)|\psi
(\gamma)\rangle+\langle\frac{\partial^{2}}{\partial\gamma_{i}\partial
\gamma_{j}}\psi(\gamma)|\psi(\gamma)\rangle
\end{array}
\right] \\
&  =-4\operatorname{Re}\!\left[  \langle\psi(\gamma)|\frac{\partial^{2}
}{\partial\gamma_{i}\partial\gamma_{j}}\psi(\gamma)\rangle+\langle\psi
(\gamma)|\frac{\partial}{\partial\gamma_{j}}\psi(\gamma)\rangle\langle
\frac{\partial}{\partial\gamma_{i}}\psi(\gamma)|\psi(\gamma)\rangle\right]  .
\label{eq:second-deriv-simplify-finalize}
\end{align}
Now consider that
\begin{align}
0  &  =\frac{\partial^{2}}{\partial\gamma_{i}\partial\gamma_{j}}\langle
\psi(\gamma)|\psi(\gamma)\rangle\\
&  =\frac{\partial}{\partial\gamma_{i}}\left[  \langle\psi(\gamma
)|\frac{\partial}{\partial\gamma_{j}}\psi(\gamma)\rangle+\langle\frac
{\partial}{\partial\gamma_{j}}\psi(\gamma)|\psi(\gamma)\rangle\right] \\
&  =\langle\frac{\partial}{\partial\gamma_{i}}\psi(\gamma)|\frac{\partial
}{\partial\gamma_{j}}\psi(\gamma)\rangle+\langle\psi(\gamma)|\frac
{\partial^{2}}{\partial\gamma_{i}\partial\gamma_{j}}\psi(\gamma)\rangle
\nonumber\\
&  \qquad+\langle\frac{\partial^{2}}{\partial\gamma_{i}\partial\gamma_{j}}
\psi(\gamma)|\psi(\gamma)\rangle+\langle\frac{\partial}{\partial\gamma_{j}
}\psi(\gamma)|\frac{\partial}{\partial\gamma_{i}}\psi(\gamma)\rangle\\
&  =2\operatorname{Re}\!\left[  \langle\frac{\partial}{\partial\gamma_{i}}
\psi(\gamma)|\frac{\partial}{\partial\gamma_{j}}\psi(\gamma)\rangle\right]
+2\operatorname{Re}\!\left[  \langle\psi(\gamma)|\frac{\partial^{2}}
{\partial\gamma_{i}\partial\gamma_{j}}\psi(\gamma)\rangle\right]  .
\end{align}
So this implies that
\begin{equation}
-\operatorname{Re}\!\left[  \langle\psi(\gamma)|\frac{\partial^{2}}
{\partial\gamma_{i}\partial\gamma_{j}}\psi(\gamma)\rangle\right]
=\operatorname{Re}\!\left[  \langle\frac{\partial}{\partial\gamma_{i}}
\psi(\gamma)|\frac{\partial}{\partial\gamma_{j}}\psi(\gamma)\rangle\right]  .
\label{eq:second-deriv-simplify}
\end{equation}
Substituting~\eqref{eq:second-deriv-simplify} into
\eqref{eq:second-deriv-simplify-finalize}, we conclude that
\begin{align}
&  \left.  2\frac{\partial^{2}}{\partial\varepsilon_{i}\partial\varepsilon
_{j}}\left[  -\ln\left\vert \langle\psi(\gamma)|\psi(\gamma+\varepsilon
)\rangle\right\vert ^{2}\right]  \right\vert _{\varepsilon=0}\nonumber\\
&  =4\operatorname{Re}\!\left[  \langle\frac{\partial}{\partial\gamma_{i}}
\psi(\gamma)|\frac{\partial}{\partial\gamma_{j}}\psi(\gamma)\rangle
-\langle\psi(\gamma)|\frac{\partial}{\partial\gamma_{j}}\psi(\gamma
)\rangle\langle\frac{\partial}{\partial\gamma_{i}}\psi(\gamma)|\psi
(\gamma)\rangle\right] \\
&  =4\operatorname{Re}\!\left[  \langle\frac{\partial}{\partial\gamma_{i}}
\psi(\gamma)|\frac{\partial}{\partial\gamma_{j}}\psi(\gamma)\rangle
-\langle\frac{\partial}{\partial\gamma_{i}}\psi(\gamma)|\psi(\gamma
)\rangle\langle\psi(\gamma)|\frac{\partial}{\partial\gamma_{j}}\psi
(\gamma)\rangle\right]  .
\end{align}

\section{Alternative proof of Proposition~\ref{prop:FB-WY-canonical-purifications}}

\label{app:proof-FB-WY-canonical-purifications}

To evaluate the expression in~\eqref{eq:formula-FB-pure} for the parameterized family $\left(  \varphi^{\sigma}(\gamma)\right)  _{\gamma\in\mathbb{R}^{L}}$, we need to compute $|\partial_{j}\varphi^{\sigma}
(\gamma)\rangle$. To this end, consider that
\begin{align}
|\partial_{j}\varphi^{\sigma}(\gamma)\rangle &  =\frac{\partial}
{\partial\gamma_{j}}|\varphi^{\sigma}(\gamma)\rangle\\
&  =\frac{\partial}{\partial\gamma_{j}}\left(  \sqrt{\sigma(\gamma)}\otimes
I\right)  |\Gamma\rangle\\
&  =\left(  \left(  \frac{\partial}{\partial\gamma_{j}}\sqrt{\sigma(\gamma
)}\right)  \otimes I\right)  |\Gamma\rangle\\
&  =\left(  \int_{0}^{\infty}dt\ e^{-t\sqrt{\sigma(\gamma)}}\left(
\frac{\partial}{\partial\gamma_{j}}\sigma(\gamma)\right)  e^{-t\sqrt
{\sigma(\gamma)}}\otimes I\right)  |\Gamma\rangle\\
&  =\left(  \int_{0}^{\infty}dt\ e^{-t\sqrt{\sigma(\gamma)}}\left(
\partial_{j}\sigma(\gamma)\right)  e^{-t\sqrt{\sigma(\gamma)}}\otimes
I\right)  |\Gamma\rangle,
\end{align}
where we made use of \cite[Theorem~1.1]{DelMoral2018} in the penultimate line.
Now consider that
\begin{align}
  \left\langle \varphi^{\sigma}(\gamma)|\partial_{j}\varphi^{\sigma}
(\gamma)\right\rangle 
&  =\langle\Gamma|\left(  \sqrt{\sigma(\gamma)}\otimes I\right)  \left(
\int_{0}^{\infty}dt\ e^{-t\sqrt{\sigma(\gamma)}}\left(  \partial_{j}
\sigma(\gamma)\right)  e^{-t\sqrt{\sigma(\gamma)}}\otimes I\right)
|\Gamma\rangle\\
&  =\int_{0}^{\infty}dt\ \langle\Gamma|\left(  \sqrt{\sigma(\gamma)}
e^{-t\sqrt{\sigma(\gamma)}}\left(  \partial_{j}\sigma(\gamma)\right)
e^{-t\sqrt{\sigma(\gamma)}}\otimes I\right)  |\Gamma\rangle\\
&  =\int_{0}^{\infty}dt\ \operatorname{Tr}\!\left[  \sqrt{\sigma(\gamma
)}e^{-t\sqrt{\sigma(\gamma)}}\left(  \partial_{j}\sigma(\gamma)\right)
e^{-t\sqrt{\sigma(\gamma)}}\right]  \\
&  =\int_{0}^{\infty}dt\ \operatorname{Tr}\!\left[  \sqrt{\sigma(\gamma
)}e^{-2t\sqrt{\sigma(\gamma)}}\left(  \partial_{j}\sigma(\gamma)\right)
\right]  \\
&  =\frac{1}{2}\operatorname{Tr}\!\left[  \sqrt{\sigma(\gamma)}\sqrt
{\sigma(\gamma)}^{-1}\left(  \partial_{j}\sigma(\gamma)\right)  \right]  \\
&  =\frac{1}{2}\operatorname{Tr}\!\left[  \partial_{j}\sigma(\gamma)\right]
\\
&  =0.
\end{align}
In the second-to-last line, we made use of the integral
\begin{equation}
\int_{0}^{\infty}dt\ e^{-2tx}=\frac{1}{2x}.
\end{equation}
So this implies that the second term in~\eqref{eq:formula-FB-pure} is
equal to zero.

Now consider that
\begin{align}
&  \langle\partial_{i}\varphi^{\sigma}(\gamma)|\partial_{j}\varphi^{\sigma
}(\gamma)\rangle \notag \\
&  =\langle\Gamma|\left(  \int_{0}^{\infty}dt_{1}\ e^{-t_{1}\sqrt
{\sigma(\gamma)}}\left(  \partial_{i}\sigma(\gamma)\right)  e^{-t_{1}
\sqrt{\sigma(\gamma)}}\otimes I\right)  \left(  \int_{0}^{\infty}
dt_{2}\ e^{-t_{2}\sqrt{\sigma(\gamma)}}\left(  \partial_{j}\sigma
(\gamma)\right)  e^{-t_{2}\sqrt{\sigma(\gamma)}}\otimes I\right)
|\Gamma\rangle\\
&  =\int_{0}^{\infty}\int_{0}^{\infty}dt_{1}\ dt_{2}\ \langle\Gamma|\left(
\ e^{-t_{1}\sqrt{\sigma(\gamma)}}\left(  \partial_{i}\sigma(\gamma)\right)
e^{-t_{1}\sqrt{\sigma(\gamma)}}e^{-t_{2}\sqrt{\sigma(\gamma)}}\left(
\partial_{j}\sigma(\gamma)\right)  e^{-t_{2}\sqrt{\sigma(\gamma)}}\otimes
I\right)  |\Gamma\rangle\\
&  =\int_{0}^{\infty}\int_{0}^{\infty}dt_{1}\ dt_{2}\ \operatorname{Tr}
\!\left[  e^{-t_{1}\sqrt{\sigma(\gamma)}}\left(  \partial_{i}\sigma
(\gamma)\right)  e^{-t_{1}\sqrt{\sigma(\gamma)}}e^{-t_{2}\sqrt{\sigma(\gamma
)}}\left(  \partial_{j}\sigma(\gamma)\right)  e^{-t_{2}\sqrt{\sigma(\gamma)}
}\right]  \\
&  =\int_{0}^{\infty}\int_{0}^{\infty}dt_{1}\ dt_{2}\ \operatorname{Tr}
\!\left[  e^{-\left(  t_{1}+t_{2}\right)  \sqrt{\sigma(\gamma)}}\left(
\partial_{i}\sigma(\gamma)\right)  e^{-\left(  t_{1}+t_{2}\right)
\sqrt{\sigma(\gamma)}}\left(  \partial_{j}\sigma(\gamma)\right)  \right]  \\
& = \sum_{k,\ell}\frac{1}{\left(  \sqrt{\lambda_{k}}+\sqrt{\lambda_{\ell}
}\right)  ^{2}}\langle k|\left(  \partial_{i}\sigma(\gamma)\right)
|\ell\rangle\!\langle\ell|\left(  \partial_{j}\sigma(\gamma)\right)
|k\rangle,
\end{align}
where the last equality follows from~\eqref{eq:WY-formula-2}.
Putting everything together, we find that the Fisher--Bures
information matrix elements of the parameterized family $\left(
\varphi^{\sigma}(\gamma)\right)  _{\gamma}$ are given by
\begin{align}
I_{ij}^{\operatorname{FB}}(\theta) &  =4\operatorname{Re}\left[
\langle\partial_{i}\varphi^{\sigma}(\gamma)|\partial_{j}\varphi^{\sigma
}(\gamma)\rangle-\left\langle \partial_{i}\varphi^{\sigma}(\gamma
)|\varphi^{\sigma}(\gamma)\right\rangle \left\langle \varphi^{\sigma}
(\gamma)|\partial_{j}\varphi^{\sigma}(\gamma)\right\rangle \right]  \\
&  =4\operatorname{Re}\left[  \langle\partial_{i}\varphi^{\sigma}
(\gamma)|\partial_{j}\varphi^{\sigma}(\gamma)\rangle\right]  \\
&  =\sum_{k,\ell}\frac{4}{\left(  \sqrt{\lambda_{k}}+\sqrt{\lambda_{\ell}
}\right)  ^{2}}\langle k|\left(  \partial_{i}\sigma(\gamma)\right)
|\ell\rangle\!\langle\ell|\left(  \partial_{j}\sigma(\gamma)\right)
|k\rangle.
\end{align}
This concludes the alternative proof of Proposition~\ref{prop:FB-WY-canonical-purifications}.

\section{Information matrix elements for evolved quantum Boltzmann machines}

\subsection{Fisher--Bures information matrix elements for evolved quantum Boltzmann machines}

\label{app:FB-elements}

\subsubsection{Proof of Theorem~\ref{thm:FB-theta}}

\label{proof:FB-theta}

Using~\eqref{eq:Fisher-info-help-alt}, consider that
\begin{align}
&  I_{ij}^{\text{FB}}(\theta)\nonumber\\
&  =\sum_{k,\ell}\frac{2}{\lambda_{k}+\lambda_{\ell}}\langle k|\left[
\frac{\partial}{\partial\theta_{i}}\omega(\theta,\phi)\right]  |\ell
\rangle\!\langle\ell|\left[  \frac{\partial}{\partial\theta_{j}}\omega(\theta
,\phi)\right]  |k\rangle\\
&  =\sum_{k,\ell}\frac{2}{\lambda_{k}+\lambda_{\ell}} \left( -\frac{1}{2}\left( \lambda_\ell + \lambda_k\right) \bra{\tilde{k}} \Phi_{\theta}(G_{i})\ket{\tilde{\ell}} + \delta_{k\ell} \lambda_\ell \left\langle G_{i}\right\rangle_{\rho(\theta)} \right) \left( -\frac{1}{2}\left( \lambda_\ell + \lambda_k\right) \bra{\tilde{\ell}} \Phi_{\theta}(G_{j}) \ket{\tilde{k}} + \delta_{\ell k} \lambda_k \left\langle G_{j}\right\rangle _{\rho(\theta)}\right)\\
& = \frac{1}{2} \sum_{k,\ell} \left( \lambda_\ell + \lambda_k \right) \bra{\tilde{k}} \Phi_{\theta}(G_{i}) \ket{\tilde{\ell}} \bra{\tilde{\ell}} \Phi_{\theta}(G_{j})\ket{\tilde{k}} + \sum_k \lambda_k \left\langle G_{i} \right\rangle_{\rho(\theta)} \left\langle G_{j} \right\rangle_{\rho(\theta)}  \nonumber\\
& \hspace{6.9cm} - \sum_k \lambda_k \bra{\tilde{k}} \Phi_{\theta}(G_{i})\ket{\tilde{k}}\left\langle G_{j} \right\rangle_{\rho(\theta)} - \sum_k \lambda_k \bra{\tilde{k}}\Phi_{\theta}(G_{j}) \ket{\tilde{k}}\left\langle G_{i} \right\rangle_{\rho(\theta)}\\
& = \frac{1}{2} \Tr\!\left[ \Phi_{\theta}(G_{i}) \rho(\theta) \Phi_{\theta}(G_{j}) \right] + \frac{1}{2} \Tr\!\left[ \Phi_{\theta}(G_{j}) \rho(\theta) \Phi_{\theta}(G_{i}) \right] + \left\langle G_{i} \right\rangle_{\rho(\theta)} \left\langle G_{j} \right\rangle_{\rho(\theta)} \nonumber\\
& \hspace{6.9cm}  - \Tr\!\left[\Phi_{\theta}(G_{i})\rho(\theta)\right] \left\langle G_{j} \right\rangle_{\rho(\theta)} - \Tr\!\left[\Phi_{\theta}(G_{j})\rho(\theta)\right] \left\langle G_{i} \right\rangle_{\rho(\theta)}\\
& = \frac{1}{2} \Tr\!\left[ \left\{ \Phi_{\theta}(G_{i}) , \Phi_{\theta}(G_{j}) \right\}\rho(\theta) \right] - \left\langle G_{i} \right\rangle_{\rho(\theta)} \left\langle G_{j} \right\rangle_{\rho(\theta)}\\
& = \frac{1}{2} \left\langle \left\{ \Phi_{\theta}(G_{i}) , \Phi_{\theta}(G_{j}) \right\} \right\rangle_{\rho(\theta)} - \left\langle G_{i} \right\rangle_{\rho(\theta)} \left\langle G_{j} \right\rangle_{\rho(\theta)}.\label{eq:last-FB-theta-app}
\end{align}
This concludes the proof of Theorem~\ref{thm:FB-theta}. Note that the first term of~\eqref{eq:last-FB-theta-app} can also be written as $\text{Re}\left[ \Tr\!\left[ \Phi_{\theta}(G_{i}) \Phi_{\theta}(G_{j}) \rho(\theta) \right] \right]$.

\subsubsection{Proof of Theorem~\ref{thm:FB-phi}}

\label{proof:FB-phi}

Using~\eqref{eq:Fisher-info-help-last}, consider that
\begin{align}
 I_{ij}^{\text{FB}}(\phi)
&  =\sum_{k,\ell}\frac{2}{\lambda_{k}+\lambda_{\ell}}\langle k|\left[
\frac{\partial}{\partial\phi_{i}}\omega(\theta,\phi)\right]  |\ell
\rangle\langle\ell|\left[  \frac{\partial}{\partial\phi_{j}}\omega(\theta
,\phi)\right]  |k\rangle\\
&  =\sum_{k,\ell}\frac{2}{\lambda_{k}+\lambda_{\ell}}\left(  i\left(
\lambda_{k}-\lambda_{\ell}\right)  \langle k|\Psi^{\dagger}_{\phi}(H_{i})|\ell
\rangle\right)  \left(  i\left(  \lambda_{\ell}-\lambda_{k}\right)
\langle\ell|\Psi^{\dagger}_{\phi}(H_{j})|k\rangle\right)  \\
&  =2\sum_{k,\ell}\frac{\left(  \lambda_{k}-\lambda_{\ell}\right)  ^{2}
}{\lambda_{k}+\lambda_{\ell}}\langle k|\Psi^{\dagger}_{\phi}(H_{i})|\ell\rangle
\langle\ell|\Psi^{\dagger}_{\phi}(H_{j})|k\rangle\\
&  =2\sum_{k,\ell}\frac{\left(  \lambda_{k}+\lambda_{\ell}\right)
^{2}-4\lambda_{k}\lambda_{\ell}}{\lambda_{k}+\lambda_{\ell}}\langle
k|\Psi^{\dagger}_{\phi}(H_{i})|\ell\rangle\langle\ell|\Psi^{\dagger}_{\phi}(H_{j})|k\rangle\\
&  =2\sum_{k,\ell}\left(  \lambda_{k}+\lambda_{\ell}\right)  \langle
k|\Psi^{\dagger}_{\phi}(H_{i})|\ell\rangle\langle\ell|\Psi^{\dagger}_{\phi}(H_{j})|k\rangle
\nonumber\\
&  \qquad-8\sum_{k,\ell}\frac{\lambda_{k}\lambda_{\ell}}{\lambda_{k}
+\lambda_{\ell}}\langle k|\Psi^{\dagger}_{\phi}(H_{i})|\ell\rangle\langle\ell|\Psi^\dag
_{\phi}(H_{j})|k\rangle\\
&  =2\sum_{k,\ell}\lambda_{k}\langle k|\Psi^{\dagger}_{\phi}(H_{i})|\ell\rangle
\langle\ell|\Psi^{\dagger}_{\phi}(H_{j})|k\rangle+2\sum_{k,\ell}\lambda_{\ell}\langle
k|\Psi^{\dagger}_{\phi}(H_{i})|\ell\rangle\langle\ell|\Psi^{\dagger}_{\phi}(H_{j})|k\rangle
\nonumber\\
&  \qquad-8\sum_{k,\ell}\frac{\lambda_{k}\lambda_{\ell}}{\lambda_{k}
+\lambda_{\ell}}\langle k|\Psi^{\dagger}_{\phi}(H_{i})|\ell\rangle\langle\ell|\Psi^\dag
_{\phi}(H_{j})|k\rangle\\
&  =2\operatorname{Tr}[\omega(\theta,\phi)\Psi^{\dagger}_{\phi}(H_{i})\Psi^\dag_{\phi}
(H_{j})]+2\operatorname{Tr}[\Psi^{\dagger}_{\phi}(H_{i})\omega(\theta,\phi)\Psi^\dag_{\phi
}(H_{j})]\nonumber\\
&  \qquad-8\sum_{k,\ell}\frac{\lambda_{k}\lambda_{\ell}}{\lambda_{k}
+\lambda_{\ell}}\langle k|\Psi^{\dagger}_{\phi}(H_{i})|\ell\rangle\langle\ell|\Psi^\dag
_{\phi}(H_{j})|k\rangle\\
&  =2\left\langle \left\{  \Psi^{\dagger}_{\phi}(H_{i}),\Psi^{\dagger}_{\phi}(H_{j})\right\}
\right\rangle _{\omega(\theta,\phi)}-8\sum_{k,\ell}\frac{\lambda_{k}
\lambda_{\ell}}{\lambda_{k}+\lambda_{\ell}}\langle k|\Psi^\dag_{\phi}(H_{i}
)|\ell\rangle\langle\ell|\Psi^{\dagger}_{\phi}(H_{j})|k\rangle
 \\
&= 2\left\langle \left\{  \Psi_{\phi}(H_{i}),\Psi_{\phi}(H_{j})\right\}  \right\rangle _{\rho(\theta)}
-8\sum_{k,\ell}\frac{\lambda_{k}
\lambda_{\ell}}{\lambda_{k}+\lambda_{\ell}}\langle k|\Psi^\dag_{\phi}(H_{i}
)|\ell\rangle\langle\ell|\Psi^{\dagger}_{\phi}(H_{j})|k\rangle.
\label{eq:proof-thm-10-1st-block}
\end{align}
Now let us recall that $|k\rangle=e^{-iH(\phi)}|\tilde{k}\rangle$, and define
the spectral decomposition
\begin{equation}
\rho(\theta)=\sum_{k}\lambda_{k}|\tilde{k}\rangle\!\langle\tilde{k}|=\frac{1}
{Z}\sum_{k}e^{-\mu_{k}}|\tilde{k}\rangle\!\langle\tilde{k}|,
\end{equation}
where a spectral decomposition of $G(\theta) = \sum_k \mu_k |\tilde{k}\rangle\!\langle\tilde{k}|$, so that the key quantity in the second term in~\eqref{eq:proof-thm-10-1st-block} can be written as
\begin{align}
&  \sum_{k,\ell}\frac{\lambda_{k}\lambda_{\ell}}{\lambda_{k}+\lambda_{\ell}
}\langle k|\Psi^{\dagger}_{\phi}(H_{i})|\ell\rangle\langle\ell|\Psi^\dag_{\phi}
(H_{j})|k\rangle\nonumber\\
&  =\sum_{k,\ell}\frac{\lambda_{k}\lambda_{\ell}}{\lambda_{k}+\lambda_{\ell}
}\langle\tilde{k}|e^{iH(\phi)}\Psi^{\dagger}_{\phi}(H_{i})e^{-iH(\phi)}|\tilde{\ell
}\rangle\langle\tilde{\ell}|e^{iH(\phi)}\Psi^{\dagger}_{\phi}(H_{j})e^{-iH(\phi)}
|\tilde{k}\rangle\\
&  =\sum_{k,\ell}\frac{\lambda_{k}\lambda_{\ell}}{\lambda_{k}+\lambda_{\ell}
}\langle\tilde{k}|\Psi_{\phi}(H_{i})|\tilde{\ell}\rangle\langle
\tilde{\ell}|\Psi_{\phi}(H_{j})|\tilde{k}\rangle\\
&  =\operatorname{Re}\!\left[  \sum_{k,\ell}\frac{\lambda_{k}\lambda_{\ell}
}{\lambda_{k}+\lambda_{\ell}}\langle\tilde{k}|\Psi_{\phi}(H_{i}
)|\tilde{\ell}\rangle\langle\tilde{\ell}|\Psi_{\phi}(H_{j})|\tilde
{k}\rangle\right]  \\
&  =\operatorname{Re}\!\left[  \sum_{k,\ell}\frac{\lambda_{k}\frac{e^{-\mu
_{\ell}}}{Z}}{\frac{e^{-\mu_{k}}}{Z}+\frac{e^{-\mu_{\ell}}}{Z}}\langle
\tilde{k}|\Psi_{\phi}(H_{i})|\tilde{\ell}\rangle\langle\tilde{\ell
}|\Psi_{\phi}(H_{j})|\tilde{k}\rangle\right]  \\
&  =\operatorname{Re}\!\left[  \sum_{k,\ell}\lambda_{k}\frac{1}{e^{-\left(
\mu_{k}-\mu_{\ell}\right)  }+1}\langle\tilde{k}|\Psi_{\phi}
(H_{i})|\tilde{\ell}\rangle\langle\tilde{\ell}|\Psi_{\phi}
(H_{j})|\tilde{k}\rangle\right]  .
\end{align}
Now observe that, for all $x\in\mathbb{R}$,
\begin{align}
\frac{1}{e^{-x}+1} &  =\frac{e^{x/2}}{e^{x/2}+e^{-x/2}}\\
&  =\frac{e^{x/2}}{e^{x/2}+e^{-x/2}}-\frac{1}{2}+\frac{1}{2}\\
&  =\frac{e^{x/2}}{e^{x/2}+e^{-x/2}}-\frac{\frac{1}{2}e^{x/2}+\frac{1}
{2}e^{-x/2}}{e^{x/2}+e^{-x/2}}+\frac{1}{2}\\
&  =\frac{1}{2}\frac{e^{x/2}-e^{-x/2}}{e^{x/2}+e^{-x/2}}+\frac{1}{2}\\
&  =\frac{1}{2}\tanh(x/2)+\frac{1}{2}\\
&  =\frac{x}{4}\frac{\tanh(x/2)}{x/2}+\frac{1}{2}.
\end{align}
Substituting above, we find that
\begin{align}
&  \operatorname{Re}\!\left[  \sum_{k,\ell}\lambda_{k}\frac{1}{e^{-\left(
\mu_{k}-\mu_{\ell}\right)  }+1}\langle\tilde{k}|\Psi_{\phi}
(H_{i})|\tilde{\ell}\rangle\langle\tilde{\ell}|\Psi_{\phi}
(H_{j})|\tilde{k}\rangle\right]  \nonumber\\
&  =\operatorname{Re}\!\left[  \sum_{k,\ell}\lambda_{k}\left(  \frac{\mu_{k}
-\mu_{\ell}}{4}\frac{\tanh(\left(  \mu_{k}-\mu_{\ell}\right)  /2)}{\left(
\mu_{k}-\mu_{\ell}\right)  /2}+\frac{1}{2}\right)  \langle\tilde{k}|\Psi
_{\phi}(H_{i})|\tilde{\ell}\rangle\langle\tilde{\ell}|\Psi_{\phi}
(H_{j})|\tilde{k}\rangle\right]  \\
&  =\frac{1}{4}\operatorname{Re}\!\left[  \sum_{k,\ell}\lambda_{k}\left(
\mu_{k}-\mu_{\ell}\right)  \frac{\tanh(\left(  \mu_{k}-\mu_{\ell}\right)
/2)}{\left(  \mu_{k}-\mu_{\ell}\right)  /2}\langle\tilde{k}|\Psi_{\phi}(H_{i})|\tilde{\ell}\rangle\langle\tilde{\ell}|\Psi_{\phi}
(H_{j})|\tilde{k}\rangle\right]  \nonumber\\
&  \qquad+\frac{1}{2}\operatorname{Re}\!\left[  \sum_{k,\ell}\lambda_{k}
\langle\tilde{k}|\Psi_{\phi}(H_{i})|\tilde{\ell}\rangle\langle
\tilde{\ell}|\Psi_{\phi}(H_{j})|\tilde{k}\rangle\right]  \\
&  =\frac{1}{4}\operatorname{Re}\!\left[  \sum_{k,\ell}\lambda_{k}\left(
\mu_{k}-\mu_{\ell}\right)  \int_{\mathbb{R}}dt\ p(t)e^{-i\left(  \mu_{k}
-\mu_{\ell}\right)  t}\langle\tilde{k}|\Psi_{\phi}(H_{i})|\tilde{\ell
}\rangle\langle\tilde{\ell}|\Psi_{\phi}(H_{j})|\tilde{k}\rangle\right]
\nonumber\\
&  \qquad+\frac{1}{2}\operatorname{Re}\!\left[  \operatorname{Tr}\!\left[
\rho(\theta)\Psi_{\phi}(H_{i})\Psi_{\phi}(H_{j})\right]
\right]  \\
&  =\frac{1}{4}\operatorname{Re}\!\left[  \sum_{k,\ell}\lambda_{k}\left(
\mu_{k}-\mu_{\ell}\right)  \int_{\mathbb{R}}dt\ p(t)e^{-i\left(  \mu_{k}
-\mu_{\ell}\right)  t}\langle\tilde{k}|\Psi_{\phi}(H_{i})|\tilde{\ell
}\rangle\langle\tilde{\ell}|\Psi_{\phi}(H_{j})|\tilde{k}\rangle\right]
\nonumber\\
&  \qquad+\frac{1}{4}\left\langle \left\{  \Psi_{\phi}(H_{i}
),\Psi_{\phi}(H_{j})\right\}  \right\rangle _{\rho(\theta)}.
\end{align}
The third equality above follows from \cite[Lemma~12]{patel2024quantumboltzmannmachine}.
Consider now that the key quantity in the first term above can be rewritten as
\begin{align}
&  \sum_{k,\ell}\lambda_{k}\left(  \mu_{k}-\mu_{\ell}\right)  \int
_{\mathbb{R}}dt\ p(t)e^{-i\left(  \mu_{k}-\mu_{\ell}\right)  t}\langle
\tilde{k}|\Psi_{\phi}(H_{i})|\tilde{\ell}\rangle\langle\tilde{\ell
}|\Psi_{\phi}(H_{j})|\tilde{k}\rangle\nonumber\\
&  =\int_{\mathbb{R}}dt\ p(t)\sum_{k,\ell}\lambda_{k}\left(  \mu_{k}-\mu
_{\ell}\right)  \operatorname{Tr}\!\left[  \Psi_{\phi}(H_{i}
)e^{i\mu_{\ell}t}|\tilde{\ell}\rangle\langle\tilde{\ell}|\Psi_{\phi}(H_{j})e^{-i\mu_{k}t}|\tilde{k}\rangle\langle\tilde{k}|\right]  \\
&  =\int_{\mathbb{R}}dt\ p(t)\sum_{k,\ell}\operatorname{Tr}\!\left[  \Psi_{\phi
}(H_{i})e^{i\mu_{\ell}t}|\tilde{\ell}\rangle\langle\tilde{\ell}
|\Psi_{\phi}(H_{j})\lambda_{k}\mu_{k}e^{-i\mu_{k}t}|\tilde{k}
\rangle\langle\tilde{k}|\right]  \nonumber\\
&  \qquad-\int_{\mathbb{R}}dt\ p(t)\sum_{k,\ell}\operatorname{Tr}\!\left[
\Psi_{\phi}(H_{i})\mu_{\ell}e^{i\mu_{\ell}t}|\tilde{\ell}\rangle
\langle\tilde{\ell}|\Psi_{\phi}(H_{j})\lambda_{k}e^{-i\mu_{k}t}
|\tilde{k}\rangle\langle\tilde{k}|\right]  \\
&  =\int_{\mathbb{R}}dt\ p(t)\operatorname{Tr}\!\left[  \Psi_{\phi}
(H_{i})e^{iG(\theta)t}\Psi_{\phi}(H_{j})\rho(\theta)G(\theta
)e^{-iG(\theta)t}\right]  \nonumber\\
&  \qquad-\int_{\mathbb{R}}dt\ p(t)\operatorname{Tr}\!\left[  \Psi_{\phi}(H_{i})G(\theta)e^{iG(\theta)t}\Psi_{\phi}(H_{j})\rho(\theta
)e^{-iG(\theta)t}\right]  \\
&  =\operatorname{Tr}\!\left[  \Phi_{\theta}(\Psi_{\phi}(H_{i}
))\Psi_{\phi}(H_{j})G(\theta)\rho(\theta)\right]  -\operatorname{Tr}\!
\left[  \Phi_{\theta}(\Psi_{\phi}(H_{i}))G(\theta)\Psi_{\phi}(H_{j})\rho(\theta)\right]  .
\end{align}
Putting everything together, we find that
\begin{align}
  I_{ij}^{\text{FB}}(\phi)
&  =2\left\langle \left\{  \Psi_{\phi}(H_{i}),\Psi_{\phi}(H_{j})\right\}  \right\rangle _{\rho(\theta)}\nonumber\\
&  \qquad-8\left(
\begin{array}
[c]{c}
\frac{1}{4}\left(
\begin{array}
[c]{c}
\operatorname{Re}\!\left[  \operatorname{Tr}\!\left[  \Phi_{\theta}(\Psi_{\phi
}(H_{i}))\Psi_{\phi}(H_{j})G(\theta)\rho(\theta)\right]
\right]  \\
-\operatorname{Re}\!\left[  \operatorname{Tr}\!\left[  \Phi_{\theta}(\Psi_{\phi
}(H_{i}))G(\theta)\Psi_{\phi}(H_{j})\rho(\theta)\right]
\right]
\end{array}
\right)  \\
+\frac{1}{4}\left\langle \left\{  \Psi_{\phi}(H_{i}),\Psi_{\phi}(H_{j})\right\}  \right\rangle _{\rho(\theta)}
\end{array}
\right)  \\
&  =2\left\langle \left\{  \Psi_{\phi}(H_{i}),\Psi_{\phi}(H_{j})\right\}  \right\rangle _{\rho(\theta)}-\operatorname{Tr}\!\left[  \Phi_{\theta
}(\Psi_{\phi}(H_{i}))\Psi_{\phi}(H_{j})G(\theta)\rho
(\theta)\right]  \nonumber\\
&  \qquad+\operatorname{Tr}\!\left[  \Phi_{\theta}(\Psi_{\phi}
(H_{i}))G(\theta)\Psi_{\phi}(H_{j})\rho(\theta)\right]  \nonumber\\
&  \qquad-\operatorname{Tr}\!\left[  \Phi_{\theta}(\Psi_{\phi}
(H_{i}))\rho(\theta)G(\theta)\Psi_{\phi}(H_{j})\right]
+\operatorname{Tr}\!\left[  \Phi_{\theta}(\Psi_{\phi}(H_{i}))\rho
(\theta)\Psi_{\phi}(H_{j})G(\theta)\right]  \nonumber\\
&  \qquad-2\left\langle \left\{  \Psi_{\phi}(H_{i}),\Psi_{\phi}(H_{j})\right\}  \right\rangle _{\rho(\theta)}\\
&  =\operatorname{Tr}\!\left[  \Phi_{\theta}(\Psi_{\phi}(H_{i}
))G(\theta)\Psi_{\phi}(H_{j})\rho(\theta)\right]  -\operatorname{Tr}\!
\left[  \Phi_{\theta}(\Psi_{\phi}(H_{i}))\Psi_{\phi}
(H_{j})G(\theta)\rho(\theta)\right]  \nonumber\\
&  \qquad+\operatorname{Tr}\!\left[  \Phi_{\theta}(\Psi_{\phi}
(H_{i}))\rho(\theta)\Psi_{\phi}(H_{j})G(\theta)\right]
-\operatorname{Tr}\!\left[  \Phi_{\theta}(\Psi_{\phi}(H_{i}))\rho
(\theta)G(\theta)\Psi_{\phi}(H_{j})\right]  \\
&  =\left\langle \left[  \left[  \Psi_{\phi}(H_{j}),G(\theta)\right]
,\Phi_{\theta}(\Psi_{\phi}(H_{i}))\right]  \right\rangle _{\rho
(\theta)}\\
& = \left\langle \left[  \Phi_{\theta}(\Psi_{\phi}(H_{i})), \left[  G(\theta) , \Psi_{\phi}
(H_{j})\right]  \right]
\right\rangle _{\rho(\theta)}.
\end{align}
This concludes the proof of Theorem~\ref{thm:FB-phi}.

\subsubsection{Proof of Theorem~\ref{thm:FB-theta-phi}}

\label{proof:FB-theta-phi}

Using~\eqref{eq:Fisher-info-help-last} and~\eqref{eq:Fisher-info-help-alt1}, consider that
\begin{align}
&  I_{ij}^{\text{FB}}(\theta, \phi)\nonumber\\
&  =\sum_{k,\ell}\frac{2}{\lambda_{k}+\lambda_{\ell}}\langle k|\left[
\frac{\partial}{\partial\theta_{i}}\omega(\theta,\phi)\right]  |\ell
\rangle\!\langle\ell|\left[  \frac{\partial}{\partial\phi_{j}}\omega(\theta
,\phi)\right]  |k\rangle\\
&  =\sum_{k,\ell}\frac{2}{\lambda_{k}+\lambda_{\ell}} \left( -\frac{1}{2}\left( \lambda_\ell + \lambda_k\right) \bra{k} e^{-iH(\phi)}\Phi_{\theta}(G_{i})e^{iH(\phi)}\ket{\ell} + \delta_{k\ell} \lambda_\ell \left\langle G_{i}\right\rangle_{\rho(\theta)} \right) \left(  i\left(  \lambda_{\ell}-\lambda_{k}\right)
\langle\ell|\Psi^{\dagger}_{\phi}(H_{j})|k\rangle\right)\\
& = \sum_{k,\ell} -i \left(\lambda_{\ell}-\lambda_{k}\right) \bra{k} e^{-iH(\phi)}\Phi_{\theta}(G_{i})e^{iH(\phi)}\ket{\ell} \langle\ell|\Psi^{\dagger}_{\phi}(H_{j})|k\rangle\\
& = \sum_{k,\ell} -i \lambda_{\ell} \bra{k} e^{-iH(\phi)}\Phi_{\theta}(G_{i})e^{iH(\phi)}\ket{\ell} \langle\ell|\Psi^{\dagger}_{\phi}(H_{j})|k\rangle 
+ \sum_{k,\ell} i \lambda_{k} \bra{k} e^{-iH(\phi)}\Phi_{\theta}(G_{i})e^{iH(\phi)}\ket{\ell} \langle\ell|\Psi^{\dagger}_{\phi}(H_{j})|k\rangle\\
& = -i \Tr\!\left[ e^{-iH(\phi)}\Phi_{\theta}(G_{i})e^{iH(\phi)} \omega(\theta,\phi) \Psi^{\dagger}_{\phi}(H_{j}) \right] 
+i \Tr\!\left[ \Psi^{\dagger}_{\phi}(H_{j}) \omega(\theta,\phi) e^{-iH(\phi)}\Phi_{\theta}(G_{i})e^{iH(\phi)}\right]\\
& = -i \Tr\!\left[ \Psi^{\dagger}_{\phi}(H_{j}) e^{-iH(\phi)}\Phi_{\theta}(G_{i})e^{iH(\phi)} \omega(\theta,\phi)  \right] 
+i \Tr\!\left[ e^{-iH(\phi)}\Phi_{\theta}(G_{i})e^{iH(\phi)} \Psi^{\dagger}_{\phi}(H_{j}) \omega(\theta,\phi) \right]\\
& = -i \Tr\!\left[ \Psi^{\dagger}_{\phi}(H_{j}) e^{-iH(\phi)}\Phi_{\theta}(G_{i}) \rho(\theta) e^{iH(\phi)} \right] 
+i \Tr\!\left[ e^{-iH(\phi)}\Phi_{\theta}(G_{i})e^{iH(\phi)} \Psi^{\dagger}_{\phi}(H_{j}) e^{-iH(\phi)} \rho(\theta) e^{iH(\phi)} \right]\\
& = -i \Tr\!\left[ e^{iH(\phi)} \Psi^{\dagger}_{\phi}(H_{j}) e^{-iH(\phi)}\Phi_{\theta}(G_{i}) \rho(\theta) \right] 
+i \Tr\!\left[ \Phi_{\theta}(G_{i})e^{iH(\phi)} \Psi^{\dagger}_{\phi}(H_{j}) e^{-iH(\phi)} \rho(\theta) \right]\\
& = -i \Tr\!\left[ \Psi_{\phi}(H_{j}) \Phi_{\theta}(G_{i}) \rho(\theta) \right] 
+i \Tr\!\left[ \Phi_{\theta}(G_{i}) \Psi_{\phi}(H_{j}) \rho(\theta) \right]\\
& = i \Tr\!\left[ \left[ \Phi_{\theta}(G_{i}) , \Psi_{\phi}(H_{j}) \right] \rho(\theta) \right]\\
& = i \left\langle \left[ \Phi_{\theta}(G_{i}) , \Psi_{\phi}(H_{j}) \right] \right\rangle_{\rho(\theta)}. \label{eq:last-FB-theta-phi-app}
\end{align}
The third equality above follows because
\begin{equation}
    \sum_{k,\ell}\frac{2}{\lambda_{k}+\lambda_{\ell}}  \delta_{k\ell} \lambda_\ell \left\langle G_{i}\right\rangle_{\rho(\theta)}  \left(  i\left(  \lambda_{\ell}-\lambda_{k}\right)
\langle\ell|\Psi^{\dagger}_{\phi}(H_{j})|k\rangle\right) = 0.
\end{equation}
This concludes the proof of Theorem~\ref{thm:FB-theta-phi}. Note that the final expression in~\eqref{eq:last-FB-theta-phi-app} can also be written as $- 2\operatorname{Im}\!\left[  \Tr\!\left[ \Phi_{\theta}(G_{i}) \Psi_{\phi}(H_{j}) \rho(\theta) \right] \right]$.

\subsection{Wigner--Yanase information matrix elements for evolved quantum Boltzmann machines}

\subsubsection{Proof of Theorem~\ref{thm:WY-partial-derivatives}}

\label{proof:WY-partial-derivatives}

Using the notations in the statement of Theorem~\ref{thm:WY-partial-derivatives}, consider that
\begin{align}
|\partial_{j}\psi(\theta,\phi)\rangle & =\frac{\partial}{\partial\theta_{j}
}\left( \sqrt{\omega(\theta,\phi)}\otimes I\right)  |\Gamma\rangle\\
& =\frac{\partial}{\partial\theta_{j}}\left(  e^{-iH(\phi)} \sqrt
{\rho(\theta)} \ e^{iH(\phi)} \otimes I\right)  |\Gamma\rangle\\
& =\left(  e^{-iH(\phi)}\left(  \frac{\partial}{\partial\theta_{j}}\sqrt
{\rho(\theta)}\right) e^{iH(\phi)} \otimes I\right)  |\Gamma\rangle\\
& =\left[  e^{-iH(\phi)}\left(  -\frac{1}{4}\left\{  \Phi_{\frac{\theta}{2}
}(G_{j}),\sqrt{\rho(\theta)}\right\}  +\frac{1}{2}\sqrt{\rho(\theta
)}\left\langle G_{j}\right\rangle _{\rho(\theta)}\right) e^{iH(\phi)} \otimes I\right]
|\Gamma\rangle\\
& =-\frac{1}{4}\left(  e^{-iH(\phi)}\left\{  \Phi_{\frac{\theta}{2}}
(G_{j}),\sqrt{\rho(\theta)}\right\} e^{iH(\phi)} \otimes I\right)  |\Gamma\rangle
+\frac{1}{2}\left\langle G_{j}\right\rangle _{\rho(\theta)}
|\psi(\theta,\phi)\rangle.
\end{align}
The third equality follows because
\begin{equation}
\sqrt{\rho(\theta)}   =\sqrt{\frac{e^{-G(\theta)}}{Z(\theta)}}
  =\frac{e^{-\frac{1}{2}G(\theta)}}{\sqrt{Z(\theta)}}
  =\frac{e^{-G(\theta/2)}}{\sqrt{Z(\theta)}},
\end{equation}
which implies that
\begin{align}
\frac{\partial}{\partial\theta_{j}}\sqrt{\rho(\theta)} &  =\frac{\partial
}{\partial\theta_{j}}\left(  \frac{e^{-G(\theta/2)}}{\sqrt{Z(\theta)}}\right)
\\
&  =\frac{\partial}{\partial\theta_{j}}\left(  e^{-G(\theta/2)}\frac{1}
{\sqrt{Z(\theta)}}\right)  \\
&  =\left(  \frac{\partial}{\partial\theta_{j}}e^{-G(\theta/2)}\right)
\frac{1}{\sqrt{Z(\theta)}}+e^{-G(\theta/2)}\left(  \frac{\partial}
{\partial\theta_{j}}\frac{1}{\sqrt{Z(\theta)}}\right)  \\
&  =-\frac{1}{4}\left\{  \Phi_{\frac{\theta}{2}}(G_{j}),e^{-G(\theta
/2)}\right\}  \frac{1}{\sqrt{Z(\theta)}}-\frac{1}{2}e^{-G(\theta/2)}\left(
\frac{\frac{\partial}{\partial\theta_{j}}Z(\theta)}{Z(\theta)^{\frac{3}{2}}
}\right) \label{eq:3rd-eq-deriv-sq-root-QBM}  \\
&  =-\frac{1}{4}\left\{  \Phi_{\frac{\theta}{2}}(G_{j}),e^{-G(\theta
/2)}\right\}  \frac{1}{\sqrt{Z(\theta)}}+\frac{1}{2}e^{-G(\theta/2)}\left(
\frac{\left\langle G_{j}\right\rangle _{\rho(\theta)}}{Z(\theta)^{\frac{1}{2}
}}\right)  \\
&  =-\frac{1}{4}\left\{  \Phi_{\frac{\theta}{2}}(G_{j}),\sqrt{\rho(\theta
)}\right\}  +\frac{1}{2}\sqrt{\rho(\theta)}\left\langle G_{j}\right\rangle
_{\rho(\theta)}.
\end{align}
To arrive at~\eqref{eq:3rd-eq-deriv-sq-root-QBM}, we applied \cite[Lemma~10]{patel2024quantumboltzmannmachine}.

Furthermore,
\begin{align}
|\partial_{i}\psi(\theta,\phi)\rangle & =\frac{\partial}{\partial\phi_{i}
}\left( \sqrt{\omega(\theta,\phi)}\otimes I\right)  |\Gamma\rangle\\
& =\frac{\partial}{\partial\phi_{i}
}\left(  e^{-iH(\phi)}\sqrt{\rho(\theta)}\ e^{iH(\phi)}\otimes I\right)  |\Gamma\rangle\\
& =\left(  \left(  \frac{\partial}{\partial\phi_{i}}e^{-iH(\phi)}\right)
\sqrt{\rho(\theta)} \ e^{iH(\phi)}\otimes I\right)  |\Gamma\rangle + \left(  e^{-iH(\phi)}
\sqrt{\rho(\theta)} \left(\frac{\partial}{\partial\phi_{i}}e^{iH(\phi)}\right)\otimes I\right)  |\Gamma\rangle\\
& =-i\left(  \Psi^{\dagger}_{\phi}(H_{i})e^{-iH(\phi)} \sqrt{\rho(\theta)}\ e^{iH(\phi)}\otimes I\right)  |\Gamma\rangle + i\left(  e^{-iH(\phi)}
\sqrt{\rho(\theta)} \ e^{iH(\phi)} \Psi^{\dagger}_{\phi}(H_{i}) \otimes I\right)  |\Gamma\rangle \\
& =i\left(  \left[ e^{-iH(\phi)}
\sqrt{\rho(\theta)} \ e^{iH(\phi)}, \Psi^{\dagger}_{\phi}(H_{i}) \right] \otimes I\right)  |\Gamma\rangle\\
& = i\left(  \left[
\sqrt{\omega(\theta,\phi)}, \Psi^{\dagger}_{\phi}(H_{i}) \right] \otimes I\right)  |\Gamma\rangle,
\end{align}
where we applied~\eqref{eq:deriv-phi-proof-1}--\eqref{eq:deriv-phi-proof-1-2}.

\subsubsection{Proof of Theorem~\ref{thm:WY-theta}}

\label{proof:WY-theta}

Here we employ the shorthands
$\partial_{i}\equiv\frac{\partial}{\partial\theta_{i}}$ and $\partial_{j}\equiv\frac{\partial}{\partial\theta_{j}}$.  Using Proposition~\ref{prop:FB-WY-canonical-purifications} and~\eqref{eq:part_der_theta_purified-evolved-QBM}, we find that
\begin{align}
&  \langle\partial_{i}\psi(\theta,\phi)|\partial_{j}\psi(\theta,\phi)\rangle\nonumber\\
&  =\left(
\begin{array}
[c]{c}
-\frac{1}{4}\left(  \langle \Gamma | e^{-iH(\phi)} \left\{  \Phi_{\frac{\theta}{2}}(G_{i}
),\sqrt{\rho(\theta)}\right\} e^{iH(\phi)} \otimes I\right)  \\
+ \frac{1}{2} \langle \psi(\theta
,\phi)| \left\langle G_{i}\right\rangle _{\rho(\theta)}
\end{array}
\right)  \left(
\begin{array}
[c]{c}
-\frac{1}{4}\left(  e^{-iH(\phi)}\left\{  \Phi_{\frac{\theta}{2}}(G_{j}
),\sqrt{\rho(\theta)}\right\} e^{iH(\phi)} \otimes I\right)  |\Gamma\rangle\\
+\frac{1}{2}\left\langle G_{j}\right\rangle _{\rho(\theta)}|\psi(\theta
,\phi)\rangle
\end{array}
\right)\\
&  = \frac{1}{16} \langle\Gamma| e^{-iH(\phi)} \left\{  \Phi_{\frac{\theta}{2}}(G_{i}
),\sqrt{\rho(\theta)}\right\} e^{iH(\phi)}e^{-iH(\phi)}\left\{  \Phi_{\frac{\theta}{2}}(G_{j}
),\sqrt{\rho(\theta)}\right\} e^{iH(\phi)} \otimes I
|\Gamma\rangle\nonumber\\
&  \qquad -\frac{1}{8} \left\langle G_{j}\right\rangle _{\rho(\theta)} \langle \Gamma | e^{-iH(\phi)} \left\{  \Phi_{\frac{\theta}{2}}(G_{i}
),\sqrt{\rho(\theta)}\right\} e^{iH(\phi)} \otimes I |\psi(\theta
,\phi)\rangle
\nonumber\\
&  \qquad -\frac{1}{8} \left\langle G_{i}\right\rangle _{\rho(\theta)} \langle \psi(\theta
,\phi) | e^{-iH(\phi)}\left\{  \Phi_{\frac{\theta}{2}}(G_{j}
),\sqrt{\rho(\theta)}\right\} e^{iH(\phi)} \otimes I  |\Gamma\rangle
\nonumber\\
&  \qquad + \frac{1}{4} \left\langle G_{i}\right\rangle _{\rho(\theta)} \left\langle G_{j}\right\rangle _{\rho(\theta)} \\
&  = \frac{1}{16} \langle\Gamma| e^{-iH(\phi)} \left\{  \Phi_{\frac{\theta}{2}}(G_{i}
),\sqrt{\rho(\theta)}\right\} \left\{  \Phi_{\frac{\theta}{2}}(G_{j}
),\sqrt{\rho(\theta)}\right\} e^{iH(\phi)} \otimes I
|\Gamma\rangle\nonumber\\
&  \qquad -\frac{1}{8} \left\langle G_{j}\right\rangle _{\rho(\theta)} \langle \Gamma | e^{-iH(\phi)} \left\{  \Phi_{\frac{\theta}{2}}(G_{i}
),\sqrt{\rho(\theta)}\right\} \sqrt{\rho(\theta)} e^{iH(\phi)} \otimes I |\Gamma\rangle
\nonumber\\
&  \qquad -\frac{1}{8} \left\langle G_{i}\right\rangle _{\rho(\theta)} \langle \Gamma| e^{-iH(\phi)} \sqrt{\rho(\theta)} \left\{  \Phi_{\frac{\theta}{2}}(G_{j}
),\sqrt{\rho(\theta)}\right\} e^{iH(\phi)} \otimes I  |\Gamma\rangle
\nonumber\\
&  \qquad + \frac{1}{4} \left\langle G_{i}\right\rangle _{\rho(\theta)} \left\langle G_{j}\right\rangle _{\rho(\theta)} \\
&  =\frac{1}{16}\langle\Gamma| e^{-iH(\phi)} \left\{  \Phi_{\frac{\theta}{2}}(G_{i}
),\sqrt{\rho(\theta)}\right\}  \left\{  \Phi_{\frac{\theta}{2}}(G_{j}
),\sqrt{\rho(\theta)}\right\} e^{iH(\phi)} \otimes I|\Gamma\rangle-\frac{1}{4}\left\langle
G_{i}\right\rangle _{\rho(\theta)}\left\langle G_{j}\right\rangle
_{\rho(\theta)}.
\end{align}
The last equality follows because
\begin{align}
      \langle \Gamma| e^{-iH(\phi)}  \left\{  \Phi_{\frac{\theta}{2}}(G_{j}
),\sqrt{\rho(\theta)}\right\} \sqrt{\rho(\theta)} e^{iH(\phi)} \otimes I  |\Gamma\rangle & =   \operatorname{Tr}\!\left[ e^{-iH(\phi)} \left\{  \Phi_{\frac{\theta}{2}}(G_{j}
),\sqrt{\rho(\theta)}\right\} \sqrt{\rho(\theta)}  e^{iH(\phi)} \right]\\
& =  \operatorname{Tr}\!\left[ \left\{  \Phi_{\frac{\theta}{2}}(G_{j}
),\sqrt{\rho(\theta)}\right\} \sqrt{\rho(\theta)}  \right]\\
& = 2 \operatorname{Tr}\!\left[    \Phi_{\frac{\theta}{2}}(G_{j}
)\rho(\theta)  \right]\\
& = 2 \operatorname{Tr}\!\left[    G_{j}
\rho(\theta)  \right] \\ 
 \langle \Gamma| e^{-iH(\phi)}\sqrt{\rho(\theta)} \left\{  \Phi_{\frac{\theta}{2}}(G_{j}
),\sqrt{\rho(\theta)}\right\} e^{iH(\phi)} \otimes I  |\Gamma\rangle 
& = \operatorname{Tr}\!\left[e^{-iH(\phi)} \sqrt{\rho(\theta)} \left\{  \Phi_{\frac{\theta}{2}}(G_{j}
),\sqrt{\rho(\theta)}\right\} e^{iH(\phi)} \right] \\
& = \operatorname{Tr}\!\left[\sqrt{\rho(\theta)} \left\{  \Phi_{\frac{\theta}{2}}(G_{j}
),\sqrt{\rho(\theta)}\right\} \right] \\
& = 2 \operatorname{Tr}\!\left[  \Phi_{\frac{\theta}{2}}(G_{j}
)\rho(\theta) \right] \\
& = 2 \operatorname{Tr}\!\left[  G_{j}
\rho(\theta) \right].
\end{align}
Also, consider that
\begin{align}
&  \langle\Gamma|e^{-iH(\phi)} \left\{  \Phi_{\frac{\theta}{2}}(G_{i}),\sqrt{\rho(\theta
)}\right\}  \left\{  \Phi_{\frac{\theta}{2}}(G_{j}),\sqrt{\rho(\theta
)}\right\} e^{iH(\phi)} \otimes I |\Gamma\rangle\nonumber\\
&  =\langle\Gamma|e^{-iH(\phi)}\Phi_{\frac{\theta}{2}}(G_{i})\sqrt{\rho(\theta)}\Phi
_{\frac{\theta}{2}}(G_{j})\sqrt{\rho(\theta)}e^{iH(\phi)}\otimes I|\Gamma\rangle\nonumber\\
&\qquad+\langle
\Gamma|e^{-iH(\phi)}\sqrt{\rho(\theta)}\Phi_{\frac{\theta}{2}}(G_{i})\Phi_{\frac{\theta}
{2}}(G_{j})\sqrt{\rho(\theta)}e^{iH(\phi)}\otimes I|\Gamma\rangle\nonumber\\
&  \qquad+\langle\Gamma|e^{-iH(\phi)}\Phi_{\frac{\theta}{2}}(G_{i})\sqrt{\rho(\theta)}
\sqrt{\rho(\theta)}\Phi_{\frac{\theta}{2}}(G_{j})e^{iH(\phi)}\otimes I|\Gamma\rangle\nonumber\\
&\qquad+\langle
\Gamma|e^{-iH(\phi)}\sqrt{\rho(\theta)}\Phi_{\frac{\theta}{2}}(G_{i})\sqrt{\rho(\theta
)}\Phi_{\frac{\theta}{2}}(G_{j})e^{iH(\phi)}\otimes I|\Gamma\rangle\\
&  =\operatorname{Tr}\!\left[  \Phi_{\frac{\theta}{2}}(G_{i})\sqrt{\rho(\theta
)}\Phi_{\frac{\theta}{2}}(G_{j})\sqrt{\rho(\theta)}\right]  +\operatorname{Tr}\!
\left[  \sqrt{\rho(\theta)}\Phi_{\frac{\theta}{2}}(G_{i})\Phi_{\frac{\theta
}{2}}(G_{j})\sqrt{\rho(\theta)}\right]  \nonumber\\
&  \qquad+\operatorname{Tr}\!\left[  \Phi_{\frac{\theta}{2}}(G_{i})\sqrt
{\rho(\theta)}\sqrt{\rho(\theta)}\Phi_{\frac{\theta}{2}}(G_{j})\right]
+\operatorname{Tr}\!\left[  \sqrt{\rho(\theta)}\Phi_{\frac{\theta}{2}}
(G_{i})\sqrt{\rho(\theta)}\Phi_{\frac{\theta}{2}}(G_{j})\right]  \\
&  =2\operatorname{Tr}\!\left[  \Phi_{\frac{\theta}{2}}(G_{i})\sqrt{\rho
(\theta)}\Phi_{\frac{\theta}{2}}(G_{j})\sqrt{\rho(\theta)}\right]
+\left\langle \left\{  \Phi_{\frac{\theta}{2}}(G_{i}),\Phi_{\frac{\theta}{2}
}(G_{j})\right\}  \right\rangle _{\rho(\theta)}.
\end{align}
So we find that
\begin{align}
&  \langle\partial_{i}\psi(\theta,\phi)|\partial_{j}\psi(\theta,\phi)\rangle\nonumber\\
&  =\frac{1}{16}\left[  2\operatorname{Tr}\!\left[  \Phi_{\frac{\theta}{2}
}(G_{i})\sqrt{\rho(\theta)}\Phi_{\frac{\theta}{2}}(G_{j})\sqrt{\rho(\theta
)}\right]  +\left\langle \left\{  \Phi_{\frac{\theta}{2}}(G_{i}),\Phi
_{\frac{\theta}{2}}(G_{j})\right\}  \right\rangle _{\rho(\theta)}\right]
-\frac{1}{4}\left\langle G_{i}\right\rangle _{\rho(\theta
)}\left\langle G_{j}\right\rangle _{\rho(\theta)}\\
&  =\frac{1}{8}\operatorname{Tr}\!\left[  \Phi_{\frac{\theta}{2}}(G_{i}
)\sqrt{\rho(\theta)}\Phi_{\frac{\theta}{2}}(G_{j})\sqrt{\rho(\theta)}\right]
+\frac{1}{16}\left\langle \left\{  \Phi_{\frac{\theta}{2}}(G_{i}),\Phi
_{\frac{\theta}{2}}(G_{j})\right\}  \right\rangle _{\rho(\theta)}-\frac{1}{4}\left\langle G_{i}\right\rangle _{\rho(\theta
)}\left\langle G_{j}\right\rangle _{\rho(\theta)}.
\end{align}
Additionally, consider that
\begin{align}
  \left\langle \psi(\theta,\phi)|\partial_{j}\psi(\theta,\phi)\right\rangle 
&  =\langle\Gamma|\sqrt{\omega(\theta,\phi)} \otimes I \left( -\frac{1}{4}  e^{-iH(\phi)}\left\{  \Phi_{\frac{\theta}{2}}
(G_{j}),\sqrt{\rho(\theta)}\right\}e^{iH(\phi)}  \otimes I  |\Gamma\rangle+\frac{1}{2}\left\langle G_{j}\right\rangle _{\rho(\theta)}
|\psi(\theta,\phi)\rangle\right)\\
&  =-\frac{1}{4}\langle\Gamma|e^{-iH(\phi)}\sqrt{\rho(\theta)}\left\{  \Phi_{\frac{\theta
}{2}}(G_{j}),\sqrt{\rho(\theta)}\right\} e^{iH(\phi)} \otimes I |\Gamma\rangle\notag\\
&\qquad+\frac{1}{2}\left\langle G_{j}\right\rangle _{\rho(\theta)} \langle\Gamma|e^{-iH(\phi)}\sqrt{\rho(\theta)}\sqrt{\rho(\theta
)} e^{iH(\phi)} \otimes I |\Gamma\rangle\\
&  =-\frac{1}{2} \Tr\!\left[ \Phi_{\frac{\theta
}{2}}(G_{j})\rho(\theta) \right] +\frac{1}{2}\left\langle G_{j}\right\rangle _{\rho(\theta)}\\
&  =-\frac{1}{2}\left\langle G_{j}\right\rangle _{\rho(\theta)} +\frac{1}{2}\left\langle G_{j}\right\rangle _{\rho(\theta)}\\
&  =0.
\end{align}
So the final expression for the Wigner--Yanase information is given by
\begin{align}
I_{ij}^{\operatorname{WY}}(\theta) &  =4\operatorname{Re}\!\left[
\frac{1}{8}\operatorname{Tr}\!\left[  \Phi_{\frac{\theta}{2}}(G_{i})\sqrt
{\rho(\theta)}\Phi_{\frac{\theta}{2}}(G_{j})\sqrt{\rho(\theta)}\right]  
+\frac{1}{16}\left\langle \left\{  \Phi_{\frac{\theta}{2}}(G_{i}),\Phi
_{\frac{\theta}{2}}(G_{j})\right\}  \right\rangle _{\rho(\theta)}
-\frac{1}{4}\left\langle G_{i}\right\rangle _{\rho(\theta)}\left\langle
G_{j}\right\rangle _{\rho(\theta)}
\right]  \\
&  =4\left[
\frac{1}{8}\operatorname{Tr}\!\left[  \Phi_{\frac{\theta}{2}}(G_{i})\sqrt
{\rho(\theta)}\Phi_{\frac{\theta}{2}}(G_{j})\sqrt{\rho(\theta)}\right]  
+\frac{1}{16}\left\langle \left\{  \Phi_{\frac{\theta}{2}}(G_{i}),\Phi
_{\frac{\theta}{2}}(G_{j})\right\}  \right\rangle _{\rho(\theta)}
-\frac{1}{4}\left\langle G_{i}\right\rangle _{\rho(\theta)}\left\langle
G_{j}\right\rangle _{\rho(\theta)}
\right]  \\
&  =\frac{1}{2}\operatorname{Tr}\!\left[  \Phi_{\frac{\theta}{2}}(G_{i}
)\sqrt{\rho(\theta)}\Phi_{\frac{\theta}{2}}(G_{j})\sqrt{\rho(\theta)}\right]
+\frac{1}{4}\left\langle \left\{  \Phi_{\frac{\theta}{2}}(G_{i}
),\Phi_{\frac{\theta}{2}}(G_{j})\right\}  \right\rangle _{\rho(\theta
)}-\left\langle G_{i}\right\rangle _{\rho(\theta)}\left\langle G_{j}
\right\rangle _{\rho(\theta)}.
\end{align}
This concludes the proof of Theorem~\ref{thm:WY-theta}.

\subsubsection{Canonical purification of quantum Boltzmann machines and thermofield double state}

\label{app:canon-pure-thermofield-double}

Let us note that the state in~\eqref{eq:canon-pure-QBM} is what is prepared by
various quantum algorithms (see, e.g.,~\cite{Holmes2022quantumalgorithms,chen2023q_Gibbs_sampl,rouze2024efficientthermalization}). Here we show briefly how this is the case. Writing an
eigendecomposition of $G(\theta)$ as $G(\theta)=\sum_{k}g_{k}|\phi_{k}
\rangle\!\langle\phi_{k}|$, the state prepared by thermal-state preparation algorithms is
\begin{equation}
\frac{1}{\sqrt{Z(\theta)}}\sum_{k}e^{-g_{k}/2}|\phi_{k}\rangle\otimes
|\phi_{k}^{\ast}\rangle,
\label{eq:thermal-state-alg-prep}
\end{equation}
where $|\phi_{k}^{\ast}\rangle$ is the complex conjugate of $|\phi_{k}\rangle$ and is defined with respect to the computational basis $\{|k\rangle\}_k$ as $|\phi_{k}^{\ast}\rangle \coloneqq \sum_{k'}\langle\phi_{k}|k'\rangle|k'\rangle$. See, e.g., \cite[Eq.~(37)]{Holmes2022quantumalgorithms}. To see why the state in~\eqref{eq:thermal-state-alg-prep} is equal to the canonical purification in~\eqref{eq:canon-pure-QBM}, consider the following steps:
\begin{align}
|\psi(\theta)\rangle & =\left(  \sqrt{\rho(\theta)}\otimes I\right)
|\Gamma\rangle\\
& =\left(  \sqrt{\rho(\theta)}\otimes I\right)  \sum_{k^{\prime}}|k^{\prime
}\rangle\otimes|k^{\prime}\rangle\\
& =\frac{1}{\sqrt{Z(\theta)}}\left(  \sum_{k}e^{-g_{k}/2}|\phi_{k}
\rangle\langle\phi_{k}|\otimes I\right)  \sum_{k^{\prime}}|k^{\prime}
\rangle\otimes|k^{\prime}\rangle\\
& =\frac{1}{\sqrt{Z(\theta)}}\sum_{k}e^{-g_{k}/2}|\phi_{k}\rangle\otimes
\sum_{k^{\prime}}\langle\phi_{k}|k^{\prime}\rangle|k^{\prime}\rangle\\
& =\frac{1}{\sqrt{Z(\theta)}}\sum_{k}e^{-g_{k}/2}|\phi_{k}\rangle\otimes
|\phi_{k}^{\ast}\rangle.
\end{align}

\subsubsection{Proof of Theorem~\ref{thm:WY-phi}}

\label{proof:WY-phi}

Here we employ the shorthands
$\partial_{i}\equiv\frac{\partial}{\partial\phi_{i}}$ and $\partial_{j}\equiv\frac{\partial}{\partial\phi_{j}}$. Using Proposition~\ref{prop:FB-WY-canonical-purifications} and~\eqref{eq:part_der_phi_purified-evolved-QBM}, we find that
\begin{align}
\langle\partial_{i}\psi(\theta,\phi)|\partial_{j}\psi(\theta,\phi)\rangle &  =
\langle\Gamma|\left(  i\left[
\sqrt{\omega(\theta,\phi)}, \Psi^{\dagger}_{\phi}(H_{i}) \right] \otimes I\right)   \left(  i\left[
\sqrt{\omega(\theta,\phi)}, \Psi^{\dagger}_{\phi}(H_{j}) \right] \otimes I\right)   |\Gamma\rangle  \\
& = -\langle\Gamma| 
\sqrt{\omega(\theta,\phi)} \Psi^{\dagger}_{\phi}(H_{i})  \sqrt{\omega(\theta,\phi)} \Psi^{\dagger}_{\phi}(H_{j}) \otimes I  |\Gamma\rangle\\
& \quad +  \langle\Gamma| 
\sqrt{\omega(\theta,\phi)} \Psi^{\dagger}_{\phi}(H_{i})  \Psi^{\dagger}_{\phi}(H_{j}) \sqrt{\omega(\theta,\phi)}  \otimes I  |\Gamma\rangle\\
& \quad +  \langle\Gamma| 
\Psi^{\dagger}_{\phi}(H_{i}) \sqrt{\omega(\theta,\phi)}   \sqrt{\omega(\theta,\phi)} \Psi^{\dagger}_{\phi}(H_{j}) \otimes I  |\Gamma\rangle\\
& \quad - \langle\Gamma| 
\Psi^{\dagger}_{\phi}(H_{i}) \sqrt{\omega(\theta,\phi)}   \Psi^{\dagger}_{\phi}(H_{j}) \sqrt{\omega(\theta,\phi)}  \otimes I  |\Gamma\rangle\\
& = - \Tr\!\left[ 
\sqrt{\omega(\theta,\phi)} \Psi^{\dagger}_{\phi}(H_{i})  \sqrt{\omega(\theta,\phi)} \Psi^{\dagger}_{\phi}(H_{j}) \right]\\
& \quad +  \Tr\!\left[ 
\sqrt{\omega(\theta,\phi)} \Psi^{\dagger}_{\phi}(H_{i})  \Psi^{\dagger}_{\phi}(H_{j}) \sqrt{\omega(\theta,\phi)}  \right]\\
& \quad + \Tr\!\left[
\Psi^{\dagger}_{\phi}(H_{i}) \sqrt{\omega(\theta,\phi)}   \sqrt{\omega(\theta,\phi)} \Psi^{\dagger}_{\phi}(H_{j}) \right]\\
& \quad - \Tr\!\left[
\Psi^{\dagger}_{\phi}(H_{i}) \sqrt{\omega(\theta,\phi)}   \Psi^{\dagger}_{\phi}(H_{j}) \sqrt{\omega(\theta,\phi)}  \right]\\
& = - \Tr\!\left[ 
e^{-iH(\phi)} \sqrt{\rho(\theta)} e^{iH(\phi)} \Psi^{\dagger}_{\phi}(H_{i}) e^{-iH(\phi)} \sqrt{\rho(\theta)} e^{iH(\phi)} \Psi^{\dagger}_{\phi}(H_{j}) \right]\\
& \quad +  \Tr\!\left[ 
 \Psi^{\dagger}_{\phi}(H_{i})  \Psi^{\dagger}_{\phi}(H_{j}) e^{-iH(\phi)} \rho(\theta) e^{iH(\phi)} \right]\\
& \quad + \Tr\!\left[
\Psi^{\dagger}_{\phi}(H_{i}) e^{-iH(\phi)} \rho(\theta) e^{iH(\phi)}\Psi^{\dagger}_{\phi}(H_{j}) \right]\\
& \quad - \Tr\!\left[
\Psi^{\dagger}_{\phi}(H_{i}) e^{-iH(\phi)} \sqrt{\rho(\theta)} e^{iH(\phi)}   \Psi^{\dagger}_{\phi}(H_{j}) e^{-iH(\phi)}\sqrt{\rho(\theta)} e^{iH(\phi)} \right]\\
& = - 2 \Tr\!\left[ \Psi_{\phi}(H_{j}) \sqrt{\rho(\theta)} \Psi_{\phi}(H_{i}) \sqrt{\rho(\theta)} \right]\\
& \quad +  \Tr\!\left[ 
 \left\{ \Psi_{\phi}(H_{i}) , \Psi_{\phi}(H_{j}) \right\} \rho(\theta)  \right]\\
 & = - 2 \Tr\!\left[ \Psi_{\phi}(H_{j}) \sqrt{\rho(\theta)} \Psi_{\phi}(H_{i}) \sqrt{\rho(\theta)} \right]\\
& \quad + \left\langle 
 \left\{ \Psi_{\phi}(H_{i}) , \Psi_{\phi}(H_{j}) \right\}   \right\rangle_{\rho(\theta)}.
\end{align}
Additionally, consider that 
\begin{align}
 \left\langle
\psi(\theta,\phi)|\partial_{j}\psi(\theta,\phi)\right\rangle\notag 
& = \langle\Gamma|\left( \sqrt{\omega(\theta,\phi)} \otimes I\right)  \left( i \left[
\sqrt{\omega(\theta,\phi)}, \Psi^{\dagger}_{\phi}(H_{j}) \right] \otimes I\right)   |\Gamma \rangle\\
& = i  \langle\Gamma| \sqrt{\omega(\theta,\phi)}
\sqrt{\omega(\theta,\phi)} \Psi^{\dagger}_{\phi}(H_{j}) \otimes I   |\Gamma\rangle\\
& \hspace{0.4cm} - i \langle\Gamma| \sqrt{\omega(\theta,\phi)}
\Psi^{\dagger}_{\phi}(H_{j}) \sqrt{\omega(\theta,\phi)}  \otimes I  |\Gamma\rangle\\
& = i \Tr\!\left[ \Psi_{\phi}(H_{j}) \rho(\theta) \right] - i \Tr\!\left[ \Psi_{\phi}(H_{j}) \rho(\theta) \right]\\
& = 0.
\end{align}
Then it follows that
\begin{align}
I_{ij}^{\operatorname{WY}} (\phi) &=   4\operatorname{Re}\!\left[  \langle\partial_{i}\psi(\theta,\phi)|\partial_{j}
\psi(\theta,\phi)\rangle-\left\langle \partial_{i}\psi(\theta,\phi)|\psi(\theta,\phi)\right\rangle
\left\langle \psi(\theta,\phi)|\partial_{j}\psi(\theta,\phi)\right\rangle \right]
\\
&  =4\operatorname{Re}\!\left[ - 2 \Tr\!\left[ \Psi_{\phi}(H_{j}) \sqrt{\rho(\theta)} \Psi_{\phi}(H_{i}) \sqrt{\rho(\theta)} \right] + \left\langle 
 \left\{ \Psi_{\phi}(H_{i}) , \Psi_{\phi}(H_{j}) \right\}   \right\rangle_{\rho(\theta)} \right] \\
&  = - 8 \Tr\!\left[ \Psi_{\phi}(H_{j}) \sqrt{\rho(\theta)} \Psi_{\phi}(H_{i}) \sqrt{\rho(\theta)} \right] + 4 \left\langle 
 \left\{ \Psi_{\phi}(H_{i}) , \Psi_{\phi}(H_{j}) \right\}   \right\rangle_{\rho(\theta)}.
\end{align}
This concludes the proof of Theorem~\ref{thm:WY-phi}.

\subsubsection{Proof of Theorem~\ref{thm:WY-theta-phi}}

\label{proof:WY-theta-phi}

Here we employ the shorthands
$\partial_{i}\equiv\frac{\partial}{\partial\phi_{i}}$ and $\partial_{j}\equiv\frac{\partial}{\partial\theta_{j}}$.
Using Proposition~\ref{prop:FB-WY-canonical-purifications} and Theorem~\ref{thm:WY-partial-derivatives}, we find that
\begin{align}
& \langle\partial_{i}\psi(\theta,\phi)|\partial_{j}\psi(\theta,\phi
)\rangle\nonumber\\
& =\langle\Gamma|\left( i \left[
\sqrt{\omega(\theta,\phi)}, \Psi^{\dagger}_{\phi}(H_{i}) \right] \otimes I\right)    \left(
\begin{array}
[c]{c}
-\frac{1}{4}\left(  e^{-iH(\phi)}\left\{  \Phi_{\frac{\theta}{2}}(G_{j}
),\sqrt{\rho(\theta)}\right\} e^{iH(\phi)} \otimes I\right)  |\Gamma\rangle\\
+\frac{1}{2}\left\langle G_{j}\right\rangle _{\rho(\theta)}|\psi(\theta
,\phi)\rangle
\end{array}
\right)  \\
& =-\frac{i}{4}\langle\Gamma|\left(   \left[
\sqrt{\omega(\theta,\phi)}, \Psi^{\dagger}_{\phi}(H_{i}) \right]
e^{-iH(\phi)}\left\{  \Phi_{\frac{\theta}{2}}(G_{j}),\sqrt{\rho(\theta
)}\right\} e^{iH(\phi)} \otimes I\right)  |\Gamma\rangle\nonumber\\
& \qquad +\frac{i}{2}\left\langle G_{j}\right\rangle _{\rho(\theta)} \langle\Gamma|  \left[
\sqrt{\omega(\theta,\phi)}, \Psi^{\dagger}_{\phi}(H_{i}) \right] \otimes I  |\psi
(\theta,\phi)\rangle\\
& =-\frac{i}{4}\langle\Gamma|\left(  
\sqrt{\omega(\theta,\phi)} \Psi^{\dagger}_{\phi}(H_{i}) 
e^{-iH(\phi)} \Phi_{\frac{\theta}{2}}(G_{j}) \sqrt{\rho(\theta
)} e^{iH(\phi)} \otimes I\right)  |\Gamma\rangle\nonumber\\
& \qquad - \frac{i}{4}\langle\Gamma|\left(  
\sqrt{\omega(\theta,\phi)} \Psi^{\dagger}_{\phi}(H_{i}) 
e^{-iH(\phi)} \sqrt{\rho(\theta)} \Phi_{\frac{\theta}{2}}(G_{j})
 e^{iH(\phi)} \otimes I\right)  |\Gamma\rangle\nonumber\\
 & \qquad + \frac{i}{4}\langle\Gamma|\left(  
 \Psi^{\dagger}_{\phi}(H_{i}) \sqrt{\omega(\theta,\phi)}
e^{-iH(\phi)} \Phi_{\frac{\theta}{2}}(G_{j}) \sqrt{\rho(\theta)}
 e^{iH(\phi)} \otimes I\right)  |\Gamma\rangle\nonumber\\
 & \qquad + \frac{i}{4}\langle\Gamma|\left(  
 \Psi^{\dagger}_{\phi}(H_{i}) \sqrt{\omega(\theta,\phi)}
e^{-iH(\phi)} \sqrt{\rho(\theta)} \Phi_{\frac{\theta}{2}}(G_{j})
 e^{iH(\phi)} \otimes I\right)  |\Gamma\rangle\nonumber\\
& \qquad +\frac{i}{2}\left\langle G_{j}\right\rangle _{\rho(\theta)} \langle\Gamma|  
\sqrt{\omega(\theta,\phi)} \Psi^{\dagger}_{\phi}(H_{i}) \sqrt{\omega(\theta,\phi)}  \otimes I  |\Gamma\rangle\nonumber\\
& \qquad - \frac{i}{2}\left\langle G_{j}\right\rangle _{\rho(\theta)} \langle\Gamma|  \Psi^{\dagger}_{\phi}(H_{i}) 
\sqrt{\omega(\theta,\phi)} \sqrt{\omega(\theta,\phi)} \otimes I  |\Gamma\rangle\\
& =-\frac{i}{4} \Tr\!\left[  
\sqrt{\rho(\theta)} \Psi_{\phi}(H_{i})  \Phi_{\frac{\theta}{2}}(G_{j}) \sqrt{\rho(\theta
)} \right]\nonumber\\
& \qquad - \frac{i}{4} \Tr\!\left[  
\sqrt{\rho(\theta)} \Psi_{\phi}(H_{i})  \sqrt{\rho(\theta)} \Phi_{\frac{\theta}{2}}(G_{j})\right] \nonumber\\
& \qquad + \frac{i}{4}\Tr\!\left[  
 \Psi_{\phi}(H_{i}) \sqrt{\rho(\theta)}
\Phi_{\frac{\theta}{2}}(G_{j}) \sqrt{\rho(\theta)} \right]\nonumber\\
& \qquad + \frac{i}{4}\Tr\!\left[  
 \Psi_{\phi}(H_{i}) \sqrt{\rho(\theta} \sqrt{\rho(\theta)} \Phi_{\frac{\theta}{2}}(G_{j}) \right]\nonumber\\
& \qquad +\frac{i}{2}\left\langle G_{j}\right\rangle _{\rho(\theta)} \Tr\!\left[ \Psi_{\phi}(H_{i}) \rho(\theta) \right] \nonumber\\
& \qquad - \frac{i}{2}\left\langle G_{j}\right\rangle _{\rho(\theta)} \Tr\!\left[ \Psi_{\phi}(H_{i}) \rho(\theta) \right] \\
& = \frac{i}{4}\Tr\!\left[  \left[ \Phi_{\frac{\theta}{2}}(G_{j}) ,
\Psi_{\phi}(H_{i}) \right] \rho(\theta) \right]\\
& = \frac{i}{4} \left\langle  \left[ \Phi_{\frac{\theta}{2}}(G_{j}) ,
\Psi_{\phi}(H_{i}) \right] \right\rangle_{\rho(\theta)}.
\end{align}
Additionally, consider that
\begin{align}
  \left\langle \psi(\theta,\phi)|\partial_{j}\psi(\theta,\phi)\right\rangle 
&  =\langle\Gamma| \sqrt{\omega(\theta,\phi)} \otimes I \left( -\frac{1}{4}  e^{-iH(\phi)}\left\{  \Phi_{\frac{\theta}{2}}
(G_{j}),\sqrt{\rho(\theta)}\right\}e^{iH(\phi)}  \otimes I  |\Gamma\rangle+\frac{1}{2}\left\langle G_{j}\right\rangle _{\rho(\theta)}
|\psi(\theta,\phi)\rangle\right)\\
&  =-\frac{1}{4}\langle\Gamma|e^{-iH(\phi)}\sqrt{\rho(\theta)}\left\{  \Phi_{\frac{\theta
}{2}}(G_{j}),\sqrt{\rho(\theta)}\right\} e^{iH(\phi)} \otimes I |\Gamma\rangle\notag\\
&\qquad+\frac{1}{2}\left\langle G_{j}\right\rangle _{\rho(\theta)} \langle\Gamma|e^{-iH(\phi)}\sqrt{\rho(\theta)}\sqrt{\rho(\theta
)} e^{iH(\phi)} \otimes I |\Gamma\rangle\\
&  =-\frac{1}{2} \Tr\!\left[ \Phi_{\frac{\theta
}{2}}(G_{j})\rho(\theta) \right] +\frac{1}{2}\left\langle G_{j}\right\rangle _{\rho(\theta)}\\
&  =-\frac{1}{2}\left\langle G_{j}\right\rangle _{\rho(\theta)} +\frac{1}{2}\left\langle G_{j}\right\rangle _{\rho(\theta)}\\
&  =0.
\end{align}
So this means that
\begin{align}
 I_{ij}^{\operatorname{WY}}(\theta,\phi)
& = 4\operatorname{Re}\!\left[  \langle\partial_{i}\psi(\theta
,\phi)|\partial_{j}\psi(\theta,\phi)\rangle-\left\langle \partial_{i}
\psi(\theta,\phi)|\psi(\theta,\phi)\right\rangle \left\langle \psi(\theta
,\phi)|\partial_{j}\psi(\theta,\phi)\right\rangle \right]  \\
& =4\operatorname{Re}\!\left[
\frac{i}{4} \left\langle  \left[ \Phi_{\frac{\theta}{2}}(G_{j}) ,
\Psi_{\phi}(H_{i}) \right] \right\rangle_{\rho(\theta)} \right] \label{eq:step-with-imag-numbers} \\
& = i\left\langle \left[   \Phi_{\frac{\theta}{2}}(G_{j}), \Psi_{\phi
}(H_{i})\right]  \right\rangle _{\rho(\theta)}.
\end{align}
This concludes the proof of Theorem~\ref{thm:WY-theta-phi}.

\subsection{Kubo--Mori information matrix elements for evolved quantum Boltzmann machines}

\subsubsection{Proof of Theorem~\ref{thm:KM-theta}}

\label{proof:KM-theta}

Using~\eqref{eq:Fisher-info-help-alt}, consider that
\begin{align}
  I_{ij}^{\text{KM}}(\theta)
&  =\sum_{k,\ell}\frac{\ln\lambda_{k}-\ln\lambda_{\ell}}{\lambda_{k}
-\lambda_{\ell}}\langle k|\left[  \frac{\partial}{\partial\theta_{i}}
\omega(\theta,\phi)\right]  |\ell\rangle\!\langle\ell|\left[  \frac{\partial
}{\partial\theta_{j}}\omega(\theta,\phi)\right]  |k\rangle\\
& = \sum_{k,\ell}\frac{\ln\lambda_{k}-\ln\lambda_{\ell}}{\lambda_{k}
-\lambda_{\ell}} \left( -\frac{1}{2}\left( \lambda_\ell + \lambda_k\right) \bra{\tilde{k}} \Phi_{\theta}(G_{i})\ket{\tilde{\ell}} + \delta_{k\ell} \lambda_\ell \left\langle G_{i}\right\rangle_{\rho(\theta)} \right) \nonumber\\
& \hspace{6cm} \times\left( -\frac{1}{2}\left( \lambda_\ell + \lambda_k\right) \bra{\tilde{\ell}} \Phi_{\theta}(G_{j})\ket{\tilde{k}} + \delta_{\ell k} \lambda_k \left\langle G_{j}\right\rangle _{\rho(\theta)}\right)\\
& = \sum_{k,\ell} \frac{1}{4} \frac{\ln\lambda_{k}-\ln\lambda_{\ell}}{\lambda_{k}-\lambda_{\ell}} (\lambda_\ell+\lambda_k)^2 \bra{\tilde{k}} \Phi_{\theta}(G_{i})\ket{\tilde{\ell}} \bra{\tilde{\ell}} \Phi_{\theta}(G_{j})\ket{\tilde{k}}\nonumber\\
& \hspace{1cm} + \sum_{k,\ell} - \frac{1}{2} \frac{\ln\lambda_{k}-\ln\lambda_{\ell}}{\lambda_{k}-\lambda_{\ell}} (\lambda_\ell+\lambda_k) \delta_{k\ell}\lambda_\ell \bra{\tilde{k}} \Phi_{\theta}(G_{i})\ket{\tilde{\ell}}\left\langle G_{j}\right\rangle _{\rho(\theta)}\nonumber\\
& \hspace{1cm} + \sum_{k,\ell} - \frac{1}{2} \frac{\ln\lambda_{k}-\ln\lambda_{\ell}}{\lambda_{k}-\lambda_{\ell}} (\lambda_\ell+\lambda_k) \delta_{k\ell}\lambda_\ell \bra{\tilde{k}}\Phi_{\theta}(G_{j})\ket{\tilde{\ell}}\left\langle G_{i}\right\rangle _{\rho(\theta)}\nonumber\\
& \hspace{1cm} + \sum_{k,\ell} \frac{\ln\lambda_{k}-\ln\lambda_{\ell}}{\lambda_{k}-\lambda_{\ell}} \delta_{k\ell}\lambda_\ell \left\langle G_{i}\right\rangle _{\rho(\theta)} \delta_{k\ell}\lambda_\ell \left\langle G_{j}\right\rangle _{\rho(\theta)}\\
& = \sum_{k,\ell} \frac{1}{4} \frac{(\ln\lambda_{k}-\ln\lambda_{\ell})(\lambda_\ell+\lambda_k)^2}{\lambda_{k}-\lambda_{\ell}}  \bra{\tilde{k}} \Phi_{\theta}(G_{i})\ket{\tilde{\ell}} \bra{\tilde{\ell}} \Phi_{\theta}(G_{j})\ket{\tilde{k}} - \sum_{k} \lambda_k \bra{\tilde{k}} \Phi_{\theta}(G_{i})\ket{\tilde{k}}\left\langle G_{j}\right\rangle _{\rho(\theta)}\nonumber\\
& \hspace{1cm} - \sum_{k} \lambda_k \bra{\tilde{k}} \Phi_{\theta}(G_{j})\ket{\tilde{k}}\left\langle G_{i}\right\rangle _{\rho(\theta)} + \sum_k \lambda_k \left\langle G_{i}\right\rangle _{\rho(\theta)} \left\langle G_{j}\right\rangle _{\rho(\theta)}\label{lim_log_formula}\\
& = \sum_{k,\ell} \frac{1}{4} \frac{(\ln\lambda_{k}-\ln\lambda_{\ell})(\lambda_\ell+\lambda_k)^2}{\lambda_{k}-\lambda_{\ell}}  \bra{\tilde{k}} \Phi_{\theta}(G_{i})\ket{\tilde{\ell}} \bra{\tilde{\ell}} \Phi_{\theta}(G_{j})\ket{\tilde{k}} - \Tr\!\left[ \rho(\theta) \Phi_{\theta}(G_{i}) \right]\left\langle G_{j}\right\rangle _{\rho(\theta)} \nonumber\\
& \hspace{1cm} - \Tr\!\left[ \rho(\theta) \Phi_{\theta}(G_{j})  \right]\left\langle G_{i}\right\rangle _{\rho(\theta)} + \left\langle G_{i}\right\rangle _{\rho(\theta)} \left\langle G_{j}\right\rangle _{\rho(\theta)}\\
& = \sum_{k,\ell} \frac{1}{4} \frac{(\ln\lambda_{k}-\ln\lambda_{\ell})(\lambda_\ell+\lambda_k)^2}{\lambda_{k}-\lambda_{\ell}}  \bra{\tilde{k}} \Phi_{\theta}(G_{i})\ket{\tilde{\ell}} \bra{\tilde{\ell}} \Phi_{\theta}(G_{j})\ket{k} - \left\langle G_{i}\right\rangle _{\rho(\theta)} \left\langle G_{j}\right\rangle _{\rho(\theta)}.\label{last_expression}
\end{align}
The~\eqref{lim_log_formula} equality is a consequence of the following fact:
\begin{equation}
    \lim_{x\to y} \frac{\ln x - \ln y}{x - y}=\frac{1}{y}.
\end{equation}
Now let us focus on the first term of~\eqref{last_expression}. Let a spectral decomposition of $G(\theta)$ be given by $G(\theta)=\sum_k \mu_k |\tilde{k}\rangle\!\langle \tilde{k}|$. This implies that for all $k$, the eigenvalues of $\omega(\theta,\phi)$, and so the eigenvalues of $\rho(\theta)$, are $\lambda_k=\frac{e^{-\mu_k}}{Z}$, where $Z$ is the partition function. Plugging this into the first term of~\eqref{last_expression}, we obtain:
\begin{align}
    &\sum_{k,\ell} \frac{1}{4} \frac{(\ln\lambda_{k}-\ln\lambda_{\ell})(\lambda_k+\lambda_\ell)^2}{\lambda_{k}-\lambda_{\ell}}  \bra{\tilde{k}} \Phi_{\theta}(G_{i})\ket{\tilde{\ell}} \bra{\tilde{\ell}} \Phi_{\theta}(G_{j})\ket{\tilde{k}}\nonumber\\
    & = \sum_{k,\ell} \frac{1}{4} \frac{(\ln \frac{e^{-\mu_k}}{Z} - \ln\frac{e^{-\mu_\ell}}{Z})(\frac{e^{-\mu_k}}{Z}+\frac{e^{-\mu_\ell}}{Z})^2}{\frac{e^{-\mu_k}}{Z}-\frac{e^{-\mu_\ell}}{Z}}  \bra{\tilde{k}} \Phi_{\theta}(G_{i})\ket{\tilde{\ell}} \bra{\tilde{\ell}} \Phi_{\theta}(G_{j})\ket{\tilde{k}}\\
    & = \sum_{k,\ell} \frac{1}{4Z} \frac{(-\mu_k + \mu_\ell)(e^{-\mu_k}+e^{-\mu_\ell})^2}{e^{-\mu_k}-e^{-\mu_\ell}}  \bra{\tilde{k}} \Phi_{\theta}(G_{i})\ket{\tilde{\ell}} \bra{\tilde{\ell}} \Phi_{\theta}(G_{j})\ket{\tilde{k}}\\
    & = \sum_{k,\ell} \frac{1}{4Z} \frac{(-\mu_k + \mu_\ell)(e^{-\mu_k}+e^{-\mu_\ell})}{\frac{e^{-\mu_k}-e^{-\mu_\ell}}{e^{-\mu_k}+e^{-\mu_\ell}}}  \bra{\tilde{k}} \Phi_{\theta}(G_{i})\ket{\tilde{\ell}} \bra{\tilde{\ell}} \Phi_{\theta}(G_{j})\ket{\tilde{k}}\\
    & = \sum_{k,\ell} \frac{1}{4Z} \frac{(-\mu_k + \mu_\ell)(e^{-\mu_k}+e^{-\mu_\ell})}{\frac{1-e^{\mu_k-\mu_\ell}}{1+ e^{\mu_k-\mu_\ell}}}  \bra{\tilde{k}} \Phi_{\theta}(G_{i})\ket{\tilde{\ell}} \bra{\tilde{\ell}} \Phi_{\theta}(G_{j})\ket{\tilde{k}}\\
    & = \sum_{k,\ell} \frac{1}{4Z} \frac{(-\mu_k + \mu_\ell)(e^{-\mu_k}+e^{-\mu_\ell})}{-\tanh\left( \frac{\mu_k-\mu_\ell}{2} \right)}  \bra{\tilde{k}} \Phi_{\theta}(G_{i})\ket{\tilde{\ell}} \bra{\tilde{\ell}} \Phi_{\theta}(G_{j})\ket{\tilde{k}}\\
    & = \sum_{k,\ell} - \frac{1}{4Z} \frac{(-\mu_k + \mu_\ell)(e^{-\mu_k}+e^{-\mu_\ell})}{\tanh\left( \frac{\mu_k-\mu_\ell}{2} \right)}  \bra{\tilde{k}} \int_{\mathbb{R}}dt\ p(t)\ e^{-iG(\theta)t} G_{i} e^{iG(\theta)t}  \ket{\tilde{\ell}} \bra{\tilde{\ell}} \Phi_{\theta}(G_{j})\ket{\tilde{k}}\\
    & = \sum_{k,\ell} -\frac{1}{4Z} \frac{(-\mu_k + \mu_\ell)(e^{-\mu_k}+e^{-\mu_\ell})}{\tanh\left( \frac{\mu_k-\mu_\ell}{2} \right)} \bra{\tilde{k}} \int_{\mathbb{R}}dt\ p(t)\ \left( \sum_m |\tilde{m} \rangle \! \langle \tilde{m}|e^{-i\mu_m t} \right) G_{i} \left( \sum_n |\tilde{n} \rangle \! \langle \tilde{n}| e^{i\mu_n t} \right)  \ket{\tilde{\ell}} \bra{\tilde{\ell}} \Phi_{\theta}(G_{j})\ket{\tilde{k}}\\
    & = \sum_{k,\ell} -\frac{1}{4Z} \frac{(-\mu_k + \mu_\ell)(e^{-\mu_k}+e^{-\mu_\ell})}{\tanh\left( \frac{\mu_k-\mu_\ell}{2} \right)} \int_{\mathbb{R}}dt\ p(t)\ e^{-i\mu_k t} \bra{\tilde{k}}G_{i}\ket{\tilde{\ell}} e^{i\mu_\ell t}  \bra{\tilde{\ell}} \Phi_{\theta}(G_{j})\ket{\tilde{k}}\\
    & = \sum_{k,\ell} -\frac{1}{4Z} \frac{(-\mu_k + \mu_\ell)(e^{-\mu_k}+e^{-\mu_\ell})}{\tanh\left( \frac{\mu_k-\mu_\ell}{2} \right)} \int_{\mathbb{R}}dt\ p(t)\ e^{-i(\mu_k-\mu_\ell) t} \bra{\tilde{k}}G_{i}\ket{\tilde{\ell}}  \bra{\tilde{\ell}} \Phi_{\theta}(G_{j})\ket{\tilde{k}}\\
     & = \sum_{k,\ell} -\frac{1}{4Z} \frac{(-\mu_k + \mu_\ell)(e^{-\mu_k}+e^{-\mu_\ell})}{\tanh\left( \frac{\mu_k-\mu_\ell}{2} \right)}  \frac{\tanh(\frac{\mu_k-\mu_\ell}{2})}{\frac{\mu_k-\mu_\ell}{2}} \bra{\tilde{k}}G_{i}\ket{\tilde{\ell}}  \bra{\tilde{\ell}} \Phi_{\theta}(G_{j})\ket{\tilde{k}} \label{eq:apply-fourier-trans-p-t}\\
     & = \sum_{k,\ell} \frac{1}{2Z} (e^{-\mu_k}+e^{-\mu_\ell}) \bra{\tilde{k}}G_{i}\ket{\tilde{\ell}}  \bra{\tilde{\ell}} \Phi_{\theta}(G_{j})\ket{\tilde{k}}\\
     & = \sum_{k,\ell} \frac{1}{2} \left(\frac{e^{-\mu_k}}{Z}+\frac{e^{-\mu_\ell}}{Z}\right) \bra{\tilde{k}}G_{i}\ket{\tilde{\ell}}  \bra{\tilde{\ell}} \Phi_{\theta}(G_{j})\ket{\tilde{k}}\\
     & = \sum_{k,\ell} \frac{1}{2} \left(\lambda_k+\lambda_\ell \right) \bra{\tilde{k}}G_{i}\ket{\tilde{\ell}}  \bra{\tilde{\ell}} \Phi_{\theta}(G_{j})\ket{\tilde{k}}\\
     & = \frac{1}{2}\Tr\!\left[ \Phi_{\theta}(G_{j})\rho(\theta) G_{i} \right] + \frac{1}{2}\Tr\!\left[ G_{i} \rho(\theta) \Phi_{\theta}(G_{j}) \right]\\
     & = \frac{1}{2} \Tr\!\left[ \left\{ G_i, \Phi_{\theta}(G_{j}) \right\} \rho(\theta) \right]\\
     & = \frac{1}{2} \left\langle \left\{ G_i, \Phi_{\theta}(G_{j}) \right\}\right\rangle_{\rho(\theta)}.\label{end_KMI_theta}
\end{align}
The equality in~\eqref{eq:apply-fourier-trans-p-t} follows from \cite[Lemma~12]{patel2024quantumboltzmannmachine}.
When combining~\eqref{end_KMI_theta} with~\eqref{last_expression}, the proof is concluded.

\subsubsection{Proof of Theorem~\ref{thm:KM-phi}}

\label{proof:KM-phi}

Using~\eqref{eq:Fisher-info-help-last}, consider
that
\begin{align}
 I_{ij}^{\text{KM}}(\phi)
&  =\sum_{k,\ell}\frac{\ln\lambda_{k}-\ln\lambda_{\ell}}{\lambda_{k}
-\lambda_{\ell}}\langle k|\left[  \frac{\partial}{\partial\phi_{i}}
\omega(\theta,\phi)\right]  |\ell\rangle\!\langle\ell|\left[  \frac{\partial
}{\partial\phi_{j}}\omega(\theta,\phi)\right]  |k\rangle\\
&  =\sum_{k,\ell}\frac{\ln\lambda_{k}-\ln\lambda_{\ell}}{\lambda_{k}
-\lambda_{\ell}}\left[  i\left(  \lambda_{k}-\lambda_{\ell}\right)  \langle
k|\Psi^{\dagger}_{\phi}(H_{i})|\ell\rangle\right]  \left[  i\left(  \lambda_{\ell
}-\lambda_{k}\right)  \langle\ell|\Psi^{\dagger}_{\phi}(H_{j})|k\rangle\right]  \\
&  =\sum_{k,\ell}\frac{\ln\lambda_{k}-\ln\lambda_{\ell}}{\lambda_{k}
-\lambda_{\ell}}\left(  \lambda_{k}-\lambda_{\ell}\right)  ^{2}\langle
k|\Psi^{\dagger}_{\phi}(H_{i})|\ell\rangle\!\langle\ell|\Psi^{\dagger}_{\phi}(H_{j})|k\rangle\\
&  =\sum_{k,\ell}\left(  \ln\lambda_{k}-\ln\lambda_{\ell}\right)  \left(
\lambda_{k}-\lambda_{\ell}\right)  \langle k|\Psi^{\dagger}_{\phi}(H_{i})|\ell
\rangle\!\langle\ell|\Psi^{\dagger}_{\phi}(H_{j})|k\rangle\\
&  =\sum_{k,\ell}\lambda_{k}\ln\lambda_{k}\langle k|\Psi^\dag_{\phi}(H_{i}
)|\ell\rangle\!\langle\ell|\Psi^{\dagger}_{\phi}(H_{j})|k\rangle\nonumber\\
&  \qquad-\sum_{k,\ell}\lambda_{\ell}\ln\lambda_{k}\langle k|\Psi^\dag_{\phi}
(H_{i})|\ell\rangle\!\langle\ell|\Psi^{\dagger}_{\phi}(H_{j})|k\rangle\nonumber\\
&  \qquad-\sum_{k,\ell}\lambda_{k}\ln\lambda_{\ell}\langle k|\Psi^\dag_{\phi}
(H_{i})|\ell\rangle\!\langle\ell|\Psi^{\dagger}_{\phi}(H_{j})|k\rangle\nonumber\\
&  \qquad+\sum_{k,\ell}\lambda_{\ell}\ln\lambda_{\ell}\langle k|\Psi^\dag_{\phi
}(H_{i})|\ell\rangle\!\langle\ell|\Psi^{\dagger}_{\phi}(H_{j})|k\rangle\\
&  =\operatorname{Tr}[\left(  \omega(\theta,\phi)\ln\omega(\theta
,\phi)\right)  \Psi^{\dagger}_{\phi}(H_{i})\Psi^{\dagger}_{\phi}(H_{j})]\nonumber\\
&  \qquad-\operatorname{Tr}[\left(  \ln\omega(\theta,\phi)\right)  \Psi^\dag_{\phi
}(H_{i})\omega(\theta,\phi)\Psi^{\dagger}_{\phi}(H_{j})]\nonumber\\
&  \qquad-\operatorname{Tr}[\omega(\theta,\phi)\Psi^{\dagger}_{\phi}(H_{i})\left(
\ln\omega(\theta,\phi)\right)  \Psi^{\dagger}_{\phi}(H_{j})]\nonumber\\
&  \qquad+\operatorname{Tr}[\Psi^{\dagger}_{\phi}(H_{i})\left(  \omega(\theta,\phi
)\ln\omega(\theta,\phi)\right)  \Psi^{\dagger}_{\phi}(H_{j})]\\
&  =\operatorname{Tr}[\left\{  \Psi^{\dagger}_{\phi}(H_{i}),\Psi^{\dagger}_{\phi}(H_{j})\right\}
\left(  \omega(\theta,\phi)\ln\omega(\theta,\phi)\right)  ]\nonumber\\
&  \qquad-2\operatorname{Re}\!\left[  \operatorname{Tr}[\left(  \ln\omega
(\theta,\phi)\right)  \Psi^{\dagger}_{\phi}(H_{i})\omega(\theta,\phi)\Psi^\dag_{\phi}
(H_{j})]\right]  .
\end{align}
Now consider that
\begin{align}
\ln\omega(\theta,\phi) &  =\ln\!\left[  e^{-iH(\phi)}\frac{e^{-G(\theta)}
}{Z(\theta)}e^{iH(\phi)}\right]  \\
&  =e^{-iH(\phi)}\ln\!\left[  \frac{e^{-G(\theta)}}{Z(\theta)}\right]
e^{iH(\phi)}\\
&  =-e^{-iH(\phi)}G(\theta)e^{iH(\phi)}-I\ln Z(\theta).
\end{align}
So then we find that
\begin{align}
&  =\operatorname{Tr}\!\left[\left\{  \Psi^{\dagger}_{\phi}(H_{i}),\Psi^{\dagger}_{\phi}(H_{j})\right\}
\left(  \omega(\theta,\phi)\ln\omega(\theta,\phi)\right)  \right]\nonumber\\
&  \qquad-2\operatorname{Re}\!\left[  \operatorname{Tr}[\left(  \ln\omega
(\theta,\phi)\right)  \Psi^{\dagger}_{\phi}(H_{i})\omega(\theta,\phi)\Psi^\dag_{\phi}
(H_{j})]\right]  \\
&  =\operatorname{Tr}\!\left[\left\{  \Psi^{\dagger}_{\phi}(H_{i}),\Psi^{\dagger}_{\phi}(H_{j})\right\}
\left(  \omega(\theta,\phi)\left(  -e^{-iH(\phi)}G(\theta)e^{iH(\phi)}-I\ln
Z(\theta)\right)  \right)  \right]\nonumber\\
&  \qquad-2\operatorname{Re}\!\left[  \operatorname{Tr}\!\left[\left(  -e^{-iH(\phi
)}G(\theta)e^{iH(\phi)}-I\ln Z(\theta)\right)  \Psi^{\dagger}_{\phi}(H_{i})\omega
(\theta,\phi)\Psi^{\dagger}_{\phi}(H_{j})\right]\right]  \\
&  =-\operatorname{Tr}\!\left[\left\{  \Psi^{\dagger}_{\phi}(H_{i}),\Psi^\dag_{\phi}(H_{j}
)\right\}  \left(  \omega(\theta,\phi)\left(  e^{-iH(\phi)}G(\theta
)e^{iH(\phi)}\right)  \right)  \right]\nonumber\\
&  \qquad-\operatorname{Tr}\!\left[\left\{  \Psi^{\dagger}_{\phi}(H_{i}),\Psi^\dag_{\phi}
(H_{j})\right\}  \omega(\theta,\phi)\right]\ln Z(\theta)\nonumber\\
&  \qquad+2\operatorname{Re}\!\left[  \operatorname{Tr}\!\left[e^{-iH(\phi)}
G(\theta)e^{iH(\phi)}\Psi^{\dagger}_{\phi}(H_{i})\omega(\theta,\phi)\Psi^\dag_{\phi}
(H_{j})\right]\right]  \nonumber\\
&  \qquad+2\operatorname{Re}\!\left[  \operatorname{Tr}\!\left[\Psi^\dag_{\phi}(H_{i}
)\omega(\theta,\phi)\Psi^{\dagger}_{\phi}(H_{j})\right]\right]  \ln Z(\theta)\\
&  =-\operatorname{Tr}\!\left[\left\{  \Psi^{\dagger}_{\phi}(H_{i}),\Psi^\dag_{\phi}(H_{j}
)\right\}  \left(  \omega(\theta,\phi)\left(  e^{-iH(\phi)}G(\theta
)e^{iH(\phi)}\right)  \right)  \right]\nonumber\\
&  \qquad+2\operatorname{Re}\!\left[  \operatorname{Tr}\!\left[e^{-iH(\phi)}
G(\theta)e^{iH(\phi)}\Psi^{\dagger}_{\phi}(H_{i})\omega(\theta,\phi)\Psi^\dag_{\phi}
(H_{j})\right]\right] \\
& = -\Tr\!\left[ \left\{ \Psi_{\phi}(H_{i}),\Psi_{\phi}(H_{j}
)\right\}  \rho(\theta) G(\theta
)\right]
+2\operatorname{Re}\!\left[  \operatorname{Tr}\!\left[
G(\theta)\Psi_{\phi}(H_{i})\rho(\theta)\Psi_{\phi}
(H_{j})\right]\right] \\
& = -\Tr\!\left[  \Psi_{\phi}(H_{i}) \Psi_{\phi}(H_{j}
)\rho(\theta) G(\theta
)\right] -\Tr\!\left[ \Psi_{\phi}(H_{j}
) \Psi_{\phi}(H_{i}) \rho(\theta) G(\theta
)\right] \nonumber\\
& \qquad + \operatorname{Tr}\!\left[
G(\theta)\Psi_{\phi}(H_{i})\rho(\theta)\Psi_{\phi}
(H_{j})\right] +  \operatorname{Tr}\!\left[ \Psi_{\phi}
(H_{j}) \rho(\theta) \Psi_{\phi}(H_{i})
G(\theta)\right] \\
& = -\Tr\!\left[  \Psi_{\phi}(H_{i}) \Psi_{\phi}(H_{j}
) G(\theta) \rho(\theta) \right] -\Tr\!\left[ G(\theta) \Psi_{\phi}(H_{j}
) \Psi_{\phi}(H_{i})  \rho(\theta) \right] \nonumber\\
& \qquad + \operatorname{Tr}\!\left[ \Psi_{\phi}(H_{j}) G(\theta) \Psi_{\phi}(H_{i}) \rho(\theta) \right] + \operatorname{Tr}\!\left[ \Psi_{\phi}(H_{i}) G(\theta) \Psi_{\phi}
(H_{j}) \rho(\theta) \right]\label{commutator-fact1}\\ 
& = \Tr\! \left[ \left( - \Psi_{\phi}(H_{i}) \Psi_{\phi}(H_{j}
) G(\theta) - G(\theta) \Psi_{\phi}(H_{j}
) \Psi_{\phi}(H_{i})  +  \Psi_{\phi}(H_{j}) G(\theta) \Psi_{\phi}(H_{i}) +  \Psi_{\phi}(H_{i}) G(\theta) \Psi_{\phi}
(H_{j}) \right) \rho(\theta) \right]\\
& = \Tr\! \left[ \left(  \Psi_{\phi}(H_{i}) \left[ G(\theta) , \Psi_{\phi}
(H_{j}) \right] - \left[ G(\theta) ,\Psi_{\phi}(H_{j}) \right] \Psi_{\phi}(H_{i})  \right) \rho(\theta) \right]\\
& = \Tr\! \bigg[ \!\left[ \Psi_{\phi}(H_{i}) , \left[ G(\theta) , \Psi_{\phi}
(H_{j}) \right] \right] \rho(\theta) \bigg]\\
& =  \Tr\! \bigg[ \!\left[ \left[\Psi_{\phi}(H_{j}) ,G(\theta) 
 \right] , \Psi_{\phi}(H_{i}) \right] \rho(\theta) \bigg]\\
& =  \left\langle \left[ \left[\Psi_{\phi}(H_{j}) ,G(\theta) 
 \right] , \Psi_{\phi}(H_{i}) \right]\right\rangle_{\rho(\theta)}\\
& = \left\langle \left[  \Psi_{\phi}(H_{i}), \left[  G(\theta) , \Psi_{\phi}
(H_{j})\right]  \right]
\right\rangle _{\rho(\theta)},
\end{align}
where in~\eqref{commutator-fact1} we used that $\left[G(\theta), \rho(\theta) \right]=0$. This concludes the proof of Theorem~\ref{thm:KM-phi}.

\subsubsection{Proof of Theorem~\ref{thm:KM-theta-phi}}

\label{proof:KM-theta-phi}

Using~\eqref{eq:Fisher-info-help-last} and~\eqref{eq:Fisher-info-help-alt1}, consider that
\begin{align}
&  I_{ij}^{\text{KM}}(\theta, \phi)\nonumber\\
&  =\sum_{k,\ell}\frac{\ln\lambda_{k}-\ln\lambda_{\ell}}{\lambda_{k}
-\lambda_{\ell}}\langle k|\left[  \frac{\partial}{\partial\theta_{i}}
\omega(\theta,\phi)\right]  |\ell\rangle\!\langle\ell|\left[  \frac{\partial
}{\partial\phi_{j}}\omega(\theta,\phi)\right]  |k\rangle\\
& = \sum_{k,\ell}\frac{\ln\lambda_{k}-\ln\lambda_{\ell}}{\lambda_{k}
-\lambda_{\ell}} \left(  -\frac{1}{2}\left( \lambda_k + \lambda_\ell\right) \bra{k} e^{-iH(\phi)}\Phi_{\theta}(G_{i})e^{iH(\phi)}\ket{\ell} + \delta_{k\ell} \lambda_\ell \left\langle G_{i}\right\rangle_{\rho(\theta)} \right) i\left(  \lambda_{\ell}-\lambda_{k}\right)  \bra{\ell}\Psi^\dag_{\phi}(H_j)\ket{k}\\
& = \frac{i}{2}  \sum_{k,\ell} \left(\ln\lambda_{k}-\ln\lambda_{\ell}\right)\left( \lambda_k + \lambda_\ell\right)  \bra{k} e^{-iH(\phi)}\Phi_{\theta}(G_{i})e^{iH(\phi)}\ket{\ell}\bra{\ell}\Psi^\dag_{\phi}(H_j)\ket{k}\\
& = \frac{i}{2}  \sum_{k,\ell} \frac{i}{2}   \lambda_{k} \ln\lambda_{k} \bra{k} e^{-iH(\phi)}\Phi_{\theta}(G_{i})e^{iH(\phi)}\ket{\ell}\bra{\ell}\Psi^\dag_{\phi}(H_j)\ket{k}\nonumber\\
& \hspace{1cm}  + \sum_{k,\ell} \frac{1}{2} i \lambda_{\ell} \ln\lambda_{k} \bra{k} e^{-iH(\phi)}\Phi_{\theta}(G_{i})e^{iH(\phi)}\ket{\ell}\bra{\ell}\Psi^\dag_{\phi}(H_j)\ket{k} \nonumber\\
& \hspace{1cm}  - \frac{i}{2} \sum_{k,\ell}  \lambda_{k} \ln\lambda_{\ell} \bra{k} e^{-iH(\phi)}\Phi_{\theta}(G_{i})e^{iH(\phi)}\ket{\ell}\bra{\ell}\Psi^\dag_{\phi}(H_j)\ket{k} \nonumber\\
& \hspace{1cm}  - \frac{i}{2} \sum_{k,\ell}  \lambda_{\ell} \ln\lambda_{\ell} \bra{k} e^{-iH(\phi)}\Phi_{\theta}(G_{i})e^{iH(\phi)}\ket{\ell}\bra{\ell}\Psi^\dag_{\phi}(H_j)\ket{k}\\
& = \frac{i}{2}  \Tr\!\left[\omega(\theta,\phi) \ln\omega(\theta,\phi)e^{-iH(\phi)}\Phi_{\theta}(G_{i}) e^{iH(\phi)}\Psi^\dag_{\phi}(H_j)  \right]\nonumber\\
& \hspace{1cm}  + \frac{i}{2}  \Tr\!\left[\ln\omega(\theta,\phi)e^{-iH(\phi)}\Phi_{\theta}(G_{i})e^{iH(\phi)} \omega(\theta,\phi) \Psi^\dag_{\phi}(H_j)  \right]\nonumber\\
& \hspace{1cm}  - \frac{i}{2}  \Tr\!\left[\Psi^\dag_{\phi}(H_j) \omega(\theta,\phi) e^{-iH(\phi)}\Phi_{\theta}(G_{i})e^{iH(\phi)}   \ln\omega(\theta,\phi) \right]\nonumber\\
& \hspace{1cm}  - \frac{i}{2}  \Tr\!\left[\omega(\theta,\phi) \ln\omega(\theta,\phi)\Psi^\dag_{\phi}(H_j)e^{-iH(\phi)}\Phi_{\theta}(G_{i}) e^{iH(\phi)}  \right].\label{eq: last_expression_KBI_mixed}
\end{align}
Now consider that 
\begin{align}
    \ln \omega(\theta,\phi) & = \ln \!\left[ e^{-iH(\phi)} \frac{e^{-G(\theta)}}{Z(\theta)}e^{iH(\phi)} \right]\\
    & = e^{-iH(\phi)} \ln\!\left[ \frac{e^{-G(\theta)}}{Z(\theta)}\right] e^{iH(\phi)}\\
    & = -e^{-iH(\phi)} G(\theta) e^{iH(\phi)} - I \ln Z(\theta).
\end{align}
So, plugging this into~\eqref{eq: last_expression_KBI_mixed}, we find that 
\begin{align}
    & = \frac{i}{2}  \Tr\!\left[\omega(\theta,\phi) \ln\omega(\theta,\phi)e^{-iH(\phi)}\Phi_{\theta}(G_{i}) e^{iH(\phi)}\Psi^\dag_{\phi}(H_j)  \right]\nonumber\\
    & \hspace{1cm}  + \frac{i}{2}  \Tr\!\left[\ln\omega(\theta,\phi)e^{-iH(\phi)}\Phi_{\theta}(G_{i})e^{iH(\phi)} \omega(\theta,\phi) \Psi^\dag_{\phi}(H_j)  \right]\nonumber\\
    & \hspace{1cm}  - \frac{i}{2}  \Tr\!\left[\Psi^\dag_{\phi}(H_j) \omega(\theta,\phi) e^{-iH(\phi)}\Phi_{\theta}(G_{i})e^{iH(\phi)}   \ln\omega(\theta,\phi) \right]\nonumber\\
    & \hspace{1cm}  - \frac{i}{2}  \Tr\!\left[\omega(\theta,\phi) \ln\omega(\theta,\phi)\Psi^\dag_{\phi}(H_j)e^{-iH(\phi)}\Phi_{\theta}(G_{i}) e^{iH(\phi)}  \right]\\
    & = - \frac{i}{2}  \Tr\!\left[ \rho(\theta) G(\theta) \Phi_{\theta}(G_{i}) \Psi_{\phi}(H_j) \right] - \frac{i}{2}  \Tr\!\left[ \rho(\theta) \Phi_{\theta}(G_{i}) \Psi_{\phi}(H_j) \right] \ln Z(\theta)  \nonumber\\
    & \hspace{1cm}  - \frac{i}{2}  \Tr\!\left[G(\theta) \Phi_{\theta}(G_{i})\rho(\theta) \Psi_{\phi}(H_j)  \right] - \frac{i}{2}  \Tr\!\left[\Phi_{\theta}(G_{i})\rho(\theta) \Psi_{\phi}(H_j)  \right] \ln Z(\theta) \nonumber\\
    & \hspace{1cm}  + \frac{i}{2}  \Tr\!\left[\Psi_{\phi}(H_j) \rho(\theta) \Phi_{\theta}(G_{i}) G(\theta) \right] + \frac{i}{2}  \Tr\!\left[\Psi_{\phi}(H_j) \rho(\theta) \Phi_{\theta}(G_{i}) \right] \ln Z(\theta) \nonumber\\
    & \hspace{1cm}  + \frac{i}{2}  \Tr\!\left[\rho(\theta) G(\theta) \Psi_{\phi}(H_j)\Phi_{\theta}(G_{i}) \right] + \frac{i}{2}  \Tr\!\left[\rho(\theta)  \Psi_{\phi}(H_j)\Phi_{\theta}(G_{i}) \right] \ln Z(\theta)\\
    & = - \frac{i}{2}  \Tr\!\left[\Phi_{\theta}(G_{i}) \Psi_{\phi}(H_j)  G(\theta) \rho(\theta) \right] - \frac{i}{2}  \Tr\!\left[\Psi_{\phi}(H_j) G(\theta) \Phi_{\theta}(G_{i})\rho(\theta)  \right]  \nonumber\\
    & \hspace{1cm}  + \frac{i}{2}  \Tr\!\left[\Phi_{\theta}(G_{i}) G(\theta) 
    \Psi_{\phi}(H_j) \rho(\theta) \right]  + \frac{i}{2}  \Tr\!\left[  G(\theta) \Psi_{\phi}(H_j)\Phi_{\theta}(G_{i}) \rho(\theta) \right]\label{commutator-fact}\\
    & = \frac{i}{2}  \Tr\!\left[ \left( - \Phi_{\theta}(G_{i}) \Psi_{\phi}(H_j) G(\theta) + \Phi_{\theta}(G_{i}) G(\theta) 
    \Psi_{\phi}(H_j) - \Psi_{\phi}(H_j) G(\theta) \Phi_{\theta}(G_{i}) + G(\theta) \Psi_{\phi}(H_j)\Phi_{\theta}(G_{i})  \right) \rho(\theta) \right]\\
    & = \frac{i}{2}  \Tr\!\left[ \left( - \Phi_{\theta}(G_{i}) \left[\Psi_{\phi}(H_j), G(\theta)\right] - \left[\Psi_{\phi}(H_j) , G(\theta)\right] \Phi_{\theta}(G_{i}) \right) \rho(\theta) \right]\\
    & = - \frac{i}{2}  \Tr\!\bigg[ \left\{  \Phi_{\theta}(G_{i}) , \left[\Psi_{\phi}(H_j) , G(\theta) \right] \right\} \rho(\theta) \bigg]\\
    & = - \frac{i}{2} \left\langle \left\{  \Phi_{\theta}(G_{i}) , \left[\Psi_{\phi}(H_j) , G(\theta) \right] \right\}\right\rangle_{\rho(\theta)}\\
    & = \frac{i}{2} \left\langle \left\{  \Phi_{\theta}(G_{i}) , \left[G(\theta) , \Psi_{\phi}(H_j)  \right] \right\}\right\rangle_{\rho(\theta)},
\end{align}
where in~\eqref{commutator-fact} we used that $\left[G(\theta), \rho(\theta) \right]=0$. This concludes the proof of Theorem~\ref{thm:KM-theta-phi}.

\section{Quantum algorithms for estimating information matrix elements}

\subsection{Quantum algorithms for estimating Fisher--Bures information matrix elements}

\subsubsection{Elements with respect to $\theta$}

\label{app:FB-theta}

Let us first recall from the statement of  Theorem~\ref{thm:FB-theta} the expression for the $(i, j)$-th element of the Fisher--Bures information matrix $I^{\operatorname{FB}}(\theta)$:
\begin{equation}
I_{ij}^{\operatorname{FB}}(\theta)=\frac{1}
{2}\left\langle \{\Phi_{\theta}(G_{i}),\Phi_{\theta}(G_{j})\}\right\rangle
_{\rho(\theta)}-\left\langle G_{i}\right\rangle _{\rho(\theta)}\left\langle
G_{j}\right\rangle _{\rho(\theta)}.\label{eq:Fisher-matrix-entries-app}
\end{equation}
Estimating the second term of the right-hand side of the above equation is relatively straightforward. So, in what follows, we present an algorithm for estimating the first term of the above equation in greater detail.

Consider the following:
\begin{align}
     \frac{1}
{2}\left\langle \{\Phi_{\theta}(G_{i}),\Phi_{\theta}(G_{j})\}\right\rangle
_{\rho(\theta)}
& = \frac{1}{2}\operatorname{Tr}[\left\{
\Phi_{\theta}(G_{i}),\Phi_{\theta}(G_{j})\right\}  \rho(\theta)]\\
& = \int_{\mathbb{R}}\int_{\mathbb{R}} dt_1\ dt_2\ p(t_1) p(t_2) \left(\frac{1}{2}\operatorname{Tr}\!\left[\left\{
e^{-iG(\theta)t_1} G_i e^{iG(\theta)t_1}, e^{-iG(\theta)t_2} G_j e^{iG(\theta)t_2}\right\} \rho(\theta)\right]\right)\label{eq:first_fb_simplify}.
\end{align}

We are now in a position to present an algorithm (Algorithm~\ref{algo:FB-theta}) to estimate the first term of~\eqref{eq:Fisher-matrix-entries-app} by using its equivalent form shown in~\eqref{eq:first_fb_simplify}. At the core of our algorithm lies the quantum circuit that estimates the expected value of the anticommutator of two operators (see Appendix~\ref{app:Hadamard_test}). In particular, one can notice that when initializing the control register in the state $\ket{0}$ instead of $\ket{1}$, the quantum circuit shown in Figure~\ref{fig:qc-primitive-anticomm} allows to estimate the quantity $\frac{1}{2} \left\langle \left\{ U, H\right\} \right\rangle_{\rho}$, where $H$ is Hermitian and $U$ is Hermitian and unitary. In this case, we choose $\rho = \rho(\theta)$, $U = e^{-iG(\theta)t_1} G_i e^{iG(\theta)t_1}$, and $H = e^{-iG(\theta)t_2} G_j e^{iG(\theta)t_2}$. We then make some further simplifications that follow because $\rho(\theta)$ commutes with $e^{-iG(\theta)t_1}$. Accordingly, the quantum circuit that estimates the integrand of~\eqref{eq:first_fb_simplify} is depicted in Figure~\ref{fig:FB-theta}.

\begin{algorithm}[H]
\caption{\texorpdfstring{$\mathtt{estimate\_first\_term\_FB\_\theta}(i, j, \theta, \{G_\ell\}_{\ell=1}^{J}, p(\cdot), \varepsilon, \delta)$}{estimate first term}}
\label{algo:FB-theta}
\begin{algorithmic}[1]
\STATE \textbf{Input:} Indices $i, j \in [J]$, parameter vector $\theta = \left( \theta_{1}, \ldots,  \theta_{J}\right)^{\mathsf{T}} \in \mathbb{R}^{J}$, Gibbs local Hamiltonians $\{G_\ell\}_{\ell=1}^{J}$, probability distribution $p(t)$ over $\mathbb{R}$, precision $\varepsilon > 0$, error probability $\delta \in (0,1)$
\STATE $N \leftarrow \lceil\sfrac{2 \ln(\sfrac{2}{\delta})}{\varepsilon^2}\rceil$
\FOR{$n = 0$ to $N-1$}
\STATE Initialize the control register to $|0\rangle\!\langle 0 |$
\STATE Prepare the system register in the state $\rho(\theta)$
\STATE Sample $t_1$ and $t_2$ independently at random with probability $p(t)$ (defined in~\eqref{eq:high-peak-tent-density})
\STATE Apply the Hadamard gate to the control register
\STATE Apply the following unitaries to the control and system registers:
\STATE \hspace{0.6cm} \textbullet~Controlled-$G_i$: $G_i$ is a local unitary acting on the system register, controlled by the control register
\STATE \hspace{0.6cm} \textbullet~$e^{-iG(\theta)(t_1 - t_2)}$: Hamiltonian simulation for time $t_1 - t_2$ on the system register
\STATE Apply the Hadamard gate to the control register
\STATE Measure the control register in the computational basis and store the measurement outcome~$b_n$
\STATE Measure the system register in the eigenbasis of $G_j$ and store the measurement outcome $\lambda_n$
\STATE $Y_{n}^{(\operatorname{FB(\theta)})} \leftarrow (-1)^{b_n}\lambda_n$
\ENDFOR

\STATE \textbf{return} $\overline{Y}^{(\operatorname{FB(\theta)})} \leftarrow \frac{1}{N}\sum_{n=0}^{N-1}Y_{n}^{(\operatorname{FB(\theta)})}$
\end{algorithmic}
\end{algorithm}

\subsubsection{Elements with respect to $\phi$}

\label{app:FB-phi}

Let us first recall from the statement of Theorem~\ref{thm:FB-phi} the expression for the $(i, j)$-th element of the Fisher–Bures information matrix $I_{ij}^{\operatorname{FB}}(\phi)$:
\begin{equation}
    I_{ij}^{\operatorname{FB}}(\phi) =\left\langle \left[  \Phi_{\theta}(\Psi_{\phi}(H_{i})), \left[  G(\theta) , \Psi_{\phi}
(H_{j})\right]  \right]
\right\rangle _{\rho(\theta)}.\label{eq:FB-phi_app}
\end{equation}
Consider the following:
\begin{align}
    & \left\langle \left[  \Phi_{\theta}(\Psi_{\phi}(H_{i})), \left[  G(\theta) , \Psi_{\phi} (H_{j})\right]  \right]\right\rangle _{\rho(\theta)}\\
    & = \left\langle \left[  \left[  \Psi_{\phi}(H_{j}),G(\theta)\right]  ,\Phi_{\theta}(\Psi_{\phi}(H_{i}))\right]\right\rangle _{\rho(\theta)}\nonumber\\
    & = \Tr\!\bigg[ \left[ \left[  \Psi_{\phi}(H_{j}),G(\theta)\right]  ,\Phi_{\theta}(\Psi_{\phi}(H_{i}))\right] \rho(\theta) \bigg] \\
    & = \int_0^1 \int_\mathbb{R} \int_0^1 dt_1\ dt_2\ dt_3\ p(t_2)\ \bigg( \Tr\!\bigg[ \left[ \left[  e^{iH(\phi)t_1} H_j e^{-iH(\phi)t_1} , G(\theta)\right] , e^{-iG(\theta)t_2} e^{iH(\phi)t_3} H_i e^{-iH(\phi)t_3}e^{iG(\theta)t_2}\right] \rho(\theta) \bigg]  \bigg).\label{eq:FB-phi-last}
\end{align}

We are now in a position to present an algorithm (Algorithm~\ref{algo:FB-phi}) to estimate~\eqref{eq:FB-phi_app} using its equivalent form shown in~\eqref{eq:FB-phi-last}. At the core of our algorithm lies the quantum circuit that estimates the expected value of two nested commutators of three operators, $\frac{1}{4} \left\langle \big[ \left[U_1 , H \right], U_0\big]\right\rangle_\rho$, where $H$ is Hermitian, and $U_0$ and $U_1$ are both Hermitian and unitary (refer to Appendix~\ref{app:Hadamard_test} and Figure~\ref{fig:prim-nest-comm}). In this case, we choose $\rho = \rho(\theta)$, $U_1 = e^{iH(\phi)t_1} H_j e^{-iH(\phi)t_1}$, $H = G(\theta)$, and $U_0 = e^{-iG(\theta)t_2} e^{iH(\phi)t_3} H_i e^{-iH(\phi)t_3}e^{iG(\theta)t_2}$. We then make some further simplifications where possible. Specifically, the quantum circuit that plays a role in estimating the integrand of~\eqref{eq:FB-phi-last} is depicted in Figure~\ref{fig:FB-phi}.

\begin{remark}
\label{rem:measuring-G-theta}
In Algorithm~\ref{algo:FB-phi} below, note that we adopt a sampling approach to measuring $G(\theta)$, similar to that used in \cite[Algorithm~1]{patel2024quantumboltzmannmachine}. This seems to be necessary, as it is not obvious how to measure directly in the eigenbasis of $G(\theta)$. We adopt a similar approach in other circuits that involve measuring $G(\theta)$.
\end{remark}

\begin{algorithm}[H]
\caption{\texorpdfstring{$\mathtt{estimate\_FB\_\phi}(i, j, \theta, \{G_\ell\}_{\ell=1}^{J}, \phi, \{H_m\}_{m=1}^{K}, p(\cdot), \varepsilon, \delta)$}{estimate first term}}
\label{algo:FB-phi}
\begin{algorithmic}[1]
\STATE \textbf{Input:} Indices $i, j \in [K]$, parameter vectors $\theta = \left( \theta_{1}, \ldots,  \theta_{J}\right)^{\mathsf{T}} \in \mathbb{R}^{J}$ and $\phi = \left( \phi_{1}, \ldots,  \phi_{K}\right)^{\mathsf{T}} \in \mathbb{R}^{K}$, Gibbs local Hamiltonians $\{G_\ell\}_{\ell=1}^{J}$ and $\{H_m\}_{m=1}^{K}$, probability distribution $p(t)$ over $\mathbb{R}$, precision $\varepsilon > 0$, error probability $\delta \in (0,1)$
\STATE $N \leftarrow \lceil\sfrac{2\left\| \theta \right\|_1^2 \ln(\sfrac{2 }{\delta})}{\varepsilon^2}\rceil$
\FOR{$n = 0$ to $N-1$}
\STATE Initialize the first control register to $|1\rangle\!\langle 1 |$
\STATE Initialize the second control register to $|1\rangle\!\langle 1 |$
\STATE Prepare the system register in the state $\rho(\theta)$
\STATE Sample $t_1$ and $t_3$ independently and uniformly at random from the interval $[0,1]$, and sample $t_2$ independently at random with probability $p(t)$ (defined in~\eqref{eq:high-peak-tent-density})
\STATE Apply the Hadamard gate and the phase gate $S$ to the first and second control registers
\STATE Apply the following unitaries to the control and system registers:
\STATE \hspace{0.6cm} \textbullet~$e^{-iH(\phi)t_3}$: Hamiltonian simulation for time $t_3$ on the system register
\STATE \hspace{0.6cm} \textbullet~Controlled-$H_i$: $H_i$ is a local unitary acting on the system register, controlled by the first control register
\STATE \hspace{0.6cm} \textbullet~$e^{iH(\phi)t_3}$: Hamiltonian simulation for time $t_3$ on the system register
\STATE \hspace{0.6cm} \textbullet~$e^{-iG(\theta)t_2}$: Hamiltonian simulation for time $t_2$ on the system register
\STATE \hspace{0.6cm} \textbullet~$e^{-iH(\phi)t_1}$: Hamiltonian simulation for time $t_1$ on the system register
\STATE \hspace{0.6cm} \textbullet~Controlled-$H_j$: $H_j$ is a local unitary acting on the system register, controlled by the second control register
\STATE \hspace{0.6cm} \textbullet~$e^{iH(\phi)t_1}$: Hamiltonian simulation for time $t_1$ on the system register
\STATE Apply the Hadamard gate to the first and second control registers
\STATE Measure the two control registers in the computational basis and store the measurement outcomes~$b_n$ and $c_n$
\STATE Sample $\ell$ according to the probability distribution $|\theta_\ell| / \left\| \theta \right\|_1$, measure the system register in the eigenbasis of sign$(\theta_\ell) G_\ell$ and store the measurement outcome $\lambda_n$
\STATE $Y_{n}^{(\operatorname{FB(\phi)})} \leftarrow (-1)^{b_n}(-1)^{c_n}\lambda_n$
\ENDFOR

\STATE \textbf{return} $\overline{Y}^{(\operatorname{FB(\phi)})} \leftarrow 4 \left\| \theta \right\|_1 \times\frac{1}{N}\sum_{n=0}^{N-1}Y_{n}^{(\operatorname{FB(\phi)})}$
\end{algorithmic}
\end{algorithm}

\subsubsection{Elements with respect to $\theta$ and $\phi$}

\label{app:FB-theta-phi}

Let us first recall from the statement of Theorem~\ref{thm:FB-theta-phi} the expression for the $(i, j)$-th element of the Fisher–Bures information matrix $I_{ij}^{\operatorname{FB}}(\theta, \phi)$:
\begin{equation}
    I_{ij}^{\operatorname{FB}}(\theta
,\phi)=i\left\langle \left[  \Phi_{\theta}(G_{i}),\Psi_{\phi}
(H_{j})\right]  \right\rangle _{\rho(\theta)}.\label{eq:FB-theta-phi_app}
\end{equation}
Consider the following:
\begin{align}
     i\left\langle \left[  \Phi_{\theta}(G_{i}),\Psi_{\phi}(H_{j})\right]  \right\rangle _{\rho(\theta)}
    & = i \left[ \Tr\!\left[ \left[ \Phi_{\theta}(G_{i}) , \Psi_{\phi}(H_{j}) \right] \rho(\theta)  \right] \right]\\
    & = \int_\mathbb{R} \int_0^1 dt_1\ dt_2\ p(t_1) \Bigg( i \Tr\!\left[ \left[ e^{-iG(\theta)t_1} G_i e^{iG(\theta)t_1}, e^{iH(\phi)t_2}H_je^{-iH(\phi)t_2} \right] \rho(\theta) \right]\Bigg),\label{eq:FB-theta-phi_est}
\end{align}

We are now in a position to present an algorithm (Algorithm~\ref{algo:FB-theta-phi}) to estimate~\eqref{eq:FB-theta-phi_app} using its equivalent form shown in~\eqref{eq:FB-theta-phi_est}. At the core of our algorithm lies the quantum circuit that estimates the expected value of the commutator of two operators (see Appendix~\ref{app:Hadamard_test}). In particular, one can notice that when initializing the control register in the state $\ket{0}$ instead of $\ket{1}$, the quantum circuit shown in Figure~\ref{fig:qc-primitive-comm} allows to estimate the quantity $\frac{i}{2} \left\langle \left[ H, U\right] \right\rangle_{\rho}$, where $H$ is Hermitian and $U$ is Hermitian and unitary. In this case, we choose $\rho = \rho(\theta)$, $H = e^{-iG(\theta)t_1} G_i e^{iG(\theta)t_1}$, and $U = e^{iH(\phi)t_2}H_je^{-iH(\phi)t_2}$. We then make some further simplifications where possible. Accordingly, the quantum circuit that plays a role in estimating the integrand of~\eqref{eq:FB-theta-phi_est} is depicted in Figure~\ref{fig:FB-theta-phi}.

\begin{algorithm}[H]
\caption{\(\mathtt{estimate\_FB\_\theta\_\phi}(i,j,\theta,\{G_\ell\}_{\ell=1}^{J},\phi,\{H_k\}_{k=1}^K,p(\cdot),\varepsilon,\delta)\)}
\label{algo:FB-theta-phi}\begin{algorithmic}[1]
\STATE \textbf{Input:} Indices $i\in[J]$, $j\in[K]$, parameter vectors $\theta = \left( \theta_{1}, \ldots,  \theta_{J}\right)^{\mathsf{T}} \in \mathbb{R}^{J}$ and $\phi = \left( \phi_{1}, \ldots,  \phi_{K}\right)^{\mathsf{T}} \in \mathbb{R}^{K}$, sets of local Hamiltonians $\{G_\ell\}_{\ell=1}^{J}$ and $\{H_k\}_{k=1}^K$, probability distribution $p(t)$ over $\mathbb{R}$, precision $\varepsilon>0$, error probability $\delta\in (0,1)$

\STATE Set $N \leftarrow \lceil 2\ln(2/\delta)/\varepsilon^2 \rceil$

\FOR{$n=0$ to $N-1$}
    \STATE Initialize the control register to $|0\rangle\!\langle 0 |$
    \STATE Prepare the system register in the state $\rho(\theta)$
    \STATE Sample $t_1$ and $t_2$ independently and at random, $t_1$ with probability $p(t)$ (defined in~\eqref{eq:high-peak-tent-density}), $t_2$ uniformly from the interval $[0,1]$
    \STATE Apply the Hadamard gate and the phase gate $S$ to the control register
    \STATE Apply the following unitaries to the control and system registers:
    \STATE \hspace{0.6cm} \textbullet~$e^{-iH(\phi)t_2}$: Hamiltonian simulation for time $t_2$ on the system register
    \STATE \hspace{0.6cm} 
    \textbullet~Controlled-$H_j$: $H_j$ is a local unitary acting on the system register, controlled by the control register
    \STATE \hspace{0.6cm} \textbullet~$e^{iH(\phi)t_2}$: Hamiltonian simulation for time $t_2$ on the system register
    \STATE \hspace{0.6cm} \textbullet~$e^{iG(\theta)t_1}$: Hamiltonian simulation for time $t_1$ on the system register
    \STATE Apply the Hadamard gate to the control register
    \STATE Measure the control register in the computational basis and store the measurement outcome $b_n$
    \STATE Measure the system register in the eigenbasis of $G_i$ and store the measurement outcome $\lambda_n$
    \STATE $Y_{n}^{(\operatorname{FB(\theta,\phi)})} \leftarrow (-1)^{b_n}\lambda_n$
\ENDFOR

\STATE \textbf{return} $\overline{Y}^{(\operatorname{FB(\theta,\phi))})} \leftarrow 2 \times \frac{1}{N}\sum_{n=0}^{N-1}Y_{n}^{(\operatorname{FB(\theta,\phi)})}$
\end{algorithmic}
\end{algorithm}

\subsection{Quantum algorithms for estimating Wigner--Yanase information matrix elements}

\subsubsection{Elements with respect to $\theta$}

\label{app:WY-theta}

Let us first recall from the statement of Theorem~\ref{thm:WY-theta} the expression for the $(i, j)$-th element of the Wigner--Yanase information matrix $I_{ij}^{\operatorname{WY}}(\theta)$:
\begin{align}
    I&_{ij}^{\operatorname{WY}}(\theta)=\frac{1}{2}\operatorname{Tr}\!\left[  \Phi_{\frac{\theta}{2}}(G_{i})\sqrt{\rho(\theta)}\Phi_{\frac{\theta}{2}}(G_{j})\sqrt{\rho(\theta)}\right]+\frac{1}{4}\left\langle \left\{  \Phi_{\frac{\theta}{2}}(G_{i}),\Phi_{\frac{\theta}{2}}(G_{j})\right\}  \right\rangle _{\rho(\theta)}-\left\langle G_{i}\right\rangle _{\rho(\theta)}\left\langle G_{j}\right\rangle _{\rho(\theta)}. \label{eq:WY-theta_app}
\end{align}
Estimating the third term of the right-hand side of the above equation is relatively straightforward. Details on the estimation of the first term are provided in Section~\ref{sec:1st-term-WY-theta}. Therefore, the focus here is on estimating the second term of~\eqref{eq:WY-theta_app}. Consider the following:
\begin{align}
    &\frac{1}{4}\left\langle \left\{  \Phi_{\frac{\theta}{2}}(G_{i}),\Phi_{\frac{\theta}{2}}(G_{j})\right\}  \right\rangle _{\rho(\theta)}\nonumber\\
    & = \frac{1}{4}\Tr\!\left[ \left\{  \Phi_{\frac{\theta}{2}}(G_{i}),\Phi_{\frac{\theta}{2}}(G_{j})\right\}  \rho(\theta) \right]\\
    & = \int_{\mathbb{R}}\int_{\mathbb{R}} dt_1\ dt_2\ p(t_1) p(t_2) \left( \frac{1}{4}\operatorname{Tr}\!\left[ \left\{
    e^{-iG(\theta/2)t_1} G_i e^{iG(\theta/2)t_1}, e^{-iG(\theta/2)t_2} G_j e^{iG(\theta/2)t_2} \right\} \rho(\theta)\right] \right).\label{eq:WY-theta_app-alt}
\end{align}

We are now in a position to present an algorithm to estimate the second term of~\eqref{eq:WY-theta_app} using its equivalent form shown in~\eqref{eq:WY-theta_app-alt}. The algorithm is similar to Algorithm~\ref{algo:FB-theta}, so here we provide a high-level description of how it works.  At its core, the algorithm relies on a quantum circuit that estimates the expected value of the anticommutator of two operators (see Appendix~\ref{app:Hadamard_test}). Specifically, if the control register in Figure~\ref{fig:qc-primitive-anticomm} is initialized in the state $\ket{0}$ instead of $\ket{1}$, the output of the circuit is $\frac{1}{2} \left\langle \left\{H,U\right\} \right\rangle_{\rho}$, where $H$ is Hermitian and $U$ is Hermitian and unitary. Accordingly, the quantum circuit that plays a role in estimating the integrand of~\eqref{eq:WY-theta_app-alt} is depicted in Figure~\ref{fig:WY-theta}. In this case, we choose $\rho = \rho(\theta)$, $H = e^{-iG(\theta/2)t_1} G_i e^{iG(\theta/2)t_1}$, and $U = e^{-iG(\theta/2)t_2} G_j e^{iG(\theta/2)t_2}$. We then make some further simplifications where possible. The algorithm involves running this circuit $N$ times, where $N$ is determined by the desired precision and error probability. During each run, the times $t_1$ and $t_2$ for the Hamiltonian evolution are sampled independently at random with probability $p(t)$ (defined in~\eqref{eq:high-peak-tent-density}). The final estimation of the second term of~\eqref{eq:WY-theta_app} is obtained by averaging the outputs of the $N$ runs and dividing the result by 2.

\subsubsection{Elements with respect to $\phi$}

\label{app:WY-phi}

Let us first recall from the statement of Theorem~\ref{thm:WY-phi} the expression for the $(i, j)$-th element of the Kubo--Mori information matrix $I_{ij}^{\operatorname{WY}}(\phi)$:
\begin{align}
    I_{ij}^{\operatorname{WY}}(\phi)& = - 8 \Tr\!\left[ \Psi_{\phi}(H_{j}) \sqrt{\rho(\theta)} \Psi_{\phi}(H_{i}) \sqrt{\rho(\theta)} \right] + 4 \left\langle \left\{ \Psi_{\phi}(H_{i}) , \Psi_{\phi}(H_{j}) \right\}   \right\rangle_{\rho(\theta)}. \label{eq:WY-phi_app}
\end{align}

\subsubsubsection{Estimation of the first element}\label{app:1st-term-WY-phi}
We begin by outlining the estimation of the first term in~\eqref{eq:WY-phi_app}. This follows the procedure detailed in~\ref{sec:1st-term-WY-theta}, where the motivations and relevant references for this method are provided. We assume that one  has access to the canonical purification $|\psi(\theta)\rangle$ of a quantum Boltzmann machine, defined in~\eqref{eq:canon-pure-QBM}, and also known as a thermofield double state.
Since many quantum algorithms for thermal-state preparation actually prepare this canonical purification, this assumption is just as reasonable as our assumption of having sample access to the thermal state $\rho(\theta)$. 
Under this assumption, the following identity implies that one can estimate the first term of~\eqref{eq:WY-phi_app} efficiently:
\begin{equation}
\Tr\!\left[ \Psi_{\phi}(H_{j}) \sqrt{\rho(\theta)} \Psi_{\phi}(H_{i}) \sqrt{\rho(\theta)} \right] =
\langle\psi(\theta)|\left(  \Psi_{\phi}(H_{j}) \otimes\left[\Psi_{\phi}(H_{i}) \right]  ^{T}\right)  |\psi(\theta)\rangle.
\label{eq:identity-WY-phi-canon-pur}
\end{equation}
The identity in~\eqref{eq:identity-WY-phi-canon-pur} follows because
\begin{align}  \langle\psi(\theta)|\left(  \Psi_{\phi}(H_{j}) \otimes\left[ \Psi_{\phi}(H_{i}) \right]  ^{T}\right)  |\psi(\theta)\rangle 
& =\langle\Gamma|\left(  \sqrt{\rho(\theta)} \Psi_{\phi}(H_{j}) \sqrt{\rho(\theta)}\otimes\left[  \Psi_{\phi}(H_{i}) \right]  ^{T}\right)  |\Gamma\rangle\\
&  =\langle\Gamma|\left(  \sqrt{\rho(\theta)}\Psi_{\phi}(H_{j}) \sqrt{\rho(\theta)}\Psi_{\phi}(H_{i}) \otimes I\right)
|\Gamma\rangle\\
&  =\operatorname{Tr}\!\left[  \sqrt{\rho(\theta)}\Psi_{\phi}(H_{j})\sqrt{\rho(\theta)}\Psi_{\phi}(H_{i})\right]  .
\end{align}
The second equality follows from the transpose trick \cite[Exercise~3.7.12]{Wbook17}.
Thus, in order to estimate the right-hand side of~\eqref{eq:identity-WY-phi-canon-pur}, we need to be able to measure the expectation of the operator $\left[  \Psi_{\phi}(H_{i}) \right]^{T}$. Consider that
\begin{align}
& \left[  \Psi_{\phi}(H_{i}) \right]^{T} \notag \\ & =\left[  \int_0^1  dt\  e^{iH(\phi)t}H_{i}e^{-iH(\phi)t}\right]
^{T}\\
& = \int_0^1  dt\ \left[  e^{iH(\phi)t}H_{i}e^{-iH(\phi)t} \right]  ^{T}\\
& =\int_0^1 dt\ \left(  e^{iH(\phi)t}\right)^{T}\!H_{i}
^{T}\!\left(  e^{-iH(\phi)t}\right)  ^{T}\\
& =\int_0^1 dt\ e^{-iH^{T}(\phi)t}H_{i}^{T}e^{iH^{T}
(\phi)t} \, .
\end{align}
If each $H_{i}$ is a Pauli string, this is easy to implement by noting that
$I^{T}=I$, $\sigma_{X}^{T}=\sigma_{X}$, $\sigma_{Y}^{T}=-\sigma_{Y}$, and
$\sigma_{Z}^{T}=\sigma_{Z}$.
Then, adopting the shorthand $\psi(\theta)\equiv|\psi(\theta)\rangle\!
\langle\psi(\theta)|$ and applying the definition of $\Psi^\dagger_{\phi}$ in~\eqref{eq:Psi-adjoint}
and cyclicity and linearity of trace, consider that
\begin{align}
\langle\psi(\theta)|\left(  \Psi_{\phi}(H_{j}) \otimes\left[ \Psi_{\phi}(H_{i}) \right]  ^{T}\right)  |\psi(\theta)\rangle
& =\operatorname{Tr}\!\left[  \left(   \Psi_{\phi}(H_{j}) \otimes\left[   \Psi_{\phi}(H_{i})\right]  ^{T}\right) \psi(\theta)\right]  \\
& =\mathbb{E}_{\tau_{1},\tau_{2}}\!\left[  \operatorname{Tr}\!\left[  \left(
H_{j}\otimes H_{i}^{T}\right)  \mathcal{U}_{\tau_{1},\tau_{2}}(\psi
(\theta))\right]  \right]  ,
\end{align}
where $\tau_{1}$ and $\tau_{2}$ are independent random variables each chosen uniformly from the interval $[0,1]$ and $\mathcal{U}_{\tau_{1},\tau_{2}}$ is the following unitary channel:
\begin{equation}
\mathcal{U}_{\tau_{1},\tau_{2}}(Y)\coloneqq \left(  e^{-iH\left(  \phi\right)  \tau_{1}}\otimes e^{iH^{T}\left(  \phi\right)\tau_{2}}\right)  Y \left(  e^{iH\left(  \phi\right)  \tau_{1}
}\otimes e^{-iH^{T}\left( \phi \right)  \tau_{2}}\right)  .
\end{equation}

Thus, a quantum algorithm for estimating the first term of~\eqref{eq:WY-phi_app} consists of
repeating the following steps and averaging: prepare the canonical purification $\psi(\theta)$ in~\eqref{eq:canon-pure-QBM}, pick $\tau_{1}$ and $\tau_{2}$ uniformly at
random from the interval $[0,1]$, apply the Hamiltonian evolution $\mathcal{U}_{\tau_{1},\tau_{2}}$ to $\psi(\theta)$, and measure the observable
$H_{j}\otimes H_{i}^{T}$. The respective quantum circuit is shown in Figure~\ref{fig:WY-phi-1}.

\subsubsubsection{Estimation of the second element}\label{app:WY-phi-2}
We now present an algorithm for estimating the first term of~\eqref{eq:WY-phi_app} in greater detail. Consider the following:
\begin{align}
    4\left\langle \left\{  \Psi_{\phi}(H_{i}),\Psi_{\phi}(H_{j}) \right\}  \right\rangle _{\rho(\theta)}
    & = 4\Tr\!\left[ \left\{  \Psi_{\phi}(H_{i}),\Psi_{\phi}(H_{j}) \right\}  \rho(\theta) \right]\\
    & = \int_0^1 \int_0^1 dt_1\ dt_2\ \left( 4\operatorname{Tr}\!\left[ \left\{
     e^{iH(\phi)t_1}  H_i e^{-iH(\phi)t_1} , e^{iH(\phi)t_2} H_j e^{-iH(\phi)t_2} \right\} \rho(\theta) \right] \right).\label{eq:WY-phi_app-alt}
\end{align}
We are now in a position to present an algorithm to estimate the second term of~\eqref{eq:WY-phi_app} using its equivalent form shown in~\eqref{eq:WY-phi_app-alt}. The algorithm is similar to Algorithm~\ref{algo:FB-theta}, so here we provide a high-level description of how it works.  At its core, the algorithm relies on a quantum circuit that estimates the expected value of the anticommutator of two operators (see Appendix~\ref{app:Hadamard_test}). In particular, if the control register in Figure~\ref{fig:qc-primitive-anticomm} is initialized in the state $\ket{0}$ instead of $\ket{1}$, the output of the circuit is $\frac{1}{2} \left\langle \left\{H,U\right\} \right\rangle_{\rho}$, where $H$ is Hermitian and $U$ is Hermitian and unitary. In this case, we choose $\rho = \rho(\theta)$, $H = e^{iH(\phi)t_1}  H_i e^{-iH(\phi)t_1}$, and $U = e^{iH(\phi)t_2} H_j e^{-iH(\phi)t_2}$. We then make some further simplifications where possible. Accordingly, the quantum circuit that plays a role in estimating the integrand of~\eqref{eq:WY-phi_app-alt} is depicted in Figure~\ref{fig:WY-phi-2}. The algorithm involves running this circuit $N$ times, where $N$ is determined by the desired precision and error probability. During each run, the times $t_1$ and $t_2$ for the Hamiltonian evolution are sampled independently and uniformly at random from the interval $[0,1]$. The final estimation of the second term of~\eqref{eq:WY-phi_app} is obtained by averaging the outputs of the $N$ runs and multiplying the result by 8.

\subsubsection{Elements with respect to $\theta$ and $\phi$} \label{app:WY-theta-phi}
Let us first recall from the statement of Theorem~\ref{thm:WY-theta-phi} the expression for the $(i, j)$-th element of the Kubo--Mori information matrix $I_{ij}^{\operatorname{WY}}(\theta,\phi)$:
\begin{align}
    I_{ij}^{\operatorname{WY}}(\theta, \phi)=i\left\langle \left[  \Phi_{\frac{\theta}{2}}(G_{j}),\Psi_{\phi}(H_{i})\right]  \right\rangle _{\rho(\theta)}. \label{eq:WY-theta-phi_app}
\end{align}
Here, we show how to estimate~\eqref{eq:WY-theta-phi_app}. Consider the following:
\begin{align}
    & i\left\langle \left[  \Phi_{\frac{\theta}{2}}(G_{j}),\Psi_{\phi}(H_{i})\right]  \right\rangle _{\rho(\theta)}\nonumber\\
    & = i \Tr\!\left[ \left[  \Phi_{\frac{\theta}{2}}(G_{j}),\Psi_{\phi}(H_{i})\right] \rho(\theta) \right]\\
    & = \int_\mathbb{R} \int_0^1 dt_1\ dt_2\ p(t_1)\ \left( i  \Tr\!\left[ \left[ e^{-iG(\theta/2)t_1}G_{j} e^{iG(\theta/2)t_1}, e^{iH(\phi)t_2} H_{i} e^{-iH(\phi)t_2}\right] \rho(\theta) \right] \right).\label{eq:WY-theta-phi_app-alt}
\end{align}

We are now in a position to present an algorithm to estimate~\eqref{eq:WY-theta-phi_app} using its equivalent form shown in~\eqref{eq:WY-theta-phi_app-alt}. The algorithm is similar to Algorithm~\ref{algo:FB-theta-phi}, so here we provide a high-level description of how it works.  At its core, the algorithm relies on a quantum circuit that estimates the expected value of the commutator of two operators (see Appendix~\ref{app:Hadamard_test}). In particular, if the control register in Figure~\ref{fig:qc-primitive-comm} is initialized in the state $\ket{0}$ instead of $\ket{1}$, the output of the circuit is $\frac{1}{2} \left\langle \left[H,U\right] \right\rangle_{\rho}$, where $H$ is Hermitian and $U$ is Hermitian and unitary. In this case, we choose $\rho = \rho(\theta)$, $H = e^{-iG(\theta/2)t_1}G_{j} e^{iG(\theta/2)t_1}$, and $U = e^{iH(\phi)t_2} H_{i} e^{-iH(\phi)t_2}$. We then make some further simplifications where possible. Accordingly, the quantum circuit that estimates the integrand of~\eqref{eq:WY-theta-phi_app-alt} is depicted in Figure~\ref{fig:WY-theta-phi}. The algorithm involves running this circuit $N$ times, where $N$ is determined by the desired precision and error probability. During each run, the times $t_1$ and $t_2$ for the Hamiltonian evolution are sampled independently at random, $t_1$ with probability $p(t)$ (defined in~\eqref{eq:high-peak-tent-density}) and $t_2$ from the interval $[0,1]$. The final estimation of the second term of~\eqref{eq:WY-theta-phi_app} is obtained by averaging the outputs of the $N$ runs and multiplying the result by 2.

\subsection{Quantum algorithms for estimating Kubo--Mori information matrix elements}

\subsubsection{Elements with respect to $\theta$}

\label{app:KM-theta}

Let us first recall from the statement of Theorem~\ref{thm:KM-theta} the expression for the $(i, j)$-th element of the Kubo--Mori information matrix $I_{ij}^{\operatorname{KM}}(\theta)$:
\begin{equation}
    I_{ij}^{\operatorname{KM}}(\theta)=\frac{1}{2}\left\langle \left\{  G_{i},\Phi_{\theta}(G_{j})\right\}  \right\rangle_{\rho(\theta)}-\left\langle G_{i}\right\rangle _{\rho(\theta)}\left\langle G_{j}\right\rangle _{\rho(\theta)}.\label{eq:KM-theta_app}
\end{equation}
Here, we show how to estimate the first term of~\eqref{eq:KM-theta_app}. Consider the following:
\begin{align}
     \frac{1}{2}\left\langle \left\{  G_{i},\Phi_{\theta}(G_{j})\right\}  \right\rangle_{\rho(\theta)}
    & = \frac{1}{2}\operatorname{Tr}[\left\{
    G_{i},\Phi_{\theta}(G_{j})\right\}  \rho(\theta)]\\
    & = \int_{\mathbb{R}} dt\ p(t)\ \left( \frac{1}{2} \operatorname{Tr}\!\left[ \left\{G_i , e^{-iG(\theta)t} G_j e^{iG(\theta)t} \right\} \rho(\theta)\right] \right).\label{eq:KM-theta_app-alt}
\end{align}

We are now in a position to present an algorithm to estimate the first term of~\eqref{eq:KM-theta_app} using its equivalent form shown in~\eqref{eq:KM-theta_app-alt}. The algorithm is similar to Algorithm~\ref{algo:FB-theta}, so here we provide a high-level description of how it works.  At its core, the algorithm relies on a quantum circuit that estimates the expected value of the anticommutator of two operators (see Appendix~\ref{app:Hadamard_test}). In particular, one can notice that when initializing the control register in the state $\ket{0}$ instead of $\ket{1}$, the quantum circuit shown in Figure~\ref{fig:qc-primitive-anticomm} allows to estimate the quantity $\frac{1}{2} \left\langle \left\{ H, U\right\} \right\rangle_{\rho}$, where $H$ is Hermitian and the $U$ is Hermitian and unitary. In this case, we choose $\rho = \rho(\theta)$, $H = G_i $, and $U = e^{-iG(\theta)t} G_j e^{iG(\theta)t}$. We then make some further simplifications that follow because $\rho(\theta)$ commutes with $e^{-iG(\theta)t}$. Accordingly, the quantum circuit that plays a role in estimating the integrand of~\eqref{eq:KM-theta_app-alt} is depicted in Figure~\ref{fig:KM-theta}. The algorithm involves running this circuit $N$ times, where $N$ is determined by the desired precision and error probability. During each run, the time $t$ for the Hamiltonian evolution is sampled at random with probability $p(t)$ (defined in~\eqref{eq:high-peak-tent-density}). The final estimation of the first term of~\eqref{eq:KM-theta_app} is obtained by averaging the outputs of the $N$ runs.

\subsubsection{Elements with respect to $\phi$}

\label{app:KM-phi}

Let us first recall from the statement of Theorem~\ref{thm:KM-phi} the expression for the $(i, j)$-th element of the Kubo--Mori information matrix $I_{ij}^{\operatorname{KM}}(\phi)$:
\begin{equation}
    I_{ij}^{\operatorname{KM}}(\phi)=\left\langle \left[  \Psi_{\phi}(H_{i}), \left[  G(\theta) , \Psi_{\phi}(H_{j})\right]  \right] \right\rangle _{\rho(\theta)}.\label{eq:KM-phi_app}
\end{equation}
Consider the following:
\begin{align}
    & \left\langle \left[  \Psi_{\phi}(H_{i}), \left[  G(\theta) , \Psi_{\phi} (H_{j})\right]  \right] \right\rangle _{\rho(\theta)}\\
    & = \left\langle \left[\left[  \Psi_{\phi}(H_{j}),G(\theta)\right]  ,\Psi_{\phi}(H_{i})\right]  \right\rangle _{\rho(\theta)}\nonumber\\
    & = \Tr\! \bigg[ \!\left[ \left[\Psi_{\phi}(H_{j}) ,G(\theta) \right], \Psi_{\phi}(H_{i}) \right] \rho(\theta) \bigg]\\
    & = \int_0^1\!\int_0^1 dt_1\ dt_2\ \Bigg(\!  \Tr\!\bigg[ \left[\left[ e^{iH(\phi)t_1} H_j e^{-iH(\phi)t_1},G(\theta)\right] , e^{iH(\phi)t_2} H_i e^{-iH(\phi)t_2} \right] \rho(\theta)\bigg] \Bigg).\label{eq:KM-phi_est}
\end{align}

We are now in a position to present an algorithm to estimate the first term of~\eqref{eq:KM-phi_app} using its equivalent form shown in~\eqref{eq:KM-phi_est}. The algorithm is similar to Algorithm~\ref{algo:FB-phi}, so here we provide a high-level description of how it works.  At its core, the algorithm relies on a quantum circuit that estimates the expected value of nested commutators of three operators, $\frac{1}{4} \left\langle \big[ \left[U_1 , H \right], U_0\big]\right\rangle_\rho$, where $H$ is Hermitian, and $U_0$ and $U_1$ are both Hermitian and unitary (refer to Appendix~\ref{app:Hadamard_test} and Figure~\ref{fig:prim-nest-comm}). In this case, we choose $\rho = \rho(\theta)$, $U_1 = e^{iH(\phi)t_1} H_j e^{-iH(\phi)t_1}$, $H = G(\theta)$, and $U_0 = e^{iH(\phi)t_2} H_i e^{-iH(\phi)t_2}$. We then make some further simplifications where possible. Specifically, the quantum circuit that plays a role in estimating the integrand of~\eqref{eq:KM-phi_est} is depicted in Figure~\ref{fig:KM-phi}. The algorithm involves running this circuit $N$ times, where $N$ is determined by the desired precision and error probability. During each run, the times $t_1$ and $t_2$ for the Hamiltonian evolutions are sampled independently and uniformly at random from the interval $[0,1]$. The final estimation of the first term of~\eqref{eq:KM-theta_app} is obtained by averaging the outputs of the $N$ runs and multiplying the result by $4 \left \| \theta \right\|_1$. For measuring $G(\theta) $, we again adopt a sampling approach (see Remark~\ref{rem:measuring-G-theta}).

\subsubsection{Elements with respect to $\theta$ and $\phi$}

\label{app:KM-theta-phi}

Let us first recall from the statement of Theorem~\ref{thm:KM-theta-phi} the expression for the $(i, j)$-th element of the Fisher–Bures information matrix $I_{ij}^{\operatorname{KM}}(\theta,\phi)$:
\begin{equation}
    I_{ij}^{\operatorname{KM}}(\theta,\phi)=\frac{i}{2}\left\langle \left\{  \Phi_{\theta}(G_{i}),\left[G(\theta),  \Psi_{\phi}(H_{j})\right]  \right\}  \right\rangle _{\rho(\theta)}
.\label{eq:KM-theta-phi_app}
\end{equation}
Consider the following:
\begin{align}
    & \frac{i}{2}\left\langle \left\{  \Phi_{\theta}(G_{i}),\left[G(\theta),   \Psi_{\phi}(H_{j})\right]  \right\}  \right\rangle _{\rho(\theta)}\notag\\
    & = \frac{i}{2} \Tr\!\bigg[ \left\{  \Phi_{\theta}(G_{i}) , \left[G(\theta) \Psi_{\phi}(H_j) \right] \right\} \rho(\theta) \bigg]\\
    & = \frac{i}{2} \int_\mathbb{R} \int_0^1 dt_1\ dt_2\ p(t_1)\ \Bigg( \Tr\!\Big[ \left\{ e^{-iG(\theta)t_1} G_{i} e^{iG(\theta)t_1}, \left[  G(\theta) , e^{iH(\phi)t_2} H_j e^{-iH(\phi)t_2} \right]\right\} \rho(\theta) \Big]
    \Bigg).\label{eq:KM-theta-phi_est}
\end{align}

We are now in a position to present an algorithm to estimate~\eqref{eq:KM-theta-phi_app} using its equivalent form shown in~\eqref{eq:KM-theta-phi_est}. The algorithm is similar to Algorithm~\ref{algo:FB-phi}, so here we provide a high-level description of how it works.  At its core, the algorithm relies on a quantum circuit that estimates the expected value of nested anticommutator and commutator of three operators, $\frac{1}{4} \left\langle \big\{ U_0,\left[H, U_1 \right]\big\}\right\rangle_\rho$, where $H$ is Hermitian, and $U_0$ and $U_1$ are both Hermitian and unitary (refer to Appendix~\ref{app:Hadamard_test} and Figure~\ref{fig:prim-nest-anticomm-comm}). In this case, we choose $\rho = \rho(\theta)$, $U_0 = e^{-iG(\theta)t_1} G_{i} e^{iG(\theta)t_1}$, $H = G(\theta)$, and $U_1 =e^{iH(\phi)t_2} H_j e^{-iH(\phi)t_2}$. We then make some further simplifications where possible. Specifically, the quantum circuit that plays a role in estimating the integrand of~\eqref{eq:KM-theta-phi_est} is depicted in Figure~\ref{fig:KM-theta-phi}. The algorithm involves running this circuit $N$ times, where $N$ is determined by the desired precision and error probability. During each run, the times $t_1$ and $t_2$ for the Hamiltonian evolutions are sampled independently and at random, $t_1$ with probability $p(t)$ (defined in~\eqref{eq:high-peak-tent-density}) and $t_2$ from the interval $[0,1]$. The final estimation of~\eqref{eq:KM-theta-phi_app} is obtained by averaging the outputs of the $N$ runs and multiplying the result by $2 \left \| \theta \right\|_1$. For measuring $G(\theta) $, we again adopt a sampling approach (see Remark~\ref{rem:measuring-G-theta}).

\end{document}